\documentclass[aps,preprint,onecolumn,eqsecnum,nofootinbib,superscriptaddress,10pt]{revtex4-1}
\usepackage[colorlinks=true,linktoc=page,citecolor=red,linkcolor=blue]{hyperref}	
\usepackage{graphicx}
\graphicspath{{./figs/}}  
\usepackage[usenames]{color}
\usepackage{hyperref}
\usepackage{comment}
\usepackage{amsmath,amssymb,amsfonts}	
\usepackage{bm,braket,array}
\usepackage{multirow}
\usepackage{amsthm}
\usepackage{amscd}
\usepackage{slashed}

\newtheorem{thm}{Theorem}[section]

\newtheorem{lem}[thm]{Lemma}

\theoremstyle{definition}
\newtheorem{dfn}[thm]{Definition}

\theoremstyle{remark}

\def\mod{{\rm\ mod\ }}
\def\Tr{{\rm Tr }}
\def\tr{{\rm tr\,}}

\def\p{\partial}

\def\wt{\widetilde}

\def\Hom{{\rm Hom}}

\def\Tor{{\rm Tor}}

\def\spin{{\rm Spin}}
\def\pin{{\rm Pin}}

\def\ua{\uparrow}
\def\da{\downarrow}

\newcommand{\C}{\mathbb{C}}

\newcommand{\R}{\mathbb{R}}
\newcommand{\Z}{\mathbb{Z}}
\newcommand{\T}{\mathbb{T}}

\newcommand{\cU}{\mathcal{U}}

\newcommand{\sT}{\sf T}

\def\widebar{\accentset{{\cc@style\underline{\mskip10mu}}}} 
\def\wideubar{\underaccent{{\cc@style\underline{\mskip10mu}}}} 

\begin{document}
\title{Many-body topological invariants
  for fermionic short-range entangled topological phases
  protected by antiunitary symmetries}
\author{Ken Shiozaki}
\thanks{The first two authors contributed equally to this work.}
\affiliation{Condensed Matter Theory Laboratory, RIKEN, Wako, Saitama, 351-0198, Japan}
\author{Hassan Shapourian}
\thanks{The first two authors contributed equally to this work.}
\affiliation{James Franck Institute and Kadanoff Center for Theoretical Physics,
University of Chicago, Illinois 60637, USA}
\author{Kiyonori Gomi}
\affiliation{Department of Mathematical Sciences, Shinshu University, Nagano, 390-8621, Japan}
\author{Shinsei Ryu}
\affiliation{James Franck Institute and Kadanoff Center for Theoretical Physics,
University of Chicago, Illinois 60637, USA}

\date{\today}

\begin{abstract}
  We present a fully many-body formulation of 
  topological invariants for 
  various topological phases of fermions
  protected by antiunitary symmetry,
  which does not refer to single particle wave functions. 
  For example, we construct the many-body $\mathbb{Z}_2$
  topological invariant for time-reversal symmetric topological
  insulators in two spatial dimensions, which is a many-body
  counterpart of the Kane-Mele $\mathbb{Z}_2$ invariant
  written in terms of single-particle Bloch wave functions. 
  We show that an important ingredient for the construction of the many-body
  topological invariants is a {\it fermionic} partial transpose which 
  is basically the standard partial transpose equipped with a sign structure 
   to account for anti-commuting property of fermion operators.   
  We also report some basic results on
  various kinds of pin structures -- a key concept behind 
  our strategy for constructing many-body topological invariants
  --
  such as
  the obstructions, isomorphism classes, and
  Dirac quantization conditions.

\end{abstract}

\maketitle

\tableofcontents

\section{Introduction}
\label{intro}

\subsection{Symmetry-protected topological phases and topological invariants}

Quantum phases of matter are characterized by their correlation functions, 
$
\Tr[ O_{R_1} O_{R_2} \cdots \rho], 
$
where $O_{R_i}$ is an operator which has a support in a space region $R_i$ and
$\rho$ is the density matrix of the system. 
Conventional spontaneous symmetry breaking phases can be detected
and characterized by a set of local order parameters,
i.e., $R_i$ refers to a region smaller than the correlation length.  
On the other hand, one can consider non-local operators $O_R$ where 
$R$ represents an extended region in space. 
For example, Wilson loop or surface operators in gauge theories
are examples of non-local order parameters, which can detect
various phases (confinement, deconfinement, etc.) in gauge theories.

In condensed matter context,
recent studies in symmetry-protected topological (SPT) phases
have shown the importance of new classes of non-local ``order parameters''
which characterize SPT phases.
(Some of such non-local order parameters are given  
as an expectation value of a non-local operator $O_R$,
while others are more complicated and involve multiple (reduced) density
matrices. See below for more details.)
Here, we recall that SPT phases are gapped and short-range entangled
phases protected by symmetry.
In other words,
ground states of SPT phases can be deformed into a trivial tensor product
state by a local unitary transformation which breaks the symmetry. 
Topological insulators and superconductors are celebrated examples of SPT phases of fermions.~\cite{Schnyder2008, Kitaev2009, Ryu2010, Hasan2010, Qi2011, Fidkowski2011, Gu2012, Freed2013, Wang2014, Kapustin2015a, HsiehChoRyu2016, witten2015fermion, Freed2016}
Examples of bosonic SPT phases have been also widely discussed.~\cite{Pollmann2010, Chen2011b, Schuch2011, Chen2013, Lu2012, Levin2012, SuleXiaoRyu2013, vishwanath2013physics, Kapustin2014symmetry, KapustinThorngren2014anomalies, Freed2016, ThorngrenElse2016gauging} 

The purpose of this paper is to construct non-local order parameters
which take quantized values for a given SPT phase,
i.e., they are topological invariants.
In particular, we are mainly interested in fermionic SPT phases
protected by time-reversal or antiunitary symmetry of various kinds.
We also emphasize that the topological invariants we construct
are completely formulated within the {\it many-body} Hilbert space,
and should be contrasted with the existing single-particle
topological invariants. 
I.e., our topological invariants can be defined and evaluated
even in the presence of interactions.
For example, we will give a many-body topological invariant
characterizing
time-reversal symmetric topological insulators in two spatial dimensions.
This should be contrasted with
the single-particle topological invariant
which has been 
 discussed previously
\cite{KaneMeleZ2, MooreBalents, RoyZ2, FuKanePump} 
by using single-particle Bloch wave functions. 

\subsection{Topological terms and many-body topological invariants}

Our construction of many-body topological invariants can be best illustrated 
by the following Euclidean path integral. 
For a given gapped quantum system, we consider the path integral
\begin{align}
 Z(X,\eta,A)=
 \int \prod_i \mathcal{D}\phi_i
 \exp[ - S(X,\eta,A,\phi_i)].
  \label{path integral}
\end{align}
Here, the path integral $\int \prod_i \mathcal{D}\phi_i$ is over all degrees of
freedom of the system for a given ``background'' or ``input'' denoted by $(X,\eta,A)$.
Among the input data, $X$ is the Euclidean spacetime.
$A$ specifies a proper external/probe background;
For SPT phases protected by a unitary on-site symmetry $G$, one can introduce
the background $G$ gauge field $A$
which couples to the ``matter field'' $\phi_i$ (e.g., electrons)
in order to detect SPT phases. 
On the other hand, SPT phases protected by orientation reversing symmetries, such as time-reversal or spatial reflection, 
can be detected by considering their ``coupling'' to unorientable spacetime.
\cite{Kapustin2014symmetry, Kapustin2014bosonic,
Kapustin2015a, 2014PhRvB..90p5134H,2015PhRvB..91s5142C, Hsieh2015}
Finally, to consider and define given matter fields $\phi_i$ on $X$,
we may need to provide additional data, which we denote by $\eta$.
In particular, to define relativistic fermion fields on $X$,
we need to specify a spin or pin structures (and proper generalization thereof). 
However, since our main concern is fermionic SPT phases in condensed matter
physics, it is not entirely obvious if we should consider spin, pin, etc.
structures.
Even for non-relativistic fermions, we may need to specify, for example,
suitable boundary conditions for matter fields. 
We will shortly discuss more about the spin (pin, and etc.) structures
both in relativistic and non-relativistic contexts.
For the time being, the reader may assume relativistic fermions for simplicity.

For gapped topological phases, we expect that $Z(X,\eta,A)$ 
includes a topological term, which is independent of local data (e.g., metric), 
\begin{align}
Z(X, \eta,A)\sim \exp [ i S_{top}(X,\eta,A) +\cdots ]. 
  \label{top term}
\end{align}
Here, observe that the topological part appears, in the Euclidean signature, 
as a phase of the partition function $Z(X,\eta,A)$.
\footnote{
In the Euclidean signature, the real (and positive) part of the effective action
$-\ln Z(X,\eta,A)$
is the Boltzmann weight, and hence related to energetics. 
As energetics is usually local, global and topological properties of the field configurations
should not enter into the real part of the effective action. 
}
The purely topological part of the partition function defines a topological quantum field theory (TQFT).   
For our purpose of distinguishing different gapped (topological) phases, 
the topological phase part $e^{i S_{top}}$ can be used as a (many-body) topological invariant or non-local order parameter of SPT phases.

Let us mention a few examples 
of the topological phase factor \eqref{top term} of the Euclidean path integral.

--
In integer quantum Hall phases,
the path integral of electron degrees of freedom in
the presence of a background $U(1)$ gauge field $A$ gives rise to
the Chern-Simons term,
$S_{top}(X,A) = \frac{k}{4\pi} \int_X A dA$ where
$k=\mbox{integer}$ is the quantized Hall conductance in unit of $e^2/h$. 

--
Another simple example is provided by $(1+1)d$ topological
superconductors (e.g., the non-trivial phase of the Kitaev chain
\cite{kitaev2001unpaired});  
Putting the Kitaev chain on the spatial circle, the imaginary time path-integral over gapped BdG fermionic quasiparticles gives rise to
$Z(T^2,\eta)$ where $X=T^2$ is the spacetime torus, and $\eta$ specifies four possible boundary conditions 
(i.e., periodic/antiperiodic boundary conditions in space and temporal directions) for BdG fermions. 
If BdG quasiparticles are considered as relativistic fermions, to give a specific boundary condition $\eta$ is 
equivalent to give a spin structure of the spacetime torus.  
The topological term in this case is given by the Arf invariant of a spin structure $\eta$.  
\cite{shiozaki2016many}

--
For $(1+1)d$ bosonic SPT phases protected by TRS
such as the Haldane spin chain,
the phase of the path integral on the real projective plane $\mathbb{R}P^2$
is either $0$ or $\pi$, and
hence serves as a $\mathbb{Z}_2$-valued topological invariant.
\cite{Pollmann2012, ShiozakiRyu2016} 

--
  In Ref.~\cite{ShapourianShiozakiRyu2016detection},
  the discretized imaginary time path integral,
  in the presence of a cross-cap,
  was evaluated for $(1+1)d$ topological superconductors
  in symmetry class BDI. 
  The spacetime is effectively $\mathbb{R}P^2$. 
  It was shown that
  the phase of the partition function  
  yields correctly the $\Z_8$ SPT topological invariant
  for $(1+1)d$ class BDI topological superconductors. \cite{Fidkowski-Kitaev, Kapustin2015a}



\subsection{Classification by cobordism theory and generating manifolds}

In the above discussion,
it is of crucial importance to choose proper input data or ``background'', $(X,\eta,A)$.
Namely, $(X,\eta,A)$ must be a suitable manifold such that
the topological term $S_{top}(X,\eta,A)$ of the partition function,
when evaluated for $(X,\eta, A)$,
is non-zero.
In addition, it is desirable to find backgrounds $(X,\eta,A)$
for which
the topological phase factor $e^{i S_{top}(X,\eta,A)}$
takes the ``smallest possible'' or ``most fundamental'' value;
for example, for SPT phases for which
we have $\mathbb{Z}_N$ classification ($N$ is an integer),
we naturally expect the topological invariant
$e^{i S_{top}(X,\eta,A)}$ to take $N$ possible values
(i.e., $e^{i S_{top}(X,\eta,A)}$ 
should be an $N$-th root of unity).
Were the input $(X,\eta, A)$ not chosen properly,
while $e^{i S_{top}(X,\eta,A)}$ could carry a non-zero phase,
it would take only a subset of possible values and hence
would not distinguish all possible phases in
the $\mathbb{Z}_N$ classification.
(See below for more details.)

The cobordism classification of
invertible topological quantum field theories (TQFTs)
provides a hint for the proper choice of $(X,\eta,A)$. 
\cite{Kapustin2014, Freed2016} 
For SPT phases for which the ground state is unique on any space manifold, 
one can classify possible behaviors $S_{top}(X,\eta, A)$, 
i.e., classify TQFTs of special kind, so-called invertible TQFTs.
The classification of a sort of TQFTs with a background field 
falls into the cobordism classification of manifolds.~\cite{Kapustin2014, Freed2016} 
In turns, this provides classification of SPT phases.

The cobordism theory~\cite{stong2015notes, Gilkey} provides
the classification of manifolds
in the presence of a background field $A$, which is introduced by gauging the onsite symmetry $G$, under the equivalence relation known as the cobordant. 
Two $d$-dimensional manifolds $(X_i,\eta_i,A_i)$ ($i=1,2$) with the structure $\eta_i$ and background field $A_i$ are said to be cobordant iff there exists a $(d+1)$-dimensional manifold $(\tilde X,\tilde \eta, \tilde A)$ with the structure $\tilde \eta$ and the background field $\tilde A$ on $\tilde X$ so that its boundary $\p \tilde X$ agrees with $X_1 \sqcup (-X_2)$ as well as the structure $\tilde \eta|_{\p \tilde X}$ and the background field $\tilde A|_{\p \tilde X}$. 
The abelian group structure is introduced by the disconnected sum
$[X_1,\eta_1,A_1] + [X_2,\eta_2,A_2] = [(X_1,\eta_1,A_1) \sqcup
(X_2,\eta_2,A_2)]$, which results in the cobordism group $\Omega^{\rm
  str}_d(BG)$,
the equivariant cobordism
group over the classifying
space of $G$
for a given type of structures. 

The partition function $Z$ of an invertible
TQFT can be considered as a homomorphism
$
\Omega^{\rm str}_d(BG) \to U(1)
$,
$
(X,\eta,A) \mapsto Z(X,\eta,A)
$.
\footnote{In typical situations when we want to compute
  SPT topological invariants,
  there are gapped excitations
  that contribute to the amplitude of the partition function.
  The cobordism invariant appears in
  the complex $U(1)$ phase $Z(X,\eta,A)/|Z(X,\eta,A)|$ of the partition function
  $Z(X,\eta,A)$.
}
Cobordism invariant partition functions
are thus labeled by elements in $\Hom (\Omega^{\rm str}_d(BG),U(1))$. 
The free part of the coboridims group leads to the theta term $e^{i \theta n}$, $n \in \Z$, parametrized by $\theta \in \R/2 \pi \Z$. 
On the other hand, the torsion part $\Tor\, \Omega^{\rm str}_d(BG)$ classifies the SPT phase.

For our purpose to construct the order parameter of SPT phases,
generating manifolds, 
that are generators of the cobordism groups, are of particular importance.
The partition function, when evaluated for
the generator of the cobordism group,
gives rise to the most fundamental phase factor.
Hence, 
the partition function on the generating manifolds serves as 
the order parameter of SPT phases.

\subsection{Orientation reversing symmetries, 
  pin structures and their variants}

In the cobordism theory, $\eta$ represents spin or 
pin structures (and proper generalizations thereof).
(Invertible) TQFTs which we expect to give an effective description of fermionic SPT phases
are relativistic and depend on spin, pin, etc. structures
(i.e., we are considering spin TQFTs, etc.).
The precise definitions of (the variants of) pin structures
will be given in Sec.\ \ref{Dirac quantization conditions}.

It should be noted that,
in defining these structures in relativistic contexts, 
fermion fields are assumed to be transformed as
a spinor under $SO(d+1)$ or $O(d+1)$ (in the Euclidean signature).
As we have warned before, 
in contrast, in condensed matter physics, there is no \emph{a priori} Poincar\'{e}
(relativistic) invariance,
and hence fermions in condensed matter systems 
are not always sensitive to spin structures.
Therefore, it is not obvious if the classification of SPT phases
by the cobordism group works for condensed matter systems.  
Nevertheless,
as we discuss in Sec.\ \ref{Dirac quantization conditions},
there are analogues of ``structures'' even
in non-relativistic contexts.
They are symmetry twists,
i.e., twisting boundary conditions by using the symmetry of the problem. 
For example,
for fermionic SPT phases (without any symmetry other than
fermion number parity conservation),
it is known that twisting boundary conditions by
fermion number parity symmetry
is a useful diagnostic tool. 
Such twisting in the path integral picture
gives rise to periodic or antiperiodic boundary conditions
of fermions fields. Non-trivial SPT phases may ``respond'' 
to such twist in a non-trivial way, and 
may be characterized and defined by such response.
The twisting by fermion number parity
is precisely a non-relativistic analogue of spin structures,
and coincides with spin structures if we consider relativistic fermions.

Put differently,
while fermions in condensed matter systems are not
relativistic and not always sensitive to spin (and pin) structures,
when some sort of topological media (topological phases)
are realized, effective relativistic fermions can emerge, which do depend on
spin structures. In short, spin may emerge from
short-range entanglement in quantum ground states.
Now, the emergent spin is described by different relativistic structures (spin, pin, and their variants)
\cite{witten2015fermion,Freed2016}, corresponding to different definitions of time-reversal symmetry
in the Altland-Zirnbauer (AZ) symmetry classes.
This correspondence is summarized in Table \ref{tab:pin},
and will be discussed in detail in Sec.\ \ref{Dirac quantization conditions}.
In particular, we discuss the Dirac quantization conditions 
for variants of pin structures, i.e.,
pin$^c$, pin$^{\tilde c}_{+}$ and pin$^{\tilde c}_-$ structures on unoriented manifolds.

\subsection{Topological invariants and fermionic partial transpose}

Instead of the spacetime path integral,  
one can also adopt the canonical operator approach which will be the focus of
our paper. 
In the canonical formalism, 
the partition functions, and
hence the topological invariants can
be expressed
in terms of
(a set of) ground state wave functions
or reduced density matrices constructed thereof
and symmetry operators.
Summarizing, our guiding principle is  
\begin{itemize}
\item[($\star$)]
  {\it Simulating} the path-integral (the partition function)
  on the generating spacetime manifold of the cobordism group by use of the ground state wave function and symmetry transformations in question. 
\end{itemize}

For example, the many-body Chern number
\cite{NiuThoulessWu1985, AvronSeiler1985}
for the quantum Hall effect in 2-space dimensions 
is the prototype of the many-body topological invariant
written in terms of ground state wave functions. 
The electric polarization operator
(Resta's $z$)
\cite{KingSmithVanderbilt,Resta1998}
is also known as a many-body characterization 
of short-range entangled phases in the presence of the electromagnetic $U(1)$
symmetry.
More generically,
the characterization of SPT phases in terms of their ground state
wave functions has been discussed
both for bosons~\cite{Pollmann2012, HaegemanPerez-GarciaCiracSchuch2012, Wen2014, Hung2014, Zaletel2014, HuangWei2015, ShiozakiRyu2016} 
and fermions.~\cite{Wen2014, ShapourianShiozakiRyu2016detection, shiozaki2016many}
In our prior publication~\cite{ShapourianShiozakiRyu2016detection, shiozaki2016many}, the strategy described above
has been employed to construct many-body topological invariants 
for various fermionic SPT phases.

To detect non-trivial SPT phases protected by TRS,
it is necessary to consider the Euclidean path integral
\eqref{path integral}
on various unoriented spacetime.
It was noted previously that,
in the operator formalism,
unoriented spacetimes can be effectively realized
by using {\it partial transpose}.
For example,  
for $(1+1)d$ bosonic SPT phases protected by TRS
such as the Haldane chain,
the generating manifold of the relevant cobordism group
is the real projective plane $\mathbb{R}P^2$.
It was shown that 
the corresponding many-body topological invariant
is given by using the partial transpose of the density matrix.~\cite{Pollmann2012, ShiozakiRyu2016}

For fermionic systems, the notion of partial transpose
must be properly introduced, i.e.,
the definition of partial transpose for fermionic systems
does not simply follow from the
definition of partial transpose for bosonic systems --
roughly speaking,
because of the fermion sign,
the fermionic Fock space does not simply factorize locally,
and hence there is an extra complication in defining the notion
of partial transpose. 
This point has been recently noticed and discussed
in Refs.\
\cite{ShapourianShiozakiRyu2016detection,shapourian2016partial,Eisler2015}. 
In particular, in Ref.\ \cite{shapourian2016partial},
it was noted that if we simply apply the bosonic
partial transpose to fermionic systems, by using the Jordan-Wigner
transformation,
the entanglement negativity, an entanglement measure which is defined
by using partial transpose, 
cannot capture entanglement between Majorana fermions.
Furthermore, our findings in
Ref.\ \cite{ShapourianShiozakiRyu2016detection} suggest that 
a fermionic version of partial transpose, if properly defined,
can be used to detect fermionic SPT phases: 
In $(1+1)d$ fermionic SPT phases with TRS $T^2 = 1$ (symmetry class BDI),
the generating manifold is the real projective plane. 
It was shown that one can use
the fermionic partial transpose to effectively simulate
the real projective plane (with a proper pin$_-$ structure),
and construct the corresponding many-body topological invariant,
which captures correctly the known $\Z_8$ classification.
\cite{2010PhRvB..81m4509F,Fidkowski2011}

\subsection{Organization of the paper}

The purpose of this paper is to construct
many-body topological invariants that can detect
fermionic SPT phases protected by time-reversal or other
antiunitary symmetries, following the strategy outlined above.  
See Table~\ref{tab:summary} and \ref{tab:summary 2}
for the many-body topological invariants studied in the present paper. 

There are two main technical steps for our goal,
discussed in Sec.\ \ref{Dirac quantization conditions}
and in Sec.\ \ref{Fermionic partial transpose and partial antiunitary
  transformations}, respectively:
\begin{itemize}
  \item
First,
in Sec.\ \ref{Dirac quantization conditions},
we give detailed descriptions of the variants of pin structures.
We discuss their origin in systems
with a given TRS in the AZ classes, 
which are not necessarily relativistic at microscopic scales,
following Ref.\ \cite{Freed2016}. 
We also derive
(i) obstructions to give structures of various kinds, and, when it is possible to have structures,
their isomorphism classes, and also
(ii) the Dirac quantization conditions for 
pin$^c$ and pin$^{\tilde c}_{\pm}$ connections,
which are necessary input for constructing many-body
topological invariants. 
Those readers who are interested in the explicit formulas for many-body
topological invariants
can safely skip this section. 

\item  
  Second, 
in Sec.\ \ref{Fermionic partial transpose and partial antiunitary transformations},
we formulate partial transpose and antiunitary transformations,
associated to orientation-reversing symmetries in the AZ symmetry classes.  
These operations will be used 
in Sec.~\ref{sec:Methodtocomputethetopologicalinvariant}
to construct many-body topological invariants in the operator formalism.

\end{itemize}

After these technical preliminaries,  
we will present our formula for many-body topological invariants
for various SPT phases
in $(1+1)d$ (Sec.\ \ref{sec:many-body_invariant}),
and in $(2+1)d$ and $(3+1)d$ 
(Sec.\ \ref{sec:many-body_invariant in 2-3d}): 
\begin{itemize}


\item 
  In Sec.\ \ref{sec:many-body_invariant}, 
  after giving an overview on how fermionic 
  partial transpose and antiunitary transformations
  can be used to ``simulate''   
  the path-integral on unoriented spacetime manifolds
  (Sec.\ \ref{sec:Methodtocomputethetopologicalinvariant}),
  we provide explicit formulas for topological invariants for
  fermionic SPT phases in five AZ symmetry classes (BDI,DIII,AIII,AI,and AII)
  in one spatial dimension.
  (Table~\ref{tab:summary})
  Our formulas are benchmarked for
  some microscopic models numerically,
  and for fixed point ground state wave functions analytically.


\item
  In Sec.\ \ref{sec:many-body_invariant in 2-3d}, 
  we provide explicit formulas for topological invariants for
  various fermionic SPT phases in
  two and three spatial
  dimensions. 
  (Table~\ref{tab:summary 2})
  In particular, we discuss the celebrated examples of
  time-reversal symmetric topological insulators in $(2+1)d$
  and $(3+1)d$.
  We will also give an in-depth discussion on the many-body topological
  invariant for the integer quantum Hall effect.
  This serves as a pedagogical introduction to
  other cases in symmetry classes A+$CR$ and AII,
  and also illustrate that for a given path-integral 
  \eqref{path integral} and \eqref{top term},
  there are more than one way to simulate it
  in the operator formalism. 
\end{itemize}

We conclude in Sec.\ \ref{sec:discussion} with outlook and open problems.
Appendices are devoted to technical and mathematical details. 
In particular: 
\begin{itemize}
\item 
  In Appendix~\ref{app:pin}, 
  the obstruction classes and topological sectors for variants pin structures
  are presented in a self-contained manner.

\item 
  In Appendix \ref{App: Dirac quantization conditions},
  we present the derivation of the Dirac quantization condition
  for pin$^c$ and pin$^{\tilde c}_{\pm}$ connections. 

\end{itemize}

\begin{table*}[!]
\caption{\label{tab:summary} List of many-body topological invariants for 
fermionic topological phases studied in the present paper.
The first column specifies the AZ symmetry class.
The second column shows the cobordisms of pin structures.~\cite{Freed2016}
The bold $\mathbf{Z}_2$ and $\mathbf{Z}$ represent topological phases 
which appear only in the presence of interaction.
$KB$ and $\mathbb{R}P^2$ represent the Klein bottle and the real projective 
plane, respectively.
}
\begin{center}
{\scriptsize
\begin{tabular}{| >{\centering\arraybackslash}m{1.5cm} | 
 >{\centering\arraybackslash}m{2cm} | 
 >{\centering\arraybackslash}m{2.5cm} | 
 >{\centering\arraybackslash}m{7cm} |>{\centering\arraybackslash}m{3cm} 
| >{\centering\arraybackslash}m{1cm} | }
\hline
  AZ class and space dim.
  & Cobordism & Generating spacetime manifold & 
Topological invariant & Comment & Section \\
\hline
  BDI in $(1+1)d$ 
  & $\Omega^{\pin_-}_2 = \Z_8$ & $\mathbb{R}P^2$ &
$$\begin{array}{l}
\mathrm{Tr}\, \Big[ \rho_I C_T^{I_1} \rho_I^{\mathsf{T}_1} [C_T^{I_1}]^{\dag} \Big]
\end{array}$$
& Adjacent partial transpose.~\cite{ShapourianShiozakiRyu2016detection} 
                                & \ref{sec:(1+1)BDI}

  \\ \hline
  DIII in $(1+1)d$ 
  & $\Omega^{\pin_+}_2 = \Z_2$ & $KB$, $(R,R)$ sector&
$$\begin{array}{l}
\mathrm{Tr}\, \Big[ \rho_{I_1 \cup I_3}\big( (-1)^{F_2} \big) C_T^{I_1} \rho_{I_1 
\cup I_3}^{\mathsf{T}_1}\big( (-1)^{F_2} \big) [C_T^{I_1}]^{\dag} \Big]
\end{array}$$
& Disjoint partial transpose with intermediate fermion parity twist. & 
                                                                       \ref{sec:(1+1)DIII}

  \\ \hline
  AIII in $(1+1)d$ 
  & $\Omega^{\pin^c}_2 = \Z_4$ & $\mathbb{R}P^2$, the flux threading 
$\mathbb{R}P^2$ is quantized to $\pm i$ &
$$\begin{array}{l}
\mathrm{Tr}\, \Big[ \rho_I U_S^{I_1} \rho_I^{\sf T_1} [U_S^{I_1}]^{\dag} \Big]
  \end{array}$$ & Adjacent partial transpose. & \ref{sec:(1+1)AIII}

  \\ \hline
  AI in $(1+1)d$
  & $\Omega^{\pin^{\tilde c}_-}_2 = \Z \times \mathbf{Z}_2$ &
$\mathbb{R}P^2$ for ${\bf Z}_2$ and a 2-manifold with a unit magnetic flux for 
$\Z$ &
$\mathrm{Tr}\, \Big[ \rho_I C_T^{I_1} \rho_I^{\sf T_1} [C_T^{I_1}]^{\dag} \Big]$ 
for ${\bf Z}_2$
& Adjacent partial transpose. $\Z_2$ phase is an interaction enabled SPT 
                  phase which is equivalent to the Haldane phase. & \ref{sec:(1+1)AI}

  \\ \hline
  AII in $(1+1)d$ 
  & $\Omega^{\pin^{\tilde c}_+}_2 = \Z$ &
$\mathbb{R}P^2$ with a half magnetic flux $\int_{RP^2} F = \pi$ &
$$\begin{array}{l}
\mathrm{Tr}\, \Big[ \rho_I \prod_{x \in I_1} e^{\frac{\pi i x \hat n(x)}{2 |I_1|}} 
C_T^{I_1} \rho_I^{\mathsf{T}_1} [C_T^{I_1}]^{\dag} \prod_{x \in I_1} 
e^{\frac{-\pi i x \hat n(x)}{2 |I_1|}} \Big]
\end{array}$$
& Adjacent partial transpose with the twist operator. & 
\ref{sec:(1+1)AII} \\ \hline
\end{tabular}
}
\end{center}
\end{table*}

\begin{table*}[!]
\caption{\label{tab:summary 2} List of many-body topological invariants for 
fermionic topological phases studied in the present paper.
The first column specifies a AZ symmetry class.
The second column shows the cobordisms of pin structures.~\cite{Freed2016}
The bold $\mathbf{Z}_2$ and $\mathbf{Z}$ represent topological phases 
which appear only in the presence of interaction.
$KB$ and $\mathbb{R}P^2$ represent the Klein bottle and the real projective 
plane, respectively.
}
\begin{center}
{\scriptsize
\begin{tabular}{| >{\centering\arraybackslash}m{1.4cm} | 
 >{\centering\arraybackslash}m{1.5cm} | 
 >{\centering\arraybackslash}m{2.3cm} | 
 >{\centering\arraybackslash}m{8.1cm} |
 >{\centering\arraybackslash}m{3cm} | 
 >{\centering\arraybackslash}m{1cm} | }
\hline
  AZ class and space dim. 
  & Cobordism & Generating spacetime manifold & 
                                                Topological invariant & Comment & Section

  \\ \hline

  DIII in $(2+1)d$ 
  & $\Omega^{\pin_+}_3 = \Z_2$ & $KB(x,y)$ $\times S^1(z)$ with 
the periodic boundary conditions for all the $x,y,z$-directions&
$$\begin{array}{l}
    \mathrm{Tr}\, \Big[ \rho_{R_1 \cup R_3}
    \big( (-1)^{F_2} \big) C_T^{R_1}
    [\rho_{R_1 \cup R_3}\big( (-1)^{F_2} \big)]^{\sf T_1}
    [C_T^{R_1}]^{\dag} \Big]
\end{array}$$ &
Disjoint partial transpose with intermediate fermion parity twist. & 
\ref{sec:(2+1)DIII} \\ \hline

  &&&
$$\begin{array}{l}
\left. \Braket{GS | \hat N |GS} \right|_{\int F = 2 \pi} - \left. 
\Braket{GS| \hat N |GS}\right|_{F=0}
\end{array}$$ & Charge pumping. Described by the Chern-Simons theory. & 
\ref{sec:2d_classA_charge_pump} \\ \cline{4-6}
  A in $(2+1)d$ 
  & Chiral phases. Topological classification is $\Z \times 
\mathbf{Z}$. &
(2d manifold with a unit magnetic flux $\int F = 2 \pi$) $\times$ ($S^1$ 
with a flux $\theta$) for $\Z$ &
$$\begin{array}{l}
\frac{i}{2 \pi} \oint d_{\theta_y} \log \Braket{GS(\theta_y) | 
\prod_{x,y} e^{\frac{2 \pi i x \hat n(x,y)}{L_x}} |GS(\theta_y)}
\end{array}$$
& Twist operator along the $x$-direction with twisted boundary condition 
along the $y$-direction. & \ref{sec:2d_classA_twisted_bc_and_twist} \\ 
\cline{4-6}
  &&&
$$\begin{array}{l}
\frac{i}{2 \pi} \oint d_\theta \log \Big\langle GS \Big|
{\displaystyle \prod_{(x,y) \in R_1 \cup R_2} } e^{\frac{2 \pi 
i y \hat n(x,y)}{L_y}} \\
\quad \cdot {\rm Swap}(R_1, R_3)
{\displaystyle \prod_{(x,y) \in R_1 \cup R_2} }
e^{i \theta \hat n(x,y)} \Big| GS \Big\rangle
\end{array}$$
& Swapping two disjoint intervals and the twist operator. & 
                                                            \ref{sec:Swap and twist operator}

  \\ \hline

  & & &
${\displaystyle \frac{\Braket{GS(\int F=2 \pi) | CR |GS(\int F = 2 
\pi)}}{\Braket{GS(F=0)| CR |GS(F=0)}}} $ & Ground state parity of $CR$ 
in the presence of a unit magnetic flux.~\cite{witten2015fermion} & 
\ref{sec:cr_z2} \\ \cline{4-6}

  A+$CR$, $(CR)^2=1$ in $(2+1)d$
  & $\Omega^{\pin^{\tilde c}_+}_3 = \Z_2$ &
$KB [(x,y) \sim (1-x,y+1)] \times S^1 (z)$ with a unit magnetic flux 
$\int F_{zx} = 2 \pi$ &
$\exp \Big[ \int_0^{2\pi} \braket{GS(KB; \theta)|\partial_\theta|GS(KB; 
\theta)} d\theta \Big]$ &
Berry phase of the ground state wave function on the Klein bottle. & 
\ref{z2_a+cr_berry_phase} \\ \cline{4-6}

  & & &
        ${\displaystyle \Braket{GS| 
        \prod_{(x,y) \in R_2} e^{\frac{2 \pi i y \hat n(x,y)}{L_y}} \prod_{(x,y) \in R_3} (-1)^{\hat n(x,y)}  {\left. CR\right|_{R_1 \cup R_3} } |GS}}$ &
Partial $CR$ flip exchanging two disjoint regions and intermediate twist 
operator.
                                                               & \ref{sec:+cr_cr_swap}
  \\ \hline

  AII in $(2+1)d$
  & $\Omega^{\pin^{\tilde c}_+}_3 = \Z_2$ &
$KB [(x,y) \sim (1-x,y+1)] \times S^1 (z)$ with a unit magnetic flux 
$\int F_{zx} = 2 \pi$ &
$$
\begin{array}{l}
\mathrm{Tr}_{R_1 \cup R_3} \Big[
  \rho_{R_1 \cup R_3}^{+} C_T^{R_1}
  [\rho_{R_1 \cup R_3}^{-}]^{\sf T_1} [C_T^{R_1}]^{\dag} \Big], \\
\rho_{R_1 \cup R_3}^{\pm} 
  = \mathrm{Tr}_{\overline{R_1 \cup R_3}}
  \Big[ \displaystyle{\prod_{(x,y)\in R_2} 
 e^{\pm \frac{2 \pi i y}{L_y} \hat n(x,y)}} 
\ket{GS} \bra{GS} \Big].
\end{array}$$ & Disjoint partial transpose along the $x$-direction with 
intermediate $U(1)$ twist which varies in the $y$-direction as $e^{2 \pi i y/L_y}$. & \ref{sec:(2+1)AII} \\
\hline

  A in $(3+1)d$ 
  & $\Omega^{\spin^{c}}_4 = \Z \times {\bf Z}$ &
Two theta terms. One of $\Z \times \Z$ is generated by a 4-manifold $X$ 
with $\frac{1}{8 \pi^2} \int_X F^2 = 1$ &
${\displaystyle \Braket{GS(\int F_{xy}= 2\pi) | e^{\frac{2 \pi i}{L_z} 
\sum_{x,y,z} z \hat n(x,y,z)} | GS(\int F_{xy}= 2\pi)} }$ &
Ground state expectation value of the twist operator along the 
$z$-direction in the presence of a unit magnetic flux in the $xy$-plane. &
\ref{theta-term in (3+1)d} \\ \cline{4-5}

  &&&
$\exp \left[ \oint_0^{2\pi} \Braket{GS
\left( \begin{array}{l}
\int F_{xy}=2 \pi, \\
\oint A_z = \theta_z
\end{array} \right)
|\partial_{\theta_z}|GS
\left( \begin{array}{l}
\int F_{xy}=2 \pi, \\
\oint A_z = \theta_z
\end{array} \right)} d\theta_z \right]$ &
Berry phase of the ground state wave function in the presence of a unit 
                                          magnetic flux in the $xy$-plane.~\cite{WangZhang2014} &

  \\ \hline
  A + $CR$, $(CR)^2=1$ in $(3+1)d$
  & $\Omega^{\pin^{\tilde c}}_4 = \Z_2 \times {\bf 
Z}_2 \times {\bf Z}_2$ &
One of $\Z_2 \times {\bf Z}_2 \times {\bf Z}_2$ is generated by a 
4-manifold $X$ with $\frac{1}{8 \pi^2} \int_X F^2 = 1$ &
${\displaystyle \frac{\Braket{GS(\int F_{xy}=2 \pi, \oint A_z = \pi)| CR 
|GS(\int F_{xy}=2 \pi, \oint A_z = \pi)}}{\Braket{GS(\int F_{xy}=2 \pi, 
\oint A_z = 0)| CR |GS(\int F_{xy}=2 \pi, \oint A_z = 0)}} }$ &
$CR$ reflection (acting the $xy$-plane) on the ground state with a unit 
magnetic flux in the $xy$-plane and $\pi$-flux along the $z$-direction. &
\ref{sec:3d_class_a+cr} \\ \hline

\end{tabular}
}
\end{center}
\end{table*}

\clearpage 
\newpage

\section{Variants of pin structures}
\label{Dirac quantization conditions}

In this section, we will discuss
time-reversal symmetries of various kinds in the AZ symmetry classes,
and the corresponding variants of pin structures in relativistic fermion theories. 
These structures are important input for the cobordism classification of fermionic TQFTs (spin or pin TQFTs),
which provides the classification of cobordism invariant partition functions
obtained by the Euclidean path integral of fermionic quantum field theories.
This section may be skipped, if the reader is interested mostly in
the explicit formulae of many-body topological invariants of fermionic SPT
phases,
which can be found in the later sections.

In Sec.\ \ref{Time-reversal symmetry in Euclidean path integral},
we start by discussing time-reversal in the Euclidean path integral
for generic systems without assuming relativistic invariance. 
Within the Euclidean path integral, 
TRS can then be used to twist boundary conditions, 
leading to unorientable spacetime such as the real projective plane, Klein
bottle, etc. 
(Sec.\ \ref{Twisting boundary condition by time-reversal}.)
As discussed in Sec.\ \ref{intro},
the path integral on these spacetime manifolds can serve as topological
invariants of fermionic SPT phases. 
While this can be done without assuming relativistic invariance,
we will then move on to relativistic quantum field theories.
In the relativistic context,
in order to consider fermionic theories on unorientable spacetime,
we can start from local definitions of fermionic spinors (on each coordinate patch).
They can be then glued together  
to give the global definitions of fermionic spinors on unorientable spacetime.
As we will see, in this construction,
TRS constitutes a part of this gluing operation.
This consideration allows us to establish a connection between 
(twisting or gluing by) TRS in the AZ symmetry classes and (variants of) pin structures
in relativistic quantum field theory 
(Sec.\ \ref{Spin and pin structures, and their variants}).
In the following Sections
\ref{Obstructions and classifications for variants of spin/pin structures}
and
\ref{Dirac quantization conditions}, 
we will discuss
variants of pin structures in more detail,
and, in particular,
derive
obstructions and isomorphism classes (classifications) of
variants of pin structures,
and 
the ``Dirac quantization conditions''
for
pin$^c$ and pin$^{\tilde c}_{\pm}$ structures.
(Most of the details of the derivations are presented in Appendices
\ref{app:pin}
and 
\ref{App: Dirac quantization conditions}.)
The Dirac quantization conditions
will be important in constructing many-body topological invariants
in the later sections. 

\subsection{Variants of pin structures and
their origins in non-relativistic systems}

\subsubsection{Time-reversal symmetry in Euclidean path integral}
\label{Time-reversal symmetry in Euclidean path integral}

In this section, we consider how TRS can be implemented in the Euclidean path integral:
This should be contrasted with TRS in the operator formalism, where
time-reversal is given by an anti-linear and anti-unitary operator acting on the
Hilbert space. 
The path-integral formulation is particularly useful when discussing
twisting boundary conditions by TRS to generate unorientable spacetimes. 
For real fermion fields (in symmetry class BDI), this point was already discussed
in Ref.\ \cite{ShapourianShiozakiRyu2016detection}.
Here, we will take complex fermion fields
in symmetry class AI and AII as an example.

Let $\{ \hat \psi^{\dag}_j, \hat \psi^{\ }_j\}$ be
a set of fermionic creation and annihilation operators.
Here, the index $j$ labels collectively all degrees of freedom
including the spatial coordinates, and internal degrees of freedom. 
Suppose that anti-unitary time-reversal $\hat{T}$
acts on these operators as 
\begin{align} 
  \hat T \hat \psi^{\dag}_j \hat T^{-1} = \hat \psi^{\dag}_k [\cU_T]^{\ }_{kj},
  \quad 
\hat T i \hat T^{-1} = -i.
\label{eq:trs}
\end{align}
The Hermite conjugate of
\eqref{eq:trs} leads to
the transformation rule 
for the fermion annihilation operators, $\hat \psi_j$. 
Here, $\cU_T$ is a unitary matrix and satisfies
\begin{align}
  \cU^{\ }_T \cU_T^* = \pm 1.
\end{align}
The $\pm$ signs correspond to
time-reversal which squares to $+1$,  
$\hat T^2=1$ (non-Kramers),
and to the fermion number parity $(-1)^F$,
$\hat T^2 = (-1)^F$ (Kramers),
respectively.

Let  us now consider a free fermionic Hamiltonian
$\hat H = \sum_{jk} \hat \psi^{\dag}_j h^{\ }_{jk} \hat \psi^{\ }_k$.
By the standard coherent state path-integral,
the partition function $Z$ can be
expressed as the Euclidean path integral, 
\begin{align}
&Z = \int D \bar \psi D \psi e^{-S} = {\rm Det} (\p_{\tau} \delta_{jk} + h_{jk}), \\
&S=\int d \tau \bar \psi_j(\tau) \big[ \p_{\tau} \delta_{jk} + h_{jk} \big] \psi_k(\tau), 
\end{align}
where $\bar \psi_j(\tau)$ and $\psi_k(\tau)$ are independent Grassmann variables. 
Requiring TRS,
$\hat T \hat H \hat T^{-1} = \hat H$,
which is equivalent to $\cU^{\ }_T h^* \cU_T^{\dag} = h$,
implies that the action $S$ is invariant under the following
renaming of path-integral variables, 
\begin{align}
T_E:&
  \psi_j(\tau) \mapsto i \bar \psi_k(-\tau) [\cU_T]_{kj},
  \nonumber \\ 
  &
\bar \psi_j(\tau) \mapsto i [\cU_T^{\dag}]_{jk} \psi_k(-\tau),
\label{eq:trs_euclid}
\end{align}
which looks like a composition of
reflection in the imaginary time direction and a particle-hole transformation.
Now it is clear that time-reversal in the operator formalism
is translated into 
reflection along the imaginary time direction in the Euclidean path integral.
\footnote{
  Here, we consider the quadratic Hamiltonian for simplicity.
  However, including interactions in the Hamiltonian gives rise to no extra difficulties,
  and leads to the same result, \eqref{eq:trs_euclid}.
}
It should be noticed that,
in contrast to the operator formalism,
both the first and second equations in
(\ref{eq:trs_euclid}) are needed to define the Euclidean TRS since the Grassmann
fields $\psi_j$ and $\bar \psi_j$ are independent.

Another important point to note is that, 
when TRS is translated into the particle-hole reflection symmetry in the Euclidean path integral, \eqref{eq:trs_euclid},
the additional pure imaginary phase factor $\pm i$ must be introduced.
This is to compensate the minus sign arising from the anticommuting property of Grassmann variables. 
We will see in Sec.~\ref{Definition of fermionic partial transpose}
that a similar phase factor arises when we introduce the fermionic partial transpose.
We also note that the square of the Euclidean TRS $T_E$ is
$T_E^2 = 1$ for $\hat T^2=(-1)^F$ and $T_E^2=-1$ for $\hat T^2=1$.

\subsubsection{Twisting boundary condition by time-reversal}
\label{Twisting boundary condition by time-reversal}

Unitary on-site symmetries in quantum many-body systems
(quantum field theories) can be used to twist
boundary conditions, or to introduce symmetry-twist defects,
by introducing a branch cut (or branch sheet).
In the presence of a symmetry-twist defect, 
quantum fields, when adiabatically transported across
a branch cut (sheet),
will be acted upon by the symmetry operator.
Twisted boundary conditions and symmetry-twist defects
have been proven to be useful diagnostic tools to detect properties of the systems.
In particular, different SPT phases protected by symmetries
respond to the twisted background in different ways,
and hence can be distinguished and characterized this way.

Twisting boundary conditions is most commonly done in particle number conserving systems,
where boundary conditions can be twisted by adding a $U(1)$ phase.
For example, in the context of the quantum Hall effect, 
from the response of the system to the twisted boundary condition,
one can extract the quantized Hall conductance (see Sec.\ref{sec:(2+1)A}).
One can also introduce a symmetry-twist defect in this context, 
a small flux tube carrying unit flux quantum, which accumulates/depletes charge.
The quantized response to the flux tube can be used 
as a characterization of the quantum Hall system.  

For fermionic systems without particle number conservation, 
e.g., those described by the BdG equations,  
boundary conditions can still be twisted by using fermion number parity --
by adding $\pm$ signs corresponding to periodic/anti-periodic boundary conditions. 
Twisting by fermion number parity can be used, for example, 
to detect the non-trivial topological superconductor phase of the Kitaev chain.

An orientation reversing symmetry,
such as time-reversal or spatial reflection,
can be used to twist boundary conditions
as well.
This procedure effectively leads to spacetime manifolds
which are not orientable in general.
Different SPT phases protected by an orientation reversing symmetry
can then be distinguished and characterized by
their responses to unoriented geometry/topology of spacetime manifold.   
For example, in Ref.~\cite{ShapourianShiozakiRyu2016detection},
it was shown that 
the $\Z_8$ SPT invariant of (1+1)d topological superconductor phases in symmetry class BDI
can be detected by the discretized imaginary time path-integral in the presence of a cross-cap
(the relevant spacetime is $\mathbb{R}P^2$).

%
%

\subsubsection{Spin and pin structures, and their variants}
\label{Spin and pin structures, and their variants}

\paragraph{Spin and pin structures}

The above construction of the path-integral on unoriented spacetime manifolds
through twisting boundary conditions by an orientation reversing symmetry
can be applied to systems without relativistic invariance. 
On the other hand, 
in the presence of relativistic invariance,
the construction of fermionic, relativistic quantum field
theories on generic (both orientable and unorientable) spacetime 
can be done once we specify how we glue local definitions of 
fermionic spinors. 
The data needed to specify the gluing rule consists of structures, such as spin or pin structures.  
We will now see that 
TRS constitutes a part of this gluing operation, 
and hence a specific choice of TRS
corresponds to a specific type of pin structures (and variants thereof).


To be specific, 
let us start by recalling the standard construction of
(unoriented) spacetime manifolds, and quantum fields on them.
We will first focus on real fermion fields,
i.e., systems without having particle number conservation symmetry,
and discuss spin and pin$_{\pm}$ structures. 
We start from a collection of ``patches'', $\{U_i\}$.
These patches are glued together to give a global definition of a
Riemannian manifold $X$.
(Here, we adopt the Euclidean signature.)
When two patches intersect, there is a transition function
relating two patches.
The transition functions take their values in $SO(d+1)$
for an oriented, $(d+1)$-dimensional Riemannian manifold.
Here, $SO(d + 1)$ is the structure group 
acting on the frame fields (vielbein) on the tangent spaces.
Similarly, global definition of a quantum field on the manifold can be made 
by first considering quantum fields on individual patches and then
gluing them together. 
For a relativistic fermion field on $X$,
$SO(d+ 1)$ is lifted to its double cover,
$\spin(d+1)$,
which is generated by $\Sigma_{ab}= [\gamma^a, \gamma^b]/4i$
where $\gamma^{a=1,\ldots,d+1}$ are gamma matrices.
The choice of signs that arises
in this lifting defines a spin structure.

Similarly, for an unoriented, $(d+1)$-dimensional Riemannian manifold $X$,
the transition functions are members of $O(d+1)$,
i.e., the structure group is $O(d + 1)$.
To define relativistic fermion fields on $X$,
$O(d+1)$ should be lifted to
its double cover, $\pin_{+}(d+ 1)$ or $\pin_-(d+1)$.
Here, these are two different double covers of $O(d+1)$.
Once again, sign ambiguities that arise in the lifting
define the pin structures.  
In addition to continuous spinor rotations 
generated by $\Sigma_{ab}$,
$\pin_{\pm}(d+ 1)$ has an element which reverses
the orientation,
and squares to $\pm 1$.
To be more specific,
$\pin_{\pm}(n)$ groups are defined as follows.
Let ${\rm Cliff}_{\pm n}$ be the algebra over $\R$ generated by $e_1, \dots, e_n$ subject to the relations $e_ie_j+e_je_i = \pm 2 \delta_{ij}$. 
Introduce the ``mode space'' $M := \{\sum_i x_i e_i | x_i \in \R\}$ with an inner product $(x,y) = \frac{1}{2} (xy+yx), x,y \in M$. 
The $\pin_{\pm}(n)$ group is defined by 
\begin{align}
\pin_{\pm}(n)
:= \{v_1 \cdots v_r | v_i \in M, (v_i,v_i) = \pm 1\}.
  \label{def pin pm groups}
\end{align}


Summarizing, relativistic fermion fields can be defined globally
on an unoriented manifold,
by patching and gluing local definitions by $\pin_{\pm}(d+ 1)$.
The gluing by the orientation-reversing element in
$\pin_{\pm}(d+ 1)$
is an analogue of twisting by orientation-reversing symmetry
which we discussed without assuming relativistic invariance.
As mentioned,
an orientation reversing symmetry in the Euclidean signature
which squares to $\pm 1$ corresponds to
time-reversal symmetry which squares to
$(-1)^F$ and $+1$, respectively, in the operator formalism.  
That is to say,
symmetry classes DIII and BDI 
correspond to
$\pin_{+}(d+ 1)$
and
$\pin_{-}(d+ 1)$
respectively in the relativistic context.

It once again should be stressed that in the above definitions, fermion fields 
are assumed to transform as a spinor under $SO(d+1)$ or $O(d+1)$.
The main difference between condensed matter systems 
and relativistic systems is the absence of \emph{a priori}
rotation symmetry in condensed matter systems. 
Put differently, fermion fields in condensed matter systems 
do not know the transformation rule under $O(d+1)$ rotations because of the lack
of rotation symmetry.
~\footnote{In lattice systems, the degrees of freedom obey a representation of space group under consideration.}
It is thus not obvious that the classification of SPT phases
by the cobordism group still works for condensed matter systems.  
Nevertheless,
when some sort of topological media (topological phases)
are realized, effective relativistic fermions can emerge,
which are sensitive to spin and pin structures.
In short, spin may emerge from short-range entanglement in quantum ground states.


\paragraph{Variants of spin/pin structures and AZ symmetry classes}

\begin{table}
\begin{center}
\begin{tabular}{l|c|c}
Symmetry & AZ class & Relativistic pin structure \\
\hline 
$(-1)^F$ & D & $\spin$ \\
$(-1)^F, \hat T$, $\hat T^2=1$ & BDI & $\pin_-$ \\
$(-1)^F, \hat T$, $\hat T^2=(-1)^F$ & DIII & $\pin_+$ \\
$U(1)$ & A & $\spin^{c}$ \\
$U(1) \rtimes \hat T$, $\hat T^2=1$ & AI & $\pin^{\tilde c}_-$ \\
$U(1) \rtimes \hat T$, $\hat T^2=(-1)^F$ & AII & $\pin^{\tilde c}_+$ \\
$U(1) \times \hat C \hat T$ & AIII & $\pin^{c}$ \\
\end{tabular}
\end{center}
\caption{
  Orientation reversing symmetries
  in the Altland-Zirnbauer (AZ) classification 
  \cite{altland1997nonstandard}
  and corresponding relativistic spin and pin structures in fermionic systems. 
  Here, $(-1)^F$ is the fermion number parity,
  $\hat T$ is antiunitary time-reversal,
  and 
  $\hat C$ is unitary particle-hole symmetry.
  The semidirect product $U(1) \rtimes \hat T$ means that the $U(1)$ charge
  $e^{i Q}$ is flipped under $\hat T$ as
  $\hat T e^{i Q} \hat T^{-1} = e^{-i Q}$.}
\label{tab:pin}
\end{table}

The above correspondence between
orientation reversing symmetries
and
pin structures can be extended
to other $T$ and $CT$ symmetries in AZ symmetry classes:
There is a one-to-one correspondence between AZ symmetry classes
and types of spin and pin structures
in Euclidean relativistic quantum field theories,
as summarized in 
Table~\ref{tab:pin}.



Let us give a brief overview of
Table~\ref{tab:pin}.
More details are given in Appendix~\ref{app:pin}. 
First, for
the symmetry class in which only fermion number parity is conserved,  
i.e., symmetry class D,
the corresponding structures are spin structures.
Next,
if, in addition to fermion number parity conservation,
there is time-reversal symmetry,
the corresponding AZ class is either class DIII or BDI,
depending on if $T^2=(-1)^F$ or $T^2=1$, respectively.
For these classes, we associate
pin$_+$ or pin$_-$ structures. 
These structures are associated to
$\pin_{+}$ or $\pin_-$ groups, respectively.

We now add the electric charge $U(1)$ conservation,
and consider symmetry classes A, AIII, AI, and AII. 
First, for symmetry class A,
we consider spin$^c$ structures.
I.e., fermions on manifolds endowed with a spin$^c$ structure, in which 
states obey the so-called spin-charge relation.~\cite{SeibergWitten2016gapped}
Next, let us consider symmetry class AIII. Symmetry class AIII
respects $CT$ (the combination of time-reversal and unitary particle-hole
symmetry). 
The 
corresponding structures are called pin$^{c}$ structures.
Finally, symmetry classes AII and AI respect, in addition to the $U(1)$ charge 
conservation, time-reversal symmetry which squares to $(-1)^F$ and $1$, respectively. 
The corresponding structures are called pin$^{\tilde{c}}_{\pm}$ structures.  
pin$^{\tilde c}_{\pm}$ structures are variants of ordinary pin$_{\pm}$
structures, which were introduced by Refs.\ \cite{Metlitski2015, Freed2016}, and
will be discussed in detail shortly.
To be precise, the pin$^c$ and pin$^{\tilde c}_{\pm}$ structures are 
associated with the patch transformations of fermionic fields 
which are elements of $\pin^c(n)$ and $\pin^{\tilde c}_{\pm}(n)$ group,
respectively. 
(See Appendix~\ref{App: Dirac quantization conditions} for details.)

In addition to these symmetry classes and their corresponding structures,
there remain symmetry classes (C, CI and CII) and structures
($G_{0,\pm}$). $G_0, G_+$, and $G_-$ structures are $SU(2)$ analogs of spin$^c$ and pin$^c$
structures,
where $U(1)$ is replaced by the internal $SU(2)$ symmetry.
(See Appendix~\ref{app:pin}.)

%

\subsection{Obstructions and classifications
 for variants of spin/pin structures}
\label{Obstructions and classifications for variants of spin/pin structures}

Manifolds for which we can endow a spin structure are called spin manifolds.
It is important to realize that not all manifolds are spin
-- there may be an obstruction. 
For spin structures, the following theorem is known:
\begin{lem}[Spin-structure]
The following holds true:
\begin{itemize}
\item
A principal $O(n)$-bundle $P \to X$ admits a spin structure if and only if $w_1(P) = 0$ and $w_2(P) = 0$.

\item
In the case when $P$ admits a spin structure, the set of isomorphism classes of spin structures on $P$ is identified with $H^1(X; \Z_2)$.
\end{itemize}
\end{lem}
From the second statement, it follows that
on $T^2$ different spin structures correspond to different choices of
periodic and antiperiodic boundary conditions --
this has been mentioned and used in the above example of the
Kitaev chain. 

It should be empathized that the second statement implies the set of spin structures on $X$ is equivalent to $H^1(X;\Z_2)$ {\it as a set}. 
Moreover, the set of spin structures is a $H^1(X;\Z_2)$-Torsor, i.e., any spin structure $\eta$ can be obtained from a shift $\eta = \eta_0 + A$ by a $\Z_2$-background field $A \in H^1(X;\Z_2)$ on a some ``reference'' spin structure $\eta_0$.  
There is no canonical choice of the reference spin structure $\eta_0$, which contrasts with $\Z_2$-background fields where there is the zero flux in $H^1(X;\Z_2)$. 
The absence of any reference elements is a feature of the spin structure (and variants of pin structures).

Similarly, 
for other types of structures (including $G_{0,\pm}$),
one can discuss their obstructions,
and isomorphism classes (``classifications''). 
These are derived and discussed in detail in Appendix \ref{app:pin}.


\subsection{Dirac quantization conditions}
\label{Dirac quantization conditions}

Let us begin this part by an important remark that in $U(1)$ charge conserving systems
on a manifold with a spin$^c$/pin$^c$/pin$^{\tilde{c}}_{\pm}$ structure, 
the background $U(1)$ gauge field $A$ is not globally defined.
Instead, it should be considered as a
spin$^c$/pin$^c$/pin$^{\tilde{c}}_{\pm}$ connection. 

This point can be clearly seen for spin$^c$ connections as follows.~\cite{SeibergWitten2016gapped}
The Dirac operator on a spin$^c$ manifold takes a form
$D_{\mu} = D^{(0)}_{\mu} - i A_{\mu}$,
where $D^{(0)}_{\mu}$ is the part independent of $A$ 
(but can include the contribution from the background gravitational field),
and we have assumed the complex fermion field carries unit charge. 
On a general spin$^c$ manifold $X$ which is not spin, $D^{(0)}_{\mu}$ is not
well-defined by itself and $A_{\mu}$ is not a $U(1)$ connection globally.~\footnote{
A $U(1)$ connection $A_{\mu}$ here means a data of 1-forms $A_i \in
\Omega^1(U_i)$
defined on local patches which are subject to the gluing condition $A_i = A_j +
d \alpha_{ij}$
on the intersection $U_i \cap U_j$.
Here the data $e^{i \alpha_{ij}}: U_i \cap U_j \to U(1)$ subject to the
cocycle condition
$e^{i \alpha_{ij}} e^{i \alpha_{jk}} e^{i \alpha_{ki}} = 1$ is a complex line bundle $L$.}
To make $D_{\mu}$ well-defined as a whole,
$A_{\mu}$ should satisfy the ``Dirac quantization condition'' modified by the gravitational contribution as 
\begin{align}
\int_C \frac{F}{2 \pi} \equiv \frac{1}{2} \int_C w_2(TX) \ \mod 1 
\label{eq:dirac_cond_spinc}
\end{align} 
for all two-cycles $C \in Z_2(X;\Z)$,
\footnote{
The r.h.s.\ of (\ref{eq:dirac_cond_spinc}) means the pairing $\Braket{w_2(TX), [C]}$.}
where $w_2(TX) \in H^2(X;\Z_2)$ is the second Stiefel-Whitney class that measures the obstruction to give a spin structure on $X$.\footnote{
Roughly speaking, a spin$^c$-connection can be thought of as a connection associated to the ``square root'' $L^{1/2}$ of a complex line bundle $L$, because the obstruction to give a square root $\pm e^{i \alpha_{ij}/2}$ of transition functions $e^{i \alpha_{ij}}$ determines a two-cocycle by $z_{ijk} = (\pm e^{i \alpha_{ij}/2})(\pm e^{i \alpha_{jk}/2}) (\pm e^{i \alpha_{ki}/2}) \in \pm 1$ and $z_{ijk}$ is equivalent to the two-cocycle of $w_2(TX)$.~\cite{scorpan2005wild}
}
The condition (\ref{eq:dirac_cond_spinc}) implies that for a two-cycle with $\int_C w_2(TX) = 1 \mod 2$ the ``monopole charge'' inside the $C$ is quantized to half integers. 
We may call (\ref{eq:dirac_cond_spinc}) the Dirac quantization condition for spin$^c$ connections.

For our purpose of finding
non-local order parameters (topological invariants)
of fermionic SPT phases realized in $U(1)$ charge conserving systems,
we need to understand the allowed background $U(1)$ gauge fields $A$ on closed
spacetime manifolds,
in particular, unoriented ones.
For complex fermion fields, the transition functions between two patches belong
to
$\spin^c$, $\pin^c$ or $\pin^{\tilde c}_{\pm}$ groups according to the symmetry, where the Dirac quantization condition for the $U(1)$ gauge field $A$ is different from that of $U(1)$ gauge fields associated with a complex line bundle. 
Once the Dirac quantization conditions of the $U(1)$ gauge fields associated with these complex spinor bundles are properly fixed, we can find all possible topological sectors of background $U(1)$ gauge fields on a given spacetime manifold.

The Dirac quantization conditions for spin$^c$, pin$^c$, and pin$^{\tilde c}_{\pm}$ connections
are summarized in Table~\ref{tab:dirac_cond}. 
(The Dirac quantization conditions for ordinary $U(1)$ and twisted $U(1)$
connections are also presented, which are relevant to complex scalar fields with
$CT$ and $T$ symmetries, respectively.)
The details of the derivation of the Dirac quantization conditions 
are presented in Appendix \ref{App: Dirac quantization conditions}.
Here, we mention two examples, which are related to the many-body topological
invariants discussed in later sections.

\begin{table*}
\begin{center}
\renewcommand{\arraystretch}{1.5}
\begin{tabular}{| >{\centering\arraybackslash}m{1.1cm} | >{\centering\arraybackslash}m{2.4cm} | >{\centering\arraybackslash}m{3cm} | >{\centering\arraybackslash}m{3.5cm} | >{\centering\arraybackslash}m{2.8cm} | >{\centering\arraybackslash}m{1.8cm} |>{\centering\arraybackslash}m{2.9cm} |}
\hline 
& Antiunitary symmetry & Reflection symmetry & Relativistic structure & Dirac quantization condition & Topological sector & Comment \\
\hline
\multirow{2}{*}{Boson} & 
$CT$ & $R$ & $O \times U(1)$ & $c \equiv 0$ & $H^2(X;\Z)$ & Bosonic paramagnet \\ \cline{2-7}
& $T$ & $CR$ & $O \ltimes U(1)$ & $\tilde c \equiv 0$ & $H^2(X;\tilde \Z)$ & Bosonic insulator \\ 
\hline 
& 
$CT$ & $R$ & $\pin^c = \pin_+ \times_{\{ \pm 1\}} U(1)$ & $c \equiv \frac{1}{2}w_2$ & $H^2(X;\Z)$ & Class AIII \\ \cline{2-7}
Fermion 
&$T$, $T^2=(-1)^F$ & $CR$, $(CR)^2=1$ & $\pin^{\tilde c}_+ = \pin_+ \ltimes_{\{ \pm 1 \}} U(1)$ & $\tilde c \equiv \frac{1}{2}w_2$ & $H^2(X;\tilde \Z)$ & Class AII \\ \cline{2-7}
&$T$, $T^2=1$ & $CR$, $(CR)^2=(-1)^F$ & $\pin^{\tilde c}_- = \pin_- \ltimes_{\{ \pm 1 \}} U(1)$ & $\tilde c \equiv \frac{1}{2}w_2 + \frac{1}{2}w_1^2$ & $H^2(X;\tilde \Z)$ & Class AI \\ 
\hline 
\end{tabular}
\renewcommand{\arraystretch}{1}
\caption{
Dirac quantization conditions. 
$T$, $C$, and $R$ means TRS, charge-conjugation (particle-hole) symmetry, and reflection symmetry, respectively, where $T$ is antiunitary. 
$C$ flips the $U(1)$ charge $Q$ as $C Q C^{-1} = -Q$, while $T$ and $R$ preserves $Q$.  
In the fifth column, $c$ is a 2-cochain $c \in C^2(X;\R)$ and $\tilde c$ is a twisted 2-cochain $\tilde c \in C^2(X;\tilde \R)$, where $\tilde \R$ is the local system associated to $w_1(TX) \in H^1(X;\Z_2)$. $w_i = w_i(TX) \in H^i(X;\Z_2)\ (i=1,2)$ are $i$-th Stiefel-Whitney classes of the tangent bundle $TX$. 
$a \equiv b$ means $a = b$ modulo $\Z$. 
\label{tab:dirac_cond}
}
\end{center}
\end{table*}

%

\subsubsection{Ex: pin$^c$ structure on $\mathbb{R}P^2$}
\label{ex:rp2_pinc}

A nontrivial example of the Dirac quantization condition of pin$^c$ structures is found in the real projective plane $\mathbb{R}P^2$. 
There are two inequivalent topological sectors as $H^2(\mathbb{R}P^2;\Z) = \Z_2$
and the condition listed in Table~\ref{tab:dirac_cond} 
is nontrivial $c \equiv \frac{1}{2} w_2(T(\mathbb{R}P^2)) \neq 0$. 
(See Fig.~\ref{fig:rp2_dirac_conditions}[a] for a representative of $w_2(T(\R P^2))$.) 
Two inequivalent pin$^c$ structures on $\mathbb{R}P^2$ can be detected by the $\pm i$ holonomy along the $\Z_2$ cycle. 
Notice that trivial pin$^c$ connection $A_{{\rm pin}^c}=0$ is forbidden due to the condition $c \neq 0$. 
Also, on unoriented manifolds, the field strength $F_{{\rm pin}^c}$ (in the sense of Euclidean spacetime) should be zero. 
An example of pin$^c$ connections $A_{{\rm pin}^c}$ on $\mathbb{R}P^2$ is shown in Fig.~\ref{fig:rp2_dirac_conditions}[b], 
where the field strength threading a simplex is given by $F_{{\rm pin}^c} = 2 \pi c + \delta A_{{\rm pin}^c}$. 
One can find that the shift by $2 \pi c$ makes $F_{{\rm pin}^c}$ flat. 
The alternative topological sector is given by $-A_{{\rm pin}^c}$. 

The $\pm i$ holonomy is relevant to the construction of the many-body $\Z_4$ SPT invariant of $(1+1)d$ SPT phases with $CT$ or $R$ symmetry. 
The generating manifold of the cobordism group $\Omega^{\pin^c}_2=\Z_4$ is given by $\mathbb{R}P^2$. 
We therefore need to simulate the $\pm i$ holonomy properly in the operator formalism. 
In fact, in the case of the reflection symmetry $R$ with $R^2=1$,
 the partial reflection combined with the $\pm \pi/2$ $U(1)$ charge rotation gives rise to the correct $\Z_4$ invariant.~\cite{shiozaki2016many}
In Sec.~\ref{sec:(1+1)AIII}, we will see that the same $\pm i$ phase rotation is involved in the definition of the $\Z_4$ SPT invariant in $(1+1)d$ class AIII systems. 
In Appendix~\ref{sec:pin_str_pr2_eta}, we present the detailed calculation of the $\Z_4$-quantized eta invariant on $\mathbb{R}P^2$ with pin$^c$ structures.

\begin{figure*}
\includegraphics[width=0.7\linewidth, trim=0cm 0cm 0cm 0cm]{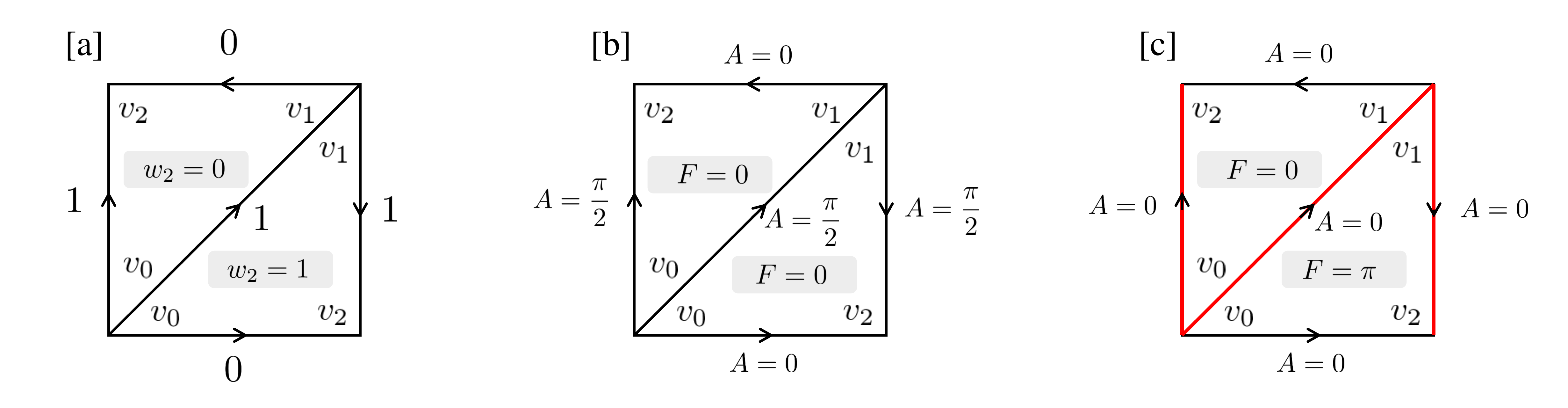}
\caption{\label{fig:rp2_dirac_conditions}
A simplicial complex of $\R P^2$. 
[a] $\Z_2$ numbers assigned in the bonds refer to a representative of orientation bundle $w_1 \in H^1(\R P^2;\Z_2)$. 
A representative for $w_2$ is given as $w_2 = w_1^2$.
[b] A pin$^c$ connection on $\R P^2$. 
[c] A pin$^{\tilde c}_+$ connection on $\R P^2$. The red lines refer to a representative of orientation bundle $w_1$. 
}
\end{figure*}


\subsubsection{Ex: pin$^{\tilde c}_+$ structure on $\mathbb{R}P^2$}
\label{sec:Ex:pinc+ structure on RP2}

In Bosonic $U(1) \rtimes T$ and pin$^{\tilde c}_{\pm}$ structures, 
thanks to the minus sign accompanied with an orientation-reversing patch transformation for $U(1)$ fields, 
the field strength (in the sense of Euclidean spacetime) can be finite. 
Interestingly, the Dirac quantization condition of pin$^{\tilde c}_+$
connections
on $\mathbb{R}P^2$ is nontrivial, which can be compared with that of pin$^{\tilde c}_-$ and $U(1) \rtimes T$ as 
\begin{align}
&\tilde c \equiv \frac{1}{2} w_2(T\mathbb{R}P^2) \qquad {\rm for\ } {\rm pin}^{\tilde c}_+, \label{eq:dirac_cond_pin_tilde_c+_rp2} \\
&\tilde c \equiv 0 \qquad {\rm for\ } {\rm pin}^{\tilde c}_- \ {\rm and\ Bosonic}  \ U(1) \rtimes T. 
\end{align}
Because $H^2(\mathbb{R}P^2;\wt \Z) \cong \Z$ is free,
the Dirac quantization condition (\ref{eq:dirac_cond_pin_tilde_c+_rp2})
implies that the monopole charge inside $\mathbb{R}P^2$ is quantized to half integers 
\begin{align}
\int_{\mathbb{R}P^2} \frac{F|_{{\rm pin}^{\tilde c}_+}}{2 \pi} \in \Z + \frac{1}{2}. 
\label{eq:magnetic_monopole_rp2_pinc+}
\end{align}
Correspondingly, the monopole charge on the double cover $S^2$ of $\mathbb{R}P^2$ is given by an odd integer. 
An explicit pin$^{\tilde c}_+$ connection $A_{{\rm pin}^{\tilde c}_+}$
with $\int_{\mathbb{R}P^2} F_{{\rm pin}^{\tilde c}_+}/2 \pi = 1/2$ is shown in Fig.~\ref{fig:rp2_dirac_conditions}[c], 
where the field strength threading a simplex is given by $F_{{\rm pin}^{\tilde c}_+} = 2 \pi \tilde c + \delta_{w} A_{{\rm pin}^{\tilde c}_+}$. 
($\delta_w$ is the twisted differential. See Appendix~\ref{app:Cohomology with local coefficient}.)

The existence of the half integer monopole is consistent with that the charge pump is quantized to even integers due to the Kramers degeneracy in the presence of TRS with $T^2 = (-1)^F$. 
The pin$^{\tilde c}_+$ cobordisms in 2d is given by $\Omega^{\pin^{\tilde
    c}_+}_2 = \Z$~\cite{Freed2016}
which is generated by $\mathbb{R}P^2$ with $1/2$ monopole. 
The cobordism invariant topological action is theta term 
\begin{align}
Z[X,A] = \exp \left[  i \theta \int_X \frac{F}{2 \pi} \right]. 
\label{eq:theta_term_rp2_pinc+}
\end{align}
In this normalization of $\theta$, the periodicity of $\theta$ is $4 \pi$ since the minimum magnetic flux is $\int_{\mathbb{R}P^2} F/2 \pi = 1/2$. 
Put differently, a TQFT describing gapped phases of complex fermions in $(1+1)d$ with TRS with $T^2 = (-1)^F$ is labeled by $\theta \in \R/4 \pi \Z$. 
Recall that the adiabatic charge pump occurs when there is a closed cycle in the parameter space. 
Because of the $4 \pi$ periodicity, the adiabatic pump is quantized to even integers, 
\begin{align}
Q_{\rm adia} = \frac{1}{2 \pi} \oint_{\rm one\ period} d \theta \in 2 \Z, 
\end{align}
as expected. 

In Appendices~\ref{sec:pinc_-_str_pr2_theta} and \ref{sec:pinc_+_str_pr2_theta}, we explicitly compute the theta terms for pin$^{\tilde c}_-$ and pin$^{\tilde c}_+$ structures, respectively, and confirm the half monopole flux in $\mathbb{R}P^2$ for the pin$^{\tilde c}_+$ structure.

\section{Fermionic partial transpose and partial antiunitary transformations}
\label{Fermionic partial transpose and partial antiunitary transformations}

In this section, we are concerned with partial transpose and related operations in
fermionic systems.
These operations are necessary technologies to construct many-body topological invariants of
fermionic SPT phases in Sec.~\ref{sec:many-body_invariant}.

For bosonic systems, the partial transpose of an operator is defined as follows.
Here, for definiteness, we consider the (reduced) density matrix $\rho_I$ of a many-body
bosonic system.
Let us consider a bosonic Hilbert space $\mathcal{H}$ and the sub Hilbert space $\mathcal{H}_I$.
The reduced density matrix $\rho_I$ 
is obtained from the density matrix $\rho_{tot}$ by tracing out all degrees of freedom in the complement
$\mathcal{H}_{\bar{I}}$,
$\rho_I = \mathrm{Tr}_{\bar{I}}\rho_{tot}$. 
We further consider the bipartition (the tensor factorization) of $\mathcal{H}_I$ to two sub Hilbert
spaces, 
$\mathcal{H}_I =\mathcal{H}_{I_1}\otimes \mathcal{H}_{I_2}$.
The partial transpose of the density matrix $\rho_I$ with respect to $\mathcal{H}_{I_1}$,
denoted by $\rho^{\sf T_1}_I$,
is defined by 
\begin{align} \label{eq:transdef}
\rho_I^{\sf T_1}= \sum_{ijkl} \ket{e_i^1,e_j^2} \braket{e_k^1,e_j^2|\rho_I |e_i^1,e_l^2} \bra{e_k^1,e_l^2}, 
\end{align}
where $\ket{e_j^{1}}$ and $\ket{e_k^{2}}$ denote an orthonormal set of states in the $I_1$ and $I_2$ regions.
The (bosonic) partial transpose has been used as a way to diagnose quantum
entanglement,
e.g., to define entanglement measures such as the entanglement negativity,
and also 
to construct topological invariants for bosonic SPT phases.~\cite{Pollmann2012}

We should note that the partial transposition introduced in
Eq.~(\ref{eq:transdef}), 
which is simply swapping indices of the first interval, 
is strictly defined for bosons.
In order to define a consistent definition for fermions, 
one needs to take into account the anti-commuting property of fermion operators.
Later in this section, we show that the fermionic definition is equivalent to matrix transposition up to a phase factor. 

\subsection{
\label{sec:Fermionicpartialtranspose} Fermionic partial transpose}

In this section, we first present a definition of the fermionic partial transpose in a basis in terms of real fermions and second recast this definition as a change of Grassmann variables in the fermionic coherent state representation. 
Using the latter definition, we derive the transformation rules in the occupation number basis. 
In later sections, we use the first approach in the Majorana basis to carry out analytical calculations for the fixed point wave functions and implement a numerical method to compute the partial transpose based on the second approach in the coherent state representation.
Our convention of the fermion coherent state is summarized in Appendix \ref{app:Fermion coherent state}. 


\subsubsection{Fermionic Fock space}
Let $f_j (j=1,\dots,N)$ be complex fermion operators satisfying the anticommutation relations 
\begin{align}
\{f_j^{\ },f_k^{\dag}\}=\delta_{jk}, \quad 
\{f_j^{\ },f_k\}=\{f_j^{\dag},f^{\dag}_k\}=0. 
\end{align}
The Fock vacuum $\ket{{\rm vac}}$ of the complex fermion operators $f^{\ }_j$ is defined
as a state that is annihilated by all $f_j$,
$f_j\ket{\rm vac}=0$ for $j=1, \dots, N$.  
The Fock space ${\cal F} = \bigoplus_{a=1}^N {\cal F}_a$ is spanned in the occupation bases by
\begin{align}
\ket{n_1n_2\dots n_N}=\ket{\{n_j\}}:=(f^{\dag}_1)^{n_1}(f^{\dag}_2)^{n_2} \cdots(f^{\dag}_N)^{n_N}\ket{\rm vac}. 
\end{align}
The fermion number parity operator $(-1)^F$ is defined by 
\begin{align}
(-1)^F:=\prod_{j=1}^N(-1)^{f^{\dag}_jf^{\ }_j}, \quad 
(-1)^F\ket{\rm vac}=\ket{\rm vac}. 
\end{align}
In the present paper, 
we always assume the fermion number parity symmetry, 
which means terms composed of odd numbers of fermions such as $f^{\ }_j$, $f^{\dag}_j f^{\ }_k f^{\ }_l$ are excluded in Hamiltonians. 
As a consequence, we can focus only on states $\ket{\phi}$ in the Hilbert space $\mathcal{F}$ with a definite fermion number parity, 
$(-1)^{F} \ket{\phi} = \pm \ket{\phi}$. 

For a fermionic Fock space ${\cal F}$ generated by $N$ complex fermion operators, 
we can introduce real fermion operators by 
\begin{align}
c_{2j-1}=f^{\dag}_j+f_j, \quad 
c_{2j}=-i(f^{\dag}_j-f_j), \quad 
j=1, \dots, N. 
\label{eq:real_fermion}
\end{align}
Any operator $A$ acting on the Fock space ${\cal F}$ is expanded in terms of the real
fermion operators $c_j (j=1, \dots, 2N)$ as 
\begin{align}
A
= \sum_{k=1}^{2N} \sum_{p_1<p_2 \cdots <p_k} A_{p_1 \cdots p_k} c_{p_1} \cdots c_{p_k}, 
\label{eq:op_expand}
\end{align}
where $A_{p_1 \dots p_k}$ are complex numbers, and are fully antisymmetric under permutations of $\{1, \dots, k\}$. 
The fermion number parity in terms of the real fermion operators is given by 
\begin{align}
(-1)^F = \prod_{j=1}^N (i c_{2j-1} c_{2j}). 
\end{align}

\subsubsection{Definition of fermionic partial transpose
\label{Definition of fermionic partial transpose}
}

A standard anti-automorphism $A \mapsto A^{\sT}$ of a Clifford algebra is defined
by reversing
the ordering of generators 
in the expansion \eqref{eq:op_expand} as 
$(c_{p_1}c_{p_2} \dots c_{p_k})^{\sf T} = c_{p_k} \cdots c_{p_2}c_{p_1}$.
This definition is extended linearly to general elements in the Clifford algebra. 
This operation is involutive $(A^{\sf T})^{\sT}=A$, linear $(zA)^{\sT}=z A^{\sT}$ for a $\C$ number $z$,
and satisfy $(AB)^{\sT}=B^{\sT}A^{\sT}$.
We hence call it {\it transpose}. 

Let us now  divide the real fermion operators $I = \{c_1, \dots, c_{2N}\}$
into two subsystems $I_1 \cup I_2 = \{c_1, \dots, c_{2N_1} \} \cup \{c_{2N_1+1},
\dots, c_{2N}\} =: \{a_1, \dots, a_{2N_1}\} \cup \{b_1, \dots, b_{2N_2} \}$,
where $N_1+N_2=N$. 
We would like to define a kind of ``partial'' transpose for the subsystem $I_1$ (and $I_2$), 
{\it focusing only on the subalgebra consisting of operators which preserve the fermion number parity}.~\footnote{
Notice that the partial trace on a fermion number parity preserving operator also preserves the fermion number parity in the reduced Fock space. 
I.e., if $A$ commutes with the fermion number parity $(-1)^{F_I}$ of the total system $I$, 
then the partially traced operator $A_{I_1} = \tr_{I \backslash I_1} A$ also commutes with the 
fermion number parity $(-1)^{F_{I_1}}$ on the subsystem $I_1$.}
For a fermion number parity preserving operator $A$,
the expansion \eqref{eq:op_expand} consists of operators with even number
of real fermion operators, 
\begin{align}
A = \sum_{k_1,k_2}^{k_1+k_2 = {\rm even}} A_{p_1 \cdots p_{k_1}, q_1 \cdots q_{k_2}} a_{p_1} \cdots a_{p_{k_1}} b_{q_1} \cdots b_{q_{k_2}}. 
\end{align}
We introduce a linear transformation $A \mapsto A^{\sT_1}$ ($A^{\sT_2}$) 
on the Clifford algebra as a vector space 
(but $A^{\sf T_1}$ is not an algebra homomorphism)
for the subsystem $I_1$ ($I_2$) so as to satisfy:

\bigskip 

\begin{itemize}
\item[{\bf 1:}] The successive transformation on $I_1$ and $I_2$ recovers the full transpose $(A^{\sT_1})^{\sT_2}=A^{\sT}$. 
\end{itemize}

\bigskip 

\noindent In addition, we demand the following ``good'' properties: 

\bigskip 

\begin{itemize}
\item[{\bf 2:}] Preserving identity 
\begin{align}
({\rm Id})^{\sf T_1} = {\rm Id}.
\end{align}
\item[{\bf 3:}] Under inner automorphisms $U$ which induce $O(2N_1) \times O(2N_2)$ rotations of real fermion operators as 
$U a_j U^{\dag} = [{\cal O}_1]_{jk} a_k$ and $U b_j U^{\dag} = [{\cal O}_2]_{jk} b_k$, 
the linear transformation $A^{\sT_1}$ satisfies~\footnote{
There is no analog of \eqref{eq:partial_tr_condition2} for the (partial) transpose in bosonic systems. 
In bosonic systems, the operator algebra is a Weyl algebra, which is simple. 
If we focus on a finite dimensional Fock space, the algebra is a matrix algebra. 
In matrix algebra, every linear anti-automorphism $\phi(A)$ is written in a form $\phi(A)=VA^{tr}V^{-1}$. 
Under the change of basis $\ket{n} \mapsto \ket{m} U_{mn}$ of the many body Fock space, 
the linear anti-automorphism is changed as $\phi(U^{\dag}AU)=U^{\dag} (UVU^{tr}) A^{tr} (UVU^{tr})^{-1} U$, which is not basis independent. 
Interestingly, the change $V \mapsto UVU^{tr}$ is the same form of the unitary matrix $V$ associated with time-reversal symmetry. 
In other words, the partial transpose in bosonic (spin) systems is really considered as a partial time-reversal transformation! 
See Sec.~\ref{sec:Partial time-reversal transformation in bosons}. }
\begin{align}
(U A U^{\dag})^{\sT_1} = U A^{\sT_1} U^{\dag}.
\label{eq:partial_tr_condition2}
\end{align}
\end{itemize}

\bigskip 

\noindent
The third condition largely restricts a possible form of $A^{\sT_1}$. 
Let $F_{{\sT}_1}(A)=A^{{\sT}_1}$ and $F_U(A)=U A U^{\dag}$ be the linear transformations as a vector space. 
The condition \eqref{eq:partial_tr_condition2} implies that these are commutative $F_{{{\sT}_1}}F_{U} = F_U F_{{\sT}_1}$ for any 
$O(2N_1) \times O(2 N_2)$ rotations $U$. 
Noticing that a $O(2N_1) \times O(2 N_2)$ rotation does not change the numbers $k_1$ and $k_2$ of real fermions because 
$A_{p_1 \cdots p_{k_1}, q_1 \cdots q_{k_2}}$ is fully anti-symmetric, 
the unitary transformation $F_{U}$ is decomposed into a block diagonal form 
$F_U=\bigoplus_{k_1+k_2={\rm even}}F_{U}(k_1,k_2)$, 
where $F_{U}(k_1,k_2)$ is a unitary transformation acting on the vector space spanned by 
the bases $\{a_{p_1}, \cdots, a_{p_{k_1}}, b_{q_1}, \cdots, b_{q_{k_2}}\}$. 
$F_{\sf T_1} F_U = F_U F_{\sf T_1}$ implies that $F_{\sf T_1}$ is also block diagonal. 
From the Schur's lemma, $A^{{\sT}_1}$ is a scalar multiplication which can depend on $k_1$: 
\begin{align}
(a_{p_1} \cdots a_{p_{k_1}} b_{q_1} \cdots b_{q_{k_2}})^{{\sT}_1}
  =z_{k_1} a_{p_1} \cdots a_{p_{k_1}} b_{q_1} \cdots b_{q_{k_2}},
\end{align}
where $z_{k_1} \in \C$ is a complex number. 
Then, the first and second conditions demand 
\begin{align}
z_0=1, \quad 
z_{k_1} z_{k_2}= \left\{\begin{array}{ll}
-1 & (k_1+k_2 =2 \ \mod\ 4), \\
1 & (k_1+k_2 = 0 \ \mod\ 4). \\
\end{array}\right.
\label{eq:partial_tr_phase_condition}
\end{align}
There are two solutions $z_k=(\pm i)^k$ $(k =0,1 \dots)$. 
We employ the convention of $z_k=i^k$ in this paper.~\footnote{
Alternative choice $z_k=(-i)^k$ is unitary equivalent to $z_k= i^k$. 
These are related by the partial fermion parity flip $(-1)^{F_1}$ defined in the second equation in (\ref{eq:pt_double}).
}
It should be noticed that the restriction into the subalgebra consisting of 
fermion number parity conserving operators is crucial to obtain a 
nontrivial solution of (\ref{eq:partial_tr_phase_condition}).~\footnote{
In fact, if we consider all possible operators, the condition (\ref{eq:partial_tr_phase_condition}) is ``enlarged'' as 
$z_{k_1} z_{k_2} = -1$ for $k_1+k_2 =1,2 \ \mod\ 4$ and 
$z_{k_1} z_{k_2} = 1$ for $k_1+k_2 = 0,3 \ \mod\ 4$, but this has no solution. 
}

\bigskip

\begin{dfn}[Fermionic partial transpose]
The fermionic partial transpose $A^{\sf T_1}$ is the following linear transformation 
on the subalgebra of operators preserving fermion number parity, 
\begin{align}
A^{\sf T_1} := \sum_{k_1,k_2}^{k_1+k_2 = {\rm even}} A_{p_1 \cdots p_{k_1}, q_1 \cdots q_{k_2}} i^{k_1} a_{p_1} \cdots a_{p_{k_1}} b_{q_1} \cdots b_{q_{k_2}}. 
\label{def:fermion_pt}
\end{align}
\end{dfn}

\bigskip

\noindent 
By definition, it follows that 

\begin{itemize}
\item[{\bf 4:}] The full fermionic transpose $A^{\sf T}=(A^{\sf T_1})^{\sf T_2}$ is an anti-automorphism 
\begin{align}
(AB)^{\sf T}=B^{\sf T} A^{\sf T}.
\end{align}
\item[{\bf 5:}] The fermionic partial transpose preserves the tensor product structure of systems: 
Let $A_i (i=1,2)$ be fermion parity preserving operators on fermionic Fock spaces ${\cal F}_i (i=1,2)$. 
It holds that 
\begin{align}
(A_1 \otimes A_2)^{\sf T_1} = A_1^{\sf T_1} \otimes A_2^{\sf T_1}. 
\end{align}
\item[{\bf 6:}]
The successive fermionic partial transpose is the partial fermion parity flip 
\begin{align}
(A^{\sf T_1})^{\sf T_1} = (-1)^{F_1} A (-1)^{F_1}, \quad 
(-1)^{F_1}=\prod_{j \in I_1}(i c_{2j-1} c_{2j}). 
\label{eq:pt_double}
\end{align}
\end{itemize}

It should be noted that our definition of 
fermionic partial transpose is different from, e.g., 
one presented in Ref.\ \cite{Eisler2015}. 
There, 
fermionic partial transpose is defined by~\footnote{
There is a typo in the original definition in \cite{Eisler2015}. 
The correct prior definition is found to be (\ref{eq:ferm_part_2})~\cite{CoserTonniCalabrese}.
}
\begin{align}
  A^{\wt{\sf T}_1} = \sum_{k_1,k_2}^{k_1+k_2 = {\rm even}} A_{p_1 \cdots p_{k_1}, q_1 \cdots q_{k_2}}
  (-1)^{f(k_1)}
  a_{p_1} \cdots a_{p_{k_1}} b_{q_1} \cdots b_{q_{k_2}} 
\label{eq:fpt_Eisler}
\end{align}
where
\begin{align}
f(k_1)
= \left\{ \begin{array}{ll}
0 & \text{if}\ k_1 \; \mathrm{mod} \; 4 \in \{0,1\},\\
1 & \text{if}\ k_1 \; \mathrm{mod} \; 4 \in \{2,3\}.
\end{array} \right.
\label{eq:ferm_part_2}
\end{align}
%
and discuss the entanglement negativity. See also Ref.~\cite{CoserTonniCalabrese} for the field theoretic description. 
Our definition (\ref{eq:fpt_Eisler}) of the fermionic partial transpose is {\it unitary inequivalent} to (\ref{eq:fpt_Eisler}). 
A crucial difference from the prior definition (\ref{eq:fpt_Eisler}) is that in our definition (\ref{def:fermion_pt}) the fermionic partial transpose $A^{\sf T_1}$ of a fermionic Gaussian state $A$ is Gaussian again, whereas the prior definition (\ref{eq:fpt_Eisler}) gives rise to a sum of two Gaussian operators. 
See Appendix A in Ref.~\cite{shapourian2016partial} for details. 
(In Ref.~\cite{shapourian2016partial}, we have called $A^{\sf T_1}$ ``partial time reversal'' and used a different symbol $A^{R_1}$. In the present paper, we decided to call it fermionic partial transpose.)
Put differently, our definition (\ref{def:fermion_pt}) results in the Euclidean path integral with a single pin$_{\pm}$ structure, whereas the prior definition (\ref{eq:fpt_Eisler}) leads to the sum of pin$_{\pm}$ structures.~\cite{CoserTonniCalabrese}


\bigskip

\begin{widetext}
\subsubsection{Fermionic partial transpose in coherent state and occupation number bases}

Let us derive the expression of the fermionic partial transpose in the coherent state basis. 
Any operator $A$ preserving fermion parity can be 
expanded by the use of fermion coherent states as 
\begin{align}
A
= \int \prod_{j \in I} d \bar \xi_j d \xi_j d \bar \chi_j d \chi_j e^{- \sum_{j \in I} (\bar \xi_j \xi_j + \bar \chi_j \chi_j)} \braket{\{\bar \xi_j\}_{j \in I}|A|\{\chi_j\}_{j \in I}} 
\ket{\{\xi_j\}_{j \in I}} \bra{\{\bar \chi_j\}_{j \in I}}, 
\end{align}
where 
\begin{align}
\ket{\{\xi_j\}_{j \in I}} = e^{-\sum_{j \in I}\xi_j f^{\dag}_j}\ket{\rm vac}, \quad 
\bra{\{\chi_j\}_{j \in I}} = \bra{\rm vac} e^{-\sum_{j \in I}f_j \chi_j}
\end{align}
are the standard fermion coherent states. 
Thanks to the anticommutation relations of Grassmannian variables, 
the basis of operators in the coherent states decomposes into a product of operators labeled by $j \in I$ 
\begin{equation}\begin{split}
\ket{\{\xi_j\}_{j \in I}} \bra{\{\chi_j\}_{j \in I}}
&=\prod_{j \in I} \Big[ e^{-\xi_j f^{\dag}_j} (1-f^{\dag}_j f^{\ }_j) e^{-f_j \chi_j} \Big] \\
&=\prod_{j \in I} \left[ \frac{1+\xi_j \chi_j}{2} - \frac{\xi_j- \chi_j}{2} c_{2j-1} - i \frac{\xi_j+ \chi_j}{2} c_{2j} + \frac{i(1-\xi_j \chi_j)}{2} c_{2j-1} c_{2j} \right], 
\end{split}\end{equation}
where we have introduced the real fermions operators by (\ref{eq:real_fermion}). 
The fermionic partial transpose on the coherent state 
$\ket{\{\xi_j\}_{j \in I}} \bra{\{\chi_j\}_{j \in I}}$
reads
\begin{equation}\begin{split}
&\big( \ket{\{\xi_j\}_{j \in I}} \bra{\{\chi_j\}_{j \in I}} \big)^{{\sT}_1} \\
&=\prod_{j \in I_1} \left[ \frac{1+\xi_j \chi_j}{2} - \frac{\xi_j- \chi_j}{2} (ic_{2j-1}) - i \frac{\xi_j+ \chi_j}{2} (ic_{2j}) + \frac{i(1-\xi_j \chi_j)}{2} (-c_{2j-1} c_{2j}) \right] \\
& \ \ \cdot 
\prod_{j \in I_2} \left[ \frac{1+\xi_j \chi_j}{2} - \frac{\xi_j- \chi_j}{2} c_{2j-1} - i \frac{\xi_j+ \chi_j}{2} c_{2j} + \frac{i(1-\xi_j \chi_j)}{2} c_{2j-1} c_{2j} \right] \\
&= C_f^{I_1} \ket{\{-i\chi_j\}_{j\in I_1},\{\xi_j\}_{j \in I_2}} \bra{\{-i\xi_j\}_{j\in I_1},\{\chi_j\}_{j \in I_2}} [C_f^{I_1}]^{\dag}, 
\end{split}\end{equation}
where $C_f^{I_1}$ is the partial particle-hole transformation $C_f^{I_1}=\prod_{j \in I_1} c_{2j-1} = \prod_{j \in I_1} (f^{\dag}_j+f_j)$ 
which exchanges the occupied and empty states in the subsystem $I_1$ as 
\begin{align}
C_f^{I_1} f^{\dag}_{j \in I_1} [C_f^{I_1}]^{\dag}=\left\{\begin{array}{ll}
f_{j \in I_1} & (N_1: {\rm odd}) \\
-f_{j \in I_1} & (N_1: {\rm even}) \\
\end{array}\right., \quad 
C_f^{I_1} f^{\dag}_{j \in I_2} [C_f^{I_1}]^{\dag}=\left\{\begin{array}{ll}
-f^{\dag}_{j \in I_2} & (N_1: {\rm odd}) \\
f^{\dag}_{j \in I_2} & (N_1: {\rm even}) \\
\end{array}\right. . 
\end{align}
The subscript of $C^{I_1}_f$ indicates that $C_f^{I_1}$ is the 
particle-hole transformation of the $f_j$ fermions.~\footnote{
$C_f^{I_1}$ is basis dependent.}
Summarizing, we have obtained:  

\bigskip

\begin{itemize}
\item[{\bf 7:}](Fermionic partial transpose in the coherent state basis) 
In the fermion coherent basis, the fermionic partial transpose is given by 
\begin{equation}\begin{split}
\big( \ket{\{\xi_j\}_{j \in I}} \bra{\{\chi_j\}_{j \in I}} \big)^{{\sT}_1}
&= C_f^{I_1} \ket{\{-i\chi_j\}_{j\in I_1},\{\xi_j\}_{j \in I_2}} \bra{\{-i\xi_j\}_{j\in I_1},\{\chi_j\}_{j \in I_2}} [C_f^{I_1}]^{\dag}, 
\end{split}\label{eq:pt_coherent}\end{equation}
where $C_f^{I_1}=\prod_{j \in I_1} c_{2j-1}$ is the partial particle-hole transformation on $I_1$. 
\end{itemize}

\bigskip


Finally, it is easy to obtain the expression of the fermionic partial transpose in the occupation number basis. 
Let us label the degrees of freedom by $I_1=\{1, \dots, L\}, I_2=\{L+1, \dots, N\}$. 
(The fermionic partial transpose is free from the ordering of degrees of freedom.)
Comparing the following two expressions of coherent states in terms of occupation basis, 
\begin{align}
&\ket{\{\xi_j\}_{j\in I_1},\{\xi_j\}_{j\in I_2}}\bra{\{\chi_j\}_{j \in I_1},\{\chi_j\}_{j \in I_2}} 
\nonumber \\
&=\sum_{\{n_j\},\{\bar n_j\}}
(-\xi_N)^{n_N}\cdots(-\xi_1)^{n_1}
\ket{\{n_j\}_{j \in I}}\bra{\{\bar n_j\}_{j \in I}}
(-\chi_1)^{\bar n_1}\cdots(-\chi_N)^{\bar n_N}, 
\\
&\ket{\{ -i \chi_j \}_{j \in I_1},\{\xi_j\}_{j \in I_2}} \bra{\{-i \xi_j\}_{j \in I_1}, \{\chi_j\}_{j \in I_2}} 
\nonumber \\
&=\sum_{\{n_j\},\{\bar n_j\}}
(-\xi_N)^{n_N}\cdots(-\xi_{L+1})^{n_{L+1}}(i\chi_{L})^{n_L}\cdots(i\chi_{1})^{n_{1}} 
\nonumber \\
& \quad \cdot \ket{\{n_j\}_{j \in I}}\bra{\{\bar n_j\}_{j \in I}}
(i\xi_{1})^{\bar n_1}\cdots(i\xi_{L})^{\bar n_L}(-\chi_{L+1})^{\bar n_{L+1}}\cdots(-\chi_N)^{\bar n_N}, 
\end{align}
we obtain: 

\bigskip 

\begin{itemize}
\item[{\bf 8:}](Fermionic partial transpose in the occupation number basis) 
In the occupation number basis, the fermionic partial transpose is given by~\footnote{
A straightforward calculation leads to the phase 
$(-i)^{\tau_1+\bar \tau_1} (-1)^{\tau_1+\bar \tau_1+\tau_1\bar \tau_1+\sum_{j<k, j,k \in I_1}n_jn_k+\sum_{j<k, j,k \in I_1}\bar n_j\bar n_k+(\tau_1+\bar \tau_1)(\tau_2+\bar \tau_2)} = i^{[\tau_1+\bar \tau_1]}(-1)^{(\tau_1+\bar \tau_1)(\tau_2+\bar \tau_2)}$.
}
\begin{equation}\begin{split}
&\big( \ket{\{n_j\}_{j \in I_1},\{n_j\}_{j \in I_2}}\bra{\{\bar n_j\}_{j \in I_1},\{\bar n_j\}_{j \in I_2}} \big)^{{\sT}_1}  \\
&= i^{[\tau_1+\bar \tau_1]}(-1)^{(\tau_1+\bar \tau_1)(\tau_2+\bar \tau_2)} 
C_f^{I_1}\ket{\{\bar n_j\}_{j \in I_1},\{n_j\}_{j \in I_2}}\bra{\{n_j\}_{j \in I_1},\{\bar n_j\}_{j \in I_2}} [C_f^{I_1}]^{\dag},
\label{eq:pt_occupation}
\end{split}\end{equation}
where we have introduced the notations 
\begin{align}
\tau_1=\sum_{j\in I_1}n_j,\quad 
\tau_2=\sum_{j\in I_2}n_j,\quad 
\bar\tau_1=\sum_{j\in I_1}\bar n_j,\quad 
\bar\tau_2=\sum_{j\in I_2}\bar n_j, \quad 
[\tau]=\left\{\begin{array}{ll}
0 & (\tau: {\rm even})\\
1 & (\tau: {\rm odd})\\
\end{array}\right.
\end{align}
\end{itemize}


\end{widetext}

\subsection{Partial anti-unitary symmetry transformations}
As an application of the fermionic partial transpose, 
in this section, we introduce the fermionic partial time-reversal transformation 
and fermionic partial anti-unitary particle-hole transformation. 

\subsubsection{Partial time-reversal transformation}
Time-reversal symmetry (TRS) $T$ is a $\Z_2$ symmetry in a physical ray space (projective Hilbert space) 
which preserves the fermion number and represented by an antiunitary operator in the Fock space ${\cal F}$. 
The action of $T$ can be written as 
\begin{align}
Tf^{\dag}_jT^{-1}
=f^{\dag}_k[\cU_T]_{kj}, \quad 
T\ket{\rm vac}=\ket{\rm vac} 
\label{eq:trs_fermion_def}
\end{align}
with a unitary matrix $\cU_T \in U(N)$. 
Under the change of basis $f^{\dag}_j=g^{\dag}_k {\cal V}_{kj}$ with ${\cal V} \in U(N)$, 
$\cU_T$ is transformed as $Tg^{\dag}_jT^{-1}=g^{\dag}_k[{\cal V}\cU_T {\cal V}^{tr}]_{kj}$.  
That $T$ is antiunitary means that $T$ acts on a state 
\begin{align}
\ket{\phi}
= \sum_{\{n_j\}} \phi(\{n_j\}) (f^{\dag}_1)^{n_1} \cdots (f^{\dag}_N)^{n_N} \ket{{\rm vac}} 
\end{align}
in the Fock space ${\cal F}$ by taking 
the complex conjugation $\phi^*(\{n_j\})$ of the wave function $\phi(\{n_j\})$ 
and changing the basis of the complex fermions in (\ref{eq:trs_fermion_def}). 
It is easy to show that 
\begin{align}
T \ket{\phi}
= \sum_{\{n_i\},\{n'_i\}, \tau=\tau'} \phi^*(\{n_i\}) \det {\cal U}_T(\{n'_j\},\{n_j\}) \ket{\{n'_j\}}, 
\label{eq:def_TRS}
\end{align}
where $\tau, \tau'$ is the fermion number $\tau=\sum_{j} n_j, \tau'=\sum_j n_j'$ 
and the $\tau \times \tau$ matrix ${\cal U}_T(\{n'_j\},\{n_j\})$ is defined by restriction of the elements of 
${\cal U}_T$ into labels with $n_j=1$ and $n_j'=1$. 

The order-two symmetry on the physical ray space
implies 
that $T^2$ is the multiplication by a pure phase 
which can depend on the sectors
of even and odd fermion number parity   
since we exclude superposition between these.
One can show that there are two possibilities 
\begin{align}
&
T^2=1: \quad {\cal U}_T^{tr}={\cal U}_T, 
\\
&
T^2=(-1)^F: \quad  {\cal U}_T^{tr}=-{\cal U}_T.
\end{align}
In the absence of any other symmetries,
following the terminology of the Altland-Zirnbauer symmetry classes,~\cite{altland1997nonstandard}
we call the former case class BDI and the latter class DIII. 
In Euclidean quantum field theories, 
the former one corresponds to pin$_-$ structures 
whereas the latter corresponds to pin$_+$ structures. 

To construct the partial time-reversal transformation, 
we first introduce the ``unitary part'' $C_T$ associated with TRS $T$ 
as the unitary operator of particle-hole type on the Fock space and defined so as to satisfy 
\begin{align}
C_Tf_jC_T^{\dag}=f^{\dag}_k[\cU_T]_{kj}, \quad 
C_T\ket{\rm vac} \sim \ket{\rm full}=f^{\dag}_1 \cdots f^{\dag}_N\ket{\rm vac}.
\label{eq:unitary_part_T}
\end{align}
The first condition uniquely fixes the unitary operator $C_T$ (up to a constant phase) and 
the second equation follows from the first. 
The point of this definition of $C_T$ is that (\ref{eq:unitary_part_T}) is independent of the basis: 
the change of basis $f^{\dag}_j=g^{\dag}_k {\cal V}_{kj}$ induces the same transformation on 
the unitary matrix $C_Tg_jC_T^{\dag}=g^{\dag}_k[{\cal V}\cU_T {\cal V}^{tr}]_{kj}$ as that of $T$.~\footnote{
There is a subtle point in the ``complex conjugation'' $K$. 
Let us consider TRS with $\cU_T=1$. 
We may write ``$T=K$'', however, under the definition (\ref{eq:def_TRS}) of TRS, 
the unitary matrix $\cU_T$ changes in general. 
In fact, the change of the basis $f^{\dag}_j \mapsto f^{\dag}_k {\cal V}_{kj}$ induces $\cU_T \mapsto {\cal V} {\cal V}^{tr}$ 
which is not identity if ${\cal V}$ is not real.  
The precise meaning of ``$K$'' is the outer automorphism taking the complex conjugate $A \mapsto A^{*}$ on the Clifford algebra. 
The ``unitary part $U$'' of TRS $T$ in the notation like ``$T=UK$'' 
means the combination of the inner automorphism $U$ and the outer automorphism $K$. }
Let us derive the explicit form of $C_T$ for each class (BDI and DIII) below. 

\bigskip

{\it Class BDI}---
In this case, $T^2=1$ and  
the symmetric unitary matrix $\cU_T$ can be written in a canonical form as 
$\cU_T = {\cal Q} {\cal Q}^{tr}$ with ${\cal Q}$ a unitary matrix.~\footnote{
This is simply because for every complex symmetric matrix $A$ there exists a unitary matrix ${\cal Q}$ such that $A = {\cal Q} \Lambda {\cal Q}^T$ where $\Lambda$ is a diagonal matrix.}
Introducing the new complex fermion operators $g^{\dag}_j = f^{\dag}_k {\cal Q}_{kj}$, 
the action of $T$ can be written as 
\begin{align}
&T g^{\dag}_j T^{-1} = g^{\dag}_j, \quad 
T\ket{\rm vac}=\ket{\rm vac}. 
\end{align}
Let us introduce the real fermion operators $c_j$ $(j=1, \dots 2N)$ by 
\begin{align}
c_{2j-1}=g^{\dag}_j+g_j, \quad 
c_{2j}=-i(g^{\dag}_j-g_j). 
\end{align}
Then, the unitary part $C_T$ of $T$ is given by 
\begin{align}
C_T=\left\{\begin{array}{ll}
\prod_{j=1}^N c_{2j-1}  & (N: {\rm odd}) \\
\prod_{j=1}^N c_{2j}  & (N: {\rm even}) \\
\end{array}\right. 
\end{align}
Notice that $C_T$ is Grassmann odd (even) if $N$ is odd (even). 

\bigskip

{\it Class DIII}---
In this case,
$T^2=(-1)^F$ and 
the antisymmetric unitary matrix $\cU_T$ can be written in a form 
\begin{align}
\cU_T = {\cal Q} \left[ \begin{pmatrix}
0 & 1 \\
-1 & 0 \\
\end{pmatrix} \otimes 1_{N \times N} \right] {\cal Q}^{tr}
\end{align}
with a unitary matrix ${\cal Q}$.~\footnote{
This fact can be easily understood since every complex skew-symmetric matrix $A$ can be written in a form $A = {\cal Q} \Sigma {\cal Q}^{tr}$ where ${\cal Q}$ is a unitary matrix and 
$\Sigma$ has a block diagonal form $\Sigma = \bigoplus_{i=1}^N \begin{pmatrix}
0 & \lambda_i \\
- \lambda_i & 0 \\
\end{pmatrix}$ with $\lambda_i \in \C$.}
Introducing the the new basis $g^{\dag}_j = f^{\dag}_k {\cal Q}_{kj}$, 
the action of $T$ can be written as 
\begin{align}
T g^{\dag}_{i \ua} T^{-1} = - g^{\dag}_{i\da}, \quad 
T g^{\dag}_{i \da} T^{-1} = g^{\dag}_{i\ua}, \quad 
T \ket{\rm vac}=\ket{\rm vac}.
\end{align}
Let us introduce the real fermion operators $a_{j\sigma}$ $(j=1, \dots 2N, \sigma=\ua,\da)$ by 
\begin{align}
a_{2j-1\sigma}=g^{\dag}_{j\sigma}+g_{j\sigma}, \quad 
a_{2j\sigma}=-i(g^{\dag}_{j\sigma}-g_{j\sigma}). 
\end{align}
Then, the unitary part $C_T$ of $T$ is given by 
\begin{align}
C_T
&=\prod_{j=1}^N e^{\frac{\pi}{4}a_{2j-1\ua}a_{2j-1\da}}e^{-\frac{\pi}{4}a_{2j\ua}a_{2j\da}} 
      \nonumber \\
&=\prod_{j=1}^N(g^{\dag}_{j\ua}g^{\dag}_{j\da}+g_{j\ua}g_{j\da}+g^{\dag}_{j\ua}g_{j\ua}+g^{\dag}_{j\da}g_{j\da}-2g^{\dag}_{j\ua}g_{j\ua}g^{\dag}_{j\da}g_{j\da}). 
\end{align}
Notice that $N$ should be even due to the Kramers degeneracy and $C_T \ket{\rm vac} \sim \ket{\rm full}$ and $(U_T)^2 \sim (-1)^F$. 

\bigskip

Combined with partial transpose $A^{{\sT}_1}$ introduced in Sec.~\ref{sec:Fermionicpartialtranspose}, 
we define the fermionic partial time-reversal transformation as follows. 

\bigskip

\begin{dfn}[Partial time-reversal transformation]
Let $I_1 \subset I$ be a subsystem of $I$. 
Let $C^{I_1}_T$ is the unitary part of the time-reversal transformation $T$ on $I_1$ defined by 
$C_T^{I_1}f_{j\in I_1}[C_T^{I_1}]^{\dag}=f^{\dag}_{k \in I_1}[\cU_T]_{kj}$.
For a fermion number parity preserving operator $A$ on the Fock space, 
the partial time-reversal transformation on a subsystem $I_1$ is defined by $C_T^{I_1} A^{{\sT}_1} [C_T^{I_1}]^{\dag}$.
\end{dfn}

\bigskip

From the properties (\ref{eq:partial_tr_condition2}) and (\ref{eq:pt_double}) of the partial transpose, 
the square of the partial time-reversal transformation
satisfies 
\begin{align}
  &
C_T^{I_1} (C_T^{I_1} A^{\sf T_1} [C_T^{I_1}]^{\dag})^{\sf T_1} [C_T^{I_1}]^\dag
= 
\left\{\begin{array}{ll}
(-1)^{F_1} A (-1)^{F_1} & (T^2=1), \\
A & (T^2=(-1)^F).  \\
\end{array}\right.
\end{align}
This means that the partial time-reversal transformations for TRS with $T^2=1$ 
and $T^2=(-1)^F$ can be interpreted as an operator formalism realizations of the 
pin$_-$ and pin$_+$ structures in Euclidean quantum field theory, respectively.

The expression of the partial time-reversal transformation in the coherent basis can be derived as follows. 
First, notice that the combined transformation $C_T^{I_1} C_f^{I_1}$ preserves the vacuum $C_T^{I_1} C_f^{I_1} \ket{\rm vac} \sim \ket{\rm vac}$ and satisfies 
\begin{align}
&\left\{\begin{array}{ll}
C_T^{I_1} C_f^{I_1} (\chi_{j \in I_1} f^{\dag}_{j\in I_1}) [C_T^{I_1} C_f^{I_1}]^{\dag}
= -\chi_{j\in I_1} f^{\dag}_{k\in I_1} [\cU_T]_{kj}, \\
C_T^{I_1} C_f^{I_1} (\xi_{j \in I_2} f^{\dag}_{j\in I_2}) [C_T^{I_1} C_f^{I_1}]^{\dag}
= \xi_{j \in I_2} f^{\dag}_{j \in I_2}, 
\end{array}\right. \\
&\left\{\begin{array}{ll}
C_T^{I_1} C_f^{I_1} (f_{j\in I_1}\xi_{j \in I_1}) [C_T^{I_1} C_f^{I_1}]^{\dag}
= -[\cU^{\dag}_T]_{jk} f_{k\in I_1} \xi_{j\in I_1}, \\
C_T^{I_1} C_f^{I_1} (f_{j\in I_2}\chi_{j \in I_2}) [C_T^{I_1} C_f^{I_1}]^{\dag}
= f_{j \in I_2} \chi_{j \in I_2}, 
\end{array}\right.  
\end{align}
where the sum is not taken for $j \in I_1$,
and we noted that, for class BDI, 
$C_T^{I_1}$ is Grassmann odd (even) for odd (even) $N_1$. 
The partial time-reversal transformation on coherent states is written as 
\begin{align}
  &
C^{I_1}_T \big( \ket{\{\xi_j\}_{j \in I}} \bra{\{\chi_j\}_{j \in I}} \big)^{\sf T_1} [C^{I_1}_T]^{\dag}
    \nonumber \\
&= \ket{\{i [\cU_T]_{jk} \chi_k\}_{j\in I_1},\{\xi_j\}_{j \in I_2}} \bra{\{i \xi_k [\cU^{\dag}_T]_{kj}\}_{j\in I_1},\{\chi_j\}_{j \in I_2}}. 
\end{align}
It is useful to introduce the unitary operator $U_T^{I_1}$ so that 
\begin{align}
U_T^{I_1} f^{\dag}_{j \in I_1} [U_T^{I_1}]^{\dag}=f^{\dag}_{k \in I_1} [\cU_T]_{kj}, \quad U_T^{I_1} \ket{\rm vac} \sim \ket{\rm vac}. 
\end{align}
Then, 
\begin{widetext}
\begin{align}
C^{I_1}_T \big( \ket{\{\xi_j\}_{j \in I}} \bra{\{\chi_j\}_{j \in I}} \big)^{\sf T_1} [C^{I_1}_T]^{\dag}
= U_T^{I_1} \ket{\{i \chi_j\}_{j\in I_1},\{\xi_j\}_{j \in I_2}} \bra{\{i \xi_j \}_{j\in I_1},\{\chi_j\}_{j \in I_2}} [U_T^{I_1}]^{\dag}. 
\end{align}
It is also useful to express the partial time-reversal transformation in the occupation basis. 
In the same way to obtain (\ref{eq:pt_occupation}), it follows that 
\begin{equation}\begin{split}
&C_T^{I_1}\Big(\ket{\{n_j\}_{j\in I_1},\{n_j\}_{j\in I_2}}\bra{\{\bar n_j\}_{j\in I_1},\{\bar n_j\}_{j\in I_2}}\Big)^{\sf T_1}[C_T^{I_1}]^{\dag} \\
&= (-i)^{[\tau_1+\bar\tau_1]}(-1)^{(\tau_1+\bar\tau_1)(\tau_2+\bar\tau_2)} 
U_T^{I_1} \ket{\{\bar n_j\}_{j\in I_1},\{n_j\}_{j\in I_2}}\bra{\{n_j\}_{j\in I_1},\{\bar n_j\}_{j\in I_2}} [U_T^{I_1}]^{\dag}. 
\end{split}\label{eq:pt_time_fermion_occupation}\end{equation}
%

\end{widetext}

\subsubsection{Partial antiunitary particle-hole transformation}
\label{sec:Partial antiunitary particle-hole transformation}
Let $CT$ be an antiunitary particle-hole symmetry (sometimes called chiral symmetry) defined by 
\begin{align}
CT f^{\dag}_j (CT)^{-1}
=[\cU_{CT}]_{jk} f_k, \quad 
CT\ket{\rm vac} \sim \ket{\rm full} 
\label{eq:chiral_sym_def}
\end{align}
with a unitary matrix $\cU_{CT} \in U(N)$. 
Under the change of basis $f^{\dag}_j=g^{\dag}_k {\cal V}_{kj}$ with ${\cal V} \in U(N)$, 
$\cU_{CT}$ is transformed in a conjugate way $CT g^{\dag}_j (CT)^{-1} = [{\cal V} \cU_{CT} {\cal V}^{\dag}]_{jk} g_k$.  
In the same manner as the partial time-reversal transformation, 
we introduce the unitary operator $U_{CT}$ which preserves the particle number so as to satisfy 
\begin{align}
U_{CT}f_jU_{CT}^{\dag}=[\cU_{CT}]_{jk}f_k, \quad 
U_{CT}\ket{\rm vac} \sim \ket{\rm vac}. 
\end{align}
This condition uniquely fixes $U_{CT}$ (up to a constant phase). 
Diagonalizing $\cU_{CT}$ in the form of $\cU_{CT} = {\cal Q} {\rm diag}(e^{i \phi_1}, e^{i \phi_2}, \dots, ) {\cal Q}^{\dag}$ 
and introducing the operators $g^{\dag}_i = f^{\dag}_j {\cal Q}_{ji}\ (g_i = {\cal Q}_{ij}^{\dag} f_j)$, 
we can explicitly write $U_{CT}$ as in
\begin{align}
U_{CT} = \prod_{j=1}^N e^{-i \phi_j (g^{\dag}_j g_j-1/2)} 
\end{align}
up to a constant phase. 

\begin{dfn}[Partial antiunitary particle-hole transformation]
Let $I_1 \subset I$ be a subsystem of $I$. 
Let $U^{I_1}_{CT}$ is the unitary part of the antiunitary particle-hole transformation $CT$ on $I_1$ defined by 
$U_{CT}^{I_1}f_{j\in I_1}[U_{CT}^{I_1}]^{\dag}=f_{k \in I_1}[\cU_{CT}]_{kj}$. 
For a fermion number parity preserving operator $A$ on the Fock space, 
the partial antiunitary particle-hole transformation on a subsystem $I_1$ is defined by $U_{CT}^{I_1} A^{{\sT}_1} [U_{CT}^{I_1}]^{\dag}$.
\end{dfn}

\bigskip

\begin{widetext}
Notice that the fermionic partial transpose $A^{\sf T_1}$ is the special case of 
the partial anti-unitary particle-hole transformation with $U_{CT}=1$.~\footnote{
In contrast to the unitary part $U_T$ of TRS $T$, 
$U_{CT} = 1$ is preserved under basis changes. 
} 
It is easy to obtain the expression of the partial anti-unitary particle-hole transformation in the coherent state basis~\footnote{
Equation (\ref{eq:pt_phs_fermion_coherent}) can be proven as follows. 
Let us introduce the unitary operator $\bar U_{CT}^{I_1}$ so that 
$\bar U_{CT}^{I_1} f^{\dag}_{j \in I_1} [\bar U_{CT}^{I_1}]^{\dag} = [\cU_{CT}]_{jk} f^{\dag}_{k \in I_1}$ and $\bar U_{CT}^{I_1} \ket{\rm vac} \sim \ket{\rm vac}$. 
It holds that $U_{CT}^{I_1} C_f^{I_1} \sim C_f^{I_1} \bar U_{CT}^{I_1}$. 
Then, $U_{CT}^{I_1} C_f^{I_1} \ket{\{ -i \chi_j \}_{j \in I_1}, \{\xi_j\}_{j \in I_2}} 
\sim C_f^{I_1} \bar U_{CT}^{I_1} \ket{\{ -i \chi_j \}_{j \in I_1}, \{\xi_j\}_{j \in I_2}} 
\sim C_f^{I_1} \ket{\{ -i \chi_k [\cU_{CT}]_{kj} \}_{j \in I_1}, \{\xi_j\}_{j \in I_2}}$. 
}
\begin{equation}\begin{split}
&U_{CT}^{I_1} \big( \ket{\{ \xi_j\}_{j \in I_1}, \{\xi_j\}_{j \in I_2}} \bra{\{ \chi_j\}_{j \in I_1}, \{\chi_j\}_{j \in I_2}} \big)^{\sf T_1} [U_{CT}^{I_1}]^{\dag} \\
&= C_f^{I_1} \ket{\{ -i \chi_k [\cU_{CT}]_{kj} \}_{j \in I_1}, \{\xi_j\}_{j \in I_2}} \bra{\{ -i [\cU_{CT}^{\dag}]_{jk} \xi_k\}_{j \in I_1}, \{\chi_j\}_{j \in I_2}} [C_f^{I_1}]^{-1} 
\label{eq:pt_phs_fermion_coherent}
\end{split}\end{equation}
and also in the occupation number basis, 
\begin{equation}\begin{split}
&U_{CT}^{I_1} \big( \ket{\{n_j\}_{j \in I_1},\{n_j\}_{j \in I_2}}\bra{\{\bar n_j\}_{j \in I_1},\{\bar n_j\}_{j \in I_2}} \big)^{{\sT}_1} [U_{CT}^{I_1}]^{\dag} \\
&= i^{[\tau_1+\bar \tau_1]}(-1)^{(\tau_1+\bar \tau_1)(\tau_2+\bar \tau_2)} 
U_{CT}^{I_1} C_f^{I_1} \ket{\{\bar n_j\}_{j \in I_1},\{n_j\}_{j \in I_2}}\bra{\{n_j\}_{j \in I_1},\{\bar n_j\}_{j \in I_2}} [C_f^{I_1}]^{\dag} [\bar U_{CT}^{I_1}]^{\dag}.
\label{eq:pt_phs_occupation}
\end{split}\end{equation}

\end{widetext}

\subsubsection{Partial time-reversal transformation in bosonic systems}
\label{sec:Partial time-reversal transformation in bosons}
Let us consider a bosonic system in any dimensions, 
where the Hilbert space is the Fock space spanned in the occupation number basis by 
\begin{align}
(b^{\dag}_1)^{n_1} (b^{\dag}_2)^{n_2} \cdots (b^{\dag}_N)^{n_N} \ket{\rm vac},  \quad 
n_j \in \{ 0, 1, 2, \dots, \}, 
\end{align}
where $b^{\dag}_j$ are boson creation operators which satisfy the commutation relations 
\begin{align}
[b^{\dag}_j, b_k]=\delta_{jk}, \quad 
[b_j,b_k]=[b_j^{\dag},b_k^{\dag}]=0, 
\end{align}
and $\ket{\rm vac}$ is the vacuum that is annihilated by all $b_j$. 
The operator algebra is known as the Weyl algebra which is simple, 
thus, the subalgebra of finite dimension is a just matrix algebra. 
An important consequence of having a matrix algebra is that there is no basis independent linear anti-automorphism, which can 
be contrasted with the case of the Clifford algebra. 

We may try to introduce the partial time-reversal transformation in a way similar to the fermionic cases. 
Let us introduce real bosonic variables by 
\begin{align}
&x_j=\frac{b_i+b^{\dag}_i}{\sqrt{2}}, \quad 
p_j=\frac{b_i-b^\dag_i}{\sqrt{2} i}, \\
&[x_j,p_k]=i \delta_{jk}, \quad [x_j,x_k]=[p_j,p_k]=0. 
\end{align}	
A desired property of ``transpose'' would be that it sends 
$x_j \mapsto x_j$ and $p_j \mapsto -p_j$, namely, $b^{\dag}_j\mapsto b_j, b_j \mapsto b^{\dag}_j$. 
This is of particle-hole type. 
However, a particle-hole transformation in bosonic systems is problematic if we consider the action on the vacuum: 
Let $C$ be a unitary operator which satisfies $Cb^{\dag}C^{\dag}=b, CbC^{\dag}=b^{\dag}$ for a single boson. 
Then, $b\ket{0}=0$ implies $b^{\dag}C\ket{0}=0$, but this condition has no solution for $C\ket{0}$ since there is no restriction of the number of bosons. 

Instead of introducing a unitary particle-hole transformation, 
we define the partial time-reversal transformation by the matrix transposition followed by the unitary part of the time-reversal transformation. 
Let $I=I_1 \cup I_2$ be a division of total system $I$. 
We write a basis of the many body Hilbert space by the tensor product $\ket{j,k}=\ket{j}_{I_1} \otimes \ket{k}_{I_2}$. 
Let us consider a time-reversal transformation $T$ which preserves the subdivision 
\begin{align}
T \ket{j}_{I_1} = \ket{k}_{I_1} [\cU_T]_{kj}.  
\label{eq:trs_boson}
\end{align}
Here, the choice of the overall phase of $[\cU_T]_{kj}$ does not matter in the partial time-reversal transformation. 
Note that 
under the basis transformation $\ket{j}_{I_1}=\ket{k}'_{I_1} {\cal V}_{kj}$ 
the unitary matrix $\cU_T$ is changed as $\cU_T'={\cal V} \cU_T {\cal V}^{tr}$. 
We define the partial time-reversal transformation associated with the time-reversal symmetry $T$ by 
\begin{align}
\big( \ket{j,k}\bra{l,m} \big)^{\sf T_1}
:= [\cU_T]_{lp} \ket{p,k}\bra{q,m} [\cU^{\dag}_T]_{qj}. 
\label{eq:partial_tr_boson}
\end{align}
By definition, the full ``time-reversal transformation'' $A^{\sf T}=(A^{\sf T_1})^{\sf T_2}$ is linear and anti-automorphism. 
An important property is that the definition (\ref{eq:partial_tr_boson}) is compatible with the change of basis 
\begin{align}
\big( \ket{j,k}'\bra{l,m}' \big)^{\sf T_1}
= [{\cal V} \cU_T {\cal V}^{tr}]_{lp} \ket{p,k}'\bra{q,m}' [{\cal V} \cU_T {\cal V}^{tr}]^{\dag}_{qj}. 
\end{align}
The usual partial transpose defined in (\ref{eq:transdef}) can be identified with the partial time-reversal transformation (\ref{eq:partial_tr_boson}) 
in the basis where $\cU_T$ is the identity matrix.


\paragraph{Partial time-reversal transformation in spin systems}
Let us consider a more restricted situation in many body bosonic systems. 
In spin systems, the many-body Hilbert space is the tensor product ${\cal H} = \bigotimes_{x} {\cal H}_x$ of the local Hilbert space ${\cal H}_x$ of the spin degree of freedom at site $x$. 
Each local Hilbert space is a projective representation of symmetry group. 
Let us consider a time-reversal symmetry $T$ which is written as $T \ket{j_x} = \ket{k_x} [\cU_T]_{kj}$ on each site, where $\cU_T \cU_T^* = \pm 1$. 
From the general definition (\ref{eq:partial_tr_boson}), 
the partial time-reversal transformation in spin systems is 
\begin{align}
&\big( \ket{\{ j_x\}_{x \in I_1}, \{k_x\}_{x_\in I_2}}\bra{\{l_x\}_{x\in I_1},\{m_x\}_{x \in I_2}} \big)^{\sf T_1} 
\nonumber \\
&= (\otimes_{x \in I_1} \cU_T) \ket{\{ l_x\}_{x \in I_1}, \{k_x\}_{x_\in I_2}} \bra{\{j_x\}_{x\in I_1},\{m_x\}_{x \in I_2}} (\otimes_{x \in I_1} \cU_T^{\dag}). 
\end{align}

\section{Many body topological invariants of fermionic short range entangled
  topological phases in one spatial dimension}
\label{sec:many-body_invariant}

In this section,
by making use of the fermionic partial time-reversal and anti-unitary transformations
introduced in the previous section \ref{Fermionic partial transpose and partial antiunitary transformations},
we construct nonlocal order parameters which detect fermionic SPT phases of various kinds
in (1+1) dimensions. 
Recall that our strategy for the construction of many-body topological invariants
is to find a way -- within the operator formalism --
to ``simulate'' the Euclidean path-integral \eqref{path integral}
on generating manifolds of the cobordism groups.
The phases of the partition functions are the desired topological invariants (Eq.\ \eqref{top term}).
This section is devoted to $(1+1)d$ fermionic SPT phases and their topological invariants.
Higher-dimensional fermionic SPT phases in $(2+1)d$ and $(3+1)d$ will be discussed in the next section.

In Sec.~\ref{sec:Methodtocomputethetopologicalinvariant},
we first give an overview of  
the use of the fermionic partial time-reversal and anti-unitary transformations
to construct non-local order parameters  
within the operator formalism that correspond to
the partition functions on unoriented spacetime manifolds. 
In the subsequent sections, 
\ref{sec:(1+1)BDI},
\ref{sec:(1+1)DIII}
\ref{sec:(1+1)AIII},
\ref{sec:(1+1)AI},
and
\ref{sec:(1+1)AII},
we present the many-body topological invariants
which are also summarized in  Table \ref{tab:summary}. 
For each case, we demonstrate that our formula can indeed detect non-trivial SPT
phases,
by using analytical calculations for fixed point wave functions
and numerics for free fermion models.

It is worth mentioning that although
we apply the fermionic partial transformation to
short-range entangled pure states here, 
our method does work for any density matrices
such as the one for the canonical ensemble
at finite temperature and critical systems. 
See Ref.~\cite{shapourian2016partial}, for example, where we discussed the fermionic cousin of the entanglement negativity of bosonic systems by use of the fermionic partial transpose introduced in Sec.~\ref{sec:Fermionicpartialtranspose}.

\subsection{Unoriented spacetime path-integral
  in the operator formalism}
\label{sec:Methodtocomputethetopologicalinvariant}


In the case of SPT phases protected by an orientation-reversing spatial symmetry
(e.g., reflection and inversion)~\cite{ShapourianShiozakiRyu2016detection, shapourian2016partial},
we explain how to implement the path-integral on unoriented spacetimes
in the operator formalism, based on which we propose formulas for many-body topological invariants.
It is shown that
the corresponding {\it partial} symmetry operation (e.g., partial reflection and inversion)
which acts only on a subregion of the total space
can be used to simulate the unoriented spacetime path integral,
e.g., path-integrals on $\mathbb{R}P^2$ and $\mathbb{R}P^4$.  
Naively pursing the analogy with this approach that works for orientation-reversing spatial symmetries,
one would then be tempted to consider a {\it partial} time-reversal or $CT$ transformation 
to construct topological invariants for SPT phases protected by these symmetries.  
However, due to the anti-unitary nature of $T$ and $CT$ transformations,
simply restricting these transformations to a given spatial region
would not work. 
In other words, what is the meaning of restricting  complex conjugation to a given region?
(The issue of ``gauging" time-reversal symmetry, i.e., promoting it to a 
local symmetry, 
has been discussed within tensor network states in Ref.~\cite{Chen_Vishwanath}.)

In this section, we show how
partial transpose and anti-unitary transformation,
introduced in Sec.\ \ref{Fermionic partial transpose and partial antiunitary transformations}, 
can be properly used to
generate a partial anti-unitary transformation and simulate the necessary unoriented spacetime path integral.

We also discuss how topological invariants can be expressed in terms of the fermion coherent
states. 
For the case of non-interacting fermionic systems, 
topological invariants can be expressed as a fermionic Gaussian integral, and
can be computed quite efficiently.
\cite{ShapourianShiozakiRyu2016detection, shapourian2016partial}


\subsubsection{Why partial transpose?}
\label{sec:whypartialtranspose}

\begin{figure*}
\includegraphics[width=0.7\linewidth, trim=0cm 0cm 0cm 0cm]{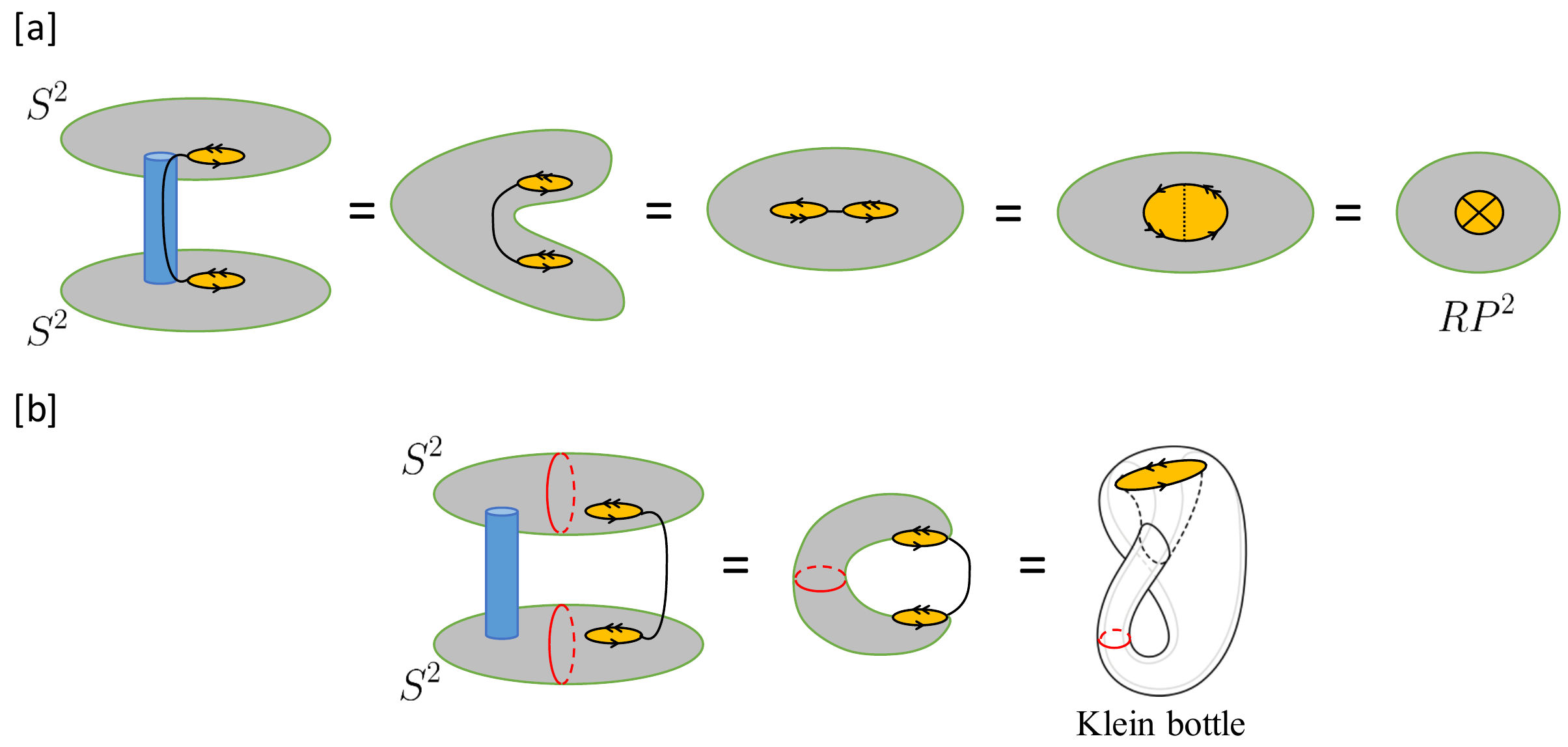}
\caption{\label{fig:manifolds} The topological equivalence of the spacetime manifolds associated with the partial transpose on [a] two adjacent intervals and [b] two disjoint intervals to the real projective plane and the Klein bottle, respectively. 
The red curve represents the holonomy from the temporal boundary condition on the intermediate region of the disjoint two intervals (See Fig.~\ref{figs/pt_disjoint_bdi}). 
}
\end{figure*}

Let us begin with the following question:
what is a physical observable
associated with time-reversal $T$?
Since $T$ is anti-unitary,
for a pure quantum state $\ket{\psi}$,
a naive expectation value
$\braket{\psi | T | \psi}$
is not a physical observable
as it depends on an unphysical $U(1)$
phase ambiguity of the state $\ket{\psi}$.
Instead, we find that the amplitude $|\braket{\psi|T|\psi}|$
is physical.
We note that amplitude square can be written as 
\begin{align}
|\braket{\psi|T|\psi}|^2
= \Tr\Big[ \rho U_T \rho^{*} U_T^{\dag} \Big]
= \Tr\Big[ \rho U_T \rho^{\sf T} U_T^{\dag} \Big], 
\end{align}
where $U_T$ is the unitary matrix defined by $[U_T]_{jk}=\braket{j|T|k}$,
$\rho=\ket{\psi}\bra{\psi}$,
$\rho^*$ is the complex conjugate of $\rho$,
and $\rho^{\sf T}$ is the matrix transpose $\big( \ket{i} \bra{j} \big)^{\sf T}
= \ket{j} \bra{i}$ in the many body Hilbert space.
Here, we have used the hermiticity $\rho^{\dag}=\rho$ of the density matrix.
The quantity $|\braket{\psi|T|\psi}|$ is useful for determining
whether time-reversal symmetry is spontaneously broken.
However,
it does not help us to differentiate topological phases protected by
time-reversal symmetry,
since $|\braket{\psi|T|\psi}|$ is identically one
for time-reversal symmetric states,
i.e., $T \ket{\psi} \sim \ket{\psi}$.

For the purpose of differentiating and detecting 
SPT phases protected by time-reversal symmetry,
it is necessary to consider
the expectation value involving 
{\it partial transpose} or {\it partial time-reversal transformation}
of the reduced density matrix
\cite{Pollmann2012},
which,
when interpreted in the path-integral picture,
corresponds to the path-integral
on unoriented spacetime.
\cite{ShiozakiRyu2016}

Let us discuss how this can be done first for bosonic SPT phases
defined on a one-dimensional space or one-dimensional lattice $I_{tot}$.
We consider a given region (segment) in $I_{tot}$, call it $I$,
and consider the reduced density matrix
$\rho_I$,
which is obtained by integrating all degrees of freedom outside of $I$,
$\rho_I
=
\mathrm{Tr}_{\bar{I}}\, |GS\rangle\langle  GS|
$
where $|GS\rangle$ is the ground state on $I_{tot}$.
We now consider bipartitioning $I$,
$I=I_1 \cup I_2$.
Here, $I_1$ and $I_2$ can be two adjacent or disjoint intervals within $I_{tot}$.
The many body Hilbert space is the tensor product
of sub Hilbert spaces on $I_1$ and $I_2$,
${\cal H}_I={\cal H}_{I_1} \otimes {\cal H}_{I_2}$.
By considering the partial transpose of the reduced density matrix
followed by the unitary transformation only on the subsystem $I_1$,
we define the quantity  
\begin{align}
Z=\mathrm{Tr}^{\ }_{I}\Big[ \rho^{\ }_I U^{I_1}_T \rho^{\sf T_1}_I [U_T^{I_1}]^{\dag} \Big], 
\label{eq:Z_pt}
\end{align}
where $\rho^{\sf T_1}_I$ is the partial transpose defined
in Eq.\ \eqref{eq:transdef}
and $U_T^{I_1}$ is the unitary matrix associated with time-reversal symmetry $T$ on the subsystem $I_1$. 

As argued in Ref.\ \cite{ShiozakiRyu2016}, 
the quantity \eqref{eq:Z_pt} can be viewed as the partition function
on an unorientable spacetime. 
In the case of adjacent intervals, the corresponding manifold is $\mathbb{R}P^2$
(Fig.~\ref{fig:manifolds} [a]),
and in the case of disjointed intervals, the corresponding manifold is the Klein
bottle~\cite{CalabreseCardyTonni2012negativity,ShiozakiRyu2016,ShapourianShiozakiRyu2016detection}
(Fig.~\ref{fig:manifolds} [b]).
Hence, the phase of the quantity \eqref{eq:Z_pt}
can serve as a topological invariant for $(1+1)d$ bosonic SPT phases protected by TRS.

In the above discussion, we have considered bosonic cases while our focus in
this paper is fermions. 
We should note that the partial transposition introduced in
Eq.~(\ref{eq:transdef}), 
which is simply swapping indices of the first interval, 
is strictly defined for bosons.
In order to define a consistent definition for fermions, 
one needs to take into account the anti-commuting property of fermion operators,
as discussed in Sec.~\ref{sec:Fermionicpartialtranspose}.

As we will show, by using the fermionic partial transpose and related operations
developed in Sec.~\ref{sec:Fermionicpartialtranspose},
we can construct many-body topological invariants for fermionic SPT phases.
We should note that the quantity $Z$ introduced above is a complex number and the complex phase is the quantized topological invariant. In general, the modulus $|Z|$ depends on microscopic details and obeys an area law away from critical points. 
In the next two subsubsections, we will introduce
a fermionic counterpart of \eqref{eq:Z_pt} both for adjacent and disjoint
intervals. 
When considering the Klein bottle, we can insert a $\pi$-flux through a non-trivial cycle (red loops in Fig.~\ref{fig:manifolds}) which leads to two possible boundary conditions: the periodic boundary condition (the ``Ramond'' sector) and the anti-periodic boundary condition (the ``Neveu-Schwarz'' sector). As we see in the following, these two boundary conditions can be realized in our calculations by applying the fermion number parity twist operator on the intermediate interval separating the two disjoint intervals.

\begin{widetext}
\subsubsection{Two adjacent intervals: cross-cap}
\label{Two adjacent intervals: cross-cap}

Let us consider two states $|\psi_{1,2}\rangle \in {\cal H}_{S_1}$ 
in the Hilbert space defined on the space manifold $S^1$,
and 
$\rho = \ket{\psi_1} \bra{\psi_2}$.
When $|\psi_1\rangle=|\psi_2\rangle$, 
$\rho$ is a pure state density matrix.
Let us introduce two adjacent intervals $I_{1,2}$,
and trace out the degrees of freedom outside
$I=I_1 \cup I_2 \subset S^1$
to obtain a reduced density matrix.
The reduced density matrix can be interpreted as
a path-integral on a cylinder with a slit as  
\begin{align}
\rho_I = \Tr_{S^1 \backslash I} \big( \ket{\psi_1} \bra{\psi_2} \big)
= 
\begin{array}{c}
\includegraphics[width=0.15\linewidth, trim=0cm 0cm 0cm 0cm]{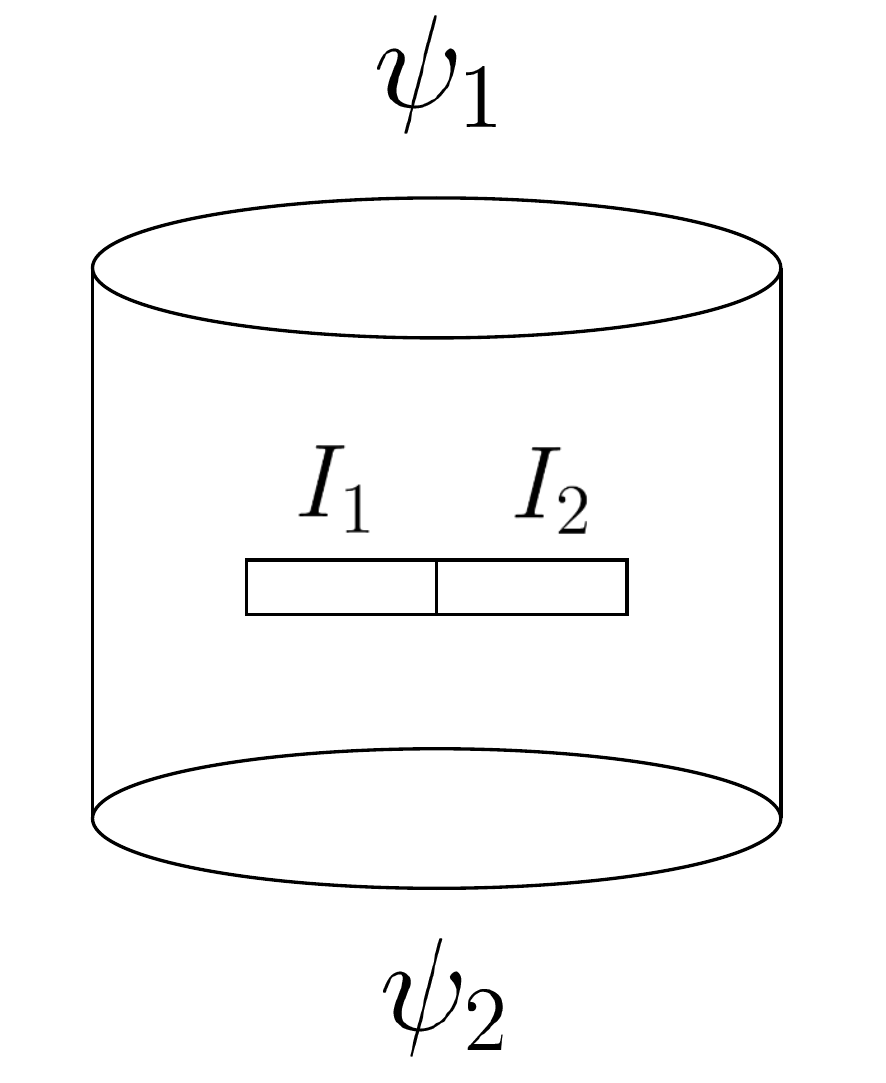}
\end{array}.
\end{align}
We now consider the fermionic partial transpose of $\rho_I$ with respect to the subregion $I_1$.
Explicitly, it is given,
in the coherent state representation, as
\begin{align}
C_T^{I_1} \rho_I^{\sf T_1} [C_T^{I_1}]^{\dag}
  &= \int \prod_{j \in I} d \bar \xi_j d \xi_j d \bar \chi_j d \chi_j e^{- \sum_{j \in I} (\bar \xi_j \xi_j + \bar \chi_j \chi_j)} \braket{\{\bar \xi_j\}_{j \in I}|\rho_I|\{\chi_j\}_{j \in I}}
  \nonumber \\
  & \quad
    \times \ket{\{i [\cU_T]_{jk} \bar \chi_k\}_{j\in I_1},\{\xi_j\}_{j \in I_2}} \bra{\{i \xi_k [\cU^{\dag}_T]_{kj}\}_{j\in I_1},\{\bar \chi_j\}_{j \in I_2}}. 
\end{align}
Let $\rho' = \ket{\psi_3} \bra{\psi_4}$ be another density matrix composed
of pure states $\ket{\psi_3}, \ket{\psi_4}$.
The quantity
\begin{align}
\Tr^{\ }_I \left[\rho'_I C_T^{I_1} \rho_I^{\sf T_1} [C_T^{I_1}]^{\dag} \right]
  \label{quantity}
\end{align}
is associated with the spacetime manifold, 
which 
is topologically equivalent to the four point function with a cross-cap 
as shown in Fig.~\ref{fig:manifolds}.
By using the coherent state basis, 
the quantity \eqref{quantity} can be expressed as
\begin{equation}\begin{split}
\Tr^{\ }_I \Big[ \rho'_I C_T^{I_1} \rho_I^{\sf T_1} [C_T^{I_1}]^{\dag} \Big] 
&= 
\int \prod_{j \in {\rm full}} [d \alpha_j d \beta_j d \gamma_j d \delta_j] 
e^{\sum_{I_1} [\alpha_j [i \cU_T]_{jk} \gamma_k + \beta_j [i \cU_T^{\dag}]_{jk} \delta_k]}
e^{\sum_{I_2} [\alpha_j \delta_j + \beta_j \gamma_j]} 
e^{\sum_{S^1 \backslash I} [\alpha_j \beta_j+\gamma_j \delta_j]} \\
&\quad \quad
\times \psi_1(\{\alpha_j\}) \psi_2^*(\{\beta_j\}) \psi_3(\{\gamma_j\}) \psi_4^*(\{\delta_j\}) . 
\end{split}\label{eq:RP2_coherent}\end{equation}
The corresponding
tensor network representation is presented in Fig.\ \ref{figs/pt_adjacent}.
When $\rho$ and $\rho'$ are composed of
the ground state $\ket{GS}$,
$\rho= \rho' = \ket{GS}\bra{GS}$,
Eq.\ \eqref{eq:4_point_func_crosscap}
can be interpreted as the partition function on the real projective plane $\mathbb{R}P^2$. 
(To see this, we note that $\ket{GS}$ is the state defined on the boundary of the disc, 
and obtained by performing the path-integral inside the disc.)

Finally, it is also possible and useful to turn on an additional 
symmetry flux along the nontrivial $\Z_2$ cycle of the cross-cap.
We insert the unitary operator $U^{I_1}_g$ of $g$ symmetry on the subsystem
$I_1$
and consider  
\begin{align}
\Tr^{\ }_I \Big[ \rho'_I U_g^{I_1} C_T^{I_1} \rho_I^{\sf T_1} [C_T^{I_1}]^{\dag} [U_g^{I_1}]^{\dag} \Big]
\sim 
\begin{array}{c}
\includegraphics[width=0.25\linewidth, trim=0cm 0cm 0cm 0cm]{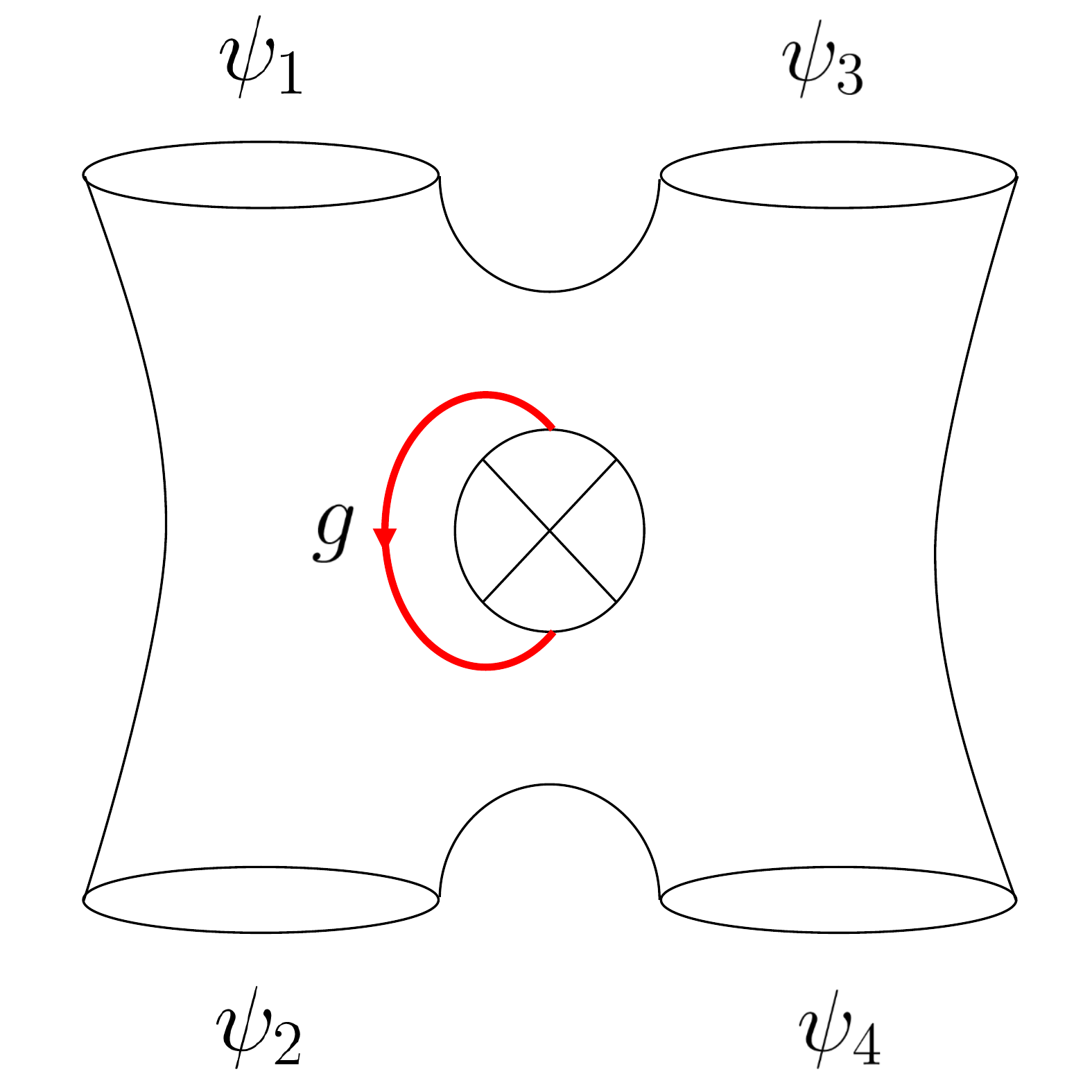}
\end{array}.
\label{eq:4_point_func_crosscap}
\end{align}

\begin{figure}[!]
	\begin{center}
	\includegraphics[width=0.8\linewidth, trim=0cm 0cm 0cm 0cm]{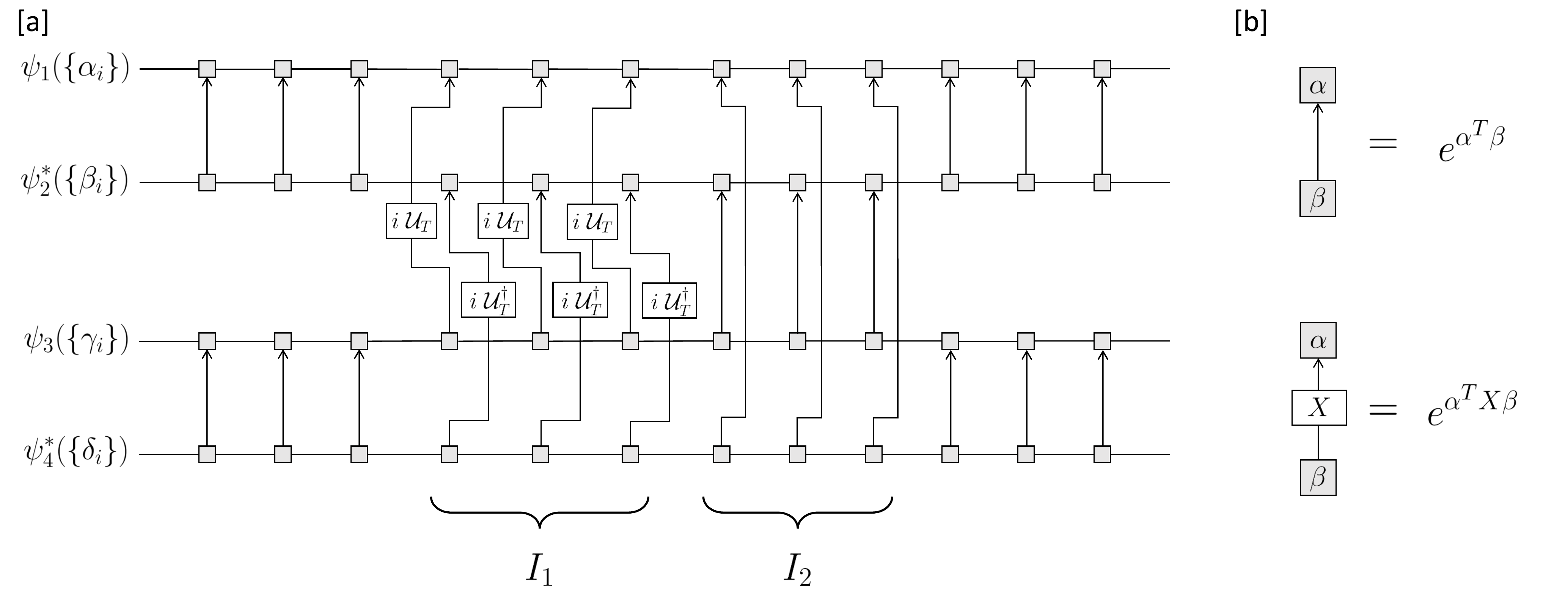}
	\end{center}
	\caption{
	[a] Adjacent partial transpose associated with the TRS. 
	[b] Connecting matrices. }
	\label{figs/pt_adjacent}
\end{figure}

\subsubsection{Two disjoint intervals: the Klein bottle}
Let us consider $\rho = \ket{\psi_1} \bra{\psi_2}$ as before.
We introduce three adjacent intervals $I=I_1 \cup I_2 \cup I_3 \subset S^1$ 
and trace out the degrees of freedom outside of $I$. 
Subsequently, we trace out the interval $I_2$ after acting with a unitary transformation $U^{I_2}$. 
Here, $U^{I_2}$ can be any unitary operator,
but we focus on the $U(1)$ transformation
$U^{I_2}_{\theta}=e^{\sum_{j \in I_2} i \theta f^{\dag}_j f^{\ }_j}$. 
The resulting reduced density matrix can be viewed
as the path-integral on a cylinder with two slits
and an intermediate symmetry defect:
\begin{align}
  \rho_{I_1 \cup I_3}(e^{i \theta}) = \Tr_{S^1 \backslash I_1 \cup I_3}
  \big( e^{i \theta \sum_{j \in I_2} f^{\dag}_j f^{\ }_j} \ket{\psi_1} \bra{\psi_2} \big)
= 
\begin{array}{c}
\includegraphics[width=0.15\linewidth, trim=0cm 0cm 0cm 0cm]{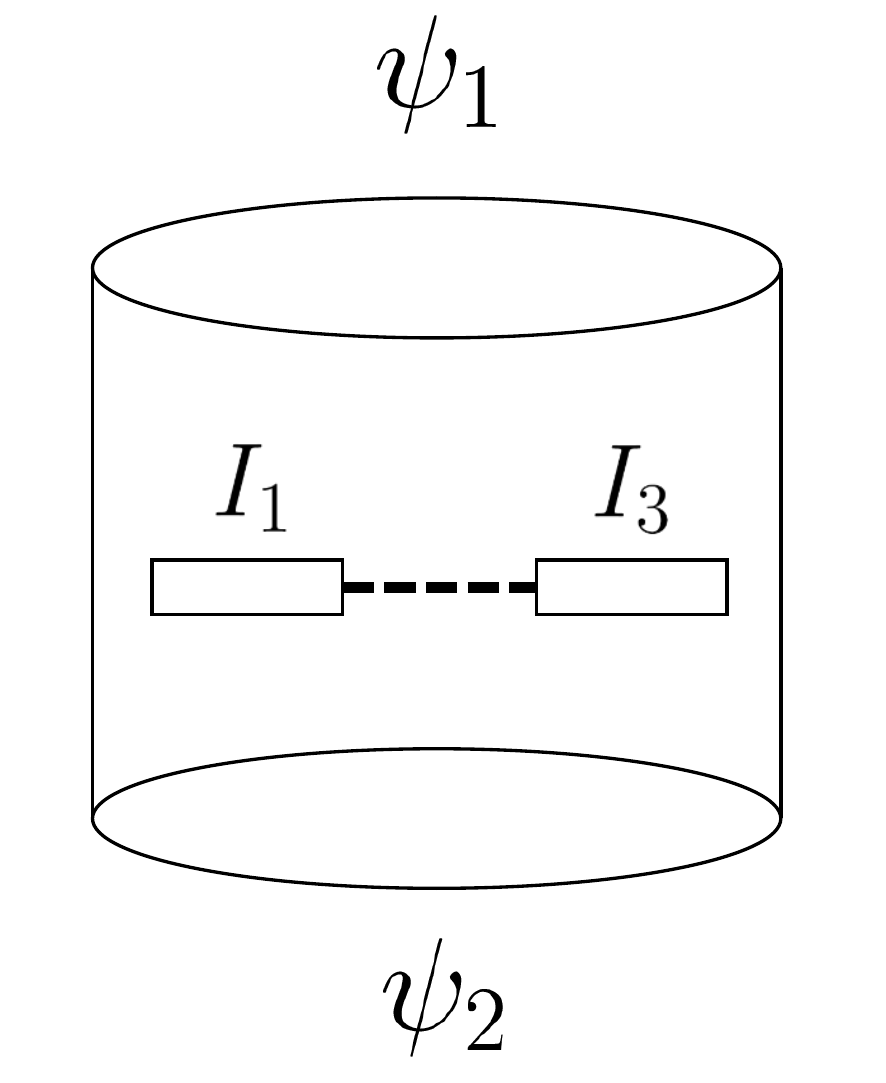}
\end{array}, 
\end{align}
where the dashed line represents the symmetry twist:
when passing this twist the complex fermion fields acquire the phase $e^{i \theta}$. 
Let $\rho' = \ket{\psi_3} \bra{\psi_4}$ be another density matrix
and consider the density matrix $\rho'_{I_1 \cup I_3}(e^{-i \theta})$
with the symmetry twist by the inverse of $U^{I_2}_{\theta}$. 
The spacetime manifold associated with the quantity 
\begin{align}
\Tr^{\ }_{I_1 \cup I_3}
\Big[ \rho'_{I_1 \cup I_3}(e^{-i \theta})
C_T^{I_1} \rho_{I_1 \cup I_3}^{\sf T_1}(e^{i \theta}) [C_T^{I_1}]^{\dag} \Big] 
  \label{quantity 2}
\end{align}
is topologically equivalent to the four point function on the Klein bottle
with the twist by $U^{I_2}_{\theta}$;
\begin{align}
\Tr_{I_1 \cup I_3} \Big[ \rho'_{I_1 \cup I_3}(e^{-i \theta}) C_T^{I_1} \rho_{I_1 \cup I_3}^{\sf T_1}(e^{i \theta}) [C_T^{I_1}]^{\dag} \Big]
\sim  
\begin{array}{c}
\includegraphics[width=0.25\linewidth, trim=0cm 0cm 0cm 0cm]{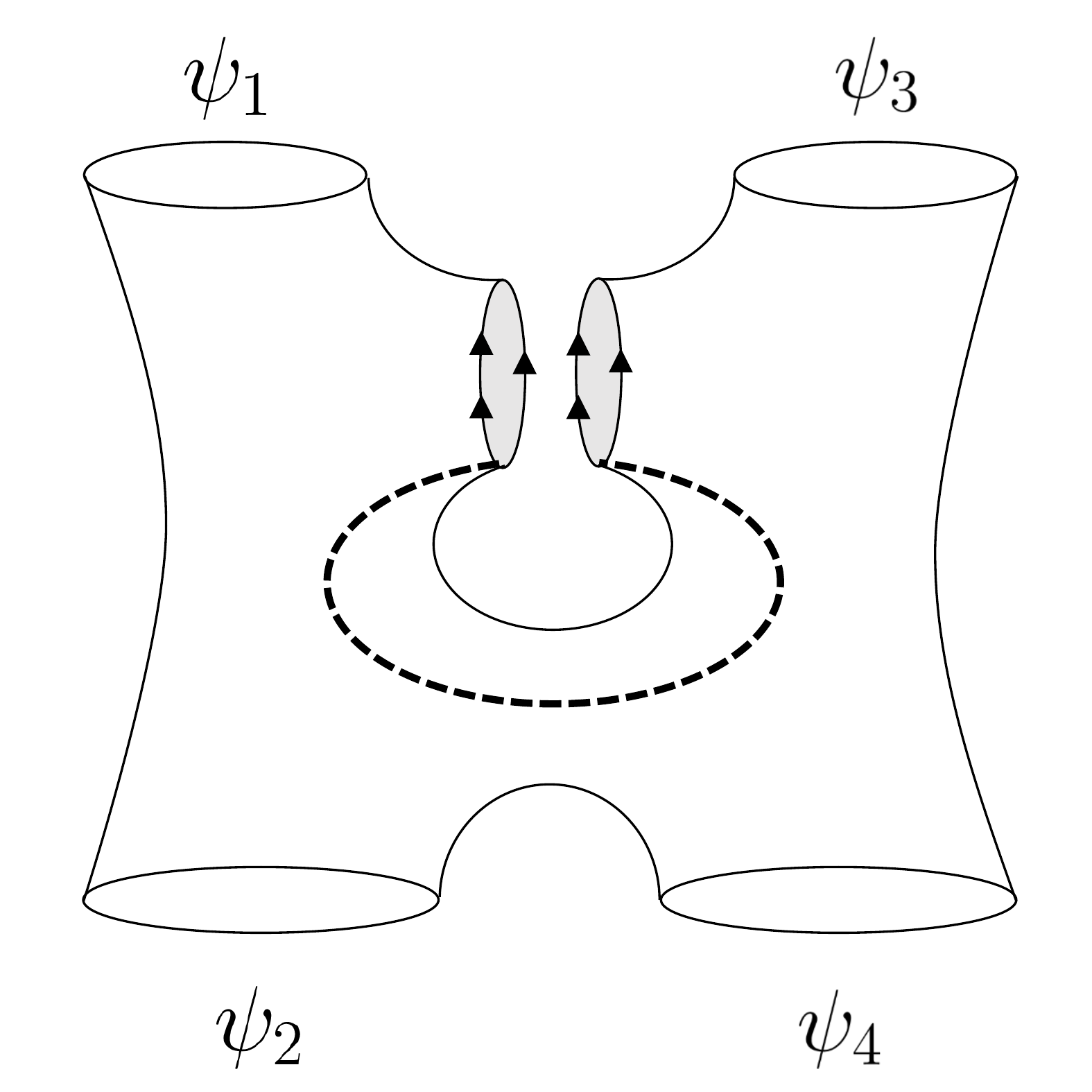}
\end{array}.
\label{eq:4point_Klein}
\end{align}
See Fig.~\ref{fig:manifolds},
and also Fig.~\ref{figs/pt_disjoint_bdi} for a tensor network representation. 
For the ground states $\ket{\psi_i} = \ket{GS} (i=1,2,3,4)$, 
the four point function (\ref{eq:4point_Klein}) is the 
partition function over the Klein bottle. 
It should be noticed that the $U(1)$ twist $e^{i \theta}$ is arbitrary
in the cases of pin$^{\tilde c}_{\pm}$ structures
due to the flip of the $U(1)$ charge on the orientation reversing patches.
On the other hand, 
in the case of pin$^c$ structures,
the $U(1)$ holonomy along the $\Z_2$ nontrivial cycle 
is quantized to $0$ or $\pi$ flux. 
As in \eqref{eq:RP2_coherent}, 
the quantity (\ref{eq:4point_Klein})
can be expressed in terms of the fermion coherent state as 
\begin{equation}\begin{split}
&\Tr_{I_1 \cup I_3} \Big[ \rho'_{I_1 \cup I_3}(e^{-i \theta}) C_T^{I_1} \rho_{I_1 \cup I_3}^{\sf T_1}(e^{i \theta}) [C_T^{I_1}]^{\dag} \Big] \\
&= 
\int \prod_{i \in {\rm full}} [d \alpha_i d \beta_i d \gamma_i d \delta_i] 
e^{\sum_{I_1} \alpha_i [i \cU_T]_{ij} \gamma_j } 
e^{\sum_{I_1} \beta_i [i \cU_T^{\dag}]_{ij} \delta_j}
e^{\sum_{I_2} \alpha_i e^{-i \theta} \beta_i}
e^{\sum_{I_2} \gamma_i e^{i \theta} \delta_i} 
e^{\sum_{I_3} \alpha_i \delta_i}
e^{\sum_{I_3} \beta_i \gamma_i} 
e^{\sum_{I_4} \alpha_i \beta_i}
e^{\sum_{I_4} \gamma_i \delta_i} \\
& \ \ \ 
\times 
\psi_1(\{\alpha_i\}) \psi_2^*(\{\beta_i\}) \psi_3(\{\gamma_i\}) \psi_4^*(\{\delta_i\}) . 
\end{split}\label{eq:Klein}\end{equation}

\begin{figure}[!]
	\begin{center}
	\includegraphics[width=0.8\linewidth, trim=0cm 0cm 0cm 0cm]{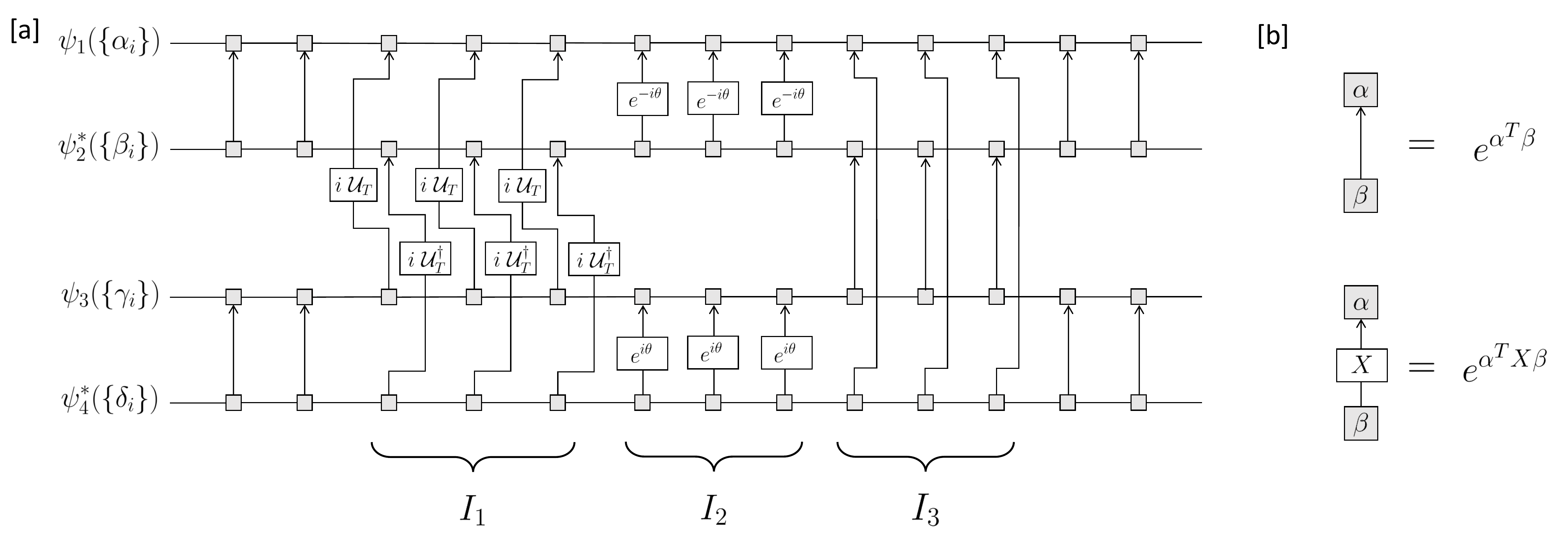}
	\end{center}
	\caption{
	[a] Disjoint partial transpose with intermediate $U(1)$ flip for TRS. 
	[b] Connecting matrices. }
	\label{figs/pt_disjoint_bdi}
\end{figure}

\end{widetext}

\subsection{$(1+1)d$ class BDI}
\label{sec:(1+1)BDI}

Let us consider $(1+1)d$ superconducting chains with TRS $T^2=1$. 
The pioneering work by Fidkowski and Kitaev
\cite{2010PhRvB..81m4509F}
showed that the ground state of 
the eight copies of the Kitaev Majorana chain~\cite{kitaev2001unpaired} 
can be adiabatically connected into a trivial ground state.
In the subsequent work,
\cite{Fidkowski2011,2011PhRvB..83g5102T}
the operator algebra with fermion parity
and TRS realized in the entanglement Hilbert space is classified by $\Z_8$. 
More recently, Kapustin {\it et al.}~\cite{Kapustin2015a} argued that 
the $\Z_8$ classification is identified with the pin$_-$ cobordism, 
which gives the classification of
the cobordism invariant topological actions
of Euclidean quantum field theory of real fermions
with reflection symmetry $R^2=(-1)^F$ or $T^2=+1$. 
See also \cite{KapustinTurzilloYou2016spin, BultinckWilliamsonHaegemanVerstraete2016fermionic}.

For our purpose, the generating manifold $\mathbb{R}P^2$ of the pin$_-$
cobordism $\Omega_2^{\pin_-} = \Z_8$ plays an important role.  
The $\Z_8$ classification implies that the partition function
of Euclidean quantum field theory on $\mathbb{R}P^2$
is given by 
\begin{align}
Z(\mathbb{R}P^2,\pm) = |Z(\mathbb{R}P^2,\pm)| e^{\pm i 2 \pi \nu /8}, \quad 
\nu \in \{0, \dots, 7\}, 
\end{align}
where $\pm$ specifies one of two pin$_-$ structure
on $\mathbb{R}P^2$,
~\footnote{There is a bijection between the set of pin$_-$ structures
  on $M$ and $H^1(M;\Z_2)$. 
For $\mathbb{R}P^2$, $H^1(M;\Z_2)=\Z_2$.}
and $\nu \in \Z_8$ labels distinct topological phases.
As discussed in the previous section,
this path-integral can be simulated by using 
partial transpose for disjoint interval,
and hence the many-body topological invariant
in the operator formalism is given by
Eq.\ \eqref{quantity} with
$\rho'=\rho=|GS\rangle\langle  GS|$:
\begin{align}
  Z=\Tr^{\ }_I \left[ \rho^{\ }_I C_T^{I_1} \rho_I^{\sf T_1} [C_T^{I_1}]^{\dag} \right].
  \label{BDI inv}
\end{align}
Moreover, the $\Z_4$ subgroup is generated
by the Klein bottle with periodic boundary condition
in the $S^1$ direction.~\cite{shiozaki2016many}
As we discussed in the previous subsection
such a spacetime manifold can be prepared
by taking partial transpose for disjoint intervals
with the fermion parity twist in the intermediate
region~\cite{ShiozakiRyu2016}:
\begin{align}
  Z=\Tr^{\ }_{I_1\cup I_3} \left[ \rho^{\ }_{I_1\cup I_3} (e^{-i \pi})
  C_T^{I_1} \rho_{I_1\cup I_3}^{\sf T_1}(e^{+i\pi}) [C_T^{I_1}]^{\dag} \right].
  \label{BDI inv, Klein}
\end{align}

\subsubsection{Numerical calculations}
\label{sec:NumcalcBDI}
A canonical model of non-trivial SPT phases
in this symmetry class is given by the Kitaev chain~\cite{kitaev2001unpaired}
\begin{align} \label{eq:BDI}
{H}= 
-\sum_{i} \Big[t f_{i+1}^\dagger f^{\ }_{i}+\Delta f_{i+1}^\dagger f^\dagger_{i} +\text{H.c.}\Big] -\mu \sum_{i} f_i^\dagger f^{\ }_i,
\end{align}
which describes a superconducting state  of spinless fermions.
Time-reversal symmetry can be introduced as 
\begin{align}
T f^{\dag}_i T^{-1} = f^{\dag}_i, \quad T^2 = 1. 
\end{align}
For simplicity, we set $t=\Delta$ in the following.
($\Delta$ is taken as a real parameter unless stated otherwise.)
The SPT phase in this model is realized when
$|\mu|/t < 2$ and protected by the time-reversal.

\begin{figure}[!]
	\includegraphics[scale=.4]{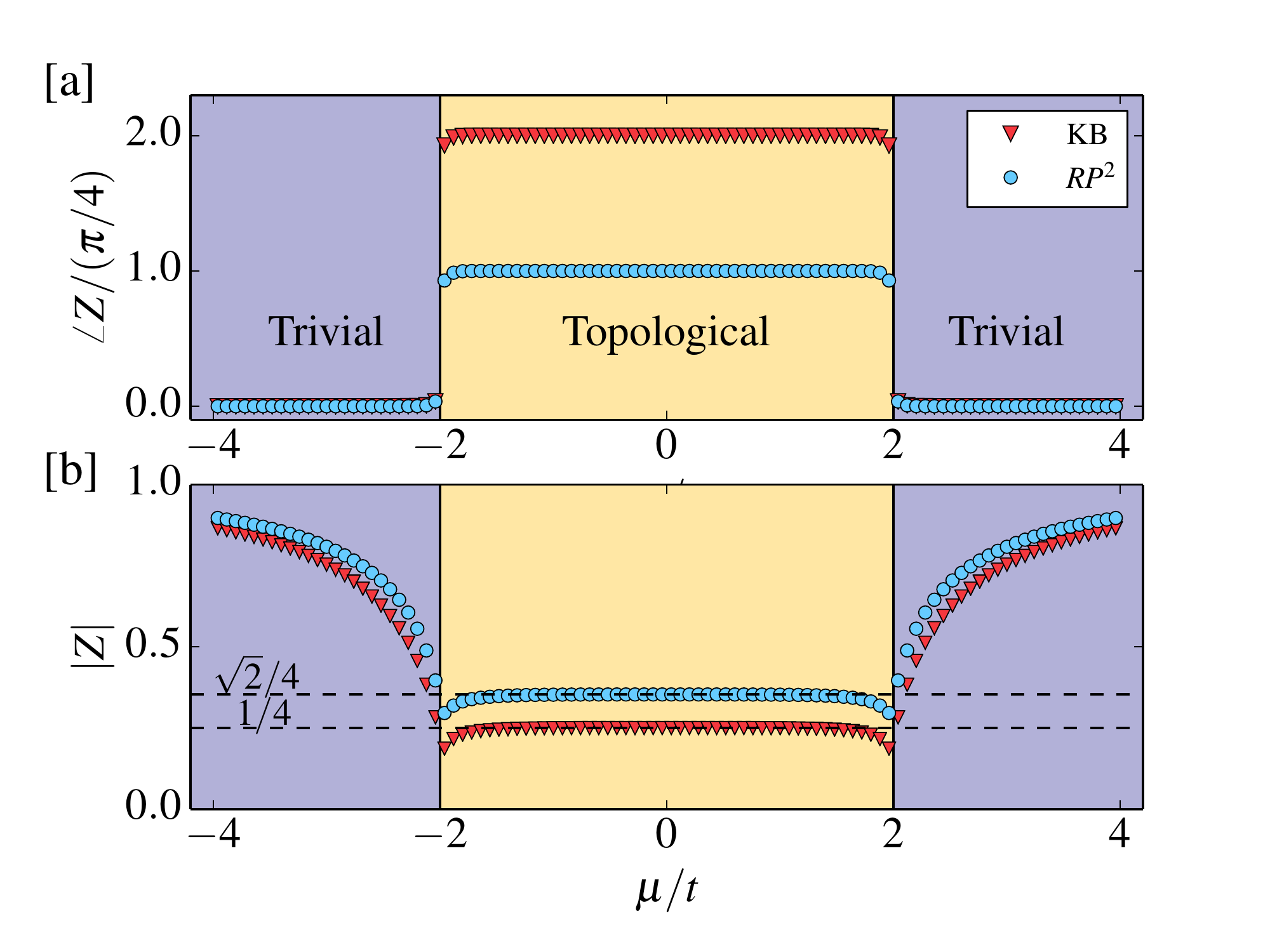}
	\caption{	\label{fig:BDI_num}
	[a] Complex phase and [b] amplitude of
  the $\mathbb{R}P^2$ \eqref{BDI inv}
  and
  Klein bottle \eqref{BDI inv, Klein} 
  partition functions
  in class BDI. 
	The generating model is given by Eq.~(\ref{eq:BDI}).
	The total number of sites is $N=80$. Each of the $I_1$, $I_2$ and $I_3$ (only Klein bottle) intervals has $20$ sites.}
\end{figure}

Figure~\ref{fig:BDI_num} shows
the evaluated complex phase and amplitude of adjacent intervals
\eqref{BDI inv} and disjoint intervals \eqref{BDI inv, Klein}
in the Ramond sector (the $r$ sector).
In the non-trivial phase, we observe that the quantization of the complex phase matches with the $e^{i\pi/4}$ and $e^{i\pi/2}$ phases associated with the spacetime manifolds $\mathbb{R}P^2$ and Klein bottle, respectively. 

It is also interesting to observe that
the amplitude reaches $1$ deep in the trivial limit, while it plateaus at values $\sqrt{2}/{4}$ and $1/4$ in the non-trivial phase. We attribute this to the fact that we choose a normalization convention such that $\Tr(\rho_I)=1$ where the spacetime manifold is actually not a sphere but a torus. We find that the overall amplitude is multiplied by the factor $d^{-(n-1)}$ in the non-trivial phase where $n$ is the number of cycles of the spacetime manifold and $d=\sqrt{2}$ is the quantum dimension of the Majorana edge modes.
 Let us check this for various examples: For the partial reflection (for class D in the presence of reflection which is CRT equivalent of BDI), manifold is $\mathbb{R}P^2$ with a handle which gives two independent cycles, hence, the amplitude is $1/\sqrt{2}$ (as expected, see Refs.~\onlinecite{ShapourianShiozakiRyu2016detection,shiozaki2016many}). For adjacent intervals, there are two handles (one per $\rho_I$) giving 4 cycles overall; thus, the amplitude is $\sqrt{2}/4$. For disjoint intervals there is an additional cycle (associated with the Klein bottle) leading to overall 5 cycles, which implies that the amplitude is $1/4$. The $d^{-(n-1)}$ factor can also be applied to other symmetry classes where $d$ is not necessarily $\sqrt{2}$ and does not depend on the boundary conditions along cycles. For clarity, we summarize all these results in Table.~\ref{tab:amplitude}.

\begin{table}[!]
\begin{center}
\begin{tabular}{|c|c|c|c|c|c|}
\hline
AZ class & Edge mode & Quant. dim $(d)$ & Adjacent & Disjoint \\
\hline
BDI & Majorana Fermion  & $\sqrt{2}$ & $\sqrt{2}/4$ & $1/4$  \\ \hline 
DIII & Complex Fermion  & $2$ & $-$ & $1/16$  \\ \hline 
AIII & Complex Fermion  & $2$ & $1/8$ & $1/16$  \\ \hline 
\end{tabular}
\caption{
\label{tab:amplitude}
Amplitude of
the partition functions on
$\mathbb{R}P^2$
and 
the Klein bottle
for various symmetry classes in $(1+1)d$
is given by $d^{-(n-1)}$ where $d$ is the quantum dimension of the edge modes and $n$ is the number of one-cycles of the spacetime manifold. There are 4 and 5 cycles in the manifolds associated with adjacent and disjoint intervals, respectively. }
\end{center}
\end{table}

\subsubsection{Analytical calculations for the fixed-point wave function}

We can also verify analytically the numerical results
by using the fixed-point wave function with vanishing correlation length.
This state is realized as the ground state of the Hamiltonian (\ref{eq:BDI})
in the $\mu=0$ limit, 
\begin{align}
H = - \sum_i \big[ f^{\dag}_i f_{i+1} + f_i f_{i+1} + \text{H.c.} \big]. 
\label{eq:bdi_hamiltonian}
\end{align}

\paragraph{\it $\Z_8$ invariant: Partition function on real projective plane}
Let $I = I_1 \cup I_2$ be two adjacent intervals on closed chain $S^1$. 
In accordance with the cut and glue construction~\cite{Qi2011b} of the reduced density matrix, 
we focus on the 6 real fermions at the boundary of $I$ as in 
\begin{align}
\cdots c^R_0 \ \ \underbrace{c^L_1 \cdots  \cdots c^R_1}_{I_1} \ \ \underbrace{c^L_2 \cdots \cdots c^R_2}_{I_2} \ \ c^L_0 \cdots 
\end{align}
Introducing the complex fermions inside the intervals as 
\begin{align}
f_i^{\dag} = \frac{c^R_i + i c^L_i}{2}, \quad (i=0,1,2), 
\end{align}
we have the gluing Hamiltonian which is
essentially identical to (\ref{eq:bdi_hamiltonian}), 
\begin{align}
H = i \sum_{i=0}^{2} c^R_i c^L_{i+1} 
= - \sum_{i=0}^{2} \big[ f^{\dag}_i f_{i+1} + f_i f_{i+1} + h.c. \big], 
\end{align}
where $c^L_3 = c^L_0, f_3 = f_0$. 
This is the fixed point (which here means zero correlation length) Kitaev chain with the periodic boundary condition. 
The ground state is given by 
\begin{align}
\ket{\Psi} = \frac{1}{2} \big[ (1-f^{\dag}_1 f^{\dag}_2) f^{\dag}_0 + (f^{\dag}_1+f^{\dag}_2) \big] \ket{0}, 
\end{align}
where $\ket{0}$ is the Fock vacuum of $f_i$ fermions. 
The reduced density matrix for the adjacent intervals $I$ is given by 
\begin{align}
\rho_{I} 
= \Tr_{0} \big( \ket{\Psi} \bra{\Psi} \big) 
= \frac{1}{4}(1-i c^R_1 c^L_2). 
\end{align}
The unitary part $C_T^{I_1}$ of $T$ is given by $C_T^{I_1}=c_1^R$ 
and the fermionic partial transpose is
$\rho_I^{\sf T_1} = \frac{1}{4}(1+c^R_1 c^L_2)$. 
We hence obtain 
\begin{align}
\Tr_I \left(\rho_I C^{I_1}_T \rho_I^{\sf T_1} [C_T^{I_1}]^{\dag}\right) 
&= \frac{1}{2 \sqrt{2}} e^{-\pi i/4}. 
\end{align}
This $\Z_8$ phase agrees with the pin$_-$ cobordsim group $\Omega^{\pin_-}_2 = \Z_8$. 

\paragraph{\it $\Z_4$ invariant: Partition function on the Klein bottle}
\label{Z4 invariant from disjoint partial transpose: partition function on Klein bottle}
Let $I=I_1\cup I_2 \cup I_3$ be three adjacent intervals on closed chain $S^1$. 
In a way similar to the previous calculation, 
we focus on the 8 Majorana fermions at the boundary of three intervals $I_1, I_2$ and $I_3$: 
\begin{align}
\cdots c^R_0 \ \ \underbrace{c^L_1 \cdots  \cdots c^R_1}_{I_1} \ \ \underbrace{c^L_2 \cdots \cdots c^R_2}_{I_2} \ \ \underbrace{c^L_3 \cdots \cdots c^R_3}_{I_3} \ \ c^L_0 \cdots 
\end{align}
Introducing the complex fermions inside the intervals as 
\begin{align}
f_i^{\dag} = \frac{c^L_i + i c^R_i}{2}, \quad (i=0,1,2,3), 
\end{align}
and the gluing Hamiltonian as 
\begin{align}
H =i \sum_{i=0}^{3} c^R_i c^L_{i+1} 
= -  \sum_{i=0}^{3} \big[ f^{\dag}_i f_{i+1} + f_i f_{i+1} + h.c. \big], 
\end{align}
where $c^L_4 = c^L_0, f_4 = f_0$, 
the ground state of $H$ is given by 
\begin{align}
\ket{\Psi}
  &= \frac{1}{\sqrt{8}} [ (f^{\dag}_1 + f^{\dag}_3) + (1 + f^{\dag}_1 f^{\dag}_3) f^{\dag}_2
  + (1- f^{\dag}_1 f^{\dag}_3) f^{\dag}_0 + (f^{\dag}_3 - f^{\dag}_1) f^{\dag}_2 f^{\dag}_0 ] \ket{0}. 
\end{align}
We introduce the reduced density matrix $\rho_{I_1 \cup I_3}\big( (-1)^{F_2} \big)$ on the disjoint intervals $I_1 \cup I_3$ with fermion party twist on the $I_2$ interval by 
\begin{equation}\begin{split}
\rho_{I_1 \cup I_3}\big( (-1)^{F_2} \big)
&= \Tr_{0,2} \big( (-1)^{f^{\dag}_2 f_2} \ket{\Psi} \bra{\Psi} \big) \\
&= \frac{i}{4} c^R_1 c^L_3. 
\end{split}\end{equation}
We have the partition function on the Klein bottle from the partial transposition as 
\begin{align}
  \Tr_{1,3}\Big[ \rho_{I_1 \cup I_3}\big( (-1)^{F_2} \big) C^{I_1}_T \rho^{\sf T_1}_{I_1 \cup I_3}\big( (-1)^{F_2}
  \big) [C_T^{I_1}]^{\dag} \Big] 
= \frac{1}{4} e^{-\pi i/2}. 
\end{align}
This is the $\Z_4$ invariant as expected. 
On the other hand, 
if we do not associate the fermion parity twist on $I_2$, 
the reduced density matrix on $I_1 \cup I_3$ is unentangled one: 
\begin{align}
  \rho_{I_1 \cup I_3} &= \Tr_{0,2} \big( \ket{\Psi} \bra{\Psi} \big) 
      = \frac{1}{4}, 
\end{align}
which leads to a trivial topological $U(1)$ phase of the Klein bottle partition function 
\begin{align}
\Tr_{1,3}\Big( \rho_{I_1 \cup I_3} C^{I_1}_T \rho^{\sf T_1}_{I_1 \cup I_3} [C_T^{I_1}]^{\dag} \Big) 
=\frac{1}{4}. 
\end{align}
This is consistent with the full reflection~\cite{shiozaki2016many} in class D with reflection symmetry which is the CRT dual of class BDI.

\subsection{$(1+1)d$ class DIII}
\label{sec:(1+1)DIII}

The pin$_+$ cobordism group in 2-spacetime dimensions is given by $\Omega^{\pin_+}_2 = \Z_2$, 
which is generated by the Klein bottle with periodic boundary condition in the $S^1$ direction.~\cite{Kirby, shiozaki2016many}. 
The many body $\Z_2$ invariant is constructed in a similar way
to Sec.~\ref{Z4 invariant from disjoint partial transpose: partition function on Klein bottle}: 
the topological invariant can be constructed
by considering the disjoint intervals (\ref{eq:Klein}) in the $r$ sector,
\begin{align} \label{eq:ZDIII1d}
  Z=
\Tr_{1,3}\Big[ \rho_{I_1 \cup I_3}\big( (-1)^{F_2} \big) C^{I_1}_T \rho^{\sf T_1}_{I_1 \cup I_3}\big( (-1)^{F_2} \big) [C_T^{I_1}]^{\dag} \Big] 
\end{align}
for a given pure state $\ket{\Psi}$.

\subsubsection{Numerical calculations}
A generating model of non-trivial SPT phases
in this symmetry class is given by two copies of the Kitaev Majorana chain Hamiltonian
\begin{align} \label{eq:DIII}
{H}&= 
-\mu \sum_{i\sigma} f_{i\sigma}^\dagger f^{\ }_{i\sigma}
-t \sum_{i\sigma} \left[ f_{i+1\sigma}^\dagger f^{\ }_{i\sigma}+\text{H.c.} \right] 
+i\Delta \sum_{i} \left[ f_{i+1\uparrow}^\dagger f^\dagger_{i\uparrow}-f_{i+1\downarrow}^\dagger f^\dagger_{i\downarrow} +\text{H.c.}\right] 
\end{align}
which describes a superconducting state  of spinful fermions, and the time-reversal symmetry is defined as
\begin{align}
T f^{\dag}_{i\ua} T^{-1} = -f^{\dag}_{i\da}, && 
T f^{\dag}_{i\da} T^{-1} = f^{\dag}_{i\ua}, && 
T^2 = (-1)^F, 
\end{align}
and hence, the unitary matrix $\cU_T$ associated with $T$ is 
\begin{align}
\mathcal{U}_T = \begin{pmatrix}
0 & 1 \\
-1 & 0 
\end{pmatrix} 
= i \sigma_y
\end{align} 
in the basis of $(\ua, \da)$. 

For simplicity, we set $t=\Delta$ in the following.
The SPT phase of the above model is realized when
$|\mu|/t < 2$ and protected by the time-reversal.
Figure~\ref{fig:DIII_num} shows the evaluated complex phase and amplitude of disjoint intervals (\ref{eq:ZDIII1d}) in both $r$ and $ns$ sectors, corresponding to periodic and anti-periodic boundary conditions along the time direction for the intermediate interval.
In the non-trivial phase and with periodic boundary condition ($r$ sector), we observe that the $\pi$ phase matches with the $\Z_2$ classification generated by putting on the  Klein bottle spacetime manifold. 
Regarding the amplitude, it reaches $1$ deep in the trivial limit, while it is quantized at $1/16$ in the non-trivial phase as explained in Sec.~\ref{sec:NumcalcBDI} (see also Table~\ref{tab:amplitude}).

\begin{figure}[!]
	\includegraphics[scale=.4]{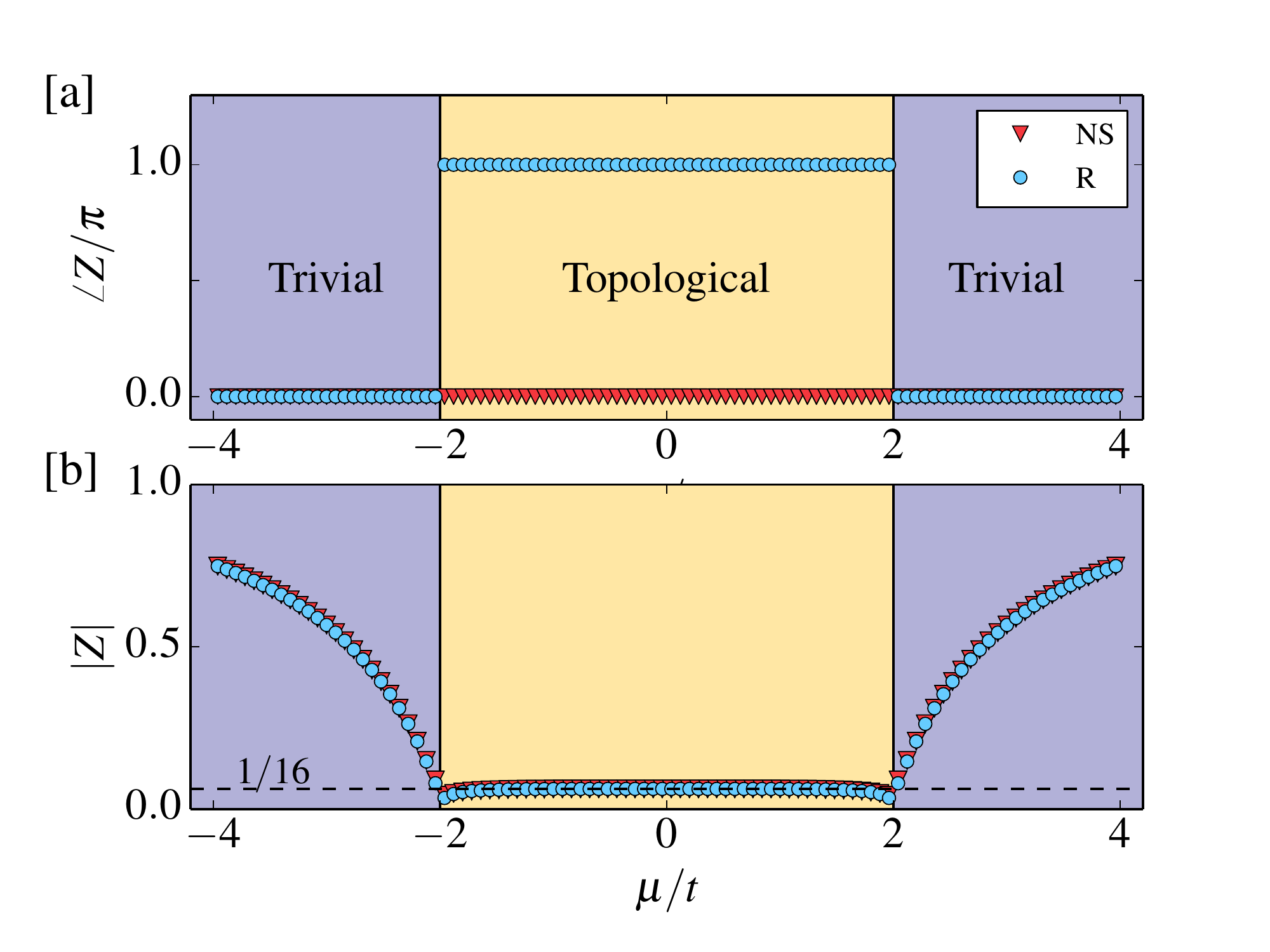}
	\caption{	\label{fig:DIII_num}
[a] Complex phase and [b] amplitude of
\eqref{eq:ZDIII1d} in class DIII. 
	The generating model is given by Eq.~(\ref{eq:DIII}). $R$ and $NS$ refer to periodic and antiperiodic boundary conditions in time direction for the intermediate interval.
	For each spin, the total number of sites is $N=80$. Each of the $I_1$, $I_2$ and $I_3$ intervals has $20$ sites.}
\end{figure}

\subsubsection{Analytical calculations for the fixed-point wave function}
Here, we verify the numerical results in the previous section. Let us consider the fixed-point wave function where the correlation length is zero. This state is realized as the ground state of the Hamiltonian (\ref{eq:DIII}) in the $\mu=0$ limit, that is
\begin{align} \label{eq:DIII_H_fp}
  H &= - \sum_{i\sigma} \big[ f^{\dag}_{i\sigma} f_{i+1,\sigma} + \text{H.c.} \big]
  +
i \sum_{i\sigma} \big[  f^\dag_{i\uparrow} f^\dag_{i+1\uparrow}-f^\dag_{i\downarrow} f^\dag_{i+1\downarrow} + \text{H.c.} \big]. 
\end{align}


\begin{widetext}
Let $I = I_1 \cup I_2 \cup I_3$ be three adjacent intervals on closed chain $S^1$. 
We focus on the $8$ Kramers pairs of Majorana fermions at the boundary of three intervals $I_1, I_2$ and $I_3$: 
\begin{align}
&\cdots c^R_{0,\ua}\ \ \ c^L_{1,\ua} \cdots  \cdots c^R_{1,\ua}\ \ c^L_{2,\ua} \cdots \cdots c^R_{2,\ua} \ \ \ c^L_{3,\ua} \cdots \cdots c^R_{3,\ua}\ \ c^L_0 \cdots \\
&\cdots c^R_{0,\da} \ \ \underbrace{c^L_{1,\da} \cdots  \cdots c^R_{1,\da}}_{I_1} \ \ \underbrace{c^L_{2,\da} \cdots \cdots c^R_{2,\da}}_{I_2} \ \ \underbrace{c^L_{3,\da} \cdots \cdots c^R_{3,\da}}_{I_3} \ \ c^L_{0,\da} \cdots 
\end{align}
We introduce complex fermions inside the intervals as 
\begin{align}
f_{i,\sigma}^{\dag} = \frac{c^R_{i,\sigma} + i c^L_{i,\sigma}}{2}, \quad (i=0,1,2,3, \sigma = \ua \da). 
\end{align}
The gluing Hamiltonian is 
\begin{align}
H &= \frac{i}{2} \sum_{\sigma = \ua,\da} \sum_{i=0}^{3} c^R_{i\sigma} c^L_{i+1\sigma} 
= - \frac{1}{2} \sum_{\sigma=\ua,\da} \sum_{i=0}^{3} \big[ f^{\dag}_{i,\sigma} f_{i+1,\sigma} + f_{i,\sigma} f_{i+1,\sigma} + h.c. \big], 
\end{align}
where $c^L_{4,\sigma} = c^L_{0,\sigma}, f_{4,\sigma} = f_{0,\sigma}$. 
By a unitary gauge transformation
$f_{i,\uparrow}\to e^{-i\pi/4} f_{i,\uparrow}$ and
$f_{i,\downarrow}\to e^{i\pi/4} f_{i,\downarrow}$,
the pairing part of this Hamiltonian 
can be brought into the form identical to \eqref{eq:DIII_H_fp}.
The ground state is given by 
\begin{align}
\ket{\Psi}
= \frac{1}{8} \prod_{\sigma=\ua,\da} \big[ 
(f^{\dag}_{1,\sigma} + f^{\dag}_{3,\sigma}) + (1 + f^{\dag}_{1,\sigma} f^{\dag}_{3,\sigma}) f^{\dag}_{2,\sigma} + (1- f^{\dag}_{1,\sigma} f^{\dag}_{3,\sigma}) f^{\dag}_{0,\sigma} + (f^{\dag}_{3,\sigma} - f^{\dag}_{1,\sigma}) f^{\dag}_{2,\sigma} f^{\dag}_{0,\sigma} \big] \ket{0}. 
\end{align}
The reduced density matrix $\rho_{I_1 \cup I_3}\big( (-1)^{F_2} \big)$ on the disjoint intervals $I_1 \cup I_3$ with fermion party twist on $I_2$ is given by 
\begin{align}
\rho_{I_1 \cup I_3}\big( (-1)^{F_2} \big)
= \Tr_{0,2} \big( (-1)^{f^{\dag}_2 f_2} \ket{\Psi} \bra{\Psi} \big) 
= \frac{1}{16} (i c^R_{1,\ua} c^L_{3,\ua} )( i c^R_{1,\da} c^L_{3,\da}). 
\end{align}
Noticing that the unitary part of TRS on subsystem $I_1$ is $C_T^{I_1} = e^{\frac{\pi}{4}(c^R_{1\ua} c^R_{1\da} - c^L_{1\ua}c^L_{1\da})}$, 
it holds that 
\begin{align}
\Tr_{I_1 \cup I_3} \Big[ \rho_{I_1 \cup I_3}\big( (-1)^{F_2} \big) C_T^{I_1} \rho^{\sf T_1}_{I_1 \cup I_3}\big( (-1)^{F_2} \big) [C_T^{I_1}]^{\dag} \Big]
= - \frac{1}{16}. 
\end{align}
This is consistent with the that is the $\Z_2$ invariant of the partition function on Klein bottle.
\end{widetext}

\subsection{$(1+1)d$ class AIII}
\label{sec:(1+1)AIII}

Class AIII insulators
are invariant under an antiunitary particle-hole symmetry (PHS), which does not flip $U(1)$ charge of
complex fermions,
$e^{i Q}\to e^{-i (-Q)}=e^{iQ}$,
and are defined on a space manifold with a spin$^c$ structure.
In Wick-rotated Euclidean spacetime, the corresponding structure is equivalent to pin$^c$. 
In $(1+1)d$ spacetime dimensions, the topological classification is given by $\Omega^{\pin^c}_2 = \Z_4$ and generated by $\mathbb{R}P^2$. 
The topological invariant is given in terms of
the partition function on $\mathbb{R}P^2$,
which can be recast in the operator formalism
by using partial particle-hole transformation,
in the way similar to the case of partial time-reversal
discussed in Sec. \ref{Two adjacent intervals: cross-cap}:
\begin{align}
Z=\Tr \big[ \rho_I U_{CT}^{I_1} \rho_I^{\sf T_1} [U_{CT}^{I_1}]^{\dag} \big],
  \label{AIII inv}
\end{align}
where $I$ consists of two adjacent interval, $I=I_1\cup I_2$.

As pointed out in Sec.~\ref{ex:rp2_pinc},
the $U(1)$ holonomy along the $\Z_2$ cycle of $\R P^2$ should
be properly chosen to yield the correct formula for
the many-body $\Z_4$ invariant.
In the current case with $(CT)^2 = 1$,
the suitable choice is $\pm i$ holonomy and
the $\pm i$ phase rotation is already involved in the definition (\ref{AIII inv}). 

As discussed
in Sec.~\ref{sec:Partial antiunitary particle-hole transformation},
the partition function \eqref{AIII inv}
can be expressed in the coherent state basis as follows.
We start from a reduced density matrix 
\begin{widetext}
\begin{align}
\rho_I 
&= \int \prod_{i} d \bar \alpha_i d \alpha_i d \bar \beta_i d \beta_i e^{- \sum_i (\bar \alpha_i \alpha_i + \bar \beta_i \beta_i)} \rho_I(\{\bar \alpha_i\}, \{\beta_i\}) \ket{\{\alpha_i\}} \bra{\{\bar \beta_i\}}.
\end{align}
The partial anti-unitary particle-hole transformed reduced density matrix is given by 
\begin{equation}\begin{split}
U_{CT}^{I_1} \rho_I^{\sf T_1} [U_{CT}^{I_1}]^{\dag} 
&= \int \prod_{i} d \bar \gamma_i d \gamma_i d \bar \delta_i d \delta_i e^{- \sum_i (\bar \gamma_i \gamma_i + \bar \delta_i \delta_i)} 
\rho_I(\{\bar \gamma_i\}, \{\delta_i\}) \\
&\times C_f^{I_1} \ket{\{ -i \bar \delta_j [\cU_{CT}]_{ji} \}_{i \in I_1}, \{\gamma_i\}_{i \in I_2}} \bra{\{-i [\cU^{\dag}_{CT}]_{ij} \gamma_j \}_{i \in I_1}, \{\bar \delta_i\}_{i \in I_2}} [C_f^{I_1}]^{\dag}, 
\label{eq:C_opt}
\end{split}\end{equation}
where $C_f^{I_1} = (f^{\dag}_1 + f_1) \cdots (f^{\dag}_{N_1} + f_{N_1})$ is the partial particle-hole transformation on $I_1= \{1, \dots, N_1\}$. 
(Here we assumed that $\rho_I$ is Grassmann even.) 
We need to know the matrix element of the particle-hole transformation. 
It is sufficient to check it for the coherent state $\ket{\alpha}=e^{-\alpha f^{\dag}}\ket{0}$ of a one complex fermion. 
For $C=f^{\dag}+f$, it holds that 
\begin{align}
\braket{\alpha|C|-\beta}
=\Tr \big[ \ket{\beta} \bra{\alpha} C^{\dag} \big]
=\alpha-\beta.
\end{align}
This is the delta function of the Grassmann variables. 
It is natural in the view point of the pin$^c$ structure. 
If we associate the additional $U(1)$ phase twist, we have 
\begin{align}
\braket{\alpha|e^{i \theta f^{\dag} f} C|-\beta}
=e^{i \theta} \alpha-\beta, 
\end{align}
which is nothing but the coordinate transformation of pin$^c$ structure with $U(1)$ twist. 
From a straightforward calculation, we obtain the coherent state formula 
\begin{align}
&\Tr \big[ \rho_I U_{CT}^{I_1} \rho_I^{\mathsf{T}_1} [U_{CT}^{I_1}]^{\dag} \big]  \nonumber \\
&= \int 
\prod_{i \in I_1 \cup I_2} [d \alpha_i d \beta_i d \gamma_i d \delta_i] 
\prod_{i \in I_1} \big[ (\delta_i+[i \cU_{CT}]_{ij} \beta_j)(\gamma_i-\alpha_j [i \cU_{CT}^{\dag}]_{ji}) \big]
e^{\sum_{i \in I_2} (\alpha_i \delta_i + \beta_i \gamma_i)} \nonumber \\
& 
\quad \times  
\rho_I(\{\alpha_i\}; \{\beta_i\}) \rho_I(\{\gamma_i\}; \{\delta_i\}) .
\label{eq:adjacent_aiii}
\end{align}
See Fig.~\ref{figs/pt_adjacent_aiii} for a network representation. 
\begin{figure}[!]
	\begin{center}
	\includegraphics[width=0.8\linewidth, trim=0cm 0cm 0cm 0cm]{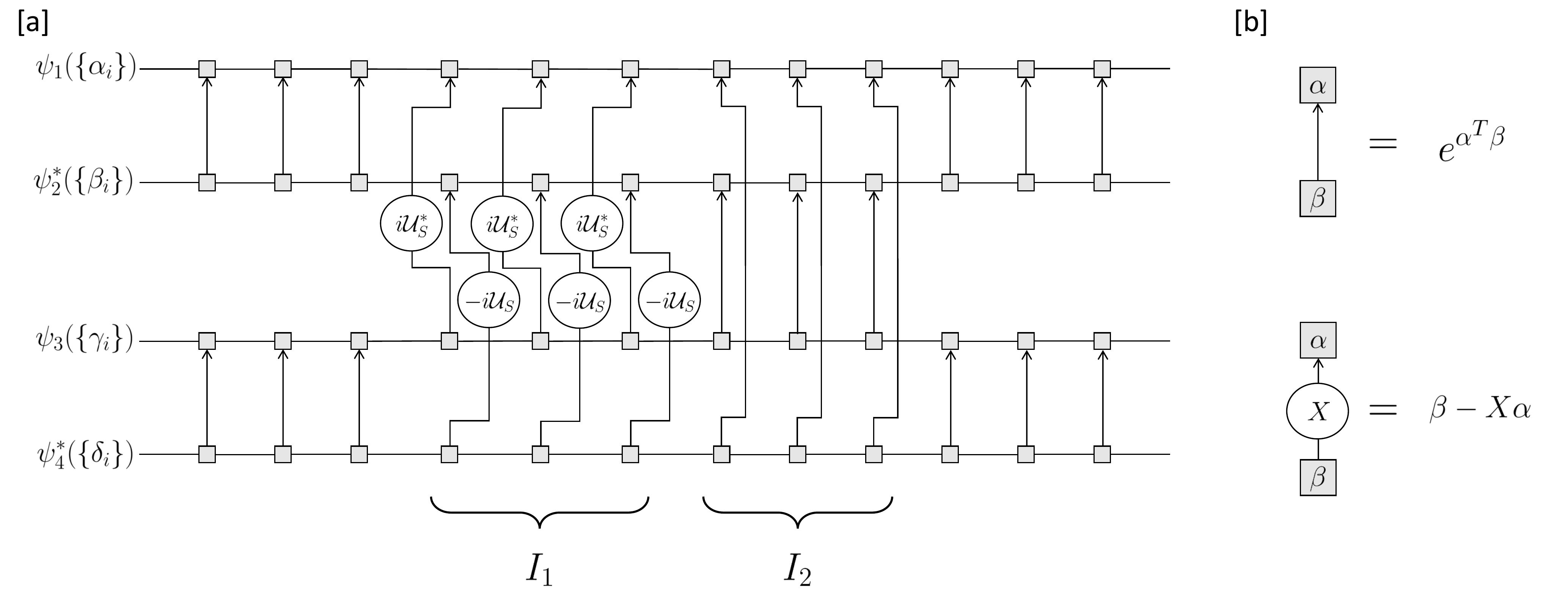}
	\end{center}
	\caption{
	[a] Adjacent partial transpose for class AIII chiral symmetry (antiunitary particle-hole symmetry).
	[b] Connecting matrices. }
	\label{figs/pt_adjacent_aiii}
\end{figure}

\end{widetext}

\subsubsection{Numerical calculations}
A canonical model of non-trivial SPT phases
in this symmetry class is given by the Su-Schrieffer-Heeger (SSH) model,
\begin{align} \label{eq:AIII}
{H}= - \sum_j [t_2 f^{\dag}_{j+1} g_j +t_1 f^{\dag}_j g_j+ \text{H.c.}],
\end{align}
where there are two fermion species living on each site $f_j$ and $g_j$.
The antiunitary PHS $S$ is defined by 
\begin{align}
S f^{\dag}_i S^{-1} = f_i, \quad  
S g^{\dag}_i S^{-1} = -g_i, \quad  
S i S^{-1} = -i.
\end{align}
This model realizes two topologically distinct phases: Topologically non-trivial phase for $t_2>t_1$, where the open chain has localized fermion zero-modes at the boundaries, and trivial phase for $t_2<t_1$ with no boundary mode.

In Fig.~\ref{fig:AIII_num}, the complex phase and amplitude of adjacent intervals (\ref{eq:adjacent_aiii}) are shown (the blue circles denoted by $S$), in the non-trivial phase we observe that the $e^{i\pi/2}$ phase which matches the $\Z_4$ classification generated by putting on the $\mathbb{R}P^2$ spacetime manifold. Moreover, the amplitude asymptotes to $1$ in the trivial limit, while it is $1/8$ in the non-trivial phase consistent with the previous discussion in Sec.~\ref{sec:NumcalcBDI} (see also Table~\ref{tab:amplitude}). We also show a reference curve denoted by $T$, where we do not include the particle-hole transformation (Eq.~(\ref{eq:C_opt}) without $C_f^{I_1}$), where the amplitude remains identical to that of $S$ curve while there is no complex phase.
This means that we must consider the original symmetry transformation, as defined for the symmetry class, within our partial transformation scheme in order to obtain the complex phase associated with the topological classification.

\begin{figure}[!]
	\includegraphics[scale=.4]{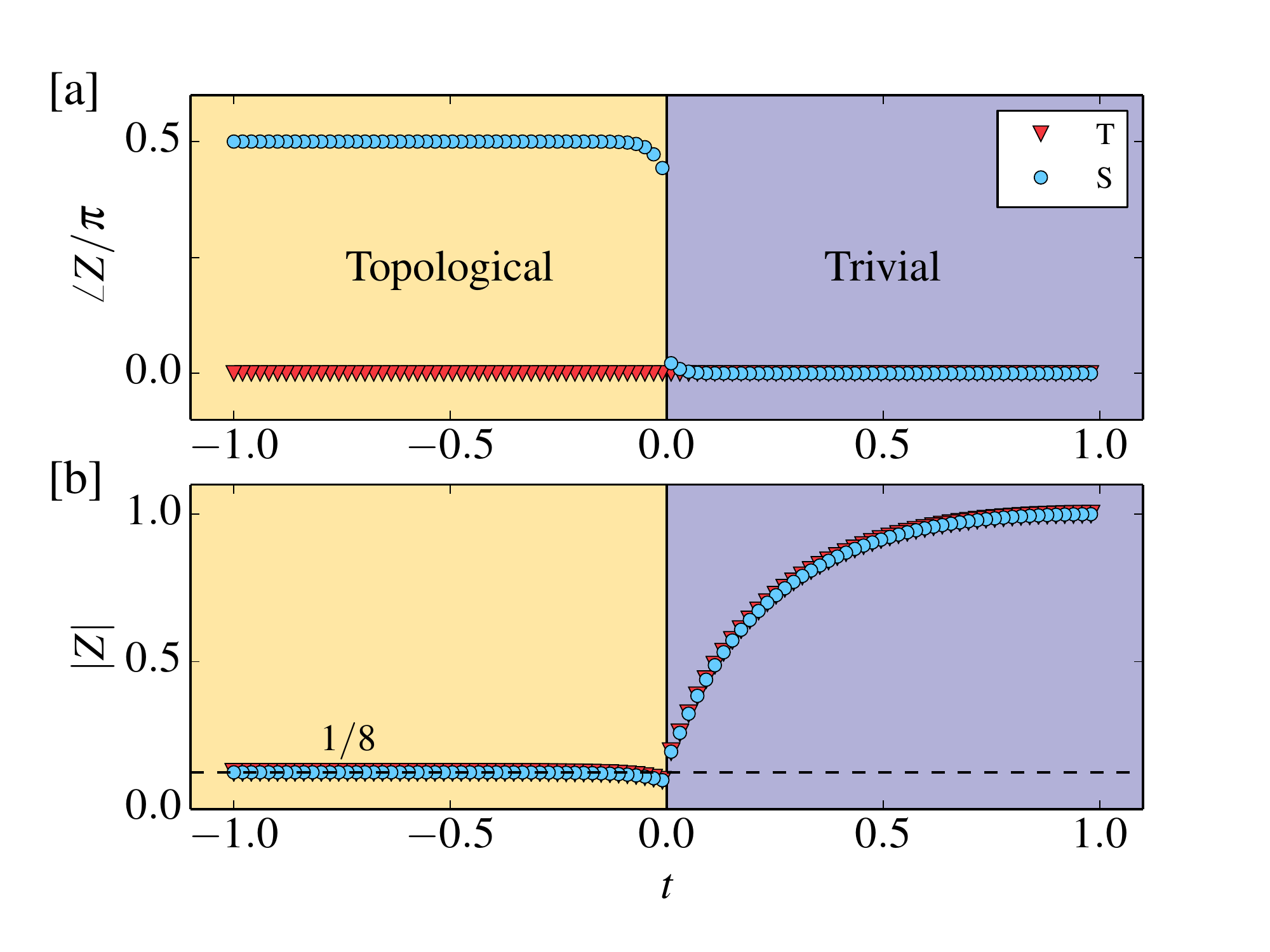}
	\caption{	\label{fig:AIII_num}
	[a] Complex phase and [b] amplitude of
  \eqref{AIII inv} in class AIII, denoted by $S$. For reference, we also include the partial transformation with no particle-hole transformation (without the $C$ operator in (\ref{eq:C_opt})), denoted by $T$.
	As a model Hamiltonian we use (\ref{eq:AIII}) with the parameterization $t_2=(1-t)/2$ and $t_1=(1+t)/2$.  The total number of sites is $N=80$. Here, $I_1$, $I_2$ and $I_3$ each has $20$ sites.}
\end{figure}


\subsubsection{Analytical calculations for the fixed-point wave function}
Here, we show that consistent results can be obtained for the fixed-point wave function and confirm the numerical results in the previous section. This zero-correlation length wave function is the ground state of the Hamiltonian (\ref{eq:AIII}) when $t_1=0$, which is
\begin{align}
H = - \sum_{i=1}^N g^{\dag}_i f_{i+1} + \text{H.c.}.
\end{align}
It is easy to show that the ground state with $N$-particle sector is fully occupied state of the ``bond'' fermions $(g^{\dag}_i+f^{\dag}_{i+1})/\sqrt{2}$, 
\begin{align}
\ket{\Psi} = \frac{1}{(\sqrt{2})^N} (g^{\dag}_1 + f^{\dag}_2) (g^{\dag}_2 + f^{\dag}_3) \cdots (g^{\dag}_N + f^{\dag}_1) \ket{0}. 
\end{align}
Let $I = I_1 \cup I_2$ be adjacent intervals on closed chain $S^1$. 
According to the cut and glue construction for the reduced density matrix, 
we focus on the 6 complex fermions at the boundaries of intervals 
\begin{align}
\cdots g_0 \ \ \underbrace{f_1 \cdots  \cdots g_1}_{I_1} \ \ \underbrace{f_2 \cdots \cdots g_2}_{I_2} \ \ f_0 \cdots 
\end{align}
The gluing Hamiltonian is 
\begin{align}
H = - ( g^{\dag}_0 f_1 + g^{\dag}_1 f_2 + g^{\dag}_2 f_0) + h.c., 
\end{align}
and its ground state is given by 
\begin{align}
\ket{\Psi} = 2^{-3/2} (g^{\dag}_0 + f^{\dag}_1) (g^{\dag}_1 + f^{\dag}_2) (g^{\dag}_2 + f^{\dag}_0) \ket{0}. 
\end{align}
\begin{widetext}
The reduced density matrix is given by 
\begin{equation}\begin{split}
\rho_I 
&= \Tr_0 \ket{\Psi} \bra{\Psi} \\
&= 2^{-3} \Big[ 
\ket{1000}\bra{1000} 
+ \ket{0001}\bra{0001} 
+ \ket{1000}\bra{0001}
+ \ket{0001}\bra{1000} \\
& \ \ \ \ \ \ \ 
+ \ket{1100}\bra{1100}
+ \ket{0101}\bra{0101} 
- \ket{1100}\bra{0101}
- \ket{0101}\bra{1100} \\
& \ \ \ \ \ \ \ 
+ \ket{1010}\bra{1010}
+ \ket{0011}\bra{0011} 
- \ket{1010}\bra{0011}
- \ket{0011}\bra{1010} \\
& \ \ \ \ \ \ \ 
+ \ket{1110}\bra{1110}
+ \ket{0111}\bra{0111} 
+ \ket{1110}\bra{0111}
+ \ket{0111}\bra{1110}
\Big] . 
\end{split}\end{equation}
Here we defined occupied states in the following order 
\begin{align}
\ket{n_1n_2n_3n_4} := (g^{\dag}_1)^{n_1} (f^{\dag}_1)^{n_2} (g^{\dag}_2)^{n_3} (f^{\dag}_2)^{n_4} \ket{0}.
\end{align}
The unitary part $U_{CT}^{I_1}$ of the anti-unitary PHS is given by the $\pi$ phase rotation on the $g_1$ fermion. 
From the formula (\ref{eq:pt_phs_occupation}), we obtain
\begin{equation}\begin{split}
U_{CT}^{I_1} \rho_I^{\sf T_1} [U_{CT}^{I_1}]^{\dag} 
&= 2^{-3} \Big[ 
\ket{0100}\bra{0100} 
+ \ket{1101}\bra{1101} 
+i \ket{1100}\bra{0101}
+i \ket{0101}\bra{1100} \\
& \ \ \ \ \ \ \ 
+ \ket{0000}\bra{0000}
+ \ket{1001}\bra{1001} 
-i \ket{1000}\bra{0001}
-i \ket{0001}\bra{1000} \\
& \ \ \ \ \ \ \ 
+ \ket{0110}\bra{0110}
+ \ket{1111}\bra{1111} 
-i \ket{1110}\bra{0111}
-i \ket{0111}\bra{1110} \\
& \ \ \ \ \ \ \ 
+ \ket{0010}\bra{0010}
+ \ket{1011}\bra{1011} 
+i \ket{1010}\bra{0011}
+i \ket{0011}\bra{1010}
\Big]. 
\end{split}\end{equation}
Hence, 
the non-local order parameter can be computed as
\begin{align}
\Tr_I \big[ \rho_I U_{CT}^{I_1} \rho_I^{\sf T_1} [U_{CT}^{I_1}]^{\dag} \big]
= - \frac{i}{8}. 
\end{align}
This is precisely the $\Z_4$ invariant. 
\end{widetext}

\subsection{$(1+1)d$ class AI}
\label{sec:(1+1)AI}

Let us consider the complex fermion operators $\{f^{\dag}_j, f_j\}$,
and time-reversal symmetry without Kramers degeneracy 
\begin{align}
T f^{\dag}_j T^{-1}=f^{\dag}_k [\cU_T]_{kj}, \quad 
\cU_T^{tr}=\cU^{\ }_T. 
\end{align}
The Wick rotated version of this TRS corresponds to pin$^{\tilde c}_-$ structure
in Euclidean quantum field theory. 
The cobordism in 2d spacetime is given
by $\Omega_2^{\pin^{\tilde c}_-}=\Z \times \Z_2$,~\cite{Freed2016}
which means the existence of $\Z_2$ SPT phases. 
(The free part corresponds to the $\theta$ term.)
This $\Z_2$ phase is interaction enabled in the sense that
it is not obtained as a ground state of free fermion systems. 

The generating manifold of $\Z_2$ subgroup of $\Omega^{\pin^c_-}_2$ is the real projective plane $\mathbb{R}P^2$. 
In the following,
by using a concrete model which is equivalent to the Haldane chain,
we demonstrate that the partial time-reversal transformation
\begin{align}
Z=\Tr_I[\rho_I C_T^{I_1} \rho_I^{\sf T_1} C_T^{I_1}]
\end{align}
on adjacent two intervals provides the many body $\Z_2$ invariant.

\subsubsection{Analytic calculation of the $\Z_2$ invariant}
A nontrivial ground state can be constructed
by representing the AKLT ground state in terms of complex fermions,
as shown by Watanabe and Fu.~\cite{Watanabe-Fu2017}
Let $a_{j\sigma},b_{j\sigma} (\sigma=\ua,\da)$ be bosons on 1d closed chain. 
We impose the TRS on these bosonic degrees of freedom by 
\begin{align}
  &
T a^{\dag}_{j\ua} T^{-1} = -a^{\dag}_{j\da}, \quad 
T a^{\dag}_{j\da} T^{-1} = a^{\dag}_{j\ua},  
  \nonumber \\
  &
T b^{\dag}_{j\ua} T^{-1} = -b^{\dag}_{j\da}, \quad 
T b^{\dag}_{j\da} T^{-1} = b^{\dag}_{j\ua}. 
\end{align}
Thus, each $a_{j\sigma}$ and $b_{j\sigma}$ boson
forms the nontrivial projective representation $\rho$ with $T^2=-1$ of TRS. 
On the other hand,
the representation of the total tensor product space $\rho \otimes \rho^*$
on a given site is linear. 
The fixed point AKLT ground state is given by the product state of singlet pair of trivial representation in $\rho \otimes \rho^*$ in the bonds, 
\begin{align}
\ket{\Psi}_{\rm boson}
=\cdots b^{\dag}_{j-1\sigma} a^{\dag}_{j\sigma} b^{\dag}_{j,\sigma'} a^{\dag}_{j+1\sigma'} b^{\dag}_{j+1,\sigma''} a^{\dag}_{j+2\sigma''} \cdots \ket{\rm vac}
\label{eq:boson_aklt}
\end{align}
up to a normalization factor. 
Notice that this ground state is composed only
of two-particle bosonic states $a^{\dag}_{j\sigma} b^{\dag}_{j\sigma'}$ at each site.  
We replace two-particle states of bosons by those of complex fermions $f^{\dag}_{j\sigma} g^{\dag}_{j\sigma'}$ as 
\begin{align}
  &
a^{\dag}_{j\ua}b^{\dag}_{j\ua}=f^{\dag}_{j\ua}g^{\dag}_{j\ua}, \quad 
a^{\dag}_{j\ua}b^{\dag}_{j\da}=f^{\dag}_{j\ua}g^{\dag}_{j\da}, 
  \nonumber \\
 & 
a^{\dag}_{j\da}b^{\dag}_{j\ua}=-f^{\dag}_{j\da}g^{\dag}_{j\ua}, \quad 
a^{\dag}_{j\da}b^{\dag}_{j\da}=f^{\dag}_{j\da}g^{\dag}_{j\da}. 
\end{align}
The relative minus sign in the third equation is essential
to represent the same ground state with complex fermions with $T^2=1$ 
\begin{align}
  &
Tf^{\dag}_{j\ua}T^{-1}=f^{\dag}_{j\da}, \quad 
Tf^{\dag}_{j\da}T^{-1}=f^{\dag}_{j\ua},  
  \nonumber \\
  &
Tg^{\dag}_{j\ua}T^{-1}=g^{\dag}_{j\da}, \quad 
Tg^{\dag}_{j\da}T^{-1}=g^{\dag}_{j\ua}.
\end{align}
\begin{widetext}
The fermionic ground state is given by 
\begin{align}
\ket{\Psi}_{\rm fermion}
=\cdots g^{\dag}_{j-1\sigma} f^{\dag}_{j\sigma} (-1)^{\chi(\sigma,\sigma')} g^{\dag}_{j,\sigma'} f^{\dag}_{j+1\sigma'} (-1)^{\chi(\sigma',\sigma'')} g^{\dag}_{j+1,\sigma''} f^{\dag}_{j+2\sigma''} \cdots \ket{\rm vac}
\end{align}
with the nonlocal phase dependence 
\begin{align}
(-1)^{\chi(\sigma,\sigma')}=\left\{\begin{array}{ll}
-1 & (\sigma=\da,\sigma'=\ua), \\
1 & ({\rm otherwise}).
\end{array}\right.
\end{align}

Now let us compute the $\Z_2$ invariant. 
Since the ground state is identical to the AKLT state, 
the $\Z_2$ invariant is
given by the Pollmann-Turner invariant $\Tr\big[\rho_I [U^{I_1}_T] \rho_I^{\sf T_1} [U^{I_1}_T]^{\dag} \big]$ 
with two adjacent interval $I=I_1 \cup I_2$.
It was already shown that the ground state (\ref{eq:boson_aklt}) gives rise to $\Tr \big[ \rho_I [U^{I_1}_T] \rho_I^{\sf T_1} [U^{I_1}_T]^{\dag} \big]=-1/8$ by use of partial time-reversal transformation of spin systems.~\cite{Pollmann2012, ShiozakiRyu2016}
However, it is worth demonstrating the same result by our formulation of the fermionic partial transpose since the fermionic partial time-reversal transformation differ from that of spin systems on the sign arising from the fermionic anti-commutation relation. 

The reduced density matrix on the subsystem $I=1, \dots, N$
can be computed explicitly,
\begin{equation}\begin{split}
\rho_I
&= \Tr_{S^1 \backslash I} \Big( \ket{\Psi}_{\rm fermion} \bra{\Psi}_{\rm fermion} \Big) \\
&= \frac{{\cal N}}{4} \sum_{\beta_L,\beta_R,\{\sigma_j\},\{\mu_j\} \in \{\ua, \da\}} 
f(\beta_L,\beta_R,\{\sigma_j\},\{\mu_j\})
\ket{\beta_L\sigma_1\sigma_1\sigma_2\cdots \sigma_{N-1} \beta_R}
\bra{\beta_L\mu_1\mu_1\mu_2\cdots \mu_{N-1}\beta_R} 
\end{split}\end{equation}
in the fermionic occupation number basis, where 
\begin{align}
f(\beta_L,\beta_R,\{\sigma_j\},\{\mu_j\})
=
(-1)^{\chi(\beta_1,\sigma_1)+\chi(\sigma_1,\sigma_2) + \cdots + \chi(\sigma_{N-1},\beta_R)} 
(-1)^{\chi(\beta_1,\mu_1)+\chi(\mu_1,\mu_2) + \cdots +\chi(\mu_{N-1},\beta_R)}, 
\end{align}
and ${\cal N}$ is a normalization constant. 
$\rho_I$ has the four-fold degeneracy from the edge fermions labeled by $\beta_L,\beta_R \in \{\ua,\da\}$. 
Applying the formula (\ref{eq:pt_time_fermion_occupation}) to $\rho_I$ and a simple algebra leads to the following partial time-reversed density matrix on $I_1=1,\dots,L$, 
\begin{equation}\begin{split}
C_T^{I_1} \rho_I^{\sf T_1} [C_T^{I_1}]^{\dag}
&=\frac{{\cal N}}{4} \sum_{\beta_L,\beta_R,\{\sigma_j\},\{\mu_j\} \in \{\ua, \da\}} 
- (-1)^{\frac{(\mu_{L-1}+\sigma_{L-1})(\mu_L+\sigma_L)}{4}}  f(\beta_L,\beta_R,\{\sigma_j\},\{\mu_j\}) \\
& \quad
\times 
\ket{\beta_L\sigma_1\cdots \sigma_{L-1} (-\mu_L) \sigma_L \sigma_{L+1} \cdots \sigma_{N-1} \beta_R}
\bra{\beta_L\mu_1\cdots \mu_{L-1}(-\sigma_L) \mu_L \mu_{L+1} \cdots \mu_{N-1}\beta_R} 
\end{split}\end{equation}
Here, we set $\sigma_j = 1 (-1)$ for $\ua$ ($\da$) and the same notation for $\mu_j$, $\beta_L$. 
Then,
we obtain
the same result $\Tr[\rho_I C_T^{I_1} \rho_I^{\sf T_1} [C_T^{I_1}]^{\dag}]=-1/8$
as
the bosonic one. 
\end{widetext}

\subsection{$(1+1)d$ class AII}
\label{sec:(1+1)AII}

Let us consider $(1+1)d$ systems of complex fermions with spin,
where TRS acts on fermion operators as 
\begin{align}
T f^{\dag}_j T^{-1}=f^{\dag}_k [\cU^{\ }_T]^{\ }_{kj}, \quad 
\cU_T^{tr}=-\cU^{\ }_T. 
\end{align}
The time-reversal $T$ squares to the fermion number parity.
The Wick rotated version of this TRS can be used to introduce a pin$^{\tilde c}_+$ structure in Euclidean quantum field theory. 
The relevant cobordism group in 2d spacetime is given by
$\Omega_2^{\pin^{\tilde c}_+}=\Z$~\cite{Freed2016},
which is generated by the real projective plane $\mathbb{R}P^2$
with the half monopole flux, $\int_{\mathbb{R}P^2} F/2 \pi = 1/2$
(Sec.~\ref{sec:Ex:pinc+ structure on RP2}). 
The cobordism invariant topological action is given by
(\ref{eq:theta_term_rp2_pinc+}),
where the periodicity of $\theta$ is $4 \pi$. 
Because ground states are parametrized by unquantized theta angles, 
$\Omega_2^{\pin^{\tilde c}_+} = \Z$ does not represent an SPT phase. 
Nevertheless, from the example treated in this section
we will learn how to realize the non-trivial topological sector of
pin$^{\tilde c}_+$ connections in the operator formalism. 

To have a better understanding of the importance of the unoriented generating
manifold $\mathbb{R}P^2$ and the half monopole flux, 
it is instructive to recall that the topological response action of one-dimensional topological insulators is given by
the theta term  
\begin{align} \label{eq:theta_oriented}
Z(X,A) = e^{\frac{i \theta}{2 \pi} \int_X d A} = e^{i \theta n}, 
\end{align}
 where $X$ is a closed oriented $(1+1)d$ spacetime manifold, 
$A$ is the $U(1)$ background gauge field,
and $n = \frac{1}{2 \pi} \int_X d A$ is the total magnetic flux which is integer-valued $n\in \Z$.  
It is important to note that this action is invariant under $\theta\to \theta+2\pi$ and hence the polarization angle $\theta$ is defined modulo $2\pi$. When we consider the symmetry class AII which consists of two spin species, the theta angle doubles $\theta=\theta_\uparrow+\theta_\downarrow$. Therefore, the theta must be $4\pi$ periodic. However, the $4\pi$ periodicity cannot be resolved from measuring the total polarization angle when the system is put an on oriented manifold as in (\ref{eq:theta_oriented}). The resolution to this is to put the system on $\mathbb{R}P^2$ which admits a half monopole due to the Dirac quantization condition. This means that the partition function of class AII on $\mathbb{R}P^2$ is given by
\begin{align} \label{eq:Z_RP2_AII}
Z(\mathbb{R}P^2,A) = e^{\frac{i \theta}{2 \pi} \int_{\mathbb{R}P^2} d A} = e^{i \theta (n+\frac{1}{2})}, 
\end{align}
where the total magnetic flux $\frac{1}{2 \pi} \int_{\mathbb{R}P^2} d A \in \Z+1/2$ is half-integer. From this, one can readily observe that the polarization angle $\theta$ is indeed $4\pi$ periodic.


The purpose of this section is to construct a many-body invariant
(the partition function on $\mathbb{R}P^2$ in the presence of half monopole)
to detect this $\theta \in \R/4 \pi \Z$ from a given ground state wave function with $T$ symmetry. 

\subsubsection{Two adjacent intervals with the Lieb-Schultz-Mattis twist operator}

Our first task is to construct a ``tensor-network'' description of the
the generating manifold ($\mathbb{R}P^2$ with the half monopole flux).
Since the projective plane $\mathbb{R}P^2$ can be created from $S^2$ by applying the
antipodal projection,
we start from $S^2$.
In the Schwinger gauge,
the pin$^{\tilde c}_+$ connection with unit monopole flux on $S^2$ is given by 
\begin{align}
A_{\theta}(\theta,\phi) = 0, \quad 
A_{\phi}(\theta,\phi) = \frac{1}{2} \cos \theta, 
\end{align}
where we use the spherical coordinate $(\theta,\phi)$ on $S^2$. 
Along the circle with the latitude $\theta$, the holonomy is given by $\oint A_{\phi} d \phi = \pi \cos \theta$. 
We can deform $A$ such that the contributions to the holonomy ``localize'' near
$\phi=0$ and $\phi = \pi$,
and then take the quotient by the antipodal map $(\theta,\phi) \mapsto
(\pi-\theta,\phi+\pi)$.
This construction gives $\mathbb{R}P^2$
with the flux line $A_{\phi} = \frac{\pi}{2} \cos \theta \delta(\phi)$,
which gives rise to the half monopole charge.
$$
\includegraphics[width=0.8\linewidth, trim=0cm 0cm 0cm 0cm]{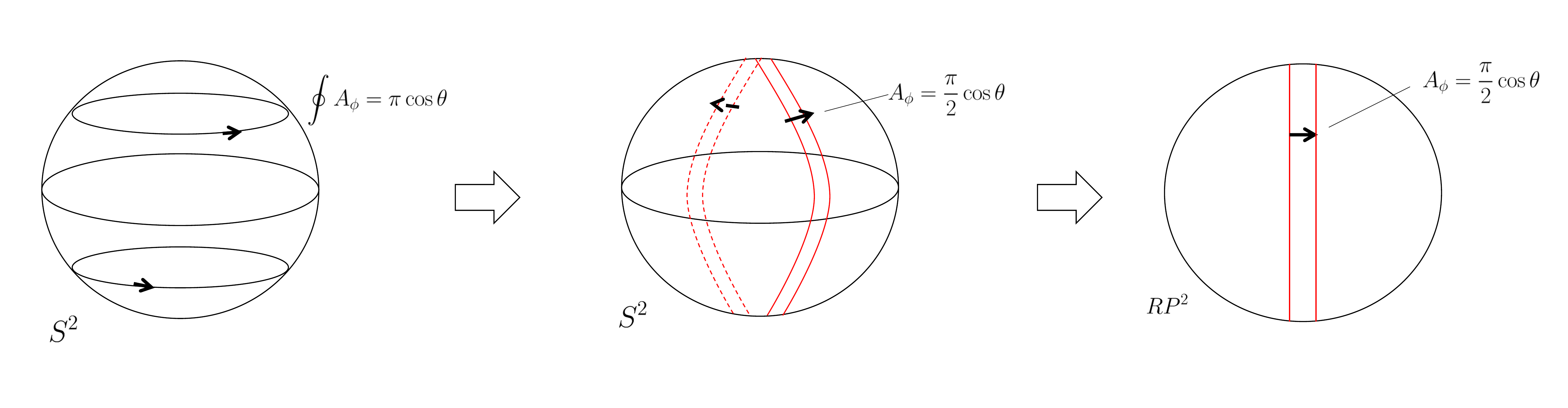}
$$
The next step is to deform this configuration.
Considering the following sequence of deformations (with a little care). 
$$
\includegraphics[width=0.8\linewidth, trim=0cm 0cm 0cm 0cm]{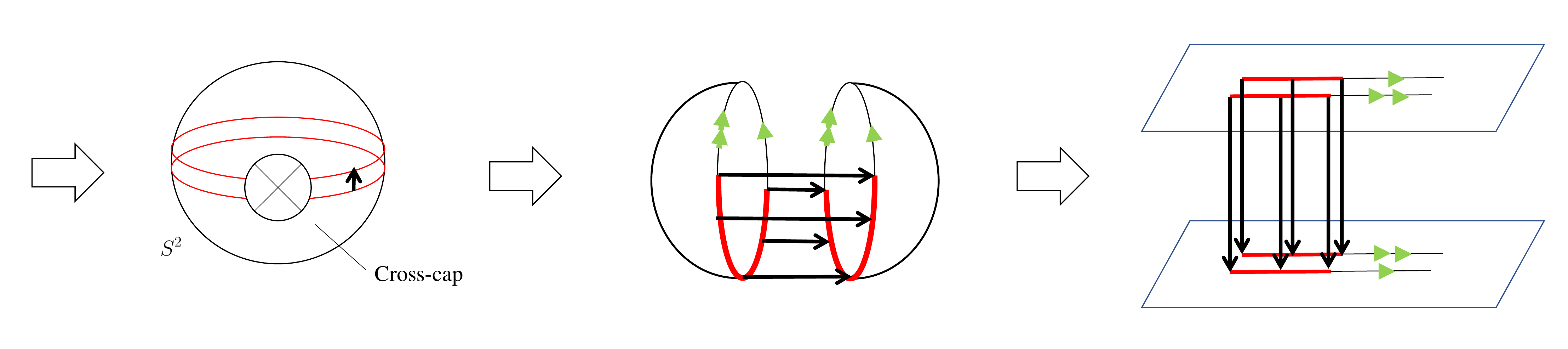}
$$
Here, the green arrows are identified with each other.
The final configuration can be readily interpreted in the canonical formalism.  
Hence, we obtain the following expression in the canonical formalism
for the path integral on the generating manifold,
\begin{align} \label{eq:ptrans_AII1d}
  &
Z_{\pin^{\tilde c}_+}\Big( \mathbb{R}P^2,\int_{\mathbb{R}P^2} \frac{F}{2 \pi} = \frac{1}{2} \Big)
 \sim 
  \Tr \Big[ \rho_I \prod_{x \in I_1} e^{\frac{\pi i x}{2 |I_1|} f^{\dag}_x f^{\ }_x}
  C_T^{I_1} \rho_I^{\sf T_1} [C_T^{I_1}]^{\dag} \prod_{x \in I_1} e^{\frac{-\pi i x}{2 |I_1|} f^{\dag}_x f^{\ }_x} \Big]. 
\end{align}
Here, $\rho_I = \Tr_{S^1_x \backslash I_1 \cup I_2} \big( \ket{\psi} \bra{\psi}
\big)$ is the reduced density matrix
of the two adjacent intervals $I_1 \cup I_2$
obtained from a pure state (ground state), 
$C^{I_1}_T \rho_I^{\sf T_1} [C_T^{I_1}]^{\dag}$ is obtained from $\rho_I$
by the partial time-reversal transformation on $I_1$ associated with $T$,
and finally, 
the operator $\prod_{x \in I_1} e^{\frac{\pi i x}{2 |I_1|} f^{\dag}_x f^{\ }_x}$ is a
quarter of the Lieb-Schultz-Mattis
twist operator of $U(1)$ charge~\cite{LIEB1961407}
($|I_1|$ is the length of the interval $I_1$).

\subsubsection{Numerical calculations}

We now explicitly compute the partition function on $\mathbb{R}P^2$ in the presence of
half monopole for a microscopic model,
following the recipe described in the previous part.
Let us consider a pair of SSH chain (\ref{eq:AIII}) (with arbitrary polarization angle parametrized by $\phi$) 
as the canonical Hamiltonian of the symmetry class AII
\begin{align} \label{eq:AII1d}
  H&= -\frac{t_2}{2} \sum_{j,\sigma}
      {\Big[} \psi^\dagger_{j+1\sigma} (\tau_x + i \tau_y) \psi_{j\sigma} +\text{H.c.} {\Big]} 
-t_1 \sum_{j,\sigma} \psi^\dagger_{j\sigma} (1+\cos\phi+i\tau_y \sin\phi)\tau_x \psi_{j\sigma}\ ,
\end{align}
where $\sigma=\uparrow,\downarrow$ are spin labels and we define a two-component fermion operator $\psi^\dag_{j,\sigma}=(f_j^\dag,g_j^\dag)_\sigma$ for each spin species in terms of the notation introduced in (\ref{eq:AIII}) and $\tau_i$ are Pauli matrices in this sublattice basis.

Figure~\ref{fig:theta_1d} shows how the total theta $\theta=\theta_\uparrow+\theta_{\downarrow}$ varies as we change $\phi$. Here, we compute the complex phase associated with the quantity introduced in (\ref{eq:ptrans_AII1d}). This way, we effectively obtain $\theta/2$ since we have placed a half monopole inside $\mathbb{R}P^2$ which corresponds to $n=0$ in (\ref{eq:Z_RP2_AII}).  From the above discussion, we expect $\theta/2$ to be $2\pi$ periodic which implies $\theta$ to be $4\pi$ periodic. This is clearly the case in Fig.~\ref{fig:theta_1d}.
As a reference, we also show the value of $\theta$ using the noninteracting formula in terms of the Berry phase
\begin{align} \label{eq:polarization1d}
\theta=\frac{1}{2\pi} \int_{\text{BZ}} d k\  \Tr\left[a_j\right],
\end{align}
where $a_j^{\mu\nu}=i \bra{u_{\mu\textbf{k}}}\partial_j\ket{u_{\nu\textbf{k}}}$ is the Berry connection defined in terms of the Bloch functions of occupied bands $\ket{u_{\nu\textbf{k}}}$ and $\partial_j=\partial/\partial k_j$. The agreement between the above non-interacting expression and the complex phase of partition function on $\mathbb{R}P^2$ is evident in Fig.~\ref{fig:theta_1d}.

\begin{figure}
\centering
\includegraphics[scale=0.4]{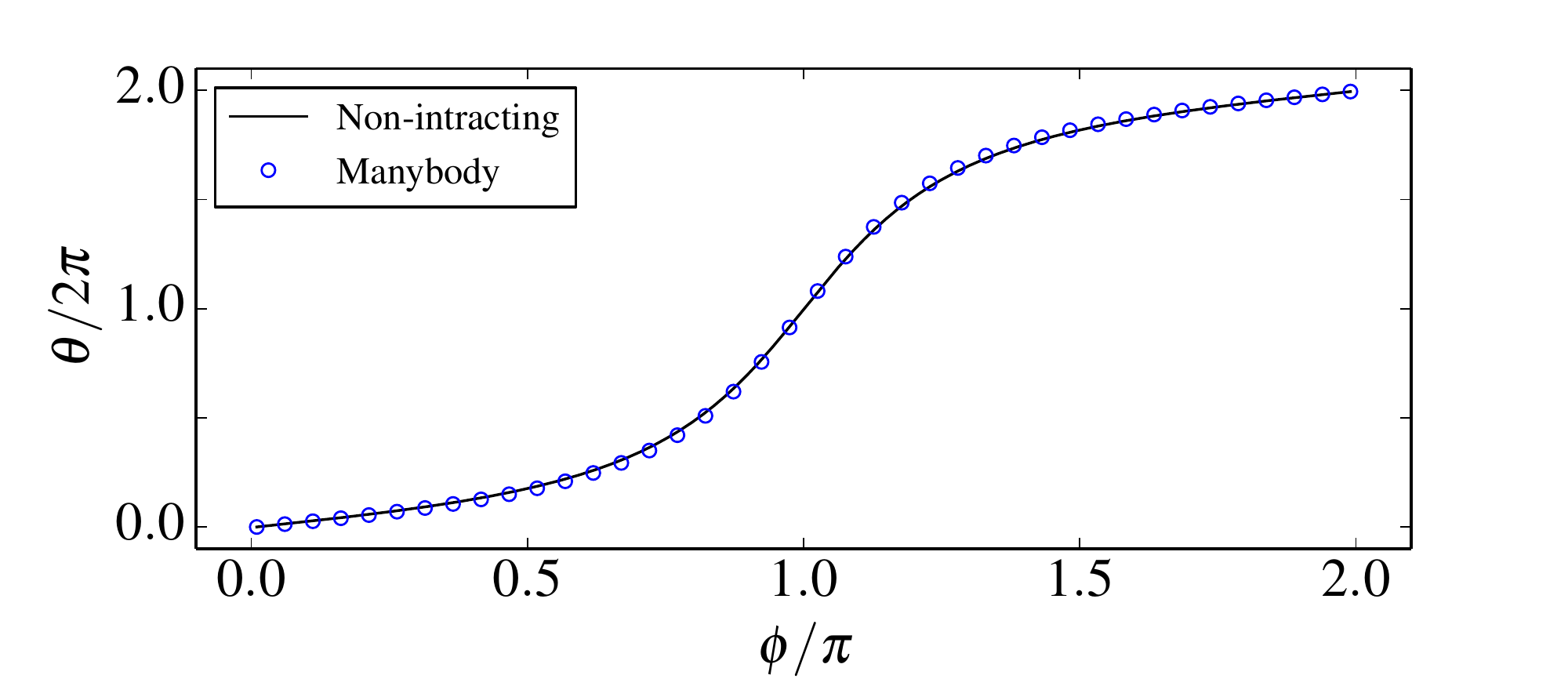}
\caption{\label{fig:theta_1d}  The polarization angle as a function of the parameter $\phi$ in the model Hamiltonian (\ref{eq:AII1d}) for class AII in one dimension. The many-body calculation refers to (\ref{eq:ptrans_AII1d}) which is equivalent to the partition function on $\mathbb{R}P^2$ in the presence of half-monopole (Eq.~(\ref{eq:Z_RP2_AII}) with $n=0$).
The non-interacting expression (solid curve) is used as a reference based on the formula (\ref{eq:polarization1d}). The system size is $N=80$ and the intervals $I_1$ and $I_2$ each contain $20$.}
\end{figure}

\section{Many body topological invariants of fermionic short range entangled
  topological phases in two and three spatial dimensions}
\label{sec:many-body_invariant in 2-3d}

Continuing from the previous section,
we develop the construction of the non-local order parameters for fermionic
short-range entangled states
protected by antiunitary symmetry in two and three spatial dimensions. 
In this section, we consider unitary symmetries as well. 
In Sec.~\ref{sec:(2+1)DIII}, we construct the $\Z_2$ invariant of $(2+1)d$ class
DIII
topological superconductors in a way similar to Sec.~\ref{sec:(1+1)DIII}. 
In Sec.~\ref{sec:(2+1)A},
we formulate the many-body $\Z$ Chern number for $(2+1)d$ class A topological insulators in various ways, which serves as a preliminary result for the subsequent sections. 
Sec.~\ref{sec:(2+1)A+CR} is devoted to developing the many-body $\Z_2$ invariant
for $(2+1)d$ class A insulators
with $CR$ particle-hole reflection symmetry, which are CPT dual to class AII insulators. 
In Sec.~\ref{sec:(2+1)AII}, we construct the many-body $\Z_2$ invariant for class
AII insulators,
which is a many-body counterpart of the Kane-Mele $\Z_2$ invariant,
based on the method explained in Sec.~\ref{sec:Methodtocomputethetopologicalinvariant}. 
We close this section with examples of many-body topological invariants in three spatial dimensions in Sec.~\ref{theta-term in (3+1)d}. 
The non-local order parameters and symmetry classes discussed in this section are summarized in Table~\ref{tab:summary 2}. 
 An analytical derivation of the topological invariants in two dimensions based on the edge theory approach~\cite{Qi2011b} should be possible. In addition, one may use exactly solvable models with zero-correlation length topological ground states (e.g., see~\cite{Wang_Chen}) to verify these results. We postpone these calculations to future studies.

\subsection{$(2+1)d$ class DIII}
\label{sec:(2+1)DIII}

The relevant structures for symmetry class DIII
are pin$_+$ structures. 
The cobordism group in $(2+1)d$ is given by $\Omega^{\pin_+}_3 = \Z_2$.
The generating manifold is
the Klein bottle$\times$ $S^1$, where $S^1$ is a spatial
direction (will be explained shortly),
with the periodic boundary condition for both the cycle of the Klein bottle and the $S^1$ direction. 
The many-body $\Z_2$ invariant is constructed in a similar way to Sec.~\ref{Z4 invariant from disjoint partial transpose: partition function on Klein bottle}. 
In order to construct the relevant spacetime manifold to detect the topological
invariant,
we first note that TRS changes the sign of the pairing terms in the $y$-direction. 
Therefore, in analogy to class DIII in $(1+1)d$, we partition the system in this direction. 
The remaining $x$-direction is left untouched and this way we realize  the Klein bottle $\times S^1$ as the spacetime manifold of the quantity,
\begin{align}   \label{eq:ZDIII2d}
  &
Z = 
\Tr_{R_1 \cup R_3} \Big[ 
\rho^{\ }_{R_1 \cup R_3}\big( (-1)^{F_2} \big) 
C_T^{R_1} [\rho^{\ }_{R_1 \cup R_3}\big( (-1)^{F_2} \big)]^{\mathsf{T}_1} [C_T^{R_1}]^{\dag} \Big],
\end{align}
where 
$R_{1,3}=I_{1,3}\times S^1_y$, 
and the reduced density matrix is found by
\begin{align}
\rho_{R_1 \cup R_3}\big( (-1)^{F_2} \big) = 
\Tr_{\overline{R_1 \cup R_3}} 
\Big[ e^{i\pi \sum_{\mathbf{r}\in R_2} n(\mathbf{r})} \ket{GS} \bra{GS} \Big],
\end{align}
and $\ket{GS}$ is the ground state of the Hamiltonian (\ref{eq:DIII2d}).
A schematic diagram of this partitioning is shown in
Fig.~\ref{fig:DIII_2d_num}[a].

\subsubsection{Numerical calculations}
A generating model of non-trivial SPT phases
in this symmetry class is given by the following
$(p_x+ i p_y)_{\uparrow}\times (p_x -i p_y)_{\downarrow}$ Hamiltonian~\cite{ReadGreen2000}
\begin{align} \label{eq:DIII2d}
{H}&= 
-\mu \sum_{i\sigma} f_{i\sigma}^\dagger f^{\ }_{i\sigma}
-\frac{t}{2} \sum_{\braket{ij}\sigma} \left[ f_{i\sigma}^\dagger f^{\ }_{j\sigma}+\text{H.c.} \right] 
       \nonumber \\
  &\quad 
    +\frac{\Delta}{2} \sum_{i} \left[ f_{i\uparrow}^\dagger f^\dagger_{i+\hat{x}\uparrow}+f_{i\downarrow}^\dagger f^\dagger_{i+\hat{x}\downarrow} +\text{H.c.}\right]
    \nonumber \\
&\quad +i\frac{\Delta}{2} \sum_{i} \left[ f_{i\uparrow}^\dagger f^\dagger_{i+\hat{y}\uparrow}-f_{i\downarrow}^\dagger f^\dagger_{i+\hat{y}\downarrow} +\text{H.c.}\right], 
\end{align}
which describes a superconducting state of spinful fermions.
Time-reversal acts on the fermion operators as 
\begin{align}
T f^{\dag}_{i\ua} T^{-1} = -f^{\dag}_{i\da}, \quad  
T f^{\dag}_{i\da} T^{-1} = f^{\dag}_{i\ua}, \quad  
T^2 = (-1)^F, 
\end{align}
and hence, the unitary matrix is $\cU_T=i\sigma_y$ in the $(\ua, \da)$ basis. 

As shown in Fig.~\ref{fig:DIII_2d_num}[b], the complex phase of the quantity
(\ref{eq:ZDIII2d})
is $\pi$ in the non-trivial phase which consistently reproduces the $\Z_2$
classification of class DIII in $(2+1)d$. 
The amplitude in the topological phase shows an area law behavior, $\sim
e^{-\alpha L_x}$ where $\alpha$ depends on microscopic details. As usual, the
amplitude reaches $1$ deep in the trivial phase regardless
of the dimensionality.

\begin{figure}[!]
	\includegraphics[scale=1]{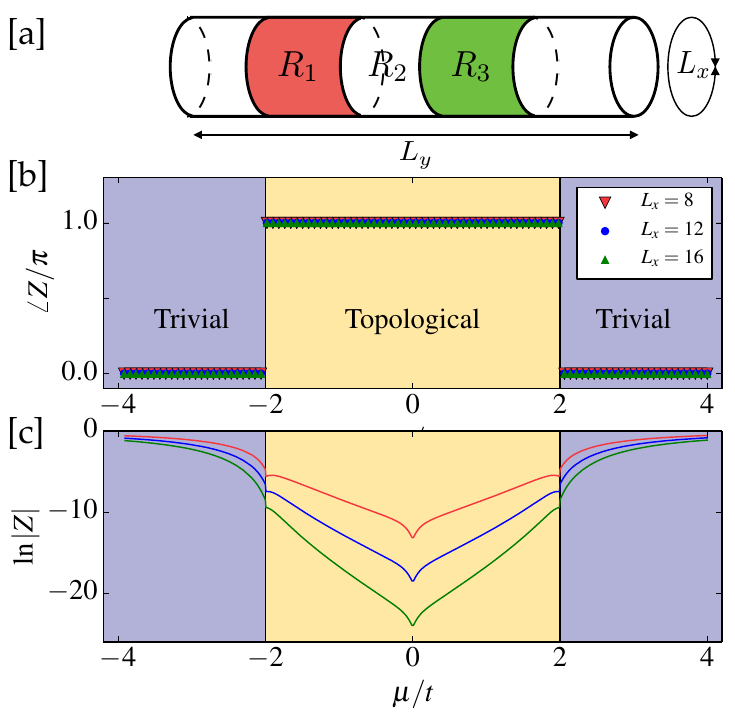}
	\caption{	\label{fig:DIII_2d_num}
	[a] Schematic of spatial partitioning for class DIII in $(2+1)d$.
	[b] Complex phase and
	[c] amplitude of the many-body invariant   \eqref{eq:ZDIII2d}
  for the model Hamiltonian~(\ref{eq:DIII2d}). 
	We set $L_y=40$, and $R_1$, $R_2$ and $R_3$ each has $10$ sites in the $y$-direction.}
\end{figure}

A simple way to explain the $\Z_2$ phase observed above
is by going to momentum space along the compactified direction (which is the
$x$-direction in our model (\ref{eq:DIII2d}),
also shown in Fig.~\ref{fig:DIII_2d_num}[a]). This way, one can view the model Hamiltonian (\ref{eq:DIII2d}) as a set of decoupled $(1+1)d$ models in the symmetry class DIII. For each $k_x$, the Hamiltonian reads 
\begin{align} \label{eq:DIII2d_k}
{H}_{k_x} &= 
(-\mu-t\cos k_x) \sum_{y\sigma} f_{k_x,y\sigma}^\dagger f^{\ }_{k_x,y\sigma}
            \nonumber \\
  &\quad
-\frac{t}{2} \sum_{y\sigma} \left[ f_{k_x,y\sigma}^\dagger f^{\ }_{k_x,y+1\sigma}+\text{H.c.} \right] 
\nonumber \\
&\quad +i\Delta\sin k_x \sum_{y\sigma} \left[ f_{k_x,y\sigma}^\dagger f^\dagger_{-k_x,y\sigma} +\text{H.c.}\right] 
               \nonumber \\
  &\quad
+i\frac{\Delta}{2} \sum_{y} \left[ f_{k_x,y\uparrow}^\dagger f^\dagger_{-k_x,y+1\uparrow}-f_{k_x,y\downarrow}^\dagger f^\dagger_{-k_x,y+1\downarrow} +\text{H.c.}\right] 
\end{align}
where $f_{k_x,y\sigma}=\frac{1}{\sqrt{L_x}} \sum_x f_{x,y\sigma}e^{-i k_x x}$.  The overall value of (\ref{eq:ZDIII2d}) is given by the product of all $k_x$ modes each evaluated by (\ref{eq:ZDIII1d}). For $k_x\neq 0, \pi$ the density matrix is a combination of both $k_x$ and $-k_x$ modes and the resulting quantity (\ref{eq:ZDIII1d}) is a complete square; thus, the associated complex phase vanishes. However, at the time-reversal invariant points $k_x= 0, \pi$ the Hamiltonian (\ref{eq:DIII2d_k}) is very similar to (\ref{eq:DIII}). When $-2t<\mu<0$, the $k_x=0$ mode is a $(1+1)d$ class DIII in the topological phase and gives the $\pi$ phase, while $k_x=\pi$ is in the trivial phase and does not have a complex phase. Hence, the overall phase which is the sum  of all corresponding phases for $k_x$ modes becomes $\pi$. A similar argument can also be applied to the regime $0<\mu<2t$ where the $k_x=\pi$ mode is described by the non-trivial phase of the $(1+1)d$ class DIII and $k_x=0$ is in the trivial phase.

\subsection{$(2+1)d$ Class A}
\label{sec:(2+1)A}

In this section,
we present various constructions of many-body $\Z$ topological
invariants
for $(2+1)d$ fermionic insulators
with the $U(1)$ charge conservation symmetry, e.g., Chern insulators. 
These topological invariants are all designed
to extract the quantized integer (the quantized Hall conductivity)
in the Chern-Simons action that arises in the
path-integral \eqref{path integral}.
Among them, 
the many-body Chern number
\cite{NiuThoulessWu1985, AvronSeiler1985} 
is the well established construction of the many-body $\Z$ invariant. 
In addition to the many-body Chern number formula, we present three
construction of the many-body $\Z$ invariants,
one of which is equivalent to the charge pump.~\cite{Laughlin}

Let us start by discussing the Chern-Simons term
for spin$^c$ connections. 
We consider a gapped phase of complex fermions
with a half-integer spin in (2+1) spacetime dimensions.  
In the Euclidean quantum field theory,
the relevant structures of manifolds are spin$^c$
which are needed to define a complex fermionic spinor with a half-integer spin on the manifolds. 
\cite{SeibergWitten2016gapped}
The topological sectors of spin$^c$ connections
are classified by the cohomology $H^2(M;\Z)$ on the manifold, 
where the free part represent background magnetic fluxes characterized by the first Chern number of the field strength. 
We consider chiral phases with short range entanglement, the low-energy
effective theory of which is given by 
the Chern-Simons action of
spin$^c$ connections
\begin{align}
  &
\int_X \left[ \frac{i k_1}{4 \pi} \big( A d A + \Omega(g) \big) + \frac{4i k_2}{2 \pi} A d A \right]
=\int_X \frac{i (k_1+8 k_2)}{4 \pi} A d A + ({\rm gravitational\ part}) 
\end{align}
with two integer parameters $k_1,k_2 \in \Z$.~\cite{SeibergWitten2016gapped}
These correspond to the cobordism invariant actions of
spin$^c$ connections in (3+1)-dimensions,~\cite{Gilkey}
\begin{align}
\frac{1}{2} \cdot \frac{F \wedge F}{(2 \pi)^2} + \frac{\sigma}{8}, \quad 4 \cdot \frac{F \wedge F}{(2 \pi)^2}, 
\end{align}
where $\sigma$ is the signature of the manifold, and $\Omega(g)$ is the gravitational correction~\cite{SeibergWitten2016gapped} representing the edge mode with a unit chiral central charge. 
The integers $k_1$ and $k_2$ are related to the electrical Hall and thermal Hall
conductivities
as $\sigma_{\rm Hall} = (k_1+8 k_2)e^2/h$ and $\kappa_{\rm Hall} = k_1 \pi^2 k_B^2 T/3 h$. 
The $(k_1,k_2)=(1,0)$ corresponds to Chern insulators realized in free fermion
systems,
whereas to realize states with $(k_1,k_2) = (0,1)$ we need many-body interactions.~\footnote{
Recall that the bosonic integer quantum Hall effect is described by the
Chern-Simons action $\frac{1}{2 \pi} \int A d A$
showing $\sigma_{\rm Hall} = 2 (e_{\rm boson})^2/h$ and $\kappa_{\rm Hall} = 0$.~\cite{Lu2012} 
The $(0,1)$ state may be viewed as the bosonic integer quantum Hall effect made of Cooper pairs with the charge $e_{\rm boson} = 2e$.}
In this section, we only consider the method to detect the pure $U(1)$ charge part $k:=(k_1+8 k_2)$.


In the rest of this section, we develop and review four constructions of the many-body Chern number detecting $k \in \Z$ from various input data of ground state wave functions, 
which are summarized in Table~\ref{tab:summary_chern_number}. 
The same construction is applicable to the many-body Chern number $k \in 2 \Z$ for bosonic integer quantum Hall state described by the Chern-Simons action $\frac{i k}{4 \pi} \int A d A$.

\begin{table*}[!]
\begin{center}
\begin{tabular}{| >{\centering\arraybackslash}m{3cm} | >{\centering\arraybackslash}m{13.5cm} | >{\centering\arraybackslash}m{2.2cm}|}
\hline
 & Many-Body Chern number & Input data \\
\hline
Charge pump~\cite{Laughlin} & 
$\Braket{GS(\int F=2 \pi) | \hat N |GS(\int F = 2 \pi)} - \Braket{GS(F=0)| \hat N |GS(F=0)} $ & $\ket{GS(F=0)}$ and $\ket{GS(\int F=2 \pi)}$ \\
\hline
Many-body Chern number~\cite{NiuThoulessWu1985, AvronSeiler1985} & 
$\frac{i}{2 \pi} \oint d_{\theta_y} \oint \Braket{GS(\theta_x,\theta_y) | d_{\theta_x} GS(\theta_x,\theta_y)}$ & $\ket{GS(\theta_x,\theta_y)}$ \\
\hline
Twisted boundary condition and twist operator~\cite{Resta1998, NakagawaFurukawa}& 
$\frac{i}{2 \pi} \oint d_{\theta_y} \log \Braket{GS(\theta_y) | \prod_{x,y} e^{\frac{2 \pi i x \hat n(x,y)}{L_x}} |GS(\theta_y)}$ & $\ket{GS(\theta_y)}$ \\
\hline
Swap~\cite{HaegemanPerez-GarciaCiracSchuch2012,ShiozakiRyu2016} and twist operator & 
$ {\displaystyle \frac{i}{2 \pi} \oint d_\theta \log \Braket{GS| \prod_{(x,y) \in R_1 \cup R_2} e^{\frac{2 \pi i y \hat n(x,y)}{L_y}} {\rm Swap}(R_1, R_3) \prod_{(x,y) \in R_1 \cup R_2} e^{i \theta \hat n(x,y)} |GS} }$ & $\ket{GS}$ \\
\hline
\end{tabular}
\caption{\label{tab:summary_chern_number}
  List of the many-body $\mathbb{Z}$ invariants
for class A topological insulators in $(2+1)d$.}
\end{center}
\end{table*}

\subsubsection{Charge pump}
\label{sec:2d_classA_charge_pump}
Let $H(A)$ be a Hamiltonian of complex fermions or bosons on a 2$d$ closed
spatial manifold $M$ in the presence of a background $U(1)$ field $A$. 
We assume $H(A)$ is gapped and has a unique ground state $\ket{GS(A)}$. 
The first example of the many-body Chern number is the increment of the particle numbers of ground states by inducing a unit magnetic flux, 
\begin{align}
  k&=
  \textstyle
     \Braket{GS(\int_M F=2 \pi)|\hat N|GS(\int_M F=2 \pi)}
  - \braket{GS|\hat N|GS},
\label{eq:class_a_z_inv}
\end{align}
where $\hat N$ is the total $U(1)$ charge operator and $\ket{GS}$ is the ground state without magnetic flux. 
This is equivalent to the charge pump~\cite{Laughlin} and derived from the Chern-Simons action $\frac{i k}{4 \pi} \int_X A d A$ as follows. 
Let $X=\R_t \times M$ be the spacetime manifold. 
With a unit magnetic flux, the partition function is recast as 
\begin{align}
  \textstyle
  Z(\R_t \times M,\int_M F=2 \pi) &= e^{\frac{i k}{2 \pi} \int_{\R_t} dt A_t \int_M dx dy F_{xy}}
= e^{i k \int_{\R_t} dt A_t}, 
\end{align}
(here, we have Wick-rotated to real time)
which represents the systems with the $U(1)$ charge $k$, where the number of $U(1)$ charge is counted from the ground state without a magnetic flux. 
The functional derivative with respect to $A_t$ leads to (\ref{eq:class_a_z_inv}). 


{\it Lattice realizations of a unit magnetic flux---}
It is worth introducing two useful realizations of a unit magnetic flux on the lattice system on the torus $M=T^2$. 
Let $A_x(x,y),A_y(x,y) \in \R/2 \pi \Z$ be the background $U(1)$ field living on bonds. 
A simple way to prepare a unit magnetic flux with uniform magnetic field is to set 
\begin{align} 
&A_x(x,y) = \frac{2 \pi y}{L_xL_y}, \quad A_y(x,y) = \left\{\begin{array}{ll}
0 & (y=1, \dots, L_y-1) \\
\frac{2 \pi x}{L_x} & (y=L_y) \\
\end{array}\right. , 
\label{eq:uniform_magnetic_flux}
\end{align}
where $L_x,L_y$ are the number of sites. 
The uniform magnetic flux $F(x,y) = A_x(x,y)+A_y(x+1,y)-A_x(x,y+1)-A_y(x,y)=2 \pi/(L_x L_y)$ is inserted per a unit cell (see Fig.~\ref{fig:flux_monopole_torus} [a]). 
Another useful choice is the following non-uniform magnetic field with a unit flux (a magnetic flux line) 
\begin{align}
&\wt A_x(x,y) = 0, \quad \wt A_y(x,y) = \left\{\begin{array}{ll}
0 & (y=1, \dots, L_y-1) \\
\frac{2 \pi x}{L_x} & (y=L_y) \\
\end{array}\right.. 
\label{eq:nonuniform_magnetic_flux}
\end{align}
The magnetic flux $F=2 \pi/L_x$ is inserted only on the $y=L_y$ line (see Fig.~\ref{fig:flux_monopole_torus} [b]).
Both configurations $A, \tilde A$ give rise to a unit magnetic flux $\int_{T^2} F = 2 \pi$ on the torus.

\begin{figure}[!]
	\includegraphics[scale=.25]{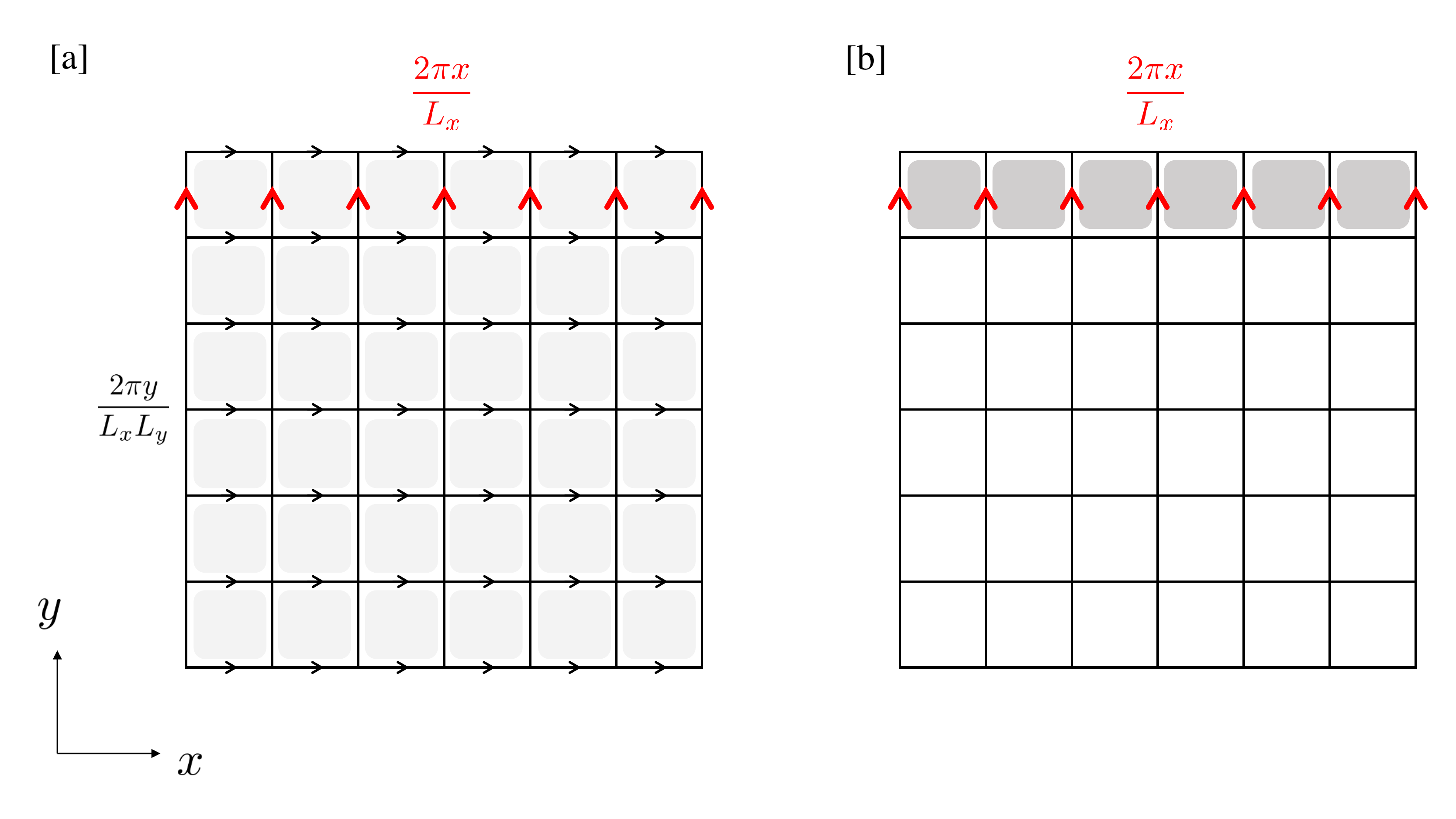}
	\caption{	\label{fig:flux_monopole_torus}
	[a] $U(1)$ field (\ref{eq:uniform_magnetic_flux}) with a unit magnetic flux $\int_{T^2} F = 2 \pi$ on the torus. The magnetic flux per unit cell is $2 \pi/L_x L_y$. 
	[b] $U(1)$ field (\ref{eq:nonuniform_magnetic_flux}) with a magnetic flux line with a unit magnetic flux $\int_{T^2} F = 2 \pi$ on the torus. This is realized by the twisted boundary condition depending on site $x$. 
}
\end{figure}

For noninteracting models of the Chern insulator, it is easy to compute the invariant (\ref{eq:class_a_z_inv}). We have numerically checked that the formula (\ref{eq:class_a_z_inv}) does work for some lattice models of the Chern insulator. 


{\it Parity anomaly for complex fermion---}
For complex fermions, the increment of the particle number by $1$ is analytically computed 
from the Dirac theory which is a critical theory between the $k=0$ and $k=1$ states.~\cite{witten2015fermion} 
Let 
\begin{align}
h(A) = (-i \p_x + A_x) \sigma_x + (-i \p_y + A_y) \sigma_y
\end{align}
be the 2$d$ Dirac Hamiltonian on a closed 2$d$ manifold $M$ with $U(1)$ background field with chiral symmetry $\{h(A),\Gamma\}=0, \Gamma=\sigma_z$. 
The index theorem states that the analytic index of $h(A)$ which is defined as the 
difference between the number of zero modes of $h(A)$ with positive and negative chirality $\Gamma = \pm 1$ is given by the magnetic flux $\int_M F/2 \pi \in \Z$. 
For the background $U(1)$ field with a unit magnetic flux, there appears a single zero mode $\ket{\phi_+}$ with positive chirality $\Gamma = 1$. 
Now we perturb the system by small mass term $m \sigma_z$. 
$m>0$ ($m<0$) corresponds to the Chern (trivial) insulator. 
The zero mode $\ket{\phi_+}$ shifts by $+m$, which means the number of negative energy modes increases by 1 as $m$ crosses zero from positive to negative. 
See, for example, Ref.\ \cite{Hasebe2015} for the detail of the analytical calculation on the sphere.

\subsubsection{Many-body Chern number}
Let $\ket{GS(\theta_x,\theta_y)}$ be the ground state wave function on the torus with twisted boundary conditions $\theta_x$ and $\theta_y$ for the $x$ and $y$-directions, respectively. 
The adiabatic change of the flux $\theta_x(t)$ by a unit period with the flux $\theta_y$ fixed gives the Berry phase. 
In spacetime geometry, the adiabatic change of the flux $\theta_x(t)$ by a unit period induces the unit electric flux 
$\oint d t d x E_x = \oint dt dx (-\p_t A_x) = - \oint d \theta_x = -2 \pi$. 
From the Chern-Simons action, the Berry phase can be written as  
\begin{align}
  &
\exp \oint \Braket{GS(\theta_x,\theta_y) | d_{\theta_x} GS(\theta_x,\theta_y)}
= \exp \frac{i k}{2 \pi} \oint d y A_y \oint dt dx E_x
= e^{-i k \theta_y}.
\label{eq:berry_phase_chern_number}
\end{align}
The many-body Chern number $k$ is extracted from the winding number of the Berry
phase as~\cite{NiuThoulessWu1985, AvronSeiler1985}
\begin{align}
k = 
\frac{i}{2 \pi} \oint d_{\theta_y} \oint \Braket{GS(\theta_x,\theta_y) | d_{\theta_x} GS(\theta_x,\theta_y)}. 
\end{align}

\subsubsection{Twisted boundary condition and twist operator}
\label{sec:2d_classA_twisted_bc_and_twist}

In \eqref{eq:berry_phase_chern_number},
we have employed the temporally varying vector potential $A_x(t)$. 
The alternative choice is the spatially varying scalar potential $A_t(x)$ to get a unit electric flux $\oint dt dx E_x = \oint dt dx (-\p_x A_t) = -2 \pi$. 
The corresponding operator is known as the Lieb-Schultz-Mattis twist operator~\cite{LIEB1961407} or the Resta's $Z$ function~\cite{Resta1998}. 
Let $\ket{GS(\theta_y)}$ be the ground state with the twisted boundary condition for the $y$-direction. 
From the Chern-Simons action, the ground state expectation value of the twist operator is given by 
\begin{align}
  &
\Braket{GS(\theta_y) | \prod_{x,y} e^{\frac{2 \pi i x \hat n(x,y)}{L_x}} |GS(\theta_y)}
= \exp \frac{i k}{2 \pi} \oint d y A_y \oint dt dx E_x
= e^{-i k \theta_y}, 
\end{align}
where $\hat n(x,y)$ is the $U(1)$ density operator and $x$ and $y$ run over all the space manifold. 
The many-body Chern number $k$ can be extracted by the phase winding again as~\cite{NakagawaFurukawa}
\begin{align}
k
= \frac{i}{2 \pi} \oint d_{\theta_y} \log \Braket{GS(\theta_y) | \prod_{x,y} e^{\frac{2 \pi i x \hat n(x,y)}{L_x}} |GS(\theta_y)}.
\end{align}

\subsubsection{Swap and twist operator}
\label{sec:Swap and twist operator}
Here we introduce the many-body Chern number made only from a single ground state wave function $\ket{GS}$.
To this end, we employ a kind of partial operation, the swap operator, exchanging two intervals, 
which was introduced to detect the $\Z_2$ SPT invariant of the Haldane chain phase protected by the $\Z_2 \times \Z_2$ symmetry.~\cite{HaegemanPerez-GarciaCiracSchuch2012}

Let us briefly review the swap operator in $1$-space dimension. 
Let $\ket{GS}$ be a pure state on a closed $1d$ chain. 
We introduce three adjacent intervals $I_1 \cup I_2 \cup I_3$. 
The swap operator ${\rm Swap}(I_1,I_3)$ is defined by exchanging the matter field $\psi^{\dag}(x)$ on $I_1$ and $I_3$ as 
\begin{equation}\begin{split}
    {\rm Swap}(I_1,I_3) \psi^{\dag}(x \in I_1) {\rm Swap}(I_1,I_3)^{-1}
    &= \psi^{\dag}(x + |I_1 \cup I_2| \in I_3), \\
    {\rm Swap}(I_1,I_3) \psi^{\dag}(x +|I_1\cup I_2|\in I_3) {\rm Swap}(I_1,I_3)^{-1}
    &= \psi^{\dag}(x \in I_1), \\
    {\rm Swap}(I_1,I_3) \psi^{\dag}(x \in I_2) {\rm Swap}(I_1,I_3)^{-1}
    &= \psi^{\dag}(x \in I_2), 
\end{split}\end{equation}
where $|I_1\cup I_2|$ is the length of the interval $I_1 \cup I_2$. 
Topologically, the swap operator ${\rm Swap}(I_1,I_3)$ induces a genus in the $(1+1)d$ spacetime manifold in the imaginary time path integral.
In the presence of the onsite symmetry $G$, one can combine the partial $G$ operation with the swap operator. 
Let us consider the following sequence~\cite{HaegemanPerez-GarciaCiracSchuch2012}
\begin{align}
\Big( \prod_{x \in I_1 \cup I_2} \hat h_x \Big) \cdot {\rm Swap}(I_1,I_3) \cdot \Big( \prod_{x \in I_1 \cup I_2} \hat g_x \Big), 
\label{eq:swap_and_symmetry_op}
\end{align}
where $\hat g_x, \hat h_x$ are the local symmetry operator of $g,h\in G$.
The operator (\ref{eq:swap_and_symmetry_op}) induces a genus with an intersection of symmetry defects between $g$ and $h$~\cite{ShiozakiRyu2016} 
as shown in Fig.~\ref{fig:swap_and_symmetry_action}. 
In the limit $|I_1|,|I_2|,|I_3| \gg \xi$ with $\xi$ the correlation length, 
the ground state expectation value of (\ref{eq:swap_and_symmetry_op}) is well quantized to get the topological invariant detecting the relevant topological phase, for example, the group cohomology $H_{\rm group}^2(G;U(1))$ if $[g,h]=0$ for $(1+1)d$ bosonic SPT phases. 

\begin{figure*}[!]
	\includegraphics[scale=.25]{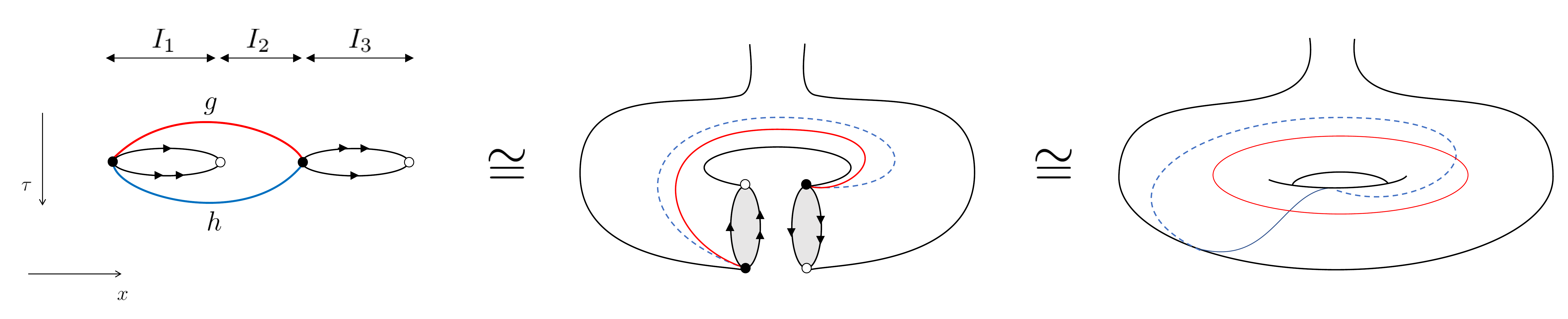}
	\caption{	\label{fig:swap_and_symmetry_action}
	Topological deformation of the operator (\ref{eq:swap_and_symmetry_op}). 
	The red and blue lines express the symmetry defects of $g$ and $h$ symmetry, respectively, which means that the matter field are charged by $g$ and $h$ when it passes these defect lines.
	}
\end{figure*}

\begin{widetext}
Let us move on to $2$-space dimensions with the $U(1)$ symmetry by adding the closed $y$-direction $S^1_y$. 
One can generalize the operator (\ref{eq:swap_and_symmetry_op}) so as to include the twist operator along the $y$-direction: 
\begin{align}
  \Big( \prod_{(x,y)\in R_1 \cup R_2} e^{\frac{2 \pi i y \hat n(x,y)}{L_y}} \Big) \cdot {\rm Swap}(R_1, R_3) \cdot
  \Big( \prod_{(x,y) \in R_1 \cup R_2} e^{i \theta \hat n(x,y)} \Big),
\end{align}
where $R_{1,2,3}=I_{1,2,3}\times S^1_y$.
This operator induces the connected sum of the spacetime manifold with the background $U(1)$ field with $\int A d A = 4 \pi \theta$. 
Thus, we have that 
\begin{align}
k=
  \frac{i}{2 \pi} \oint d_\theta \log \Braket{GS|
  \prod_{(x,y) \in R_1 \cup R_2}
  e^{\frac{2 \pi i y \hat n(x,y)}{L_y}}
  {\rm Swap}(R_1, R_3)
  \prod_{x \in R_1 \cup R_2} e^{i \theta \hat n(x,y)} |GS}
\end{align}
in the limit $|I_1|,|I_2|,|I_3| \gg \xi$ with $\xi$ the correlation length of bulk.
\end{widetext}

\subsection{$(2+1)d$ Class A + $CR$ with $(CR)^2=1$}
\label{sec:(2+1)A+CR}

In this section, we discuss the many-body topological invariant of
topological insulators in $(2+1)d$ protected by
the combined particle-hole and reflection symmetry, 
i.e., $CR$ symmetry, with $(CR)^2=1$. 
\cite{2014PhRvB..90x5111H, 2014PhRvB..90p5134H, 2015PhRvB..91s5142C}
Here, $CR$
is unitary, and acts on complex fermion operators as  
\begin{align}
  &
(CR) \psi_j^{\dag}(x,y) (CR)^{-1}
  = [\cU_{CR}]_{jk} \psi_k(-x,y),
  \nonumber \\
  &
\cU_{CR}^{tr} = \cU_{CR}. 
\end{align}
The CR symmetric topological insulators are
CRT (CPT) dual of time-reversal symmetric topological insulators
in class AII in $(2+1)d$, and hence
this section can be considered as a preliminary for the next section
where the many-body $\Z_2$ invariant for class AII topological insulators in $(2+1)d$ is defined. 

In the Euclidean field theory, 
this orientation-reversing symmetry corresponds
to a pin$^{\tilde c}_+$ structure
of unoriented manifolds.~\cite{Metlitski2015, Freed2016}
Since the $U(1)$ charge is flipped under the $CR$ transformation, 
the topological sectors of pin$^{\tilde c}_+$ structures are
classified by the integer cohomology $H^2(X;\tilde \Z)$ twisted by the orientation bundle.
For example,
$H^2(\mathbb{R}P^2;\tilde \Z) = \Z$; 
$H^2(KB \times S^1;\tilde \Z) = \Z \oplus \Z \oplus \Z_2$. 
\footnote{
To know possible topological sectors,
the Poincare duality of twisted cohomology
$H^2(X;\tilde \Z) \cong H_{d-2}(X;\Z)$ is useful, where $d$ is the dimension of $X$. 
Also, the K\"{u}nneth formula holds true: 
$H^d(X \times Y;\tilde \Z) \cong \bigoplus_{i+j=d} H^{i}(X;\tilde \Z) \otimes H^j(Y;\tilde \Z) \bigoplus_{i+j=d-1} {\rm Tor}_1^{\Z}(H^i(X;\tilde \Z), H^j(Y;\tilde \Z))$. 
For example, $H^2(\mathbb{R}P^2;\tilde \Z) \cong H_0(\mathbb{R}P^2;\Z) = \Z$; 
$H^2(KB \times S^1;\tilde \Z) \cong H^2(KB; \tilde \Z) \otimes H^0(S^1;\tilde \Z) \oplus H^1(KB;\tilde \Z) \otimes H^1(S^1;\tilde \Z)= \Z \oplus \Z \oplus \Z_2$.
}


The cobordism was computed by Freed and Hopkins~\cite{Freed2016}
and is given by $\Omega_3^{\pin^{\tilde c}_+}=\Z_2$. 
The generating manifold was discussed by Witten:~\cite{witten2015fermion} 
Let $(T^2,A_1)$ be the space torus with a unit magnetic flux. 
Thanks to the fact that $CR$ flips the sign of the $U(1)$ charge, 
the background spin$^c$ field can be $CR$-symmetric. 
Suppose that $CR$ acts freely on $T^2 \times S^1$ as  
$CR: (x,y,t) \mapsto (-x,y,t+1/2)$ where $S^1 = [0,1]$ is the time-direction. 
Then,
(the phase of) the partition function on the quotient
$T^2\times S^1/CR \cong {\rm Klein\ bottle} \times S^1$
is the $\Z_2$ invariant. 
Among the topological sectors classified
by $H^2(KB \times S^1;\tilde \Z)$, 
the generating manifold belongs to $\Z$,
which is given by the free part of $H^1(KB;\tilde \Z) = \Z \oplus \Z_2$ and $H^1(S^1;\Z) = \Z$ in the view of the K\"{u}nneth formula.

%

Similar to the integer topological invariants of the $(2+1)d$
quantum Hall effect, 
it is possible to construct different (expressions of) many-body topological invariants,
which all detect the $\Z_2$ phase in symmetry class A+CR in $(2+1)d$.
See Table~\ref{tab:summary_z2_number} for the summary of
the many-body $\Z_2$ invariants constructed in this section.
We note that
our $\Z_2$ invariants are also applicable
to $\Z_2$ bosonic topological insulators protected
by $U(1) \rtimes R$ symmetry, where the semi direct product means the reflection $R$ flips the $U(1)$ charge. 

\begin{table*}[!]
\begin{center}

\begin{tabular}{| >{\centering\arraybackslash}m{3cm} | >{\centering\arraybackslash}m{10.5cm} | >{\centering\arraybackslash}m{2.2cm}|}
\hline
 & Many-Body $\Z_2$ number & Input data \\
\hline
$CR$ parity pump~\cite{witten2015fermion} & 
$\Braket{GS(\int F=2 \pi) | CR |GS(\int F = 2 \pi)}/\Braket{GS(F=0)| CR |GS(F=0)} $ & $\ket{GS(F=0)}$ and $\ket{GS(\int F=2 \pi)}$ \\
\hline
Berry phase on the Klein bottle & 
$\exp \oint \Braket{GS(KB;\theta) | d_{\theta} GS(KB;\theta)}$ & $\ket{GS(KB;\theta)}$ \\
\hline
$CR$ swap and twist operator & 
$ \Braket{GS| \prod_{(x,y) \in R_2} e^{\frac{2 \pi i y \hat n(x,y)}{L_y}}  \prod_{(x,y) \in R_3} (-1)^{\hat n(x,y)}  {\left. CR\right|_{R_1 \cup R_3} } |GS}$ & $\ket{GS}$ \\
\hline
\end{tabular}
\end{center}
\caption{\label{tab:summary_z2_number}
  List of the many-body $\Z_2$ invariants detecting
  fermionic topological insulators protected by $U(1) \rtimes CR$ symmetry
  with $(CR)^2=1$ and bosonic topological insulators protected by $U(1) \rtimes R$ symmetry.}
\end{table*}

\subsubsection{$CR$ parity pump
}
\label{sec:cr_z2}

It is straightforward to translate 
the $\mathbb{Z}_2$ invariant given as 
the partition function
on
$T^2\times S^1/CR \cong {\rm Klein\ bottle} \times S^1$
into the operator formalism.
Let $(T^2,A)$ be the 2d closed torus with background spin$^c$ field $A$ with unit magnetic flux. 
The configuration $A$ can be chosen to be $CR$ symmetric. 
The many-body $\Z_2$ invariant is the ground state expectation value of $CR$ operator 
\begin{align}
\nu = \Braket{\Psi(T^2,A)| CR | \Psi(T^2,A)} \in \{ \pm 1\}. 
\label{eq:cr_z2_inv}
\end{align}
I.e., 
the ground state parity of $CR$ operator. 
The quantity $\nu$ is clearly quantized and invariant under
deformation of the Hamiltonian.
The remaining problem is to show that
there exists a model Hamiltonian and ground state with $\nu=-1$
relative to the trivial one with $\nu=1$.

\paragraph{Ground state expectation value of a particle-hole operator---}

Before showing the existence of non-trivial ground states with $\nu=-1$,
we discuss a technical point:
Because $CR$ includes a particle-hole, 
there is a subtle point to make the expectation
value of $CR$ free from a $U(1)$ phase of states. 
In the following, we give a proper gauge fixing procedure
(see Eq.\ \eqref{eq:cr_gauge_fixing} below). 
Let us consider a gapped Hamiltonian of complex fermions $\psi_{j}, \psi_{-j} (j=1, \dots, N)$ with particle-hole symmetry 
\begin{align}
C H C^{-1} = H, \quad 
C \psi^{\dag}_j C^{-1}
= [\cU_C]_{jk} \psi_k, \quad 
\cU_C^{tr} = \cU_C, 
\end{align}
where $\cU_C$ is a unitary matrix. 
We {\it define} the action of the $C$ operator on the 
vacuum of $\psi_j$ fermion by 
\begin{align}
C \ket{0} = \psi^{\dag}_{-N} \cdots \psi^{\dag}_{N-1} \psi^{\dag}_{N} \ket{0} 
\end{align}
up to a sign. 
For simplicity, we consider free fermionic model 
\begin{align}
H=\sum_{jk} \psi^{\dag}_j h_{jk} \psi_k = \sum_j {\cal E}_j \chi^{\dag}_j \chi_j, \quad {\cal E}_j < 0 \ {\rm for\ }j<0. 
\end{align}
Here, we introduced the diagonalizing basis $\chi^{\dag}_j = \psi^{\dag}_i u_{ij}$ with $h=u {\cal E} u^{\dag}$. 
The ground state is given by occupying states with negative energy as
\begin{align}
  \ket{GS}
  = \chi_{-1}^{\dag} \cdots \chi^{\dag}_{-N} \ket{0}. 
\end{align}
\footnote{
  For simplicity, here, we describe free fermion ground states. 
  However, our procedure applies to
  any ground states since in the presence of particle-hole symmetry $C$, 
  all states in the Fock space (in the presence of an energy gap)
  come in particle-hole symmetric pairs. }
Because of PHS, one can choose unoccupied states satisfying 
\begin{align}
C \chi^{\dag}_{-j} C^{-1} = \chi_{j}, 
\label{eq:cr_gauge_fixing}
\end{align}
which is equivalent to imposing the condition 
\begin{align}
u_j = \cU_C^{\dag} u^*_{-j} 
\label{eq:cr_gauge_fixing_2}
\end{align}
on the state vector $u_j = (u_{ij})$. 
Under the phase fixing (\ref{eq:cr_gauge_fixing}), 
the ground state expectation value of $C$ reads 
\begin{align}
\braket{GS | C | GS}
&= \bra{0} \chi_{-N} \cdots \chi_{-1} C \chi^{\dag}_{-1} \cdots \chi^{\dag}_{-N} \ket{0} 
   \nonumber \\
&\sim (\det U)^{-1}, 
\end{align}
where $U$ is the $2N$ by $2N$ matrix 
\begin{align}
U = (u_{-N}, \cdots, u_{-1}, \cU_C^{\dag} u_{-1}^*, \cdots, \cU_C^{\dag} u_{-N}^*). 
\end{align}
Notice that $\det U$ is well-defined only with the phase fixing condition (\ref{eq:cr_gauge_fixing_2}). 

\paragraph{Existence of $\Z_2$ nontrivial model---}
Now we show there exists a model with nontrivial $\Z_2$ invariant (\ref{eq:cr_z2_inv}). 
It is best understood in the presence of an additional $\Z_2$ symmetry which anticommutes with $CR$. 
In such cases, one-particle eigenstates are labeled
by $\Z_2$ eigenvalues, $\ua$ and $\da$, say,
and $CR$ exchanges $\ua$ and $\da$
as well as occupied and unoccupied states as 
\begin{align}
  u^+_{i\ua} = \cU_{CR}^{\dag} (u^-_{i\da})^*,
  \quad 
u^+_{i\da} = \cU_{CR}^{\dag} (u^-_{i\ua})^*. 
\end{align}
The crucial point is the parity anomaly of Chern insulators. 
The number of occupied states of a Chern insulator
with Chern number $ch_1=m$ in the presence of the unit magnetic flux 
increases (decrease) by $m$
as compared to the number of occupied states
in the absence of the magnetic flux. 
Let $m$ be the spin Chern number, namely, the $ch_{\ua}=-ch_{\da}=m$. 
In such case, the numbers of occupied states are as follows: 
\begin{align}
&u^+_{i\ua} = \cU_{CR}^{\dag} (u^-_{i\da})^* \quad (i=1, \dots, N+m), 
  \nonumber \\
&u^+_{i\da} = \cU_{CR}^{\dag} (u^-_{i\ua})^* \quad (i=1, \dots, N-m). 
\end{align}
Then, because the determinant of the odd permutation is $(-1)$, it follows that 
\begin{align}
&\det U 
  = \det \cU_{CR}^* (-1)^{N+m}
    \times 
    \Big| \det \Big[ u^-_{1\ua}, \dots , u^-_{N-m\ua}, \cU_{CR}^{\dag} (u^-_{1\da})^*, \dots,\cU_{CR}^{\dag} (u^-_{N+m\da})^* \Big] \Big|^2.
\end{align} 
Clearly, $\det U$ is proportional to the parity of the spin Chern number $m$.

\begin{figure}[!]
	\includegraphics[scale=.4]{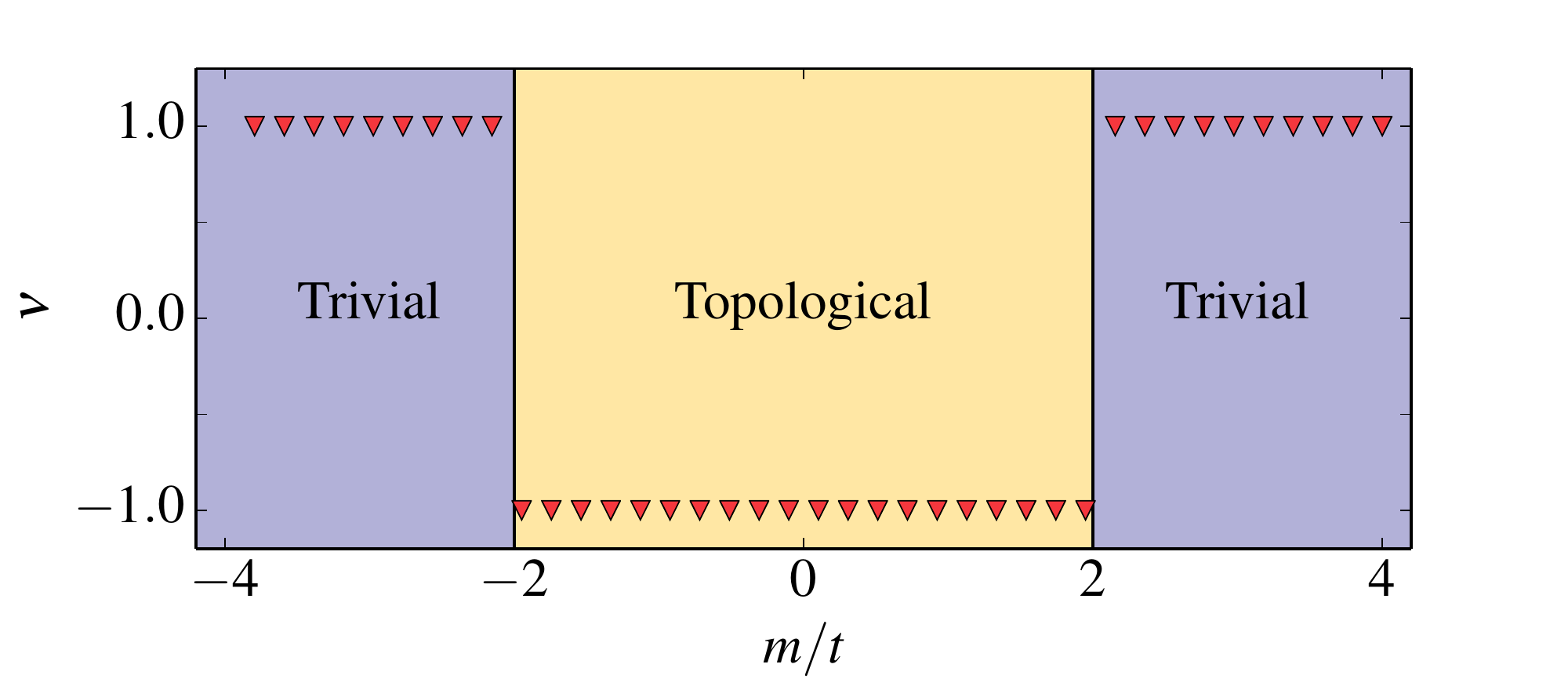}
	\caption{	\label{fig:A_CR_num}
	 The topological invariant $\nu = \Braket{\Psi(T^2,A)| CR | \Psi(T^2,A)}$ of class A+$CR$, for the model Hamiltonian~(\ref{eq:A_CR}) on torus ($L_x=L_y=12$) in the presence of one unit of background magnetic flux.}
\end{figure}


\paragraph{CR symmetric lattice model}
We now explicitly show that the $\Z_2$ 
invariant (\ref{eq:cr_z2_inv})
can be computed for a microscopic model.
We consider the Hamiltonian
\begin{align} \label{eq:A_CR}
H&= \frac{1}{2} \sum_{\substack{\textbf{r}\\ s=1,2}} {\Big[} \psi^\dagger({\textbf{r}+\hat{x}_s}) (i \Delta\Gamma_s - t \Gamma_3) \psi (\textbf{r}) +\text{H.c.} {\Big]} 
  + m \sum_i \psi^\dagger(\textbf{r}) \Gamma_3 \psi (\textbf{r})
\end{align}
on the square lattice, which describes two copies of the Chern insulator (class A) with opposite chiralities, 
where $\psi (\textbf{r})$ is a four-component fermion operator and the hopping amplitudes are $4\times4$ matrices given by
$\Gamma_s= (\sigma_z\tau_x,  \sigma_0\tau_y,  \sigma_0\tau_z)$.
where the $\sigma$ and $\tau$ are Pauli matrices and the $0$ subscript denotes the identity matrix. The CR symmetry is defined by
\begin{align}
(CR) \psi^{\dag}(x,y) (CR)^{-1}
= \sigma_x\tau_x \ \psi(-x,y).
\end{align}
The background magnetic flux is implemented in the lattice model through the
Pierels substitution
where the hopping matrices are modified according to the rule,
$
h_{\textbf{r}\textbf{r}''} \to h_{\textbf{r}\textbf{r}'} e^{i\int_{\textbf{r}}^{\textbf{r}'} A(\textbf{r})\cdot d\textbf{r}}.
$
A few important remarks regarding the spectrum of Chern insulators 
in the presence of $m$ units of magnetic flux $\phi$ are in order. 
Let $L_x$ and $L_y$ be the dimensions of the lattice model (\ref{eq:A_CR}). 
There are $2L_xL_y$ sites available per each Chern insulator. 
In the absence of background magnetic field, 
there are $n_{\text{occ}}(\phi=0)=L_x L_y$ occupied states in each Chern insulator separated 
from the other $n_{\text{unocc}}(\phi=0)=L_x L_y$ states by a gap. 
Inserting a magnetic flux of strength $\phi=m$ creates an imbalance between occupied and unoccupied states, such that $n_{\text{occ}}(\phi=m)=L_xL_y+ m\ ch_\sigma $ and $n_{\text{unocc}}(\phi=m)=L_xL_y- m\ ch_\sigma$ where $ch_\sigma=\pm 1$ is the Chern number of the Chern insulator $\sigma=\uparrow/\downarrow$ and the spectral gap remains the same as in the absence of $\phi$. This result is independent of the choice of gauge fields (\ref{eq:uniform_magnetic_flux}) or (\ref{eq:nonuniform_magnetic_flux}) as expected from the index theorem (\ref{eq:class_a_z_inv}). However, for more practical purposes one should use the uniform gauge field (\ref{eq:uniform_magnetic_flux}) in which single particle wave functions are smoother. Finally, we should note that the overall number of occupied states in CR symmetric model (\ref{eq:A_CR}) remains unchanged.
 
In order to compute the topological invariant (\ref{eq:cr_z2_inv}), we calculate the ground state of (\ref{eq:A_CR}) subject to a unit magnetic flux and fix the $U(1)$ gauge freedom through (\ref{eq:cr_gauge_fixing_2}). The result is shown in Fig.~\ref{fig:A_CR_num}. Interestingly, there is a sharp distinction between the trivial and non-trivial phases.

\subsubsection{$\Z_2$ quantized Berry phase on the Klein bottle}
\label{z2_a+cr_berry_phase}

\begin{figure*}[!]
	\includegraphics[scale=.6]{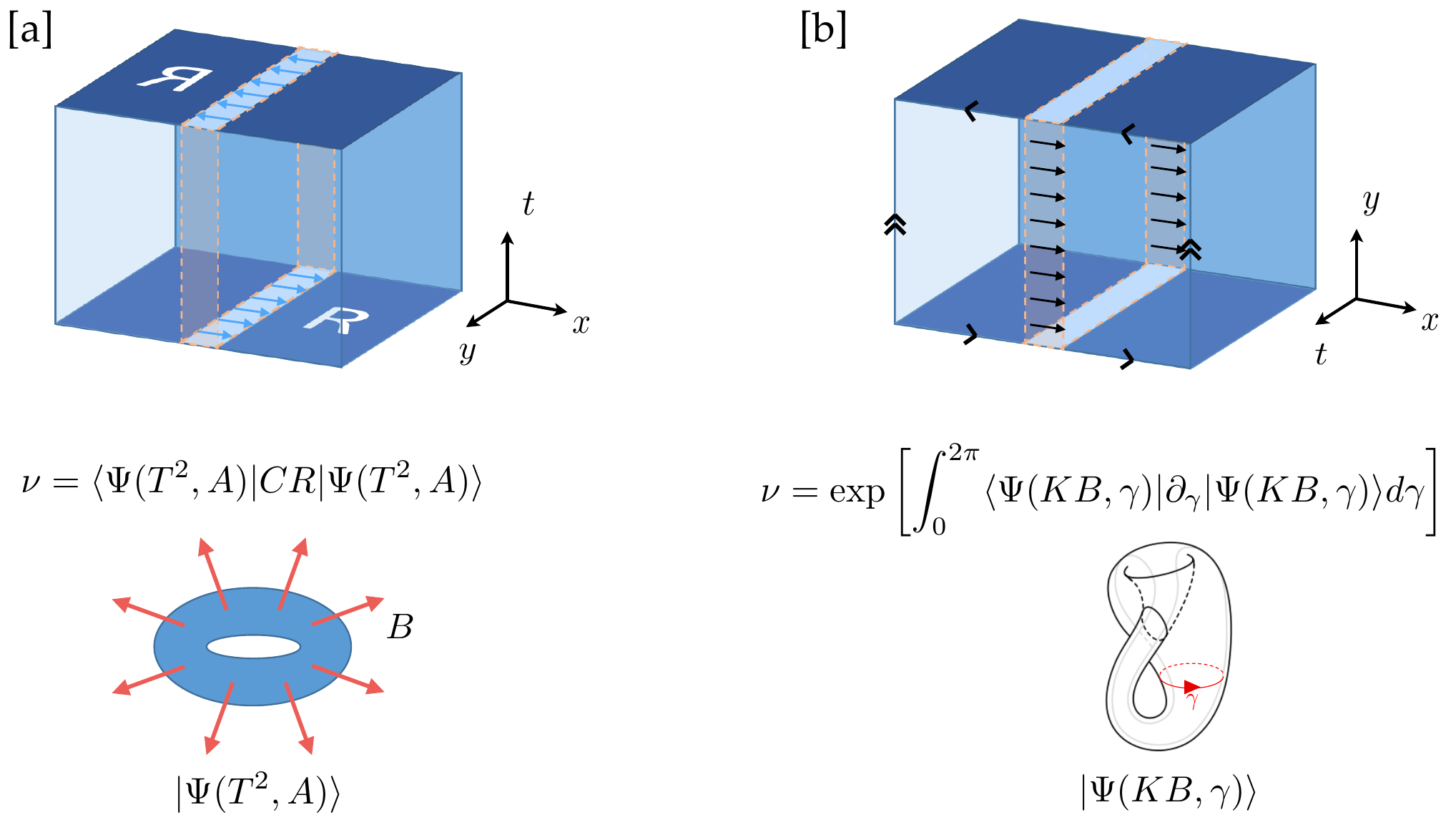}
	\caption{	\label{fig:CR_manifold}
	 Two methods of computing the topological invariant $\nu$ of class A+$CR$ in $(2+1)d$, for the model Hamiltonian~(\ref{eq:A_CR}). Top panels show the spacetime manifold and the small arrows represent the background gauge field $A=(A_t,A_x,A_y)$. Lower panels show the spatial manifold and the background gauge fields, which are non-flat (unit magnetic flux) and flat (twisted boundary condition) in [a] and [b], respectively. [a] $\nu$ is defined as the expectation value of the $CR$ operator where the spatial manifold is torus and the system is subjected to one unit of background magnetic flux $A=(0,2\pi y/L_y,0)$. [b] $\nu$ is computed in terms of Berry phase as the twisted boundary condition $\gamma=2\pi t/L_t$, i.e., $A=(0,2\pi t/L_t,0)$, along the cycle of (spatial) Klein bottle is swept from $0$ to $2\pi$. The Klein bottle is obtained through twisting by the $CR$ symmetry along the $y$-direction of the original torus.}
\end{figure*}

Alternatively, the topological invariant can be derived from an adiabatic
process. Compared to the previous method, the adiabatic method has a merit that
there is no need to introduce a non-flat background gauge field over the space
manifold.  Let us first motivate this idea by constructing the spacetime
picture. Starting from the spacetime picture of the topological invariant
(\ref{eq:cr_z2_inv}), shown in Fig.\ref{fig:CR_manifold}[a], we choose the
spatial part of the spacetime manifold to contain the $CR$ twist and a flat
gauge field, as shown in Fig.\ref{fig:CR_manifold}[b]. This means that the
spatial manifold  is the Klein bottle and the gauge field is $A_x=2\pi t/L_t$
which varies in the temporal direction. To simulate the partition function on
this manifold, we consider slices of fixed time manifolds where the ground state
is computed from the $CR$-twisted Hamiltonian and the background field is simply
a twisted boundary condition along the cycle of the Klein bottle
($x$-direction).
The details of this Hamiltonian can be found in
Appendix~\ref{Class A+CR: Twisting by CR symmetry}.
Hence, the partition function can be written as 
\begin{align} \label{eq:ACRtwist_2}
  Z&=e^{\int_0^{2\pi} \braket{\Psi(\gamma)|\partial_\gamma|\Psi(\gamma)} d\gamma }
     \nonumber \\
     &=
       \braket{\Psi(\gamma_0)|\Psi(\gamma_1)}\braket{\Psi(\gamma_1)|\Psi(\gamma_2)}
       \braket{\Psi(\gamma_2)|\Psi(\gamma_3)}
\cdots \braket{\Psi(\gamma_{N-1})|\Psi(\gamma_0)}
\end{align}
where the topological invariant is the Berry phase associated with a closed loop
in the space of twisted boundary condition $\gamma$, which is the twist angle
along the $x$-direction. In the second identity, we explicitly show how this
quantity can be computed numerically,
where $\gamma_j= 2\pi j/N$ ($0\leq j<N$), and $N$ is the number of steps in discretizing the $\gamma\in[0,2\pi)$ interval.

The Berry phase \eqref{eq:ACRtwist_2},
when computed for non-trivial topological insulators
(such as the non-trivial phase realized in the model \eqref{eq:A_CR}),
is quantized and equal to $\pi$ (mod $2\pi$).
We have tested the formula \eqref{eq:ACRtwist_2} 
numerically for the model \eqref{eq:A_CR},
and confirmed the $\pi$ Berry phase. 
The details of the numerical procedure can be found
in
Appendix~\ref{Class A+CR: Twisting by CR symmetry}.
The $\pi$ Berry phase is consistent with the boundary 
calculations presented in Ref.\ \cite{2014PhRvB..90p5134H};
In Ref.\ \cite{2014PhRvB..90p5134H},
the gapless boundary theory of class A+$CR$ topological
insulators is put on the spacetime Klein bottle. 
The anomalous phase acquired by the partition function 
of the boundary theory 
in large $U(1)$ gauge transformations
is quantized to be $\pi$ mod $2\pi$. 
We have thus demonstrated the bulk-boundary correspondence
for class A+$CR$ topological insulators in $(2+1)d$.

  \subsubsection{$CR$ swap and the twist operator}
\label{sec:+cr_cr_swap}

\begin{figure*}[!]
	\includegraphics[scale=.27]{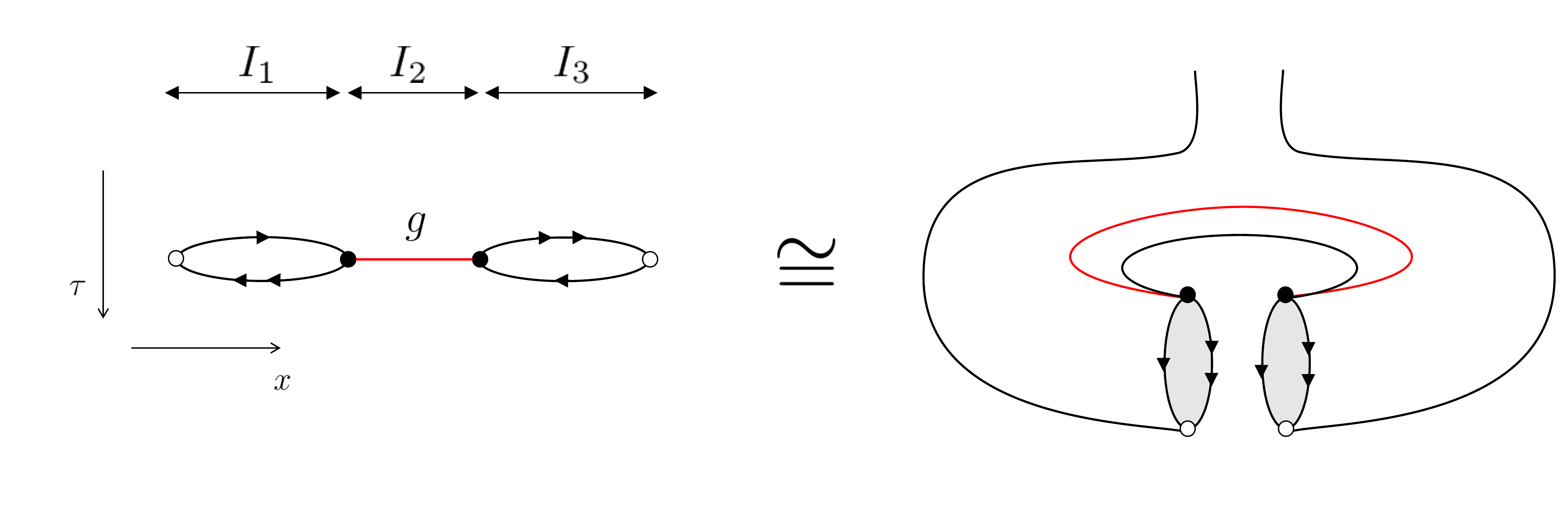}
	\caption{	\label{fig:swap_CR}
	Topological deformation of the operator (\ref{eq:cr_swap_and_symmetry_op}). 
	The red line expresses the symmetry defects of $g$ symmetry, which means that the matter field are charged by $g$ when it passes the defect line.
}		
\end{figure*}

The third expression of the many-body $\Z_2$ invariant needs only a single ground state $\ket{GS}$,
similar to Sec.~\ref{sec:Swap and twist operator}. 
In this scheme, we introduce the ``$CR$ swap'' operation to simulate the Klein bottle as follows. 
Let us consider, as a warm up, 
the pure ground state $\ket{GS}$ on a closed $1d$ chain and $I_1, I_2, I_3$ be
three adjacent intervals where $I_1$ and $I_3$ are of the same length. 
We assume $I_1 \cup I_3$ is $CR$-symmetric: $I_2$ is centered at the $CR$ symmetric site $x=0$. 
We define the $CR$ swap operator by restricting the $CR$ reflection to the
intervals $I_1$ and $I_3$, 
\begin{equation}\begin{split}
    CR|_{I_1 \cup I_3} \psi^{\dag}(x \in I_1 \cup I_3) (CR|_{I_1 \cup I_3})^{-1} 
    &= CR \psi^{\dag}(x) (CR)^{-1}
    = {\cal U}_{CR} \psi(-x), \\
CR|_{I_1 \cup I_3} \psi^{\dag}(x \in I_2) (CR|_{I_1 \cup I_3})^{-1} &= \psi^{\dag}(x). 
\end{split}\end{equation}
Topologically, the $CR$ swap operator $CR|_{I_1 \cup I_3}$ induces a genus one Klein bottle $(1+1)d$ spacetime manifold in the imaginary time path integral (see Fig.~\ref{fig:swap_CR}). 
Naively, we may expect that the ground state expectation value of the $CR$ swap operator gives the Klein bottle partition function. 
However, it turns out that the ground state expectation value $\Braket{GS | CR|_{I_1 \cup I_3} | GS}$ vanishes identically, which implies that $CR|_{I_1 \cup I_3}$ operator itself cannot be used to create a well-defined background pin$^{\tilde c}_+$ structure on the Klein bottle. 
To fix this issue, we find that combining the $CR$ swap operator with the
partial fermion parity flip on the subsystem $I_3$
provides the desired Klein bottle partition function.  
The corresponding composite operator is given by
\begin{align}
\left( \prod_{x \in I_3} (-1)^{\hat n(x)} \right) \cdot CR|_{I_1 \cup I_3}.
\label{eq:cr_swap_add_fermion_parity}
\end{align}
The necessity to introduce the additional fermion parity flip may be understood as follows. 
First, let us recall that in relativistic fermion systems the spacetime $2 \pi$ rotation $e^{2 \pi i \Sigma}$ is identical to the fermion parity flip $(-1)^{\hat N}$, where $\Sigma$ is the generator of the Spin(2) rotation. 
Therefore, one must be careful in performing spacetime rotations.
In fact, Fig~\ref{fig:swap_to_Klein} suggests that the proper way to make the Klein bottle from the $CR$ swap is 
(i) taking the spacetime $(-\pi)$-rotation on $I_1$, 
(ii) swapping two intervals $I_1 \cup I_3$ by the $CR$ reflection, and 
(iii) gluing back to the state $\ket{GS}$ after the the spacetime $(-\pi)$-rotation on $I_3$. 
The sequence of operations is summarized as 
\begin{align}
\left[ e^{-i \pi \Sigma} \Big|_{I_3} \right]^{\dag} \times CR_{I_1 \cup I_3} \times e^{-i \pi \Sigma} \Big|_{I_1}.
\label{eq:cr_swap_with_rotation}
\end{align}
Because the spacetime $\pi$-rotation anticommutes with the $CR$ reflection, i.e., $CR e^{-i \pi \Sigma} = (-1)^{\hat N} e^{-i \pi \Sigma} CR$, 
we have the fermion parity flip as in (\ref{eq:cr_swap_add_fermion_parity}).
We should note that the prescription (\ref{eq:cr_swap_add_fermion_parity})
is not limited to $CR$-symmetric systems, but required in any fermionic systems. 
In Appendix~\ref{app:1d_a+reflection_r_swap}, we show that in $(1+1)d$ class A
systems with reflection symmetry the fermion parity flip
in (\ref{eq:cr_swap_add_fermion_parity}) is also need to yield the Klein bottle partition function. 

\begin{figure}[!]
	\begin{center}
	\includegraphics[width=0.7\linewidth, trim=0cm 0cm 0cm 0cm]{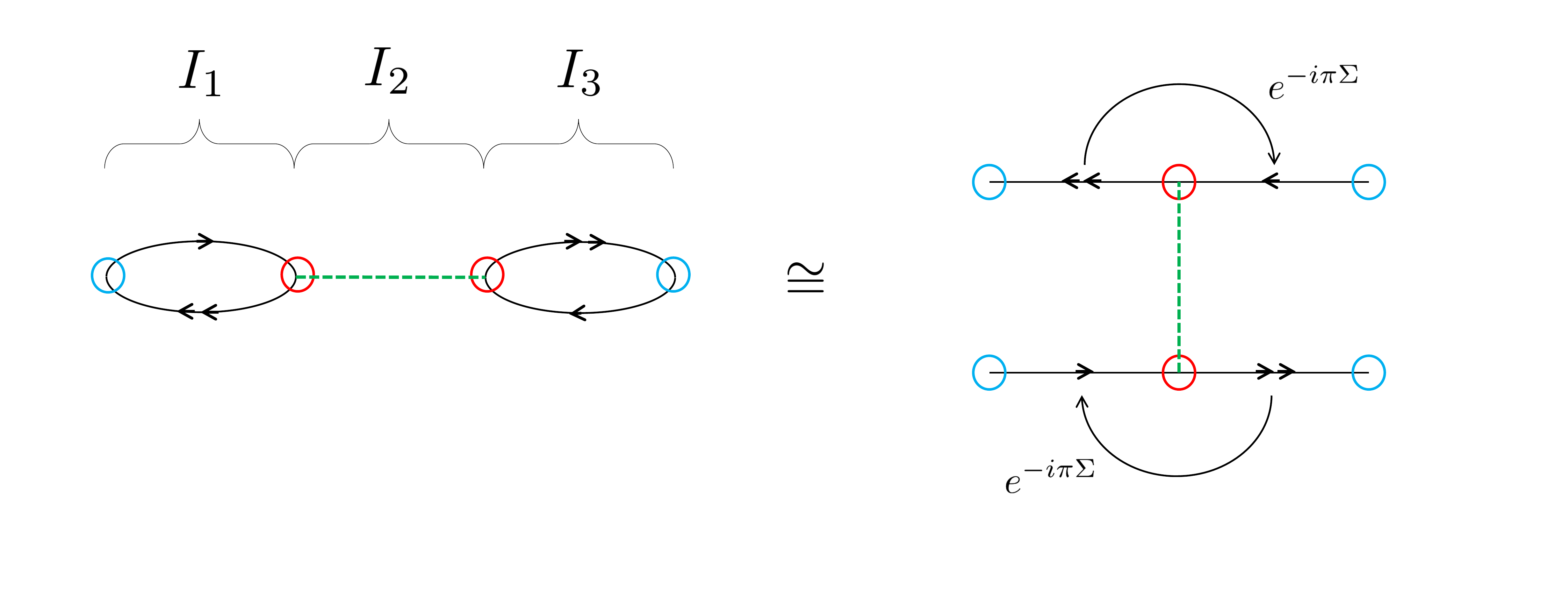}
	\end{center}
	\caption{From the $CR$ swap to the Klein bottle. }
	\label{fig:swap_to_Klein}
\end{figure}

In the presence of an abelian onsite symmetry $G$ that is flipped by the $CR$ reflection as $CR \hat g (CR)^{-1} = - \hat g$, 
one can consistently associate the intermediate $g$ action between the disjoint intervals $I_1 \cup I_3$ as 
\begin{align}
\left( \prod_{x \in I_2} \hat g_x \right) \cdot \left( \prod_{x \in I_3} (-1)^{\hat n(x)} \right) \cdot CR|_{I_1 \cup I_3}, 
\label{eq:cr_swap_and_symmetry_op}
\end{align}
where $\hat g_x$ is the local symmetry operator corresponding to $g \in G$. 
The operator (\ref{eq:cr_swap_and_symmetry_op}) induces
a genus one Klein bottle with the symmetry flux $g \in G$ (see Fig.~\ref{fig:swap_CR}).  
In the limit $|I_1|,|I_2|,|I_3| \gg \xi$ with $\xi$ the correlation length, 
the ground state expectation value of (\ref{eq:swap_and_symmetry_op}) is well quantized and provide a many-body topological invariant.

Let us move on to $2$-space dimensions with the $U(1) \rtimes CR$ symmetry by adding the closed $y$-direction $S^1_y$. 
One can generalize the operator (\ref{eq:cr_swap_and_symmetry_op}) so as to include the twist operator along the $y$-direction: 
\begin{align}
\left( \prod_{(x,y) \in R_2} e^{\frac{2 \pi i y \hat n(x,y)}{L_y}} \right) \cdot \left( \prod_{(x,y) \in R_3} (-1)^{\hat n(x,y)} \right) \cdot CR|_{R_1 \cup R_3}, 
\label{eq:swap_cr_twist}
\end{align}
where $R_{1,2,3}=I_{1,2,3}\times S^1_y$.
This operator induces the same spacetime manifold as (\ref{eq:cr_z2_inv}) and (\ref{eq:ACRtwist_2}). 
The ground state expectation value of (\ref{eq:swap_cr_twist}) is well quantized in the limit $|I_1|,|I_2|,|I_3| \gg \xi$ with $\xi$ the correlation length of bulk and gives the many-body $\Z_2$ invariant.

\subsection{$(2+1)d$ Class AII}
\label{sec:(2+1)AII}

Symmetry class AII is characterized by TRS $T$ which squares to $-1$.
It acts on the fermionic creation/annihilation operators as
\begin{align}
T f^{\dag}_j T^{-1}=f^{\dag}_k [\cU_T]_{kj}, \quad 
\cU_T^{tr}=-\cU^{\ }_T. 
\end{align}
The Wick rotated version of this TRS corresponds to a pin$^{\tilde c}_+$ structure in the Euclidean quantum field theory. 
The pin$^{\tilde c}_+$ cobordism group in (2+1)$d$ is given by
$\Omega_2^{\pin^{\tilde c}_+}= \Z_2$, 
which implies the existence of $\Z_2$ SPT phases,
i.e, the celebrated time-reversal symmetric topological insulators.
\cite{KaneMeleZ2,RoyZ2,MooreBalents}
The generating manifold is the Klein bottle$\times$ $S^1$, where $S^1$ is a
spatial direction (similar to class DIII),
with a magnetic flux piercing through the two-dimensional subspace
consisting of the cycle of the Klein bottle and $S^1$.


\subsubsection{Many-body $\Z_2$ invariant}

The construction of the many-body topological invariant
for class AII topological insulators in $(2+1)d$
is analogous
to the case of class A+$CR$ discussed in the previous section
(because of $CRT$ ($CPT$) theorem).
In the case of class A+$CR$,
the relevant spacetime manifold
is the Klein bottle$\times S^1_y$
where the cycle of the Klein bottle is $S^1_x$ (see Fig.~\ref{fig:CR_manifold}[a]). 
The magnetic flux is inserted in the $S^1_x\times S^1_y$ subspace. 
Similarly,
for class AII topological insulators in $(2+1)d$,
the relevant generating spacetime is the Klein bottle$\times S^1_y$
where
the cycle of the Klein bottle is along
the time direction $S^1_t$ (recall Fig.~\ref{fig:manifolds}[b]). 
This spacetime manifold can be realized, in the operator formalism,
by using the partial transpose on the disjoint intervals
(Fig.~\ref{fig:AII_TRS_2d_num}[a]).
%
%
Furthermore, in analogy to the $CR$-symmetric case,
we need to insert a unit magnetic flux through the $S^1_t\times S^1_y$
sub-manifold,
which can be realized by turning on 
the temporal component of the gauge field
$A_t(t,y)=\frac{2\pi y}{L_y} \delta(t-t_0)$. 
Putting together, 
we can write the desired many-body topological invariant as
the phase of
\begin{align}    \label{eq:ZAII}
Z= 
\Tr_{R_1 \cup R_3} 
\Big[ 
\rho_{R_1 \cup R_3}^{+} C_T^{I_1} [\rho_{R_1 \cup R_3}^{-}]^{\mathsf{T}_1} [C_T^{I_1}]^{\dag} \Big],
\end{align}
where the two-dimensional spatial manifold is partitioned as in
(Fig.~\ref{fig:AII_TRS_2d_num}[a])
with 
$R_{1,2,3}=I_{1,2,3}\times S^1_{y}$
where $I_{1,2,3}$ is an interval in the $x$-direction, 
and 
we introduce the reduced density matrix on $R_1 \cup R_3$ with 
the intermediate magnetic flux on $R_2$ by 
\begin{align}
\rho_{R_1 \cup R_3}^{\pm} = 
\Tr_{\overline{R_1 \cup R_3} } \Big[ 
e^{\pm \sum_{\mathbf{r}\in R_2} \frac{2 \pi i y}{L_y} n(\mathbf{r})} \ket{GS} \bra{GS} \Big]. 
\end{align}
Note that the effect of temporal gauge field is incorporated as a phase twist in the above expression.
A schematic diagram of the spatial partitioning is shown in Fig.~\ref{fig:AII_TRS_2d_num}[a].

The $\mathbb{Z}_2$ many-body topological invariant
\eqref{eq:ZAII}
can be tested for a specific microscopic model.
A generating model of non-trivial SPT phases
in class AII is the celebrated quantum spin Hall effect, 
which consists of two copies of Chern insulator with Chern numbers $ch_{\uparrow}=-ch_{\downarrow}= 1$,~\cite{KaneMeleZ2}
\begin{align} \label{eq:AII}
H&= \frac{1}{2} \sum_{\substack{\textbf{r}\\ s=1,2}} {\Big[} \psi^\dagger ({\textbf{r}+\hat{x}_s}) (i \Delta\Gamma_s - t \Gamma_3) \psi (\textbf{r}) +\text{H.c.} {\Big]} 
+ m \sum_i \psi^\dagger  (\textbf{r}) \Gamma_3 \psi (\textbf{r})
\end{align}
where $\psi (\textbf{r})=(\psi_{a\uparrow} (\textbf{r}),\psi_{b\uparrow} (\textbf{r}),\psi_{a\downarrow} (\textbf{r}),\psi_{b\downarrow} (\textbf{r}))^T$ is a four-component fermion operator in spin $(\uparrow,\downarrow)$ and orbit $(a,b)$ bases and the hopping amplitudes are given by
$\Gamma_s= (\sigma_z \tau_x, \sigma_0 \tau_y,  \sigma_0 \tau_z)$.
The $\sigma$ and $\tau$ are Pauli matrices which act on the spin and orbital degrees of freedom respectively and the $0$ subscript denotes the identity matrix.
Time-reversal symmetry is defined by $T= i\sigma_y K$. It is worth noting that the above Hamiltonian commutes with $\sigma_z$ and hence the overall $SU(2)$ spin rotation symmetry is reduced to $U(1)$ rotation symmetry around the $z$-axis.
The Hamiltonian in the momentum space can be written as  $H=\sum_{\textbf{k}} \psi^\dagger(\textbf{k}) h(\textbf{k}) \psi (\textbf{k})$, where
\begin{align}
h(\textbf{k})= \sum_{s=1,2}{\Big[} t \Gamma_s \sin k_s - r\Gamma_3 \cos k_s {\Big]}+ m\Gamma_3 .
\end{align}
The $\Z_2$ classification from the complex phase of the quantity (\ref{eq:ZAII}) is obvious in Fig.~\ref{fig:AII_TRS_2d_num}[b]. 
Moreover, the amplitude in the non-trivial phase assumes an area law behavior, $\sim e^{-\alpha L_x}$, where $\alpha$ depends on microscopic details similar to class DIII, and the amplitude approaches $1$ deep in the trivial phase.

\begin{figure}[!]
	\includegraphics[scale=1]{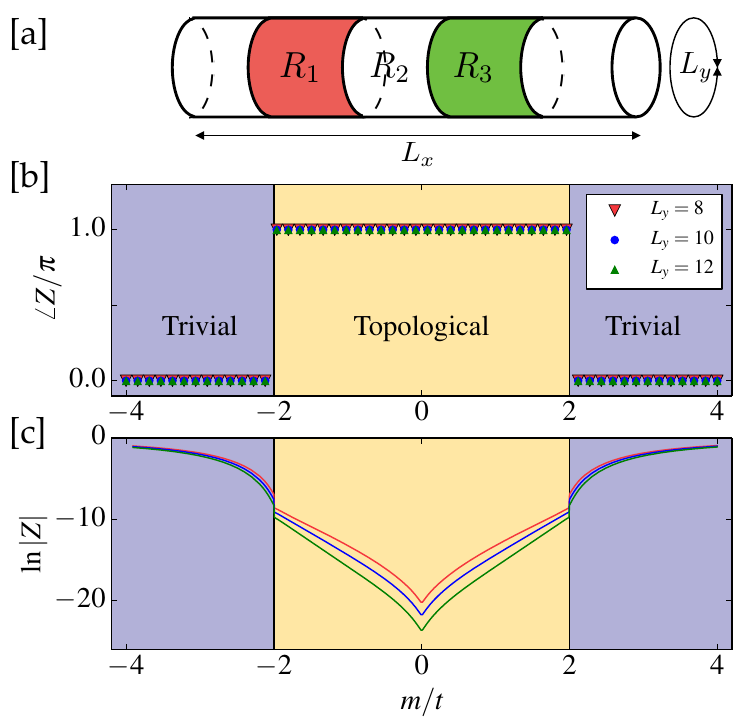}
	\caption{\label{fig:AII_TRS_2d_num}
	[a] Schematic of spatial partitioning for class AII in (2+1)$d$.
	[b]  Complex phase and
	[c] amplitude of the many-body topological invariant (\ref{eq:ZAII}) for the model Hamiltonian~(\ref{eq:AII}). 
	We set $L_y=32$, and $I_1$, $I_2$ and $I_3$ each has $8$ sites in the $y$-direction.}
\end{figure}

\subsection{$\theta$-term in $(3+1)d$}
\label{theta-term in (3+1)d}

In this section, we study the many-body topological invariant of three-dimensional insulators, whose effective topological response is given by the theta (axion) terms~\cite{SeibergWitten2016gapped, Gilkey}
\begin{align}
  &
    Z(X,A)
    = \exp \left[ i \theta_1 \int_X \Big( 
\frac{1}{2} \cdot \frac{F \wedge F}{(2 \pi)^2} + \frac{\sigma}{8} \Big) 
+ i \theta_2 \int_X 4 \cdot \frac{F \wedge F}{(2 \pi)^2} 
\right], 
\end{align}
 where $X$ is a closed oriented $(3+1)d$ spacetime manifold, 
$A$ is the spin$^c$ background field. 
Insulators with non-zero $\theta_1$ are realized in free fermions
(i.e., band insulators),
whereas realizing systems with non-zero $\theta_2$ requires
many-body interactions. 
In the following, we focus on the electromagnetic response of class A insulators which are described only by the total axion angle $\theta=\theta_1 + 8 \theta_2$. 
The topological materials described by the above action exhibit several interesting properties such as the magneto-electric response, surface quantum anomalous Hall effect, etc. 
It is important to note that this action is invariant under $\theta\to \theta+2\pi$ and hence $\theta$  is defined modulo $2\pi$. In the presence of time-reversal symmetry (class AII), $\theta$ is quantized to $0$ or $\pi$. 
For a noninteracting band structure, a Berry phase expression for $\theta$ has been given~\cite{Qi2008}
\begin{align}
\theta=\frac{1}{2\pi} \int_{\text{BZ}} d^3 k\ \epsilon_{ijk} \Tr\left[ a_i \partial_j a_k - i \frac{2}{3} a_i a_j a_k \right],
\end{align}
where $a_j^{\mu\nu}=i \bra{u_{\mu\textbf{k}}}\partial_j\ket{u_{\nu\textbf{k}}}$ is the Berry connection defined in terms of the Bloch functions of occupied bands $\ket{u_{\nu\textbf{k}}}$ and $\partial_j=\partial/\partial k_j$. Another way to compute $\theta$
is the layer-resolved Chern number  in the slab geometry~\cite{Essin2009,Shapourian2015}.
To find how different $z$ layers contribute to $\theta$ we can define a projection operator $\tilde{\cal P}_z= |z\rangle \langle z|$ onto the $z$-th layer and compute
\begin{align} \label{eq:layerres}
C(z) = \frac{1}{2\pi i} \int dk_x dk_y\  
\text{Tr}{\big[} {\cal P}_{\textbf{k}} \epsilon_{ij} (\partial_i {\cal P}_{\textbf{k}}) \tilde{\cal P}_z (\partial_j {\cal P}_{\textbf{k}}) {\big]}, 
\end{align}
in which  ${\cal P}_{\textbf{k}}=\sum_\nu |u_{\nu\textbf{k}}\rangle \langle u_{\nu\textbf{k}}|$ is the projection operator onto the occupied states of the Hamiltonian  in the slab geometry.

In contrast to the previous definitions,
here we look for a method which can be applied beyond the single particle description. In other words, we are interested in an approach analogous to the many-body Chern number(which was introduced for quantum Hall effect in $(2+1)d$)  for the 3D topological insulators. The idea is to compute the zero-temperature partition function in the presence of a background field with a unit second Chern number.
For a spacetime manifold $T^4$, such an example for $n=1$ background field is 
\begin{align} \label{eq:gauge1}
  &  A_x(x,y,z,t) = 0,
    \quad
A_y(x,y,z,t) = \frac{2 \pi x}{L_x} \delta(y-y_0),  
                                                     \nonumber \\
  &A_z(x,y,z,t) = 0, 
    \quad
A_t(x,y,z,t) = \frac{2 \pi z}{L_z} \delta(t - t_0). 
\end{align}
This background field can be formulated in the operator formalism as follows:
Spatial components of $A$ are implemented as boundary conditions:
\begin{align}
  \psi(x+L_x,y,z) &= \psi(x,y,z),
  \nonumber \\
  \psi(x,y+L_y,z) &=e^{i\frac{2 \pi x}{L_x}} \psi(x,y,z),
  \nonumber \\
\psi(x,y,z+L_z) &= \psi(x,y,z).
\end{align}
Let $\ket{\Psi}$ be the ground state of the Hamiltonian with the above boundary conditions. 
The nonlocal operator for $A_t$ becomes 
\begin{align}
\hat{U} = e^{\frac{2 \pi i}{L_z} \sum_{x,y,z} z \hat n(x,y,z)}. 
\end{align} 
Then, the $\theta$ invariant can be defined as a ground state expectation value of $\hat U$, that is 
\begin{align} \label{eq:AII3dnonflat}
e^{i \theta} = \braket{\Psi | \hat{U} |\Psi}. 
\end{align}
We should note that similar to the 2D case (\ref{eq:nonuniform_magnetic_flux}), the gauge choice (\ref{eq:gauge1}) can be made uniform in the real space as done in (\ref{eq:uniform_magnetic_flux}).

Alternatively, the gauge configuration in (\ref{eq:gauge1}) for $n=1$ can be chosen such that $\theta$ appears as an adiabatic Berry phase as a result of traversing a closed loop in the twisted boundary condition along the $z$-direction,
\begin{align} \label{eq:gauge2}
  &A_x(x,y,z,t) = 0,\quad 
    A_y(x,y,z,t) = \frac{2 \pi x}{L_x} \delta(y-y_0),  
    \nonumber \\
  &                                                     
    A_z(x,y,z,t) = \frac{2 \pi t}{L_t} \delta(z-z_0),
    \quad
    A_t(x,y,z,t) = 0. 
\end{align}
This choice has been also derived in Ref.~\cite{WangZhang2014} through a
dimensional reduction from the Hall conductivity of
the quantum Hall effect in $(4+1)d$.
In this case, we define $\ket{\Psi_\gamma}$ as the ground state of the Hamiltonian with the twisted boundary conditions
\begin{align}
\psi(x+L_x,y,z) &= \psi(x,y,z), 
\nonumber \\
\psi(x,y+L_y,z) &=e^{i\frac{2 \pi x}{L_x}} \psi(x,y,z), 
\nonumber \\
\psi(x,y,z+L_z) &= e^{i\gamma} \psi(x,y,z).
\end{align}
Hence, $\theta$ can be obtained by
\begin{align} \label{eq:AII3dberry}
  e^{i\theta}&= e^{i\int_0^{2\pi} \braket{\Psi_\gamma|\partial_\gamma|\Psi_\gamma}d\gamma}
  \nonumber \\
&= \braket{\Psi_0|\Psi_1} \braket{\Psi_1|\Psi_2}\cdots \braket{\Psi_{N-1}|\Psi_0}
\end{align}
where $\ket{\Psi_n}$ is the ground state with twisted boundary condition $\gamma_n=2\pi/N$, i.e., the interval $\gamma\in[0,2\pi)$ is discretized into $N$ steps.

\begin{figure}
\centering
\includegraphics[scale=0.4]{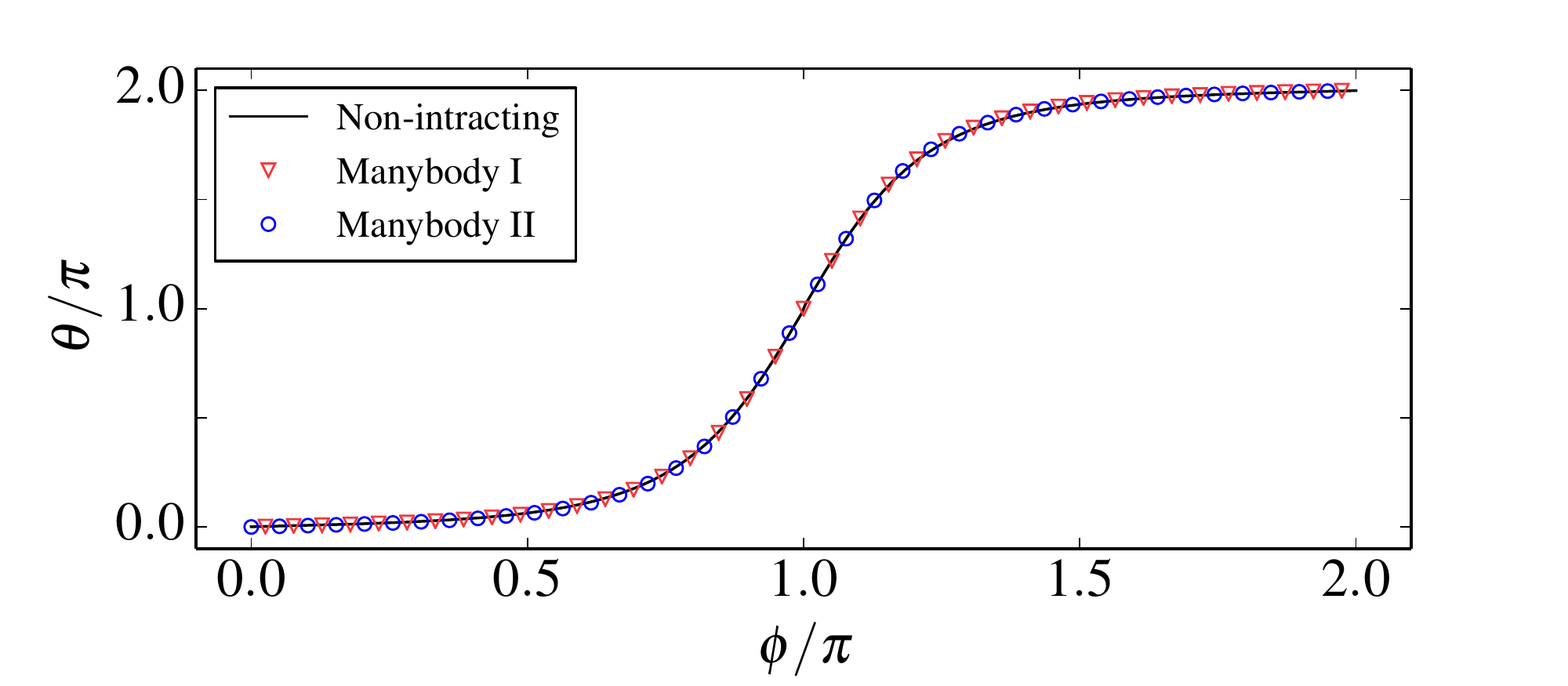}
\caption{\label{fig:theta_m5}  The axion angle as a function of the rotation angle $\phi$ in the mass parameter $m (\cos\phi + i\gamma_5 \sin\phi)$ of the broken time-reversal symmetry Dirac Hamiltonian. In the legend, Many-Body I and II refer to many-body topological invariants computed by (\ref{eq:AII3dnonflat}) and (\ref{eq:AII3dberry}), respectively. The non-interacting expression (\ref{eq:layerres}) is used as a reference. The system size is $L^3=12^3$.}
\end{figure}

To benchmark our many-body method introduced above, we apply it to the Dirac Hamiltonian
\begin{align}
h(\textbf{k})= \textbf{k}\cdot\boldsymbol{\alpha} + m [\cos(\phi)+i \sin(\phi) \gamma_5]\beta
\end{align}
where $\boldsymbol{\alpha}=(\alpha_1,\alpha_2,\alpha_3)$ and the $\gamma_5$ term is added to break the time-reversal symmetry. As we change $\phi$ from $0$ to $2\pi$ the Dirac mass changes sign, while the bulk gap remains constant at $2m$. Throughout this process, $\theta$ changes from $0$ (trivial insulator) to $2\pi$ and importantly at $\phi=\pi$, we have $\theta=\pi$ which corresponds to the time-reversal symmetric topological insulator. We simulate this phenomenon by applying our method to the Wilson-Dirac Hamiltonian on a cubic lattice
\begin{align} \label{eq:AII3d}
  H&= \frac{1}{2} \sum_{\substack{\textbf{r}\\ s=1,2,3}}
  {\Big[} \psi^\dagger ({\textbf{r}+\hat{x}_s}) (i\Delta \alpha_s - t\beta) \psi (\textbf{r}) +\text{H.c.} {\Big]} 
  \nonumber \\
  &\quad
+ \sum_i \psi^\dagger  (\textbf{r}) [m(\cos\phi+i\gamma_5 \sin\phi)+ 3t]\beta \psi (\textbf{r}),
\end{align}
where the Dirac matrices are given by
\begin{align*}
  &
\alpha_s= \tau_1\otimes \sigma_s=\left(\begin{array}{cc}
0 & \sigma_s \\ \sigma_s & 0
\end{array} \right), 
\quad 
\beta= \tau_3\otimes 1=\left(\begin{array}{cc}
\mathbb{I} & 0 \\ 0 & -\mathbb{I}
\end{array} \right), 
\quad 
\gamma_5= \tau_1\otimes \sigma_0=\left(\begin{array}{cc}
0 & \mathbb{I} \\ \mathbb{I} & 0
\end{array} \right) .
\end{align*}
Figure~\ref{fig:theta_m5} shows how $\theta$ evolves as a function of $\phi$ and there is a clear agreement among different gauge choices (\ref{eq:gauge1}) and (\ref{eq:gauge2}) as well as with the non-interacting result based on (\ref{eq:layerres}).

\subsubsection{(3+1)$d$ Class AII}

Here, we exclusively focus on time-reversal symmetric topological band
insulators.
Time-reversal in symmetry class AII
corresponds to the pin$_+^{\tilde c}$ structure in the quantum field theory in the Euclidean $(3+1)d$ spacetime. The interacting classification is $\Z_2 \times \Z_2\times \Z_2$.~\cite{Wang2014,Freed2016} 
A $\Z_2$ subgroup can be captured by the quantized axion angle $\theta \in \{0, \pi\}$. 
We use the background field introduced in (\ref{eq:gauge1}) and (\ref{eq:gauge2}) to capture the $\Z_2$ invariant.
 The model Hamiltonian for this class is given by (\ref{eq:AII3d}) with $\phi=0$. In this convention the $\sigma$ and $\tau$ matrices act on the spin and orbital degrees of freedom respectively. The computed $\theta$ for  a wide range of mass values is plotted in Fig.~\ref{fig:theta_TRS} using the many-body expressions (\ref{eq:AII3dnonflat}) and (\ref{eq:AII3dberry}). It is clear that at the many-body level, only the topological insulator with an odd number of Dirac cones on each boundary surface yields $\theta=\pi$.

\begin{figure}
\centering
\includegraphics[scale=0.4]{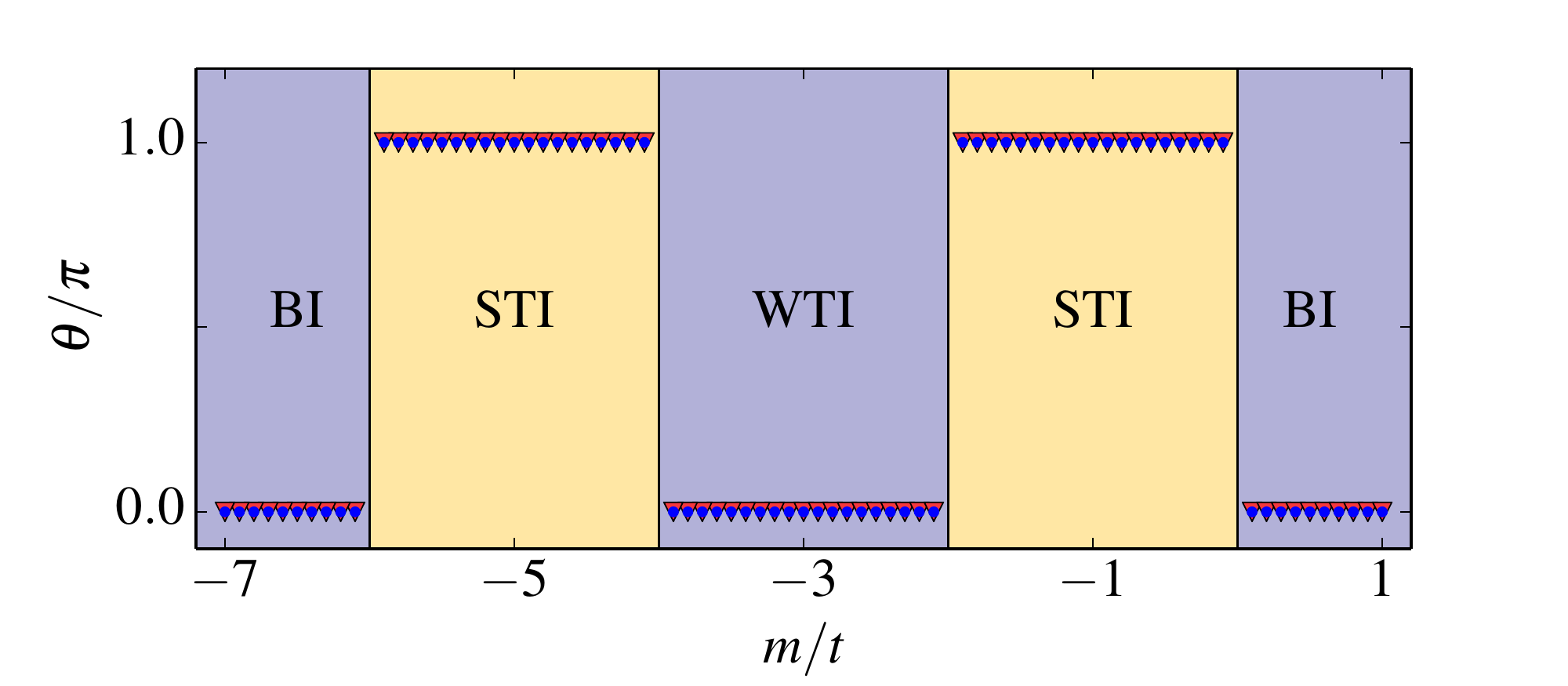}
\caption{\label{fig:theta_TRS} The axion angle $\theta$ for the time-reversal symmetric insulators (class AII) computed by the many-body formulas  (\ref{eq:AII3dnonflat})[red triangles] and (\ref{eq:AII3dberry})[blue circles]. The labels BI, STI, WTI refer to the non-interacting classification, i.e., band insulator, strong and weak topological insulators, respectively. The system size is $L^3=10^3$.}
\end{figure}

\subsubsection{(3+1)$d$ Class A + CR, $(CR)^2 = 1$}
\label{sec:3d_class_a+cr}

Class AII topological insulators $(T^2=(-1)^F)$
are $CPT$ dual to class A topological insulators
with $CR$ reflection symmetry with $(CR)^2=1$.
The interacting classification is still $\Z_2 \times \Z_2 \times \Z_2$.
Here we present another way to give one of the $\Z_2$ invariants in the  
presence of the $CR$ symmetry.
We start with the surface theory of the topological insulator which is  
the free Dirac fermion
\begin{align}
  &
    H(A) = \int d^2x\, \psi^{\dag}(x,y)
    \sum_{i=x,y}
    \big[ \sigma_i (-i\partial_i+A_i)  \big] \psi(x,y),
  \nonumber \\
  &
CR \psi^{\dag}(x,y) (CR)^{-1} = \psi(x,-y).
\end{align}
One can apply a unit magnetic flux without breaking the $CR$ symmetry,  
yielding a complex fermion zero mode $\hat \phi_0$
which is also $CR$ symmetric,
$(CR) \hat \phi^{\dag}_0 (CR)^{-1} = \hat \phi_0$.
This $(2+1)d$ surface theory is anomalous:
there is no unique ground state compatible
with the $U(1) \rtimes CR$ symmetry,
since 
the $U(1)$ symmetry forbids 
the superposition of $\ket{vac}$ and $\phi^{\dag}_0 \ket{vac}$  
while and the particle-hole $CR$ symmetry  
exchanges those.
Now the situation is the same as the $(0+1)d$ edge of the $(1+1)d$ class  
D insulator (= class A with PHS $C, C^2=1$), where the many-body $\Z_2$  
invariant is known to be the ratio of  the ground state expectation  
values of the parity of $CR$ under the $\pi$ flux and no  
flux.~\cite{shiozaki2016many}
Translating this into $(3+1)d$, the following is 
one of many-body  
$\Z_2$ invariants
\begin{widetext}
\begin{align}
\nu :=
\frac{\Braket{GS(\oint F_{xy}=2 \pi, \oint A_z = \pi) | CR | GS(\oint  
F_{xy}=2 \pi, \oint A_z = \pi)}}
{\Braket{GS(\oint F_{xy}=2 \pi, \oint A_z = 0) | CR | GS(\oint F_{xy}=2  
\pi, \oint A_z = 0)}}
\in \{\pm 1\}.
\end{align}
\end{widetext}
The cut and glue construction~\cite{Qi2011b} proves that $\nu=-1$ for  
the topological insulator.

\section{Discussion and outlook}
\label{sec:discussion}

In this paper, we developed an approach to detect interacting fermionic 
short-range entangled topological phases, focusing on those protected by antiunitary symmetry.
To this end, we introduced the fermionic analogue of the bosonic partial 
transpose in Sec.\ \ref{Fermionic partial transpose and partial antiunitary transformations},
and showed how the fermionic partial transpose combined with the unitary part of
the antiunitary symmetry
can be used to simulate the path-integral on unoriented spacetime manifold in 
Sec.~\ref{sec:Methodtocomputethetopologicalinvariant}.
The many-body topological invariants introduced in 
Sec.~\ref{sec:many-body_invariant} are based on (i) (a family of) ground 
state wave function(s) on bulk and (ii) symmetry operators in question. 
In this sense, our many-body topological invariants are the ``order 
parameters'' of fermionic short-range entangled topological phases.

It should be re-empathized that our definition of the fermionic partial 
transpose 
(\ref{def:fermion_pt})~\cite{ShapourianShiozakiRyu2016detection} is 
unitary inequivalent to that of Ref.~\cite{Eisler2015}.
The fermionic partial transpose defined in Ref.~\cite{Eisler2015} was based
on the Jordan-Wigner transformation from the bosonic partial transpose.
On the other hand, our definition of the fermionic partial transpose is 
 purely intrinsic for fermionic systems.
The resulting spacetime manifold
after taking the fermionic partial transpose
is an unoriented spacetime manifold with a single spin structure ---
this very characteristic enables us to 
compute the partition function on the generating manifold corresponding 
to the many-body topological invariant.

Also, in Sec.~\ref{Dirac quantization conditions}, we presented the 
details of the flux quantization conditions of pin$^{\tilde c}_{\pm}$ 
structures that are relativistic structures of manifolds for class AI 
and AII TRS in complex fermions. We showed that the pin$^{\tilde c}_+$ 
structure admits a half monopole flux on the real projective plane 
$\mathbb{R}P^2$, which is consistent with the even integer valued charge pump due 
to the Kramers degeneracy in class AII complex fermion systems.

Let us close by mentioning a number of interesting future directions.

--- The fermionic partial transpose we defined in (\ref{def:fermion_pt}) 
employs only the fermion parity symmetry: It was not based on the 
existence of any antiunitary symmetries.
This suggests that the fermionic partial transpose 
(\ref{def:fermion_pt}) serves as a tool to measure an intrinsic property 
of ferminic pure states density matrices.
It would be interesting to extend known applications of the bosonic 
partial transpose to fermionic ones, and also beyond bosonic ones.
For example, in Ref.~\cite{shapourian2016partial} we showed that a 
fermionic analog of the entanglement negativity defined by the fermionic 
partial transpose (\ref{def:fermion_pt}) captures the Majorana bond as 
well as the entanglement negativity at criticality derived by the CFT.

--- The generalization to finite temperature is interesting.
However, one needs to be cautious here, as 
there is a recent claim that SPT phases are unstable at any finite temperature~\cite{Yoshida}. 
Historically, Ulmann phases~\cite{UHLMANN1986229} was introduced as a generalization of Berry phases to probe the topology of density matrices~\cite{Viyuela1,Viyuela2,Huang,Budich,Linzner2016,Grusdt2017}. This approach usually involves adiabatic processes in some parameter space, which differs from our approach in terms of non-local order parameters.
Our constructions of many-body topological invariants are 
straightforwardly generalized to a density matrix as well as a pure 
state, thus, it can be applied to finite temperature canonical ensembles. 
See also recent studies~\cite{Bardyn1,Bardyn2} which discuss the mixed state charge polarization as a finite-temperature topological invariant.
Further, the many-body topological invariants for the thermal pure 
states
\cite{2012PhRvL.108x0401S,2013PhRvL.111a0401S}
is worth studying.

--- 
In the bulk of the paper, we are concerned with
the behaviors of
the phase and modulus of ``the partition functions''  
(deep) inside SPT phases. 
It would however be interesting to study
their behaviors at or near criticalities
which are in proximity to SPT phases.   
As we cross a critical point
between two distinct SPT phases, 
the modulus becomes zero in the thermodynamic limit,
and the $U(1)$ phase jumps at the criticality.
More generally,
the phase and modulus should  
depend on (the ratios between)
the system size,
the partial space region,
the distance to the critical point, etc., 
and 
a natural question is whether there is a scaling relation 
similar to usual order parameters.
Also, strictly at a critical point or within a critical phase, 
the modulus may show an interesting scaling behavior.
For instance in $(1+1)d$ systems, 
the modulus at the criticality admits a logarithmic behavior
in the subsystem size,
$\log (Z)=-\# \log L+\dots$,
as opposed to the ''area law'' 
in the in gapped phases $\log (Z)=\text{const}$,
as known in the behavior of the entanglement entropy and
other entanglement measures. 
Universal data of criticalities (such as central charges)
may be extracted from the scaling. 
For the scaling law of the closely-related quantity,
the entanglement negativity, at $(1+1)d$ critical points  
described by conformal field theory, 
see Ref.\ \cite{shapourian2016partial}.
For a recent study of the partition function (free energy) of 
$(1+1)d$ lattice systems put on a spacetime Klein bottle, 
see Ref.\ \cite{2017arXiv170705812T, 2017arXiv170804022T}.
 

--- It is known that
class DIII (TRS with $T^2=(-1)^F$) topological 
superconductors
in $(3+1)d$
are classified by $\Z_{16}$ and the generating manifold 
is the 4-dimensional real projective space ($\mathbb{R}P^4$).~\cite{FidkowskiChenVishwanath2013, Kapustin2015a}
Following the spirit of non-local order parameters discussed in
the present paper,
the corresponding many-body topological invariant 
should be constructed only by using the symmetry operator
of the problem, i.e., TRS.
However, 
we have not succeeded so far in figuring out
the construction of the many-body 
$\Z_{16}$ invariant in this manner.
As shown in Sec.~\ref{sec:Methodtocomputethetopologicalinvariant}, the 
partial transpose combined with the TRS simulates some unoriented 
spacetime manifolds, the Klein bottle and real projective plane and its 
product with other space, where the partial transpose essentially 
behaves as a reflection in a spacetime manifold.
However, 
it seems rather a challenging problem to
find a way to simulate the $\mathbb{R}P^4$ 
only by using TRS.

--- In this paper, we have not dealt with class C, CI and CII 
topological insulators/superconductors (except for Appendix~\ref{app:pin}).
The relativistic pin structures are $SU(2)$-analogs of spin$^c$ and 
pin$^c$ structures:
Instead of the $U(1)$ charge symmetry, the $SU(2)$ color symmetry is 
assumed for Majorana fermions, and the $\spin(n)$ $2\pi$ spin rotation 
is identified with the $SU(2)$ $2 \pi$ color rotation. Class C, CI and 
CII naturally appear in fermionic systems with $SU(2)$ spin rotation 
symmetry, for example, superconductors with $SU(2)$ spin rotation symmetry.
It is interesting to explore the many-body topological invariant for 
class C, CI, and CII, which remains as a future work.

%
%
%
%
%
%
%
%
%
%
%

\acknowledgements
We thank 
Matthew F.\ Lapa 
and 
Xueda Wen
for useful discussions.
This work was supported in part by the National Science Foundation grant DMR 1455296. 
K.S.\ is supported by RIKEN Special Postdoctoral Researcher Program. 
K.G.\ is supported by JSPS KAKENHI Grant Number JP15K04871.

\appendix

\section{Fermion coherent states}
\label{app:Fermion coherent state}
In this section, we summarize the fermion coherent states 
which are used in the main text of the paper. 
Let $\xi_i$, $\chi_i$ be anticommuting Grassmann variables. 
We define the integral of Grassmann variables by 
\begin{align}
  &
\int d \xi = 0, \quad  \int d \xi \xi = - \int \xi d \xi = 1, \quad 
  \nonumber \\
& 
\int d \xi_1 \cdots d \xi_n f(\{\xi_i\}) = c \quad {\rm with\ } f(\{\xi_i\}) = \cdots + c \xi_n \cdots \xi_1. 
\end{align}
Let $f_i,f^{\dag}_i$ be fermion annihilation/creation operators. 
In addition to the anticommutation relations of Grassmann variables, 
we impose the anticommutation relations between $\xi_i,\chi_i$ and $f_i,f^{\dag}_i$ as 
\begin{align}
\{f_i, \xi_j\} = \{f_i,\chi_j\} = \{f^{\dag}_i, \xi_j \} = \{f^{\dag}_i,\chi_j\} = 0. 
\end{align}
The fermion coherent states $\ket{\{\xi_i\}}, \bra{\{ \chi_i\} }$ are defined as 
\begin{align}
  \ket{\{\xi_i\}} &= e^{-\sum_i \xi_i f^{\dag}_i} \ket{\rm vac}
= \prod_{i} (1-\xi_i f^{\dag}_i) \ket{\rm vac},
  \nonumber \\ 
\bra{\{\chi_i\}} &= \bra{\rm vac} e^{\sum_i \chi_i f_i} = \bra{\rm vac} \prod_{i} (1+\chi_i f_i), 
\end{align}
which satisfy 
\begin{align}
  &f_i \ket{\{\xi_j\}} = \xi_i \ket{\{\xi_j \}} = \ket{\{\xi_j \}} \xi_i ,
    \quad
   \bra{\{\chi_j\}} f^{\dag}_i = \chi_i \bra{\{\chi_j \}},
  \nonumber \\
&\braket{\{\chi_i\} | \{\xi_i\}} = e^{\sum_i \chi_i \xi_i},
\quad \int \prod_{i} d \chi_i d \xi_i e^{- \sum_i \chi_i \xi_i} \ket{\{\xi_i\}} \bra{\{\chi_i\}} = 1, 
\end{align}
Here, $\prod_{i} d \chi_i d \xi_i$ means 
\begin{align}
\prod_{i} d \chi_i d \xi_i 
= d \chi_1 d \xi_1 d \chi_2 d \xi_2 \cdots d \chi_N d \xi_N.
\end{align}
Note that $\ket{\{\xi_i\}}$ is Grassmann even.
For a given pure state $\ket{\phi}$, its wave function $\phi(\{\chi_i\})$ in the coherent state basis is given by 
\begin{align}
  &
  \ket{\phi} = \int \prod_i d \chi_i d \xi_i e^{- \sum \chi_i \xi_i} \ket{\{\xi_i\}} \phi(\{\chi_i\}),
  \nonumber \\
  &
\phi(\{\chi_i\}) :=  \braket{\{\chi_i \} | \Psi}. 
\end{align}
The trace of the coherent state basis obeys the ``antiperiodic boundary condition'':
\begin{align}
\Tr ( \ket{\{\xi_i\}} \bra{\{\chi_i\}} )=\braket{\{-\chi_i\}|\{\xi_i\}}. 
\end{align}
Any operator $A$ on $I$ can be expanded in the coherent states as 
\begin{align}
A
  &= \int \prod_{i} d \bar \xi_i d \xi_i d \bar \chi_i d \chi_i e^{- \sum_i (\bar \xi_i \xi_i + \bar \chi_i \chi_i)}
    A(\{\bar \xi_i\}, \{\chi_i\}) \ket{\{\xi_i\}} \bra{\{\bar \chi_i\}}, 
\end{align}
where we introduced the matrix element $A(\{\bar \xi_i\}, \{\chi_i\}) = \braket{\{\bar \xi_i\} | A | \{\chi_i\}}$ 
and used the fact that the coherent states are Grassmann even. 
The trace of a Grassmann even operator $A$ is given by 
\begin{align}
\Tr A 
&= \int \prod_{i} d \bar \xi_i d \xi_i e^{- \sum_i \bar \xi_i \xi_i} A(\{-\bar \xi_i\}, \{\xi_i\}) . 
\end{align}
Finally, the Grassmann Gaussian integral can be computed
according to:
\begin{align}
\int \prod_{i} [d \xi_i d \chi_i] e^{ \bar \xi_i \xi_i + \bar \chi_i \chi_i + \chi_i M_{ij} \xi_j} 
= \frac{1}{\det M} e^{ \bar \xi_i M^{-1}_{ij} \bar \chi_j}. 
\end{align}
(Here, the Einstein summation convention is used.)

\section{Class A+CR: Twisting by CR symmetry}
\label{Class A+CR: Twisting by CR symmetry}

In this appendix, we explain how the representative model for class A with $CR$ symmetry, described by the Hamiltonian (\ref{eq:A_CR}),  can be placed on the Klein bottle. The essential idea is that the Klein bottle can be viewed as a torus in the presence of a reflection twist along one of its cycles. In what follows, we show how to twist by the $CR$ symmetry in a tight-binding lattice model.
Recall that, the CR symmetry in the model Hamiltonian (\ref{eq:A_CR}) is defined by
\begin{align}
(CR) \psi^{\dag}(x,y) (CR)^{-1}
= \sigma_x\tau_x \ \psi(-x,y).
\end{align}
where the unitary part is $\cU_{CR}=\sigma_x\tau_x$ and $\cU_{CR}^{tr} = \cU_{CR}$.
Note that the choice $\cU_{CR}=\sigma_y\tau_x$, in which  $\cU_{CR}^\text{tr}=-\cU_{CR}$, corresponds to the trivial phase.
In order to create a spatial Klein bottle, the fermion operator needs to undergo a $CR$-twist
\begin{align}
\psi^\dag(x,y+L_y)= (CR) \psi^\dag (x,y) (CR)^{-1}=  \cU_{CR}\ \psi (-x,y).
\end{align}
We can implement the above twist, using only the links connecting the edges at $y=1$ and $y=L_y$. The hopping terms between these two edges are modified such that
\begin{align}
    H_{BC_y}&=\sum_x \psi^\dag(x,y=1) [i\Delta \Gamma_2 -t\Gamma_3] \psi(x,y=L_y) +\text{H.c.}
    \nonumber  \\
&\to
  \nonumber \\
    \tilde{H}_{BC_y} &=\sum_x  (CR) \psi^\dag(x,y=1) (CR)^{-1}
    [i\Delta \Gamma_2 -t\Gamma_3]   \psi(x,y=L_y) +\text{H.c.}
  \nonumber \\
            &= \sum_x \psi(-x,y=1) \cU_{CR}
              [i\Delta \Gamma_2 -t\Gamma_3]   \psi(x,y=L_y) +\text{H.c.}
\end{align}
The resulting lattice is shown in Fig.~\ref{fig:KBlattice}.
It is worth noting that the translational symmetry along the $x$-direction is still preserved in the current topology and we can still apply a phase twist,
\begin{align}
\psi^\dag (x+L_x,y)= e^{i\gamma} \psi^\dag(x,y).
\end{align}
In the lattice model, this phase twist can be achieved by the hopping terms connecting the edges $x=1$ and $x=L_x$,
\begin{align}
  &
    H_{BC_x}=\sum_y  \psi^\dag(x=1,y) [i\Delta \Gamma_1 -t\Gamma_3] \psi(x=L_x,y) +\text{H.c.}
    \nonumber \\
&\to
                   \nonumber \\
  &
\tilde{H}_{BC_x} = e^{i\gamma} \sum_y  \psi^\dag(x=1,y)  [i\Delta \Gamma_2 -t\Gamma_3]   \psi(x=L_x,y)  +\text{H.c.}
\end{align}
Hence, the new Hamiltonian with these modifications reads
\begin{align} \label{eq:ACR_KB}
H_{KB}=& \frac{1}{2} \sum_{x=1}^{L_x-1} \sum_{y=1}^{L_y} {\Big[} \psi^\dagger({\textbf{r}+\hat{x}}) (i \Delta\Gamma_1 - t \Gamma_3) \psi (\textbf{r}) +\text{H.c.} {\Big]} \nonumber \\
&+ \frac{1}{2} \sum_{x=1}^{L_x} \sum_{y=1}^{L_y-1} {\Big[} \psi^\dagger({\textbf{r}+\hat{y}}) (i \Delta\Gamma_2 - t \Gamma_3) \psi (\textbf{r}) +\text{H.c.} {\Big]} \nonumber \\
&+ m \sum_i \psi^\dagger(\textbf{r}) \Gamma_3 \psi (\textbf{r}) + \tilde{H}_{BC_x} + \tilde{H}_{BC_y}.
\end{align}
Aside from the critical points, the above Hamiltonian is always gapped at half filling. We should note that this Hamiltonian contains pairing terms in $H_{BC_y}$ as a result of the $CR$ symmetry twist and the ground state
wave function takes a BCS form.

\begin{figure}[!]
	\begin{center}
	\includegraphics[scale=1]{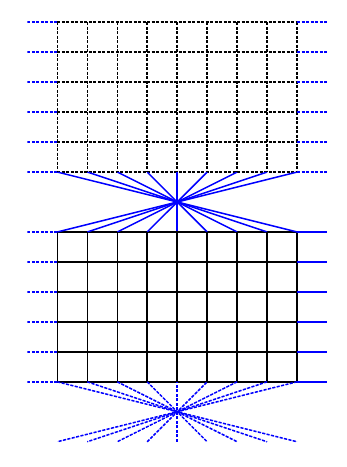}
	\end{center}
	\caption{ Implementation of the $CR$ twist on the square lattice. Links represent hopping terms and the dotted lattice is a shifted copy of the original lattice to show the new identification scheme (Klein bottle) more clearly.
	\label{fig:KBlattice}}
\end{figure}


As explained in the main text, we would like to compute the Berry phase as a result of the adiabatic process $\gamma: 0 \to 2\pi$,
\begin{align} \label{eq:ACRtwist}
    Z&=e^{\oint \braket{\Psi(\gamma)|\partial_\gamma|\Psi(\gamma)} d\gamma }
    \nonumber \\
    &=
      \braket{\Psi(\gamma_0)|\Psi(\gamma_1)}\braket{\Psi(\gamma_1)|\Psi(\gamma_2)}
      \braket{\Psi(\gamma_2)|\Psi(\gamma_3)}
\cdots \braket{\Psi(\gamma_{N-1})|\Psi(\gamma_0)}
\end{align}
where  the wave function $\ket{\Psi(\gamma_j)}$ is the ground state of the Hamiltonian (\ref{eq:ACR_KB}) defined on the Klein bottle, and for numerical purposes the $\gamma\in[0,2\pi]$ interval is discretized into $N$ steps such that 
$\gamma_n= 2\pi n/N$. In order to unravel the $\Z_2$ Berry phase, given that we are dealing with the inner product of BCS-type wave functions, it is more convenient to write the above expression in terms of a Pfaffian rather than a square root of a determinant which is blind to any possible sign (or $\Z_2$ phase). 

In general, the ground state of the BdG Hamiltonian in which the pairing potential is non-vanishing everywhere within the Brillouin zone can be written in the form
\begin{align} \label{eq:BdGwf}
\ket{\Psi} &= \exp\left(\frac{1}{2} \sum_{\textbf{r},\textbf{r}'} \phi_{\textbf{r},\textbf{r}'} \psi^\dag(\textbf{r}) \psi^\dag(\textbf{r}') \right) \ket{0},
\end{align}
where $\phi_{\textbf{r},\textbf{r}'}$ is an ${\cal N}\times {\cal N}$ skew-symmetric matrix and ${\cal N}$ is the number of sites. For a given configuration of $n$ particles $\{\textbf{r}_j\}$, $j=1,\cdots,n$ the wave function is given by
\begin{align} \label{eq:projwf}
\Psi (\{\textbf{r}_j\}) &:= \braket{0|\psi(\textbf{r}_1)\cdots \psi(\textbf{r}_n)|\Psi}=\text{Pf}\, [\phi_{\{\textbf{r}_j\},\{\textbf{r}_j\}}] 
\end{align}
where $\phi_{\{\textbf{r}_j\},\{\textbf{r}_j\}}$ is a $n\times n$ sub-matrix of the full ${\cal N}\times {\cal N}$ matrix $\phi_{\textbf{r},\textbf{r}'}$. To evaluate the inner product, it is more convenient to write the BdG ground state in the coherent state basis
\begin{align} \label{eq:MPS_BdG}
\ket{\Psi}=& \int {d}[ \xi] \  e^{\sum_{\textbf{r},\textbf{r}'} {\xi}_\textbf{r} [\phi^{-1}]_{\textbf{r},\textbf{r}'} {\xi}_{\textbf{r}'} }\ e^{\sum_\textbf{r} {\xi}_\textbf{r} \psi^\dag(\textbf{r})} \ket{0}.
\end{align}
where $\xi_j$ are Grassmann variables. The inner product can then be written as
\begin{align} \label{eq:inprod}
\braket{\Psi_1|\Psi_2}= \text{Pf}\left(\begin{array}{cc}
[\phi^{-1}_1]^\dag & \mathbb{I}_N \\
 -\mathbb{I}_N & \phi^{-1}_2
\end{array}\right),
\end{align}
where $\phi_j$  is the matrix introduced in (\ref{eq:BdGwf}) associated with the state $\ket{\Psi_j}$. This way, we can compute the inner product of each pair of wave functions and eventually obtain the Berry phase of (\ref{eq:ACRtwist}). 
We should note that the definition (\ref{eq:BdGwf}) assumes that the pairing potential is non-vanishing everywhere within the Brillouin zone which requires anti-periodic boundary condition in one or both directions on the lattice. However, one can compute the Berry phase for a general set of BCS wave functions without this assumption in terms of a single Pfaffian using the density matrix representation,
\begin{align}  \label{eq:Bphaseshot}
Z= \Tr \left[ \rho (\gamma_1)  \rho (\gamma_2)   \cdots \rho (\gamma_{N-1})\right]
&= \text{Pf}({\cal M}),
\end{align}
where $\rho(\gamma_n) =\ket{\Psi(\gamma_n)}\bra{\Psi(\gamma_n)}$ is the full density matrix,
\begin{align}
{\cal M}=\left(
\begin{array}{ccccc}
S_0 & -\Xi & 0 & \cdots& -\Xi^{\text{tr}}\\
\Xi^{\text{tr}} & S_2 & -\Xi &  & \vdots \\
0 & \Xi^{\text{tr}} & S_1 &  \ddots & \vdots \\
\vdots &  & \ddots & \ddots & -\Xi \\
\Xi & \cdots &  \cdots & \Xi^{\mathrm{tr}} & S_{N-1} \vphantom{\vdots}\\
\end{array}\right)
\end{align}
and
we use the coherent state representation of the density matrix,
\begin{align}
\rho (\gamma_m) &=\frac{1}{{\cal Z}_\rho} \int {d}[ \xi] {d} [\bar{\xi}]  \ e^{\frac{1}{2} \sum_{i,j} \boldsymbol{\xi}_i^T S^{ij}_m \boldsymbol{\xi}_j }  
  \ket{\{ \xi_j \} } \bra{ \{ \bar{\xi}_j \}},
\end{align}
in the basis $\boldsymbol{\xi}_j^T=(\bar\xi_j,\xi_j)$. The Pfaffian is a direct consequence of the Gaussian integrals over Grassmann variables.
Here, the $S_{ij}$ matrix is given by
\begin{align}
S^{ij}= Q_{ij} + i\sigma_2\ \delta_{ij},
\end{align}
 and the $Q$ matrix  is related to the single particle correlators through
\begin{align} \label{eq:Gammadef}
[Q^{-1}]_{ij}=\left( \begin{array}{cc}
[F^\dag]_{ij} & -[C^T]_{ij} \\
C_{ij} & F_{ij}
\end{array} \right),
\end{align}
where $F_{ij}=\bra{\Psi}f_i^\dag f_j^\dag\ket{\Psi}$ and $C_{ij}=\bra{\Psi}f_i^\dag f_j\ket{\Psi}$ are particle-hole and particle-particle correlators, respectively.
The dot product $\braket{\bar\xi|\chi}=e^{\bar{\xi}\chi}$ gives rise to the interconnectivity matrix
\begin{align}
\Xi=\left(
\begin{array}{cc}
0 & 0\\
\mathbb{I}_{\cal {N}} & 0
\end{array}
\right),
\end{align}
in the basis $\boldsymbol{\xi}_j^T=(\bar\xi_j,\xi_j)$ and ${\cal N}=4L_x L_y$ is the total number of lattice sites.
As mentioned, the latter method (\ref{eq:Bphaseshot}) has an advantage over the former method (\ref{eq:inprod}) in that it does not require any assumption about the pairing potential; however, (\ref{eq:Bphaseshot}) is computationally inefficient as it requires computing the Pfaffian of a very large matrix (of dimension $N{\cal N}$), while the former method (\ref{eq:inprod}) involves computing the products of Pfaffians of $N$ small matrices (of dimension ${\cal N}$).

\section{Cohomology with local coefficient
\label{app:Cohomology with local coefficient}
}
This section is devoted to briefly introducing the simplicial cochain  
complex twisted by the orientation bundle
represented by $w_1(TX) \in H^1(X;\Z_2)$, which is an example of the  
cohomology with a local coefficient.~\cite{hatcher2002algebraic}
The twisted cohomology relevant to the main text is the  
second cohomology $H^2(X;\tilde \Z)$ that provides the topological  
sectors of pin$^{\tilde c}_{\pm}$ structures and bosonic $U(1) \rtimes  
T$ and $U(1) \rtimes R$ symmetries.~\cite{Kapustin2014bosonic}

Let $X$ be an $n$-dimensional unoriented manifold.
We fix an triangulation of $X$ and assign the numbers $v_0,v_1,\cdots  
,v_n$ to the vertexes of each $n$-simplex $\Delta^n$
so that the induced orientations $v_i \to v_j (i<j)$ associated with  
low-dimensional simplexes $\Delta^{n-1}, \Delta^{n-2}, \dots$ agree at  
the boundaries.
See Fig.~\ref{fig:simplicial_str} for some examples.

\begin{figure*}[!]
        \begin{center}
        \includegraphics[scale=0.3]{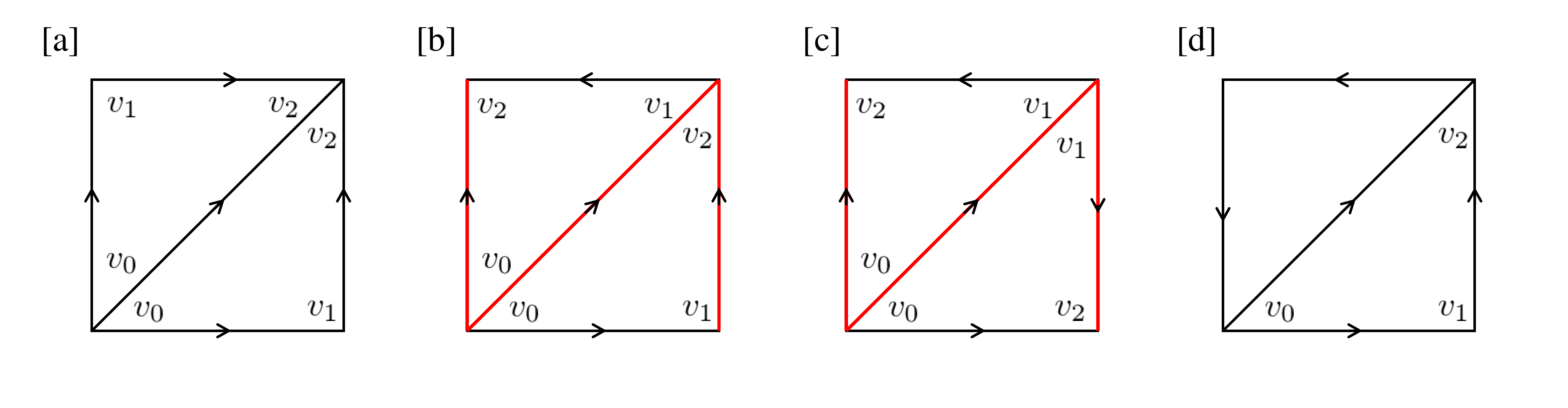}
        \end{center}
        \caption{ Examples of simplicial structures for [a] torus, [b] Klein  
bottle, and [c] real projective plane.
        [d] shows an improper assignment for the real projective plane.
        In [b] and [c], the red lines express the $\Z_2$ field representing  
orientation reversing patches (the first Stiefel-Whitney class of  
manifold).
        \label{fig:simplicial_str}}
\end{figure*}

Let $A$ be an abelian group and $C^p(X;A)$ be the set of functions from  
$p$-simplexes in $X$ to $A$.
$C^p(X;A)$ inherits the abelian structure as $(f+g)(v_0\cdots v_p):=  
f(v_0\cdots v_p)+g(v_0 \cdots v_p)$.
The differential $\delta: C^p(X;A) \to C^{p+1}(X;A)$ is defined by
\begin{align}
  (\delta f)(v_0 v_1 \cdots v_{p+1}) : =
  \sum_{j=0}^{p+1} (-1)^j f(v_0 \cdots \check{v_j} \cdots v_{p+1}),
\end{align}
where $\check{v_j}$ means $v_j$ is removed.
One can show $\delta^2=0$. 
The cocycle, coboundary, and cohomology  
are defined in a usual manner:
$Z^p(X;A):={\rm Ker}[\delta: C^p(X;A) \to C^{p+1}(X;A)]$,
$B^p(X;A):={\rm Im}[\delta: C^p(X;A) \to C^{p+1}(X;A)]$,
and $H^p(X;A):=Z^p(X;A)/B^p(X;A)$.
If $A$ is a commutative ring, the cup product $\cup: H^p(X;A) \times  
H^q(X;A) \to H^{p+q}(X;A)$ is defined by
\begin{align}
(f \cup g)(v_0 \cdots v_{p+q}) :=f(v_0 \cdots v_p) g(v_p \cdots v_{p+q}).
\end{align}

Let us move on to the twisted cohomology.
Let $\alpha \in Z^1(X;G)$ be a 1-cocycle taking values in an abelian  
group $G$.
Let $A$ be an (left) $G$-module, that is, there is an action $g \cdot a  
\in A$ for $g \in G, a\in A$, which is associative $(g_1 g_2) \cdot a =  
g_1 \cdot (g_2 \cdot a)$.
The differential operator $\delta_{\alpha}$ twisted by the 1-cocycle  
$\alpha$ is defined by
\begin{align}
  &
(\delta_{\alpha} f)(v_0 \cdots v_{p+1})
:= \alpha(v_0 v_1) \cdot f(v_1 \cdots v_{p+1}) + \sum_{j=1}^{p+1} (-1)^j  
f(v_0 \cdots \check{v_j} \cdots v_{p+1})
\end{align}
on the 1-cochain $f(v_0 \cdots v_p)$.
The 1-cocycle condition $(\delta  
\alpha)(v_0v_1v_2)=\alpha(v_1v_2)-\alpha(v_0v_2)+\alpha(v_0v_1)=0$  
ensures $\delta_{\alpha}^2=0$, thus
the cocycle $Z^p(X;A_{\alpha})$, coboundary $B^p(X;A_{\alpha})$ and  
cohomology $H^p(X;A_{\alpha})$ are defined in the same way as the  
untwisted cohomology.

For our purpose in the body of the paper,
$G=\Z_2=\{\pm 1\}$ and the 1-cocycle $w(v_0v_1) \in Z^1(X;\Z_2)$  
represents the orientation reversing patches.
The $\Z_2$ action on the abelian groups is defined by $(-1) \cdot a = -a$.
We have denoted the integer cohomology twisted by $w$ by $H^p(X;\wt \Z)  
:= H^p(X;\Z_{w})$.
Since the $\Z_2$ action on $\Z_2$ group is trivial, $H^p(X;\wt \Z_2) =  
H^p(X;\Z_2)$.
Also, the Poincar\'{e} duality $H^p(X;\tilde \Z) \cong H_{n-p}(X;\Z)$  
holds.

\subsection{Example: real projective plane $\mathbb{R}P^2$}
The twisted cohomology on $\mathbb{R}P^2$ is given by
\begin{align}
\left\{\begin{array}{ll}
H^0(\mathbb{R}P^2;\tilde \Z) \cong H_2(\mathbb{R}P^2;\Z) \cong 0, \\
H^1(\mathbb{R}P^2;\tilde \Z) \cong H_1(\mathbb{R}P^2;\Z) \cong \Z_2, \\
H^2(\mathbb{R}P^2;\tilde \Z) \cong H_0(\mathbb{R}P^2;\Z) \cong \Z.
\end{array}\right.
\label{eq:twisted_cohomology_rp2}
\end{align}
To illustrate the twisted cochain complex, let us reproduce these by a  
direct calculation.
First, we fix a simplicial structure of $\mathbb{R}P^2$ and $\Z_2$ field $w \in  
Z^1(\mathbb{R}P^2;\Z_2)$ representing the orientation reversing bonds as in  
Fig.~\ref{fig:simplicial_str} [c].
The twisted differential for 0 and 1-cochains are given as
\begin{align}
&\delta_w \left(
\begin{array}{ll}
\includegraphics[width=0.2\linewidth, trim=0cm 0cm 0cm  
0cm]{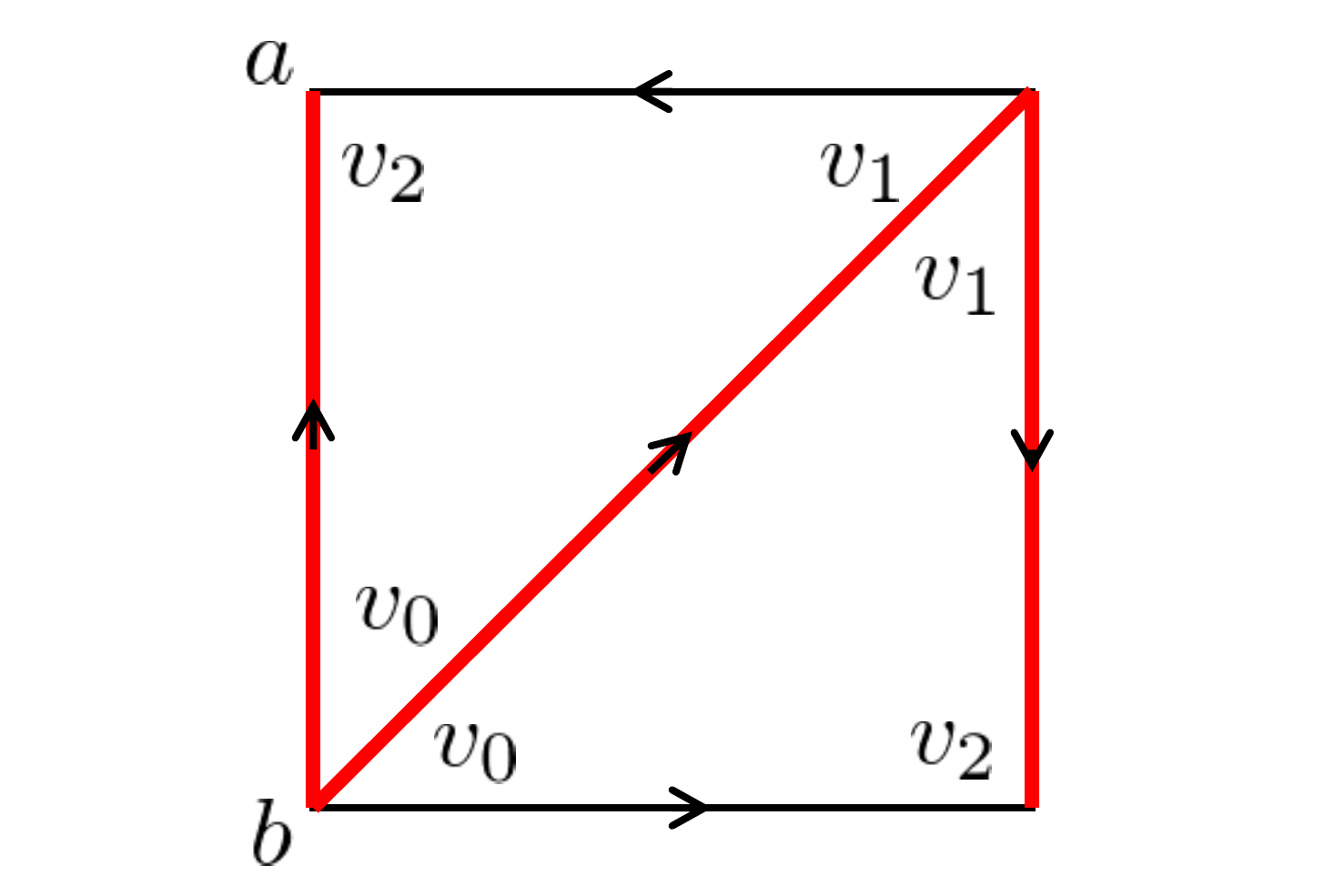}
\end{array}\right)
\quad = \quad
\begin{array}{ll}
\includegraphics[width=0.2\linewidth, trim=0cm 0cm 0cm  
0cm]{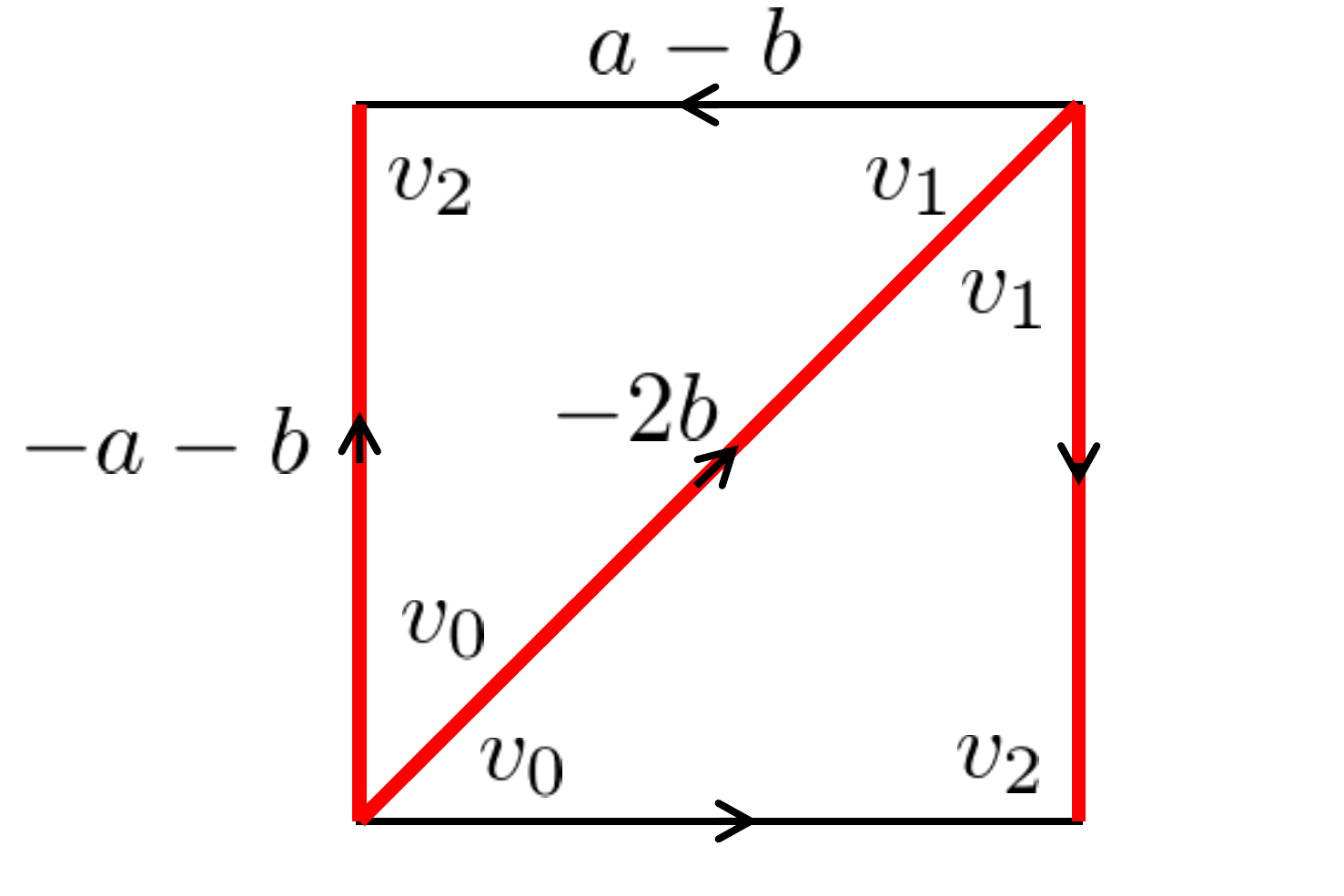}
\end{array}, \\
&\delta_w \left(
\begin{array}{ll}
\includegraphics[width=0.2\linewidth, trim=0cm 0cm 0cm  
0cm]{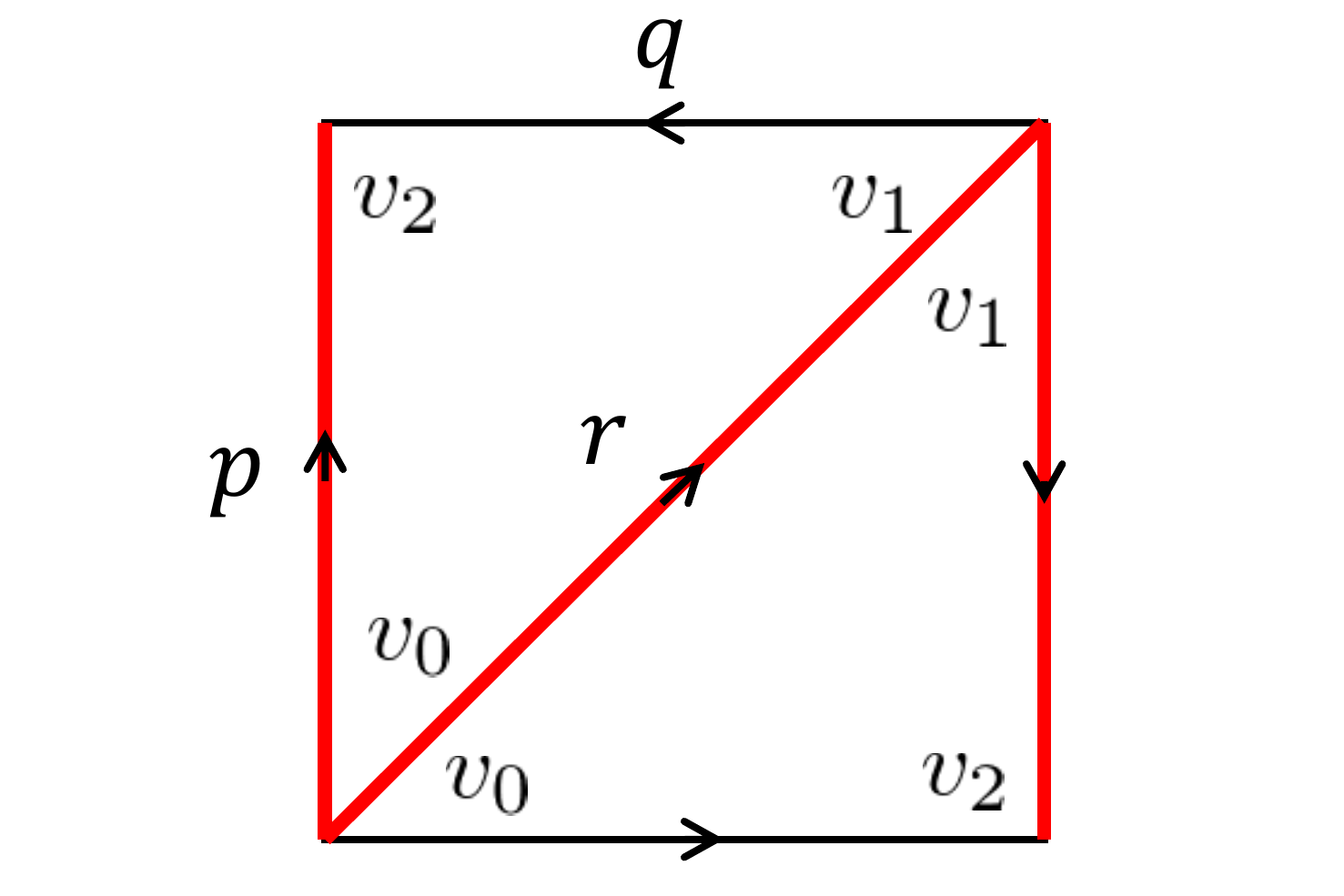}
\end{array}\right)
\quad = \quad
\begin{array}{ll}
\includegraphics[width=0.2\linewidth, trim=0cm 0cm 0cm  
0cm]{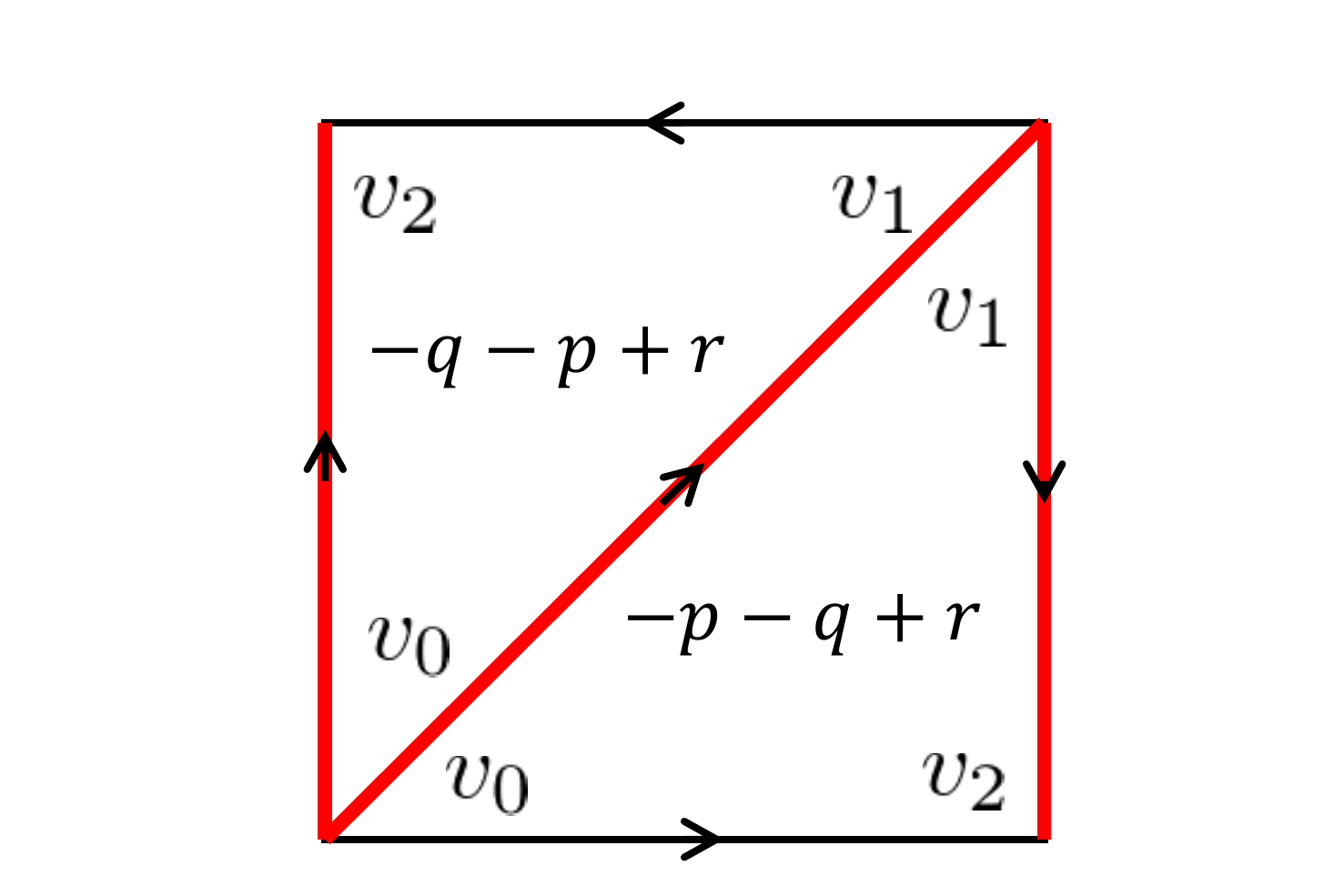}
\end{array},
\end{align}
and $\delta_w c^2=0$ for a 2-cochain.
It is easy to show (\ref{eq:twisted_cohomology_rp2}).

\section{Variants of $\pin$-structures}
\label{app:pin}

In topology and geometry, the notion of spin-structures is well-known.
As its variants, $\pin_\pm$-structures, $\spin^c$-structures 
and $\pin^c$-structures associated to $\pin_\pm$-groups, $\spin^c$-groups and
$\pin^c$-groups, respectively,
also appear. 
The other variants of $\spin$-structures associated to the groups,
\begin{align*}
&  \pin^{\tilde{c}}_\pm = \pin_\pm \ltimes_{\{ \pm 1\}} \T,
  \nonumber \\
  &
  G_\pm = \pin_\pm \times_{\{ \pm 1 \}} SU(2),
  \quad 
G_0 = \spin \times_{\{ \pm 1\}} SU(2),
\end{align*}
appear in the paper by Freed and Hopkins \cite{Freed2016} in the context of a `relativistic $10$-fold way' (p. 57, (9.22), (9.23)):
$$
\begin{array}{c|l|c|c|c}
\hline
s & H^c & K & \mbox{Cartan} & D \\
\hline
0 & \spin^c & \T & \mbox{A} & \C \\
1 & \pin^c & \T & \mbox{AIII} & \mathrm{Cliff}^{\C}_{-1} \\
\hline
\end{array}
$$
$$
\begin{array}{c|l|c|c|c}
\hline
s & H^c & K & \mbox{Cartan} & D \\
\hline
0 & \spin & \{ \pm 1 \} & \mbox{D} & 
\R \\
\hline
-1 & \pin_+ & \{ \pm 1 \} & \mbox{DIII} & 
\mathrm{Cliff}_{-1} \\
\hline
-2 & \pin_+ \ltimes_{\{ \pm 1 \}} \T & \T & \mbox{AII} &
\mathrm{Cliff}_{-2} \\
\hline
-3 & \pin_- \times_{\{ \pm 1 \}} SU(2) & SU(2) & \mbox{CII} &
\mathrm{Cliff}_{-3} \\
\hline
4 & \spin \times_{\{ \pm 1 \}} SU(2) & SU(2) & \mbox{C} &
\mathbb{H} \\
\hline
3 & \pin_+ \times_{\{ \pm 1 \}} SU(2) & SU(2) & \mbox{CI} & 
\mathrm{Cliff}_{+3} \\
\hline
2 & \pin_- \ltimes_{\{ \pm 1 \}} \T & \T & \mbox{AI} &
\mathrm{Cliff}_{+2} \\
\hline
1 & \pin_- & \{ \pm 1 \} & \mbox{BDI} &
\mathrm{Cliff}_{+1} \\
\hline
\end{array}
$$
In these tables, $K$ refers the subgroup of onsite symmetry in fermionic
symmetries specified by $H^c$,
where $\{\pm 1\}, \T$ and $SU(2)$ are the $\Z_2$ fermion parity, $U(1)$ particle number and $SU(2)$ flavor symmetries, respectively. 
$D$ means the super division algebra. 
See Ref.~\cite{Freed2016} for the detail.

The purpose of this appendix is to discuss the variants of $\spin$-structures
associated to the groups $\pin_\pm^{\tilde{c}}$, $G_0$ and $G_\pm$. In
particular, the cohomological obstructions to the existence of these structures
will be discussed.


\subsection{The variants of $\spin$-groups}

We summarize the relevant variants of $\spin$-groups here.


\subsubsection{The standard variants}

Let $\mathrm{Cliff}_{\pm n}$ be the algebra over $\R$ generated by $e_1, \ldots e_n$ subject to the relations:
$$
e_i e_j + e_j e_i = \pm 2 \delta_{ij}.
$$
Hence $\mathrm{Cliff}_{-n}$ is the standard Clifford algebra. The $\spin$-group $\spin(n)$ is defined to be the group of unit norm elements in the even part of $\mathrm{Cliff}_{-n}$ or equivalently that of $\mathrm{Cliff}_{+n}$, which gives rise to the double cover of $SO(n)$:
$$
1 \to \{ \pm 1 \} \to \spin(n) \to SO(n) \to 1.
$$ 
Considering the unit norm elements in the full Clifford algebra instead, we get the $\pin$-groups $\pin_\pm(n) \subset \mathrm{Cliff}_{\pm n}$, which fit into the exact sequences:
$$
1 \to \{ \pm 1 \} \to \pin_\pm(n) \to O(n) \to 1.
$$ 

Let $\mathrm{Cliff}^\C_n = \mathrm{Cliff}_{\pm n} \otimes_{\R} \C$ be the complexified Clifford algebra. The same consideration as above leads to the $\spin^c$-group $\spin^c(n)$,
$$
1 \to \T \to \spin^c(n) \to SO(n) \to 1,
$$
and the $\pin^c$-group $\pin^c(n)$,
$$
1 \to \T \to \pin^c(n) \to O(n) \to 1.
$$
These groups can be realized as:
\begin{align*}
&\spin^c(n) = \spin(n) \times_{\{ \pm 1 \}} \T, \\
&\pin^c(n) = \pin_\pm(n) \times_{\{ \pm 1 \}} \T, 
\end{align*}
where, for instance, $\spin(n) \times_{\{ \pm 1 \}} \T$ is the quotient of $\spin(n) \times \T$ by the central subgroup $\{ (1, 1), (-1, -1) \}$.


\subsubsection{The variants $\pin^{\tilde{c}}_\pm$}

The group $\pin^{\tilde{c}}_\pm$-group $\pin^{\tilde{c}}_\pm(n)$ is defined as follows:
First of all, notice the homomorphism $\phi : \pin_\pm(n) \to \{ \pm 1 \}$ given by 
the composition of the surjection $\pin_\pm(n) \to O(n)$ and the determinant $\det: O(n) \to \{ \pm 1 \}$.
By means of $\phi : \pin_\pm(n) \to \{ \pm 1 \}$, we let $\pin_\pm(n)$ act on $\T$ as its automorphisms. 
Namely, for any $g \in \pin_\pm(n)$, we get a homomorphism $\T \to \T$, ($u \mapsto u^{\phi(g)}$). Using this action of $\pin_\pm(n)$ on $\T$, we can form the semi-direct product $\pin_\pm(n) \ltimes \T$. In this semi-direct product is a central subgroup $\{ (1, 1), (-1, -1) \}$, the quotient by which eventually defines the $\pin^{\tilde{c}}_\pm$-group
$$
\pin^{\tilde{c}}_\pm(n) 
= \pin_\pm (n) \ltimes_{\{ \pm 1 \}} \T.
$$
It should be noticed that this group is a $\phi$-twisted extension of $O(n)$ by $\T$ in the sense of \cite{Freed2013}:
$$
\begin{CD}
1 @>>> \T @>>> \pin^{\tilde{c}}_\pm(n) @>{\pi}>> O(n) @>>> 1, \\
@. @. @. @VV{\phi}V @. \\
@. @. @. \{ \pm 1 \} @. 
\end{CD}
$$
in which $\phi = \det$. This means $gu = u^{\phi(\pi(g))}g$ for any $g \in \pin^{\tilde{c}}_\pm(n)$ and $u \in \T$.


\subsubsection{The variants $G_0$ and $G_\pm$}

The groups $G_0(n)$ and $G_\pm(n)$ are defined by
\begin{align*}
G_0(n) &= \spin(n) \times_{\{ \pm 1\}} SU(2), \quad 
  \nonumber \\
G_\pm(n) &= \pin_\pm(n) \times_{\{ \pm 1 \}} SU(2),
\end{align*}
which are the quotients of the product groups by the subgroup $\{ (1, 1), (-1, -1) \}$. By construction, $G_0(n)$ fits into the exact sequence
$$
1 \to SU(2) \to G_0(n) \to SO(n) \to 1,
$$
and $G_\pm(n)$ into
$$
1 \to SU(2) \to G_\pm(n) \to O(n) \to 1.
$$
The following exact sequences can also be derived from the definitions:
\begin{gather*}
1 \to \spin(n) \to G_0(n) \to SO(3) \to 1, \\
1 \to \pin_\pm(n) \to G_\pm(n) \to SO(3) \to 1,
\end{gather*}
where $SO(3) = SU(2)/\{ \pm 1 \}$.


\subsection{The obstructions}


\subsubsection{Lifting and obstruction in general}

Suppose that we have a (surjective) homomorphism $\pi : \tilde{G} \to G$ and a principal $G$-bundle $q : P \to X$ on a space $X$. Then a \textit{lifting} or a \textit{lift} of the structure group $G$ of $P$ to $\tilde{G}$ is defined to be a principal $\tilde{G}$-bundle $\tilde{q} : \tilde{P} \to X$ together with a map $\varpi : \tilde{P} \to P$ such that:
\begin{itemize}
\item
$q \circ \varpi = \tilde{q}$,

\item
$\varpi(\tilde{p}\tilde{g}) = \varpi(\tilde{p})\pi(\tilde{g})$
for $\tilde{p} \in \tilde{P}$ and $\tilde{g} \in \tilde{G}$.
\end{itemize}
The notion of equivalence for lifting can be defined in an apparent way. 
In the case when $G = O(n)$ is the orthogonal group, a lifting may often be called a $\tilde{G}$-structure. A $\tilde{G}$-structure of a real vector bundle is defined through its unoriented frame bundle, and a $\tilde{G}$-structure of a manifold through its tangent bundle.

Then the existence of a lifting is a topological problem. By means of the \v{C}ech cohomology classification of principal bundles,~\cite{Bry} we can prove the following:

\begin{lem} \label{lem:general_obstruction}
Let $G^\tau$ be a $\phi$-twisted extension of $G$ by an abelian group $K$.
$$
\begin{CD}
1 @>>> K @>>> G^\tau @>{\pi}>> G @>>> 1. \\
@. @. @. @VV{\phi}V @. \\
@. @. @. \{ \pm 1 \} @. 
\end{CD}
$$
\begin{itemize}
\item[(a)]
For any principal $G$-bundle $P$ on $X$, the obstruction to the existence of a lifting of $P$ is a cohomology class
$$
\mathfrak{o}(P) \in H^2(X; \underline{K}_{P}),
$$
where $\underline{K}_{P}$ is the sheaf of germs of $K$-valued functions twisted by the principal $\Z_2$-bundle $P \times_G \Z_2 \to X$ associated to the homomorphism $\phi : G \to \Z_2$.

\item[(b)]
Suppose that $P$ admits a lifting. Then the set of isomorphism classes of liftings of $P$ is identified with $H^1(X; \underline{K}_{P})$.
\end{itemize}
\end{lem}

The obstruction class can be made more explicit if the groups in question are specified concretely. This is the case for $\spin$, $\pin_\pm$, $\spin^c$, $\pin^c$ and $\pin^{\tilde{c}}_\pm$.


\subsubsection{The cases of $\mathrm{Spin}$, $\mathrm{Pin}_\pm$ and $\mathrm{Spin}^c$}
In the cases of $\mathrm{Spin}$,
$\mathrm{Pin}_\pm$ and $\mathrm{Spin}^c$, the obstruction classes are
well-known.~\cite{Kirby,L-M}
For a principal $O(n)$-bundle $P \to X$, we denote by
$$
w_i(P) \in H^{2i}(X; \Z_2)
$$
the $i$th Stiefel-Whiteny class.

\begin{lem}[$\spin$-structure]
The following holds true:
\begin{itemize}
\item
A principal $O(n)$-bundle $P \to X$ admits a $\spin$-structure if and only if $w_1(P) = 0$ and $w_2(P) = 0$.

\item
In the case when $P$ admits a $\spin$-structure, the set of isomorphism classes of $\spin$-structures on $P$ is identified with $H^1(X; \Z_2)$.
\end{itemize}
\end{lem}

\begin{lem}[$\pin_-$-structure]
The following holds true:
\begin{itemize}
\item
A principal $O(n)$-bundle $P \to X$ admits a $\pin_-$-structure if and only if $w_2(P) + w_1(P)^2 = 0$.

\item
In the case when $P$ admits a $\pin_-$-structure, the set of isomorphism classes of $\pin_-$-structures on $P$ is identified with $H^1(X; \Z_2)$.
\end{itemize}
\end{lem}

\begin{lem}[$\pin_+$-structure]
The following holds true:
\begin{itemize}
\item
A principal $O(n)$-bundle $P \to X$ admits a $\pin_+$-structure if and only if $w_2(P) = 0$.

\item
In the case when $P$ admits a $\pin_+$-structure, the set of isomorphism classes of $\pin_+$-structures on $P$ is identified with $H^1(X; \Z_2)$.
\end{itemize}
\end{lem}

We write $\beta : H^2(X; \Z_2) \to H^3(X; \Z)$ for the Bockstein connecting homomorphism associated to the short exact sequence $0 \to \Z \to \Z \to \Z_2 \to 0$ in coefficients.

\begin{lem}[$\spin^c$-structure]
The following holds true:
\begin{itemize}
\item
A principal $O(n)$-bundle $P \to X$ admits a $\spin^c$-structure if and only if $w_1(P) = 0$ and $\beta(w_2(P)) = 0$. 
Equivalently, $P$ admits a $\spin^c$-structure if and only if $w_1(P) = 0$ and there exists a class $c \in H^2(X; \Z)$ such that $w_2(P) = c \mod 2$ in $H^2(X; \Z_2)$

\item
In the case when $P$ admits a $\spin^c$-structure, the set of isomorphism classes of $\spin^c$-structures on $P$ is identified with $H^2(X; \Z)$.
\end{itemize}
\end{lem}

The cohomology class $W_3(P) = \beta(w_2(P)) \in H^3(X; \Z)$ is called the third integral Stiefel-Whiteny class. Because of the exact sequence:
$$
\cdots \longrightarrow
H^2(X; \Z) \overset{\mod 2}{\longrightarrow}
H^2(X; \Z_2) \overset{\beta}{\longrightarrow}
H^3(X; \Z) \longrightarrow
\cdots,
$$
we have $W_3(P) = 0$ if and only if $w_2(P)$ is the mod $2$ reduction of a class in $H^2(X; \Z)$, as stated in the lemma above.

\begin{lem}[$\pin^c$-structure]
The following holds true:
\begin{itemize}
\item
A principal $O(n)$-bundle $P \to X$ admits a $\pin^c$-structure if and only if $W_3(P) = \beta(w_2(P)) = 0$. 
Equivalently, $P$ admits a $\pin^c$-structure if and only if there exists a class $c \in H^2(X; \Z)$ such that $w_2(P) = c \mod 2$.

\item
In the case when $P$ admits a $\pin^c$-structure, the set of isomorphism classes of $\pin^c$-structures on $P$ is identified with $H^2(X; \Z)$.
\end{itemize}
\end{lem}


\subsubsection{The cases of $\pin^{\tilde{c}}_\pm$}

Recall that the principal $\Z_2$-bundle $P \times_{O(n)} \Z_2 \to X$ is associated to a principal $O(n)$-bundle $P \to X$ and the homomorphism $\det : O(n) \to \Z_2$. The bundle $P \times_{O(n)} \Z_2$ is classified by $w_1(P) \in H^1(X; \Z_2)$, and provides us a local system for the cohomology. Concretely, the $\Z_2$-bundle twists the cohomology $H^*( - ; A)$ with coefficients in the abelian group $A$ to define the cohomology $H^*( - ; \underline{A}_P)$ with coefficients in a local system $\underline{A}_{P}$ twisted by $P$. In the case of $A = \Z_2$, it turns out that $\underline{(\Z_2)}_P = \Z_2$ since the automorphism group of the group $\Z_2$ is trivial. Accordingly, associated to the exact sequence $0 \to \Z \to \Z \to \Z_2 \to 0$ in coefficients, we get the long exact sequence
\begin{align*}
\cdots &\to
H^n(X; \underline{\Z}_P) \overset{2 \times}{\to}
H^n(X; \underline{\Z}_P) \overset{\mod 2}{\to}
H^n(X; \Z_2) \\
&\overset{\tilde{\beta}}{\to}
H^{n+1}(X; \underline{\Z}_P) \to
\cdots.
\end{align*}
The connecting homomorphism in this exact sequence will be denoted with $\tilde{\beta}$.

\begin{lem}[$\pin^{\tilde{c}}_-$-structure]
The following holds true:
\begin{itemize}
\item
A principal $O(n)$-bundle $P \to X$ admits a $\pin^{\tilde{c}}_-$-structure if and only if $\tilde{\beta}(w_2(P) + w_1(P)^2) = 0$ in $H^3(X; \underline{\Z}_{P})$. 
Equivalently, $P$ admits a $\pin^{\tilde{c}}_-$-structure if and only if there is $\tilde{c} \in H^2(X; \underline{\Z}_P)$ such that $w_2(P) + w_1(P)^2 = \tilde{c} \mod 2$.

\item
In the case where $P$ admits a $\pin^{\tilde{c}}_-$-structure, the set of isomorphism classes of $\pin^{\tilde{c}}_-$-structures on $P$ is identified with $H^2(X; \underline{\Z}_{P})$.
\end{itemize}
\end{lem}

\begin{lem}[$\pin^{\tilde{c}}_+$-structure]
The following holds true:
\begin{itemize}
\item
A principal $O(n)$-bundle $P \to X$ admits a $\pin^{\tilde{c}}_+$-structure if and only if $\tilde{\beta}(w_2(P)) = 0$ in $H^3(X; \underline{\Z}_{P})$. 
Equivalently, $P$ admits a $\pin^{\tilde{c}}_-$-structure if and only if there is $\tilde{c} \in H^2(X; \underline{\Z}_P)$ such that $w_2(P) = \tilde{c} \mod 2$.

\item
In the case where $P$ admits a $\pin^{\tilde{c}}_+$-structure, the set of isomorphism classes of $\pin^{\tilde{c}}_+$-structures on $P$ is identified with $H^2(X; \underline{\Z}_{P})$.
\end{itemize}
\end{lem}


\subsubsection{The cases of $G_0$ and $G_\pm$}

The groups $G_0(n) = \spin(n) \times_{\{ \pm 1 \}} SU(2)$ and $G_\pm(n) = \pin_\pm(n) \times_{\{ \pm 1 \}} SU(2)$ have the non-abelian group $SU(2)$ as the kernels of the surjective homomorphisms onto $SO(n)$ or $O(n)$:
\begin{gather*}
1 \to SU(2) \to G_0(n) \overset{\pi}{\to} SO(n) \to 1, \\
1 \to SU(2) \to G_\pm(n) \overset{\pi}{\to} O(n) \to 1.
\end{gather*}
Accordingly, we cannot apply Lemma \ref{lem:general_obstruction}. Instead, the exact sequences
\begin{gather*}
1 \to \spin(n) \to G_0(n) \to SO(3) \to 1, \\
1 \to \pin_\pm(n) \to G_\pm(n) \to SO(3) \to 1,
\end{gather*}
provide the clue to the obstructions.

\begin{lem}[$G$-structures] \label{lem:G_structure}
The following holds true:
\begin{itemize}
\item
($G_0$-structure)
A principal $O(n)$-bundle $P \to X$ admits a $G_0$-structure if and only if $w_1(P) = 0$ and there is a principal $SO(3)$-bundle $Q \to X$ such that $w_2(P) = w_2(Q)$.

\item
($G_-$-structure)
A principal $O(n)$-bundle $P \to X$ admits a $G_-$-structure if and only if there is a principal $SO(3)$-bundle $Q \to X$ such that $w_2(P) + w_1(P)^2 = w_2(Q)$.

\item
($G_+$-structure)
A principal $O(n)$-bundle $P \to X$ admits a $G_+$-structure if and only if there is a principal $SO(3)$-bundle $Q \to X$ such that $w_2(P) = w_2(Q)$.
\end{itemize}
\end{lem}

\begin{proof}
We only give the proof in the case of $G_+$, since the other cases can be shown similarly. For the proof, we choose an open cover $\{ U_i \}$ of $X$ so that $P|_{U_i}$ admits a trivialization $s_i$. As usual, we write $U_{i \cdots j} = U_i \cap \cdots \cap U_j$ for the intersections. Then, by $s_j = s_i g_{ij}$, we get the transition functions $g_{ij} : U_{ij} \to O(n)$ subject to the cocycle condition $g_{ij} g_{jk} = g_{ik}$ on $U_{ijk}$. 

Now, to prove the `only if' part, let us suppose that $P$ admits a $G_+$-structure. Then, from the principal $G_+$-bundle $\tilde{P}$, we get the transition functions $\tilde{g}_{ij} : U_{ij} \to G_+(n)$. By the definition of $G_+(n)$, we may express $\tilde{g}_{ij}$ as $\tilde{g}_{ij} = [\tilde{f}_{ij}, \tilde{h}_{ij}]$ by using maps $\tilde{f}_{ij} : U_{ij} \to \pin_-(n)$ such that $\pi(\tilde{f}_{ij}) = g_{ij}$ and $\tilde{h}_{ij} : U_{ij} \to SU(2)$, where $\pi : \pin_+(n) \to O(n)$ is the surjection. On the one hand, if $\{ \tilde{f}_{ij} \}$ glue together to form a $\pin_+(n)$-structure, then $w_2(P) = 0$, so that we can take $Q$ to be the trivial principal $SO(3)$-bundle to get the relation $w_2(P) = w_2(Q)$. On the other hand, if $P$ does not admit any $\pin_+$-structure, then the failure of the cocycle condition for $\tilde{f}_{ij}$ defines a \v{C}ech cocycle $(z_{ijk}) \in Z^2(\{ U_i \}; \Z_2)$ representing $w_2(P) \neq 0$. Since the existence of $G_+$-structure is assumed, it turns out that the failure of the cocycle condition for $\tilde{h}_{ij}$ agrees with $z_{ijk}$. From this viewpoint, the cocycle $(z_{ijk})$ represents $w_2(Q)$ of the principal $SO(3)$-bundle $Q$ whose transition function is $\pi(\tilde{h}_{ij})$, where $\pi : SU(2) \to SO(3)$ is the surjection. To summarize, we have the relation $w_2(P)  = w_2(Q)$, showing the `only if' part.

To show the `if' part, let us assume that there is a principal $SO(3)$-bundle $Q$ such that $w_2(P) = w_2(Q)$. Choosing $\tilde{f}_{ij} : U_{ij} \to \pin_+(n)$ such that $\pi(\tilde{f}_{ij}) = g_{ij}$, we get a representative $(z_{ijk}) \in Z^2(\{ U_i \}; \Z_2)$ of $w_2(P)$ by $\tilde{f}_{ij}\tilde{f}_{jk} = z_{ijk} \tilde{f}_{ik}$, where $\pi : \pin_+(n) \to O(n)$ is the surjection. By the assumption, we can get $h_{ij} : U_{ij} \to SU(2)$ such that $h_{ij}h_{jk} = z_{ijk}h_{ik}$ and $\{ \pi(h_{ij}) \}$ is a set of transition functions of $Q$, where $\pi : SU(2) \to SO(3)$ is the surjection. If we define $\tilde{g}_{ij} : U_{ij} \to G_+(n)$ by $\tilde{g}_{ij} = [\tilde{f}_{ij}, \tilde{h}_{ij}]$, then $\{ \tilde{g}_{ij} \}$ satisfy the cocycle condition to form transition functions of a $G_+$-structure of $P$. 
\end{proof}


\subsection{Examples}


\subsubsection{Computing cohomology}
\label{subsec:compute_cohomology}

To see the (non-)existence of a $\spin$-structure and its variants, we often need the information of the cohomology groups of the space under study. Thus, before the consideration of examples, we summarize here some ways to compute cohomology groups.

\medskip

Let $\pi : P \to X$ be a principal $\Z_2$-bundle. Associated to $P$ is the local system $\underline{\Z}_P$. Also, we have the associated real line bundle $\ell = P \times_{\Z_2} \R$. The Thom isomorphism gives us
$$
H^n(X; \underline{\Z}_P) \cong H^{n+1}(D(\ell), S(\ell); \Z),
$$
where $D(\ell)$ and $S(\ell)$ are the disk bundle and the sphere bundle of $\ell$ respectively. Clearly, $D(\ell)$ is homotopic to $X$, and $S(\ell)$ is identified with $P$. Therefore the exact sequence for the pair $(D(\ell), S(\ell))$ leads to the exact sequence:
\begin{align*}
\cdots &\to
H^n(X; \Z) \overset{\pi^*}{\to}
H^n(P; \Z) \to
H^n(X; \underline{\Z}_P) \\
&\to
H^{n+1}(X; \underline{\Z}_P) \to
\cdots,
\end{align*}
where $\pi^* : H^n(X; \Z) \to H^n(P; \Z)$ is the pull-back under the projection $\pi : P \to X$. The Thom isomorphism theorem also gives us
$$
H^n(X; \Z) \cong H^{n+1}(D(\ell), S(\ell); \underline{\Z}_{\pi^*P}),
$$
where the pull-back $\pi^* P \to D(\ell)$ is used to define the local system $\underline{\Z}_{\pi^*P}$. It is easy to see the identifications:
\begin{align*}
H^n(D(\ell); \underline{\Z}_{\pi^*P}) 
\cong H^n(X; \underline{\Z}_P), 
\quad 
H^n(S(\ell); \underline{\Z}_{\pi^*P}) 
\cong H^n(P; \Z).
\end{align*}
Hence the exact sequence for the pair $(D(\ell), S(\ell))$ also leads to:
\begin{align*}
\cdots &\to
H^n(X; \underline{\Z}_P) \overset{\pi^*}{\to}
H^n(P; \Z) \to
H^n(X; \Z) \\
&\to
H^{n+1}(X; \Z) \to
\cdots.
\end{align*}
In the case of the cohomology with coefficients in $\Z_2$, the two exact sequences as derived above agree to give the Gysin exact sequence for $P$:
\begin{align*}
\cdots &\to
H^n(X; \Z_2) \overset{\pi^*}{\to}
H^n(P; \Z_2) \to
H^n(X; \Z_2) \\
&\overset{w_1(P)}{\to}
H^{n+1}(X; \Z_2)
\cdots,
\end{align*}
where $w_1(P) : H^{n-1}(X; \Z_2) \to H^n(X; \Z_2)$ is the cup product with the first Stiefel-Whiteny class $w_1(P) \in H^1(X; \Z_2)$ of $P$.

%


\subsubsection{Circle}

Let $S^1$ be the circle. Its cohomology ring with coefficients in $\Z$ is the exterior ring generated by a single generator $a \in H^1(S^1; \Z) \cong \Z_2$:
$$
H^*(S^1; \Z) \cong \Lambda_{\Z} a.
$$
The cohomology with $\Z_2$-coefficients is then computed by the universal coefficient theorem:
$$
H^*(S^1; \Z_2) \cong \Z_2[a]/(a^2),
$$
where $a \in H^1(S^1; \Z_2) \cong \Z_2$ is the generator. Associated to $a$ is a principal $\Z_2$-bundle $P$ and the local system $\tilde{\Z} = \underline{\Z}_P$. Since $P = S^1$ and $\pi : P \to S^1$ is the double covering, we use the exact sequences in Subsection \ref{subsec:compute_cohomology} to see:
$$
\begin{array}{|c|c|c|c|c|c|}
\hline
& n = 0 & n = 1 & n = 2 & n = 3 & n = 4 \\
\hline
H^n(S^1; \Z) & \Z & \Z & 0 & 0 & 0 \\
\hline
H^n(S^1; \tilde{\Z}) & 0 & \Z_2 & 0 & 0 & 0 \\
\hline
\end{array}
$$
Because $H^3(S^1; \Z)$ and $H^3(S^1; \tilde{\Z})$ are trivial, all the obstruction classes living in these cohomology groups are automatically trivial.

\medskip

Now, the Stiefel-Whiteny classes of the real line bundle $\ell = P \times_{\Z_2} \R$ on $S^1$ representing the generator $a \in H^1(\R P^2; \Z_2)$ are:
\begin{align*}
w_1(\ell) = a, \quad
w_2(\ell) = 0, \quad 
w_3(\ell) = 0, \quad 
w_4(\ell) = 0, \quad 
\cdots.
\end{align*}
Let $T \to S^1$ be the tangent bundle $T = TS^1$ of $S^1$. This is isomorphic to the product bundle $T = S^1 \times \R$, so that $w_i(T) = 0$ for all $i > 0$. Then we can see the following (non-)existences of a $\spin$-structure and its variants ($\bigcirc$ and $\times$ respectively mean the existence and the non-existence):
$$
\begin{array}{|c|c|c|c|}
\hline
S^1 & \ell & T \\
\hline
\spin^c & \times & \bigcirc \\
\hline
\pin^c & \bigcirc & \bigcirc \\
\hline
\spin & \times & \bigcirc \\
\hline
\pin_- & \bigcirc & \bigcirc \\
\hline
\pin_+ & \bigcirc & \bigcirc \\
\hline
\pin^{\tilde{c}}_- & \bigcirc & \bigcirc \\
\hline
\pin^{\tilde{c}}_+ & \bigcirc & \bigcirc \\
\hline
G_0 & \times & \bigcirc \\
\hline
G_- & \bigcirc & \bigcirc \\
\hline
G_+ & \bigcirc & \bigcirc \\
\hline
\end{array}
$$

We give comments and clarifications: 
(i) For our application in the main text, 
we are concerned with the lifting of the tangent bundle
(i.e., related to fermionic spinors).
On the other hand, here, we are presenting more general examples 
(i.e., not just $T$, but also $\ell$). 
(ii)
The table means:
When the tangent bundle $T$ can have (cannot have) a given structure, 
we have $\bigcirc$ ($\times$).
Similarly, 
When the $\ell$-bundle on $S^1$ can have (cannot have) a given structure, 
we have $\bigcirc$ ($\times$).
(iii)
$S^1$ is orientable and so is the tangent bundle. 
A $\pin_{\pm}$-structure on $T$ means a $\spin$-structure on $T$. 
On the one hand, $\ell$ is nonorientable, meaning the absence of $\spin$-, $\spin^c$- and $G_0$-structures. 
The table states that $\ell$ admits other pin structures.



\subsubsection{$2$-dimensional sphere}

For the sphere $S^2$, its cohomology group is as follows:
$$
\begin{array}{|c|c|c|c|c|c|}
\hline
& n = 0 & n = 1 & n = 2 & n = 3 & n = 4 \\
\hline
H^n(S^2; \Z) & \Z & 0 & \Z & 0 & 0 \\
\hline
H^n(S^2; \Z_2) & \Z_2 & 0 & \Z_2 & 0 & 0 \\
\hline
\end{array}
$$
Since $H^1(S^2; \Z_2) = 0$, there is no non-trivial local system. 

\medskip

Let $L \to S^2$ be the complex line bundle of degree $1$. Forgetting the complex structure, we can think of $L$ as a real vector bundle of rank $2$. Its Stiefel-Whiteny classes are as follows:
\begin{align*}
w_1(L) = 0, \quad 
w_2(L) = a, \quad 
w_3(L) = 0, \quad
w_4(L) = 0, \quad  
\cdots,
\end{align*}
where $a \in H^2(S^2; \Z_2) \cong \Z_2$ is a generator. (Being a complex vector bundle, $L$ is orientable and $w_1(L) = 0$. 
Since $L$ is oriented, $w_2(L) = c_1(L) \mod 2$ is non-trivial. 
The vanishing of the higher classes is the dimensional reason.) Let $T = TS^2$ be the tangent bundle of $S^2$. As in the case of $L$, we can check that its Stiefel-Whiteny classes are as follows:
\begin{align*}
w_1(T) = 0, \quad 
w_2(T) = 0, \quad 
w_3(T) = 0, \quad 
w_4(T) = 0, \quad 
\cdots.
\end{align*}
Then the existences of $\spin$-structures, etc. are summarized as follows:
$$
\begin{array}{|c|c|c|c|}
\hline
S^2 & L & T \\
\hline
\spin^c & \bigcirc & \bigcirc \\
\hline
\pin^c & \bigcirc & \bigcirc \\
\hline
\spin & \times & \bigcirc \\
\hline
\pin_- & \times & \bigcirc \\
\hline
\pin_+ & \times & \bigcirc \\
\hline
\pin^{\tilde{c}}_- & \bigcirc & \bigcirc \\
\hline
\pin^{\tilde{c}}_+ & \bigcirc & \bigcirc \\
\hline
G_0 & \bigcirc & \bigcirc \\
\hline
G_- & \bigcirc & \bigcirc \\
\hline
G_+ & \bigcirc & \bigcirc \\
\hline
\end{array}
$$
We here make comments on the existence of $G_n$-structures: The existence of $G_0$-structure is the consequence of that of $\spin$-structure. To show the existence of $G_\pm$-structure, notice that the oriented frame bundle of $L$ is a principal $SO(2)$-bundle with non-trivial $w_2(L)$. Then, associated to the inclusion $SO(2) \to SO(3)$, we get a principal $SO(3)$-bundle $Q = F(L) \times_{SO(2)} SO(3)$ on $S^2$. By construction, $w_2(Q) = w_2(L)$, so that we see the existence of $G_\pm$-structures.


\subsubsection{Real projective plane}

Let us consider the real projective space $\R P^2$ of dimension $2$. The cohomology with integer coefficients $H^*(\R P^2; \Z)$ is well-known.~\cite{bott2013differential} The cohomology with coefficients in $\Z_2$ is known to be the truncated polynomial ring
$$
H^*(\R P^2; \Z_2) \cong \Z_2[a]/(a^3),
$$
where $a \in H^1(\R P^2; \Z_2) \cong \Z_2$ is the generator. Associated to this generator $a$, we have a principal $\Z_2$-bundle $P$ and the local system $\tilde{\Z} = \underline{\Z}_P$. Notice that $P = S^2$ as a space. Then the cohomology $H^*(\R P^2; \tilde{\Z})$ can be computed by appealing to Subsection \ref{subsec:compute_cohomology}:
$$
\begin{array}{|c|c|c|c|c|c|}
\hline
& n = 0 & n = 1 & n = 2 & n = 3 & n = 4 \\
\hline
H^n(\R P^2; \Z) & \Z & 0 & \Z_2 & 0 & 0 \\
\hline
H^n(\R P^2; \tilde{\Z}) & 0 & \Z_2 & \Z & 0 & 0 \\
\hline
\end{array}
$$
Because $H^3(\R P^2; \Z)$ and $H^3(\R P^2; \tilde{\Z})$ are trivial, all the obstruction classes in these cohomology groups are automatically trivial. (The relevant connecting homomorphisms $\beta$ and $\tilde{\beta}$ are also trivial.)

\medskip

Now, let $\ell \to \R P^2$ be the real line bundle representing the generator $a \in H^1(\R P^2; \Z_2)$. It is then easy to see:
\begin{align*}
w_1(\ell) = a, \quad
w_2(\ell) = 0, \quad 
w_3(\ell) = 0, \quad 
w_4(\ell) = 0, \quad  
\cdots.
\end{align*}
We also let $T \to \R P^2$ be the tangent bundle $T = T\R P^2$ of $\R P^2$. Since $\ell \oplus T$ is isomorphic to the product real bundle of rank $3$, the Whiteny formula of the total Stiefel-Whiteny class leads to:
\begin{align*}
w_1(T) = a, \quad
w_2(T) = a^2, \quad
w_3(T) = 0, \quad 
w_4(T) = 0.
\end{align*}
These information on $w_i$, the ring structure of $H^*(\R P^2; \Z_2)$ and the vanishings of $\beta$ and $\tilde{\beta}$ allow us to determine whether there are $\spin^c$-, $\pin^c$-, $\spin$-, $\pin_\pm$- and $\pin^{\tilde{c}}_\pm$-structures. It is easy to see the non-existence of $G_0$-structures by $w_1 \neq 0$. To see the existence of $G_\pm$-structures, let us cosider the real vector bundle $\ell^{\oplus 2} = \ell \oplus \ell$ of rank $2$. We have
\begin{align*}
  w_1(\ell^{\oplus 2}) =0,
                         \quad
w_2(\ell^{\oplus 2}) = a^2.
\end{align*}
Therefore $\ell^{\oplus 2}$ is oriented but its oriented frame bundle does not admit a $\spin$-structure. The frame bundle is a principal $SO(2)$-bundle. Associated to this bundle and the inclusion $SO(2) \to SO(3)$, we have a principal $SO(3)$-bundle $Q \to \R P^2$ such that $w_2(Q) = a^2$. This helps us to determine whether there are $G_\pm$-structures on $\ell$ and $T$. We can similarly study the (non-)existence of $\spin$-structures, etc. for $\ell^{\oplus 2}$. The results are summarized as follows:
$$
\begin{array}{|c|c|c|c|c|}
\hline
\R P^2 & \ell & T & \ell^{\oplus 2} \\
\hline
\spin^c & \times & \times & \bigcirc \\
\hline
\pin^c & \bigcirc & \bigcirc & \bigcirc \\
\hline
\spin & \times & \times & \times \\
\hline
\pin_- & \times & \bigcirc & \times \\
\hline
\pin_+ & \bigcirc & \times & \times \\
\hline
\pin^{\tilde{c}}_- & \bigcirc & \bigcirc & \bigcirc \\
\hline
\pin^{\tilde{c}}_+ & \bigcirc & \bigcirc & \bigcirc \\
\hline
G_0 & \times & \times & \bigcirc \\
\hline
G_- & \bigcirc & \bigcirc & \bigcirc \\
\hline
G_+ & \bigcirc & \bigcirc & \bigcirc \\
\hline
\end{array}
$$


\subsubsection{Klein bottle}

Let $K = KB$ denote the Klein bottle. A way to realize $K$ is to let the group $G = \langle a, b |\ baba^{-1} \rangle$ act on $\R^2$ by
\begin{align*}
  a(z) = \bar{z} + \frac{1}{2},
  \quad 
b(z) = z + i.
\end{align*}
Then this action is free (actually properly discontinuous), and the quotient space is the Klein bottle $K = \R^2/G$. Notice that $G$ is nothing but the fundamentl group of $K$. Another way to realize $K$ is to take the connected sum of two copies of the real projective space $\R P^2$, that is, to glue together two M\"{o}bius bands along their boundary circles. This allows us to compute the integer cohomology of $K$ by means of the Mayer-Vietoris exact sequence, and then the cohomology with $\Z_2$-coefficients by the universal coefficient theorem:
$$
\begin{array}{|c|c|c|c|c|c|}
\hline
& n = 0 & n = 1 & n = 2 & n = 3 & n = 4 \\
\hline
H^n(K; \Z) & \Z & \Z & \Z_2 & 0 & 0 \\
\hline
H^n(K; \Z_2) & \Z_2 & \Z_2 \oplus \Z_2 & \Z_2 & 0 & 0 \\
\hline
\end{array}
$$

We here specify a basis of $H^1(K; \Z_2) \cong \Z_2 \oplus \Z_2$. 
For this aim, notice that principal $\Z_2$-bundles on $K = \R^2/G$ are in a bijective correspondence with $G$-equivariant principal $\Z_2$-bundles on $\R^2$. Because $\R^2$ is contractible, any principal $\Z_2$-bundle on $\R^2$ is trivial. Then the action of the discrete group $G$ on the trivial $\Z_2$-bundle amounts to a homomorphism $G \to \Z_2$, realizing the isomorphismsm:
$$
\mathrm{Hom}(G, \Z_2)
\cong \mathrm{Hom}(\pi_1(K), \Z_2)
\cong H^1(K; \Z_2).
$$ 
We now choose the following basis of $\mathrm{Hom}(G, \Z_2) \cong H^1(K; \Z_2)$:
\begin{align*}
\rho_\alpha :
\left\{
\begin{array}{l}
a \mapsto -1, \\
b \mapsto 1.
\end{array}
\right.
  \quad
\rho_\beta :
\left\{
\begin{array}{l}
a \mapsto 1, \\
b \mapsto -1.
\end{array}
\right.
\end{align*}
The basis $\alpha$ admits the following realization: Let us consider the subgroup $G' = \langle a^2, b |\ a^2ba^{-2}b^{-1} \rangle$ in $G$. This is a free abelian group of rank $2$, and a normal subgroup of $G$, with the quotient the cyclic group of order $2$. Thus, taking the quotient of $\R^2$ by $G' \subset G$, we get the $2$-dimensional torus $T^2 = \R^2/G'$. With the residual free action of $G/G' \cong \langle a | a^2 \rangle \cong \Z_2$, the torus defines a principal $\Z_2$-bundle $\pi : T^2 \to K$. The first Stiefel-Whiteny class of this bundle is $\alpha$. To see this, note that, under the identification $\R^2/G' \cong S^1 \times S^1$, where $S^1 \subset \C$ is the unit circle, we can describe the $\Z_2$-action as $(u, v) \mapsto (-u, \bar{v})$. Thus the first projection $p : S^1 \times S^1 \to S^1$ is equivariant with respect to the free $\Z_2$-action on $S^1$, and descends to give a map $K \to S^1/\Z_2 = S^1$. Thus, we have a map of principal $\Z_2$-bundles:
$$
\begin{CD}
\R^2/G' @>>> \R/\langle a^2 \rangle \\
@VVV @VVV \\
\R^2/G @>>> \R/\langle a \rangle.
\end{CD}
$$
The principal $\Z_2$-bundle on $S^1$ is represented by the generator of $H^1(S^1; \Z_2) \cong \Z_2$, and its pull-back under the projection is $\alpha$. Consequently, the principal $\Z_2$-bundle $T^2 \to K$ realizes $\alpha$. We remark that $H^1(S^1; \Z_2)$ is a direct summand of $H^1(S^1; \Z_2)$, since the projection $K \to S^1$ admits a section $S^1 \to K$ induced from the equivariant map $s : S^1 \to S^1 \times S^1$ given by $s(u) = (u, 1)$. 

\smallskip

Now, the ring structure of $H^*(K; \Z_2)$ is determined by:

\begin{lem}
$\alpha^2 = 0$ and $\beta^2 = \alpha\beta \neq 0$.
\end{lem}

\begin{proof}
Recall that, by a projection $K \to S^1$, the cohomology $H^*(S^1; \Z_2)$ is a direct summand of $H^*(K; \Z_2)$. The class $\alpha$ is then the generator of $H^1(S^1; \Z_2) \cong \Z_2$. By the dimensional reason, $\alpha^2 = 0$ on $S^1$. Thus, we also have $\alpha^2 = 0$ on $K$. To prove the remaining formulae, let $\ell_\beta$ denote the real line bundle such that $w_1(\ell_\beta) = \beta$. Since $\ell_\beta$ is of rank $1$, we have $w_2(\ell_\beta) = 0$ automatically. We then study the existence of a $\pin_-$-structure on $\ell_\beta$. A $\pin_-$-structure on $K$ is equivalent to a $G$-equivariant $\pin_-$-structure on $\R^2$. Since a $\pin_-$-structure on $\R^2$ is unique, the problem is whether we can let $G$ act on the $\pin_-$-structure $\R^2 \times \pin_-(1)$ in the compatible way. A candidate of such an action is to let $a$ and $b$ in $G$ act on the $\pin_-$-structure as follows:
\begin{align*}
  (z, g) \overset{a}{\mapsto} (a(z), i g),
  \quad 
(z, g) \overset{b}{\mapsto} (b(z), g),
\end{align*}
where $g \in \pin_-(1) = \Z_4 = \{ \pm 1, \pm i \}$. But, this transformation is not compatible with the relation $bab = a$ in $G$. In the same way, we can verify that the other candidates are incompatible with the relation in $G$. Thus, we conclude that there is no $\pin_-$-structure on $\ell_\beta$, so that:
$$
0 \neq w_2(\ell_\beta) + w_1(\ell_\beta)^2 = w_1(\ell_\beta)^2 = \beta^2.
$$
The computation of $\alpha\beta$ is similar: Let $\ell_\alpha$ be the real line bundle such that $w_1(\ell_\alpha) = \alpha$. We easily see:
\begin{align*}
  w_1(\ell_\alpha \oplus \ell_\beta) = \alpha + \beta,
  \quad 
w_2(\ell_\alpha \oplus \ell_\beta) = \alpha\beta.
\end{align*}
Thus, we check whether $\ell_\alpha \oplus \ell_\beta$ admits a $\pin_+$-structure or not. A $\pin_+$-structure on the vector bundle on $K$ is equivalent to a $G$-equivariant $\pin_+$-structure on the trivial rank $2$ vector bundle on $\R^2$. There is the unique $\pin_+$-structure $\R^2 \times \pin_+(2)$ for the trivial rank $2$ bundle on $\R^2$. A candidate of a $G$-action on this $\pin_+$-structure is
\begin{align*}
  (z, g) \overset{a}{\mapsto} (a(z), e_1 g),
  \quad
(z, g) \overset{b}{\mapsto} (b(z), e_2 g),
\end{align*}
where $e_1, e_2 \in \pin_+(2) \subset \mathrm{Cliff}_{+2}$ are the generators of the Clifford algebra. Since $e_1e_2e_1 = -e_2$, the transformations above do not define a $G$-action. In the same way, we find that the $\pin_+$-structure cannot be $G$-equivariant. Hence there is no $\pin_+$-structure on $\ell_\alpha \oplus \ell_\beta$, and $\alpha\beta \neq 0$. Since $H^2(K; \Z_2) \cong \Z_2$, the non-trivial elements $\beta^2$ and $\alpha\beta$ must agree.
We remark that the computation above is consistent with \cite{hatcher2002algebraic} (Example 3.8) under a base change. 
\end{proof}

\medskip

Let $T = TK$ be the tangent bundle of the Klein bottle. Inspecting the action of $G$, we can see $TK \cong \underline{\R} \oplus \ell_\alpha$, where $\underline{\R}$ is the trivial real line bundle and $\ell_\alpha$ is the real line bundle such that $w_1(\ell_\alpha) = \alpha$. The cohomology with coefficients in the local system $\underline{\Z}_P$ with $P$ the orientation bundle of $TK$ can be computed by using a Gysin exact sequence. (We regard the total space of the $\Z_2$-bundle $T^2 \to K$ as a $\Z_2$-equivariant principal $S^1$-bundle on the circle $S^1$ with the $\Z_2$-action $z \mapsto \bar{z}$. Then $H^1(K; \underline{\Z}_P) \cong \Z \oplus \Z_2$ and $H^2(K; \underline{\Z}_P) \cong \Z$, with the other cohomology trivial.) But, for the present purpose, we only need the fact that the third cohomology $H^3(K; \underline{\Z}_P)$ is trivial, which can be shown by appealing to Subsection \ref{subsec:compute_cohomology}. Now we have:
$$
\begin{array}{|c|c|}
\hline
K & T  \\
\hline
\spin^c & \times  \\
\hline
\pin^c & \bigcirc  \\
\hline
\spin & \times  \\
\hline
\pin_- & \bigcirc  \\
\hline
\pin_+ & \bigcirc  \\
\hline
\pin^{\tilde{c}}_- & \bigcirc  \\
\hline
\pin^{\tilde{c}}_+ & \bigcirc  \\
\hline
G_0 & \times   \\
\hline
G_- & \bigcirc  \\
\hline
G_+ & \bigcirc  \\
\hline
\end{array}
$$

\section{Dirac quantization conditions}
\label{App: Dirac quantization conditions}

In this section, we derive Dirac quantization conditions
for pin$^c$ and pin$^{\tilde c}_{\pm}$ connections. 
We first warm up with
the $U(1)$ and twisted $U(1)$ connections which arise
when the complex scalar field is endowed with $U(1) \times CT$ and $U(1) \rtimes T$ symmetry, respectively. 
Next, we will move on to the pin$^c$ and pin$^{\tilde c}_{\pm}$ connections. 
The convention of this section follows Ref.~\cite{Metlitski2015}. 
A good example of manifold having the nontrivial Dirac condition for pin$^c$
and pin$^{\tilde c}_{\pm}$ structures
is the real projective plane $\mathbb{R}P^2$, which will be described in
Appendix~\ref{sec:pin_rp2}.
We will also compute the partition functions on $\mathbb{R}P^2$ for free fermions with these pin structures.

\subsection{Bosonic $U(1) \times CT$ symmetry}
\label{Bosonic U(1)xCT symmetry}

We start with the usual Dirac quantization condition of $U(1)$ connections on unoriented manifolds. 
For simplicity, we employ the complex scalar field. 
The relevant symmetry is the reflection symmetry $R$ or the antiunitary
charge-conjugation symmetry $CT$
which preserves the $U(1)$ charge $e^{i Q}$.  
Let $X$ be an unoriented manifold. The Euclidean action on $X$ is given by 
\begin{align}
S 
= \int_X d^n x \sqrt{g} \Big[ D^{\mu} \bar \phi D_{\mu} \phi + V(\bar \phi \phi) \Big], 
\label{eq:action_u1}
\end{align}
where $\phi$ and $\bar \phi$ are independent bosonic fields, and their covariant derivatives are given as 
\begin{align}
D_{\mu} \phi = (\p_{\mu} - i A_{\mu}) \phi, \quad 
D_{\mu} \bar \phi = (\p_{\mu} + i A_{\mu}) \bar \phi. 
\label{eq:covariant_der_u1}
\end{align}
When $X$ is not an Euclidean space $\R^n$, $X$ needs a cover $\{U_i\}$ to give
its global definition.
We thus start with the fields $\phi,\bar \phi$ and $A$ which are
locally defined on each patch $U_i$. 
At the intersection $U_i \cap U_j$, the complex scalar fields are glued
by the patch transformation 
\begin{align}
CT:\,   
\wt \phi(y) = e^{i \alpha_{ij}(y)} \phi(x(y)), \quad 
\wt {\bar \phi}(y) = e^{-i \alpha_{ij}(y)} \bar \phi(x(y)) 
\end{align}
for both orientation-preserving/reversing intersections $U_i \cap U_j$, 
where $x$ and $y$ are local coordinates of $U_i$ and $U_j$, respectively. 
The volume form is invariant $d^n x \sqrt{g} = d^n y \sqrt{\tilde g}$. 
The patch transformation for the $U(1)$ connection $A$ is given by imposing the covariant condition $\wt D_{\mu} \wt \phi(y) = e^{i \alpha_{ij}(y)} \frac{\p x^{\beta}}{\p y^{\mu}} D_{\beta} \phi(x(y))$ to get 
\begin{align}
CT:\,  
\wt A_{\mu}(y) = \frac{\p x^{\beta}}{\p y^{\mu}} A_{\beta}(x(y)) + \p_{\mu} \alpha_{ij}(y).  
\label{eq:u1_conn_patch}
\end{align}
The patch transformations determine the data,
$e^{i \alpha_{ij}}: U_i \cap U_j \to U(1)$, $e^{i \alpha_{ji}} = e^{-i
  \alpha_{ij}}$, of transition functions.
The consistency condition on the intersections $U_i \cap U_j \cap U_k$ of three patches leads to the cocycle condition 
\begin{align}
e^{i \alpha_{ij}} e^{i \alpha_{jk}} e^{i \alpha_{ki}} = 1 \quad {\rm on}\ \ U_i \cap U_j \cap U_k.
\label{eq:u1_cocycle_cond}
\end{align}
$\{ e^{i \alpha_{ij}} \}$ determines a complex line bundle $L$ and the complex
scalar fields are considered as sections of the line bundle $L$. 
$A$ is a $U(1)$ connection associated to the line bundle $L$ which is subject to the patch transformation (\ref{eq:u1_conn_patch}). 

Topologically distinct field configurations of $U(1)$ connections
can be characterized by the first Chern class $c_1(L) \in H^2(X;\Z)$ of the line bundle $L$, which is explained below. 
Let $\theta_{ij}: U_i \cap U_j \to \R$, $\theta_{ji} = - \theta_{ij}$, be a lift of $\alpha_{ij}$, that is, $\theta_{ij}$ satisfies $e^{2 \pi i \theta_{ij}} = e^{i \alpha_{ij}}$. 
The cocycle condition (\ref{eq:u1_cocycle_cond}) implies that $\theta_{ij}$ obey 
\begin{align}
\theta_{ij} + \theta_{jk} + \theta_{ki} \in \Z, 
\label{eq:2-cocycle_u1}
\end{align}
hence $c_{ijk}:= \theta_{ij} + \theta_{jk} + \theta_{ki}$ is a 2-cochain
which takes its values in $\mathbb{Z}$.
A direct computation shows $\delta c = 0$ and changing the lift as $\theta_{ij} \mapsto \theta_{ij} + b_{ij}$ with $b \in C^1(X;\Z)$ induces the 2-coboundary $c \mapsto c + \delta b$. 
In sum, $c_{ijk}$ is a 2-cocycle $c \in Z^2(X;\Z)$ and the transition functions $\{e^{i \alpha_{ij}}\}$ determines an element $[c] \in H^2(X;\Z)$. 
$c_1(L) = [c]$ is called the first Chern class of $L$. 
For examples: $H^2(T^2;\Z) \cong \Z$ represents the monopole flux inside $T^2$;  
$H^2(\mathbb{R}P^2;\Z) \cong \Z_2$ represents the $\Z_2$-valued quantized $U(1)$ holonomy in $\mathbb{R}P^2$. 

For a systematic description of $U(1)$ connections,
it is useful to introduce the notation in the differential cohomology which can be extended to other pin structures.~\cite{Hopkins-Singer} 
Let $(C^*(X;\R), \delta)$ be the singular cochain complex. A $U(1)$ connection $A$ is a triple 
\begin{align}
A = (c,h,\omega) \in C^2(X;\R) \times C^1(X;\R) \times \Omega^2(X)
\label{eq:triple}
\end{align}
subject to the following conditions 
\begin{align}
&c \equiv 0  \ \mod \ \Z \quad ({\rm Dirac\ quantization\ condition}), \label{eq:boson_dirac_cond} \\
&\delta c = 0, \label{eq:closed_c_u1} \\
&d \omega = 0, \label{eq:connection_cond_2} \\
&\omega = c + \delta h. \label{eq:2_form_u1}
\end{align}
Two $U(1)$ connections $A = (c,h,\omega)$ and $A'=(c',h',\omega')$ are gauge equivalent if and only if there exists a pair $(b,k) \in C^1(X;\Z) \times C^0(X;\R)$ so that 
\begin{align}
c'=c-\delta b, \quad 
h'=h+b+\delta k, \quad 
\omega'=\omega.
\end{align}
In Eq.\ \eqref{eq:2_form_u1}, the l.h.s.\ of $\omega = c+\delta h$ means the integral of the 2-form $\omega$ over the 2-simplex.
The Dirac quantization condition (\ref{eq:boson_dirac_cond}) together with the cocycle condition $(\ref{eq:closed_c_u1})$ means $c \in Z^2(X;\Z)$. 
$h$ can be understood as $U(1)$ fields $e^{i h_{ij}}$ living on bonds of lattice systems. 
$c$ is the 2-cocycle introduced in (\ref{eq:2-cocycle_u1}). 
The gauge invariant 2-form $w$ is the field strength $F = 2 \pi \omega$. 
The usual Dirac quantization condition $\int_C F/2 \pi \in \Z$ for 2-cycles $C \in Z_2(X;\Z)$ follows from (\ref{eq:2_form_u1}). 

We illustrate the expression of the $U(1)$ connection by the use of the triple $(c,h,\omega)$ in Fig.~\ref{fig:flux_u1} in terms of simplicial complexes. 

\begin{figure}[!]
	\begin{center}
\includegraphics[width=0.5\linewidth, trim=0cm 0cm 0cm 0cm]{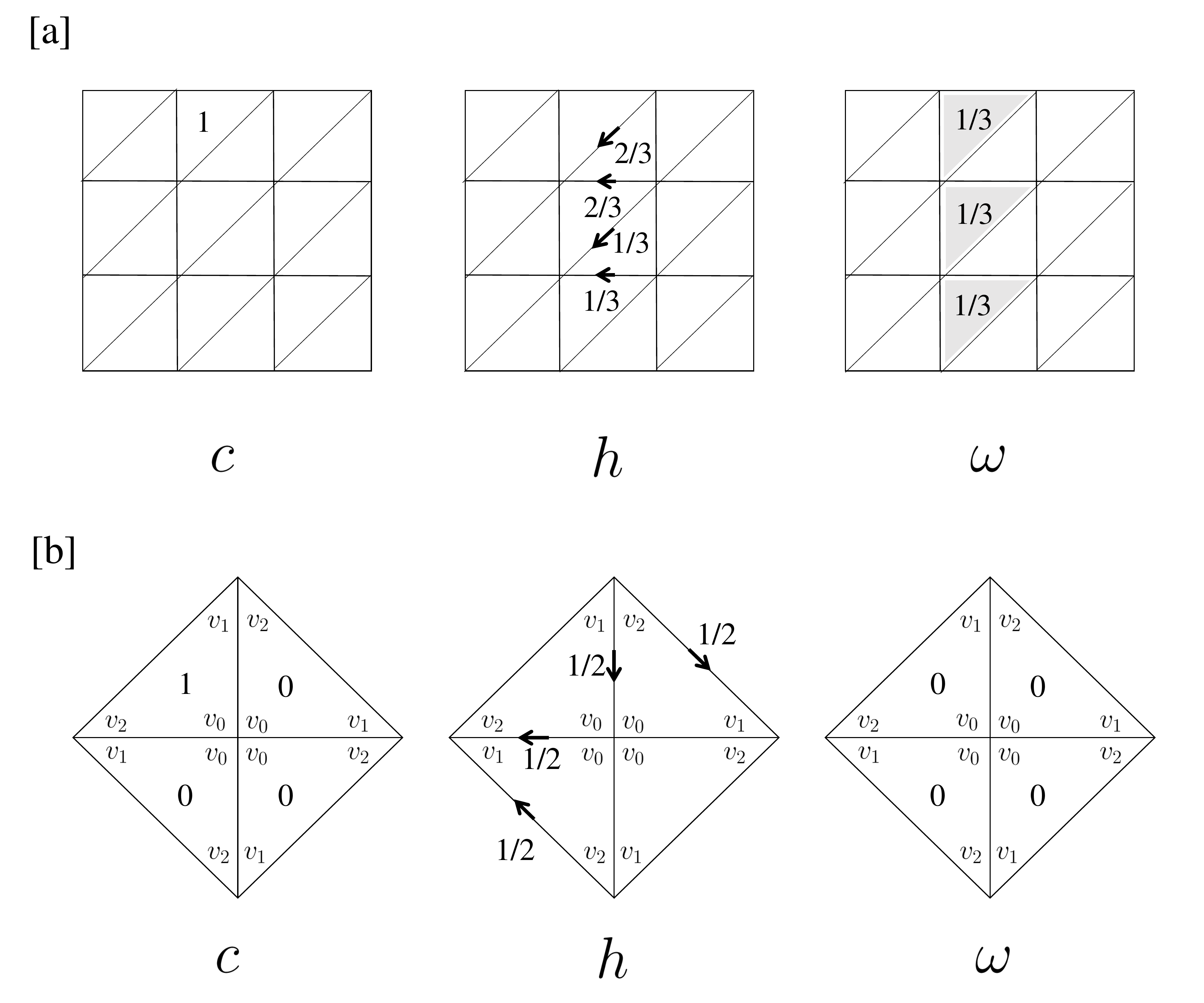}
	\end{center}
	\caption{
[a] A $U(1)$ connection $A=(c,h,\omega)$ with a unit magnetic flux $\int F/2 \pi = 1$ on the torus $T^2$. $c$ represent a generator of $H^2(T^2;\Z) = \Z$. 
[b] A $U(1)$ connection $A=(c,h,\omega)$ with the quantized $\pi$-flux on the
real projective plane $\mathbb{R}P^2$. 
$c$ represent the nontrivial topological sector classified by $H^2(\mathbb{R}P^2;\Z) = \Z_2$. 
}
	\label{fig:flux_u1}
\end{figure}

\subsection{Bosonic $U(1) \rtimes T$ symmetry}
\label{Bosonic U(1)xT symmetry}

Next, let us consider the charge-conjugation reflection symmetry $CR$ or TRS $T$ which flips the $U(1)$ charge $e^{i Q}$. 
Let $X$ be an unoriented manifold. 
We consider the Euclidean action and the covariant derivatives 
given by the same forms as before
(\ref{eq:action_u1}) and (\ref{eq:covariant_der_u1}), respectively.
On the other hand, for the patch transformations with orientation-reversing,
we consider  
\begin{align}
T:\,  
\wt \phi(y) = e^{i \alpha_{ij}(y)} \bar \phi(x(y)), \quad 
\wt{\bar\phi}(y) = e^{-i \alpha_{ij}(y)} \phi(x(y)). 
\end{align}
The patch transformations preserving the orientation are the same as before,
Eq.\ \eqref{eq:u1_conn_patch}. 
It is useful to introduce the ``real'' boson fields by 
\begin{align}
\chi = \begin{pmatrix}
\chi_{\ua} \\
\chi_{\da} \\
\end{pmatrix}, \quad 
\chi_{\ua} = \frac{\phi+\bar \phi}{2}, \quad 
\chi_{\da} = \frac{\phi-\bar \phi}{2i},
\end{align}
in terms of which 
the action is written as 
\begin{align}
  &
S 
= \int d^n x \sqrt{g} \Big[ (D^{\mu} \chi)^{\tr} D_{\mu} \chi + V(\chi^{\tr} \chi) \Big], 
  \nonumber \\
  &
D_{\mu} = \p_{\mu} + i A_{\mu} \sigma_y, 
\end{align}
and the patch transformation on $U_i \cap U_j$ is 
\begin{align}
T:\,  
\wt \chi_i(y) = 
e^{-i \alpha_{ij}(y) \sigma_y} (\sigma_z)^{w_{ij}} \chi_j(x(y)), 
\end{align}
where $(w_{ij}) \in Z^1(X;\Z/2)$ is a Cech cocycle representing $w_1(TX)$ and $\sigma_{\mu} (\mu=x,y,z)$ are the Pauli matrices on the $(\ua,\da)$ space. 
The cocycle condition $h_{ij} h_{jk} = h_{ik}$ and the inverse $h_{ji} = (h_{ij})^{-1}$ are written as 
\begin{align}
  &
(\delta_w \alpha)_{ijk}
  = (-1)^{w_{ij}} \alpha_{jk} - \alpha_{ik} + \alpha_{ij} = 0 \in \R/2 \pi \Z,
  \nonumber \\
  &
\alpha_{ji} = - (-1)^{w_{ij}} \alpha_{ij}, 
\label{eq:twisted_u1_cocycle_cond}
\end{align}
where $\delta_w$ is the differential twisted by $(w_{ij})$ (See Appendix \ref{app:Cohomology with local coefficient}). 
The patch transformation of the twisted $U(1)$ connection is determined by the condition 
\begin{align}
\wt D_{\mu} \wt \chi_i(y) = e^{-i \alpha_{ij} \sigma_y} (\sigma_z)^{w_{ij}} \frac{\p x^{\beta}}{\p y^{\mu}} D_{\beta} \chi_j(x(y)).
\end{align}
We have 
\begin{align}
T:\,  
\wt A_{\mu}(y) = (-1)^{w_{ij}} \Big[ \frac{\p x^{\beta}}{\p y^{\mu}} A_{\beta}(x(y)) + \p_{\mu} \alpha_{ij}(y)\Big]. 
\label{eq:twisted_u1_connection_patch}
\end{align}

Equations (\ref{eq:twisted_u1_cocycle_cond}) and (\ref{eq:twisted_u1_connection_patch}) indicate that the twisted $U(1)$ connection fits into the cochain complex twisted by $(w_{ij})$. 
Let $\theta_{ij}: U_i \cap U_j \to \R$ be a lift of $\alpha_{ij}$ with $\theta_{ji} = - (-1)^{w_{ij}} \theta_{ij}$. 
Then, the twisted first Chern class $\tilde c_1(L) \in H^2(X;\wt \Z)$ is represented by the twisted 2-cocycle 
\begin{align}
\tilde c_{ijk} := (-1)^{w_{ij}} \theta_{jk} - \theta_{ik} + \theta_{ij} \in \Z. 
\end{align}
Let $(C^*(X;\wt \R), \delta_w)$ be the singular cochain complex twisted by $(w_{ij})$ (See Appendix \ref{app:Cohomology with local coefficient}). 
A twisted $U(1)$ connection $A$ is a triple 
\begin{align}
A = (\tilde c,h,\omega) \in C^2(X;\wt \R) \times C^1(X;\wt \R) \times \wt \Omega^2(X)
\label{eq:triple}
\end{align}
subject to the following conditions 
\begin{align}
  &\tilde c \equiv 0  \ \mod \ \Z \quad ({\rm Dirac\ quantization\ condition}),
    \label{eq:boson_dirac_cond_twisted_u1} \\
&\delta_w \tilde c = 0, \label{eq:closed_c_twisted_u1} \\
&d \omega = 0, \label{eq:closed_omega_twisted} \\
&\omega = \tilde c + \delta_w h. \label{eq:2_form_twisted_u1}
\end{align}
Here, $\wt \Omega^2(X) = \Omega(X) \otimes P$ ($P$ is the orientation bundle) is
the group of ``$2$-densities''~\cite{bott2013differential}
in which orientation-reversing patch transformations are accompanied with the minus sign as $\omega_i = (-1)^{w_{ij}} \omega_j$. 
Two twisted $U(1)$ connections $A = (\tilde c,h,\omega)$ and $A'=(\tilde c',h',\omega''')$ are gauge equivalent if and only if there exists a pair $(b,k) \in C^1(X;\wt \Z) \times C^0(X;\wt \R)$ so that 
\begin{align}
\tilde c'=\tilde c-\delta_w b, \quad 
h''=h+b+\delta_w k, \quad 
\omega'=\omega.
\end{align}
The field strength is given by $F = 2 \pi \omega$. 
Thanks to the additional minus sign on the patch transformations, 
the 2-density can be integrated out on unoriented manifolds. 
For example, the monopole charges inside the real projective plane are quantized into integers $\int_{\mathbb{R}P^2} F/2 \pi \in \Z$. 
See Fig.~\ref{fig:flux_twisted_u1} [a] for an example of the twisted $U(1)$ connection on $\mathbb{R}P^2$ with unit magnetic flux. 
An important difference from untwisted $U(1)$ connections is that the holonomy
(the twisted boundary condition) along the nontrivial $\Z_2$ cycle 
(in the sense of untwisted homology) on the real projective plane
$\mathbb{R}P^2$ and the Klein bottle $K$ 
is not quantized due to the sign flip of the twisted $U(1)$ connection (\ref{eq:twisted_u1_connection_patch}). 
See Fig.~\ref{fig:flux_twisted_u1} [b] for the twisted boundary condition on the Klein bottle. 
\begin{figure*}[!]
	\begin{center}
	\includegraphics[width=\linewidth, trim=0cm 0cm 0cm 0cm]{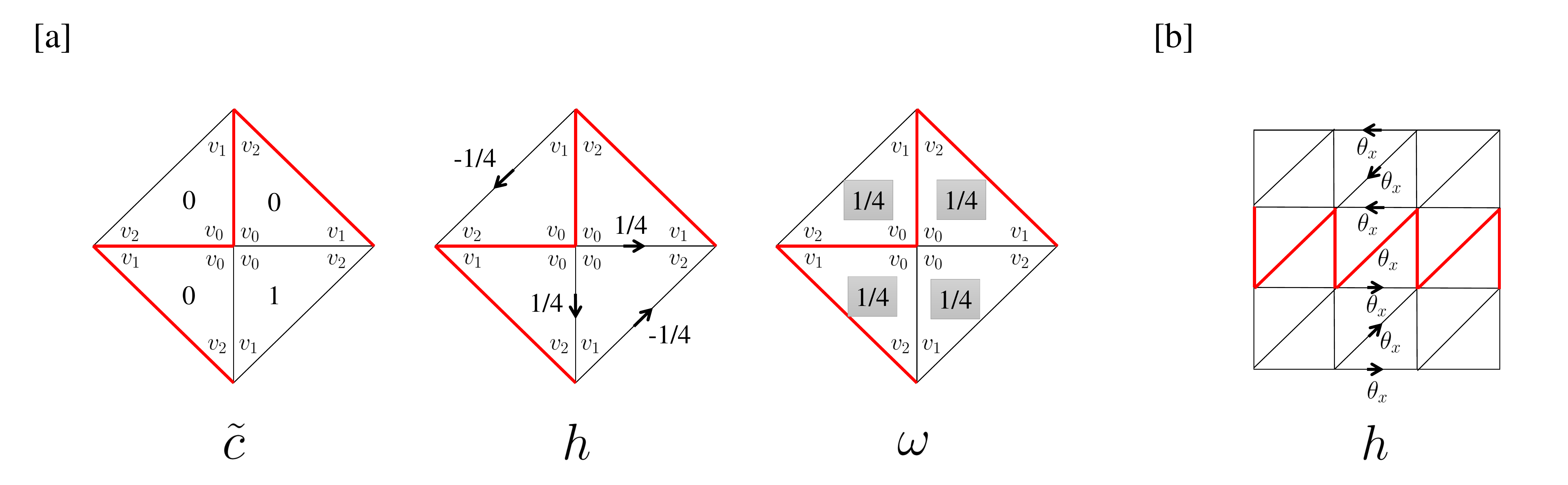}
	\end{center}
	\caption{
[a] A twisted $U(1)$ connection $A=(\tilde c,h,\omega)$ with a unit magnetic flux $\int_{RP^2} F/2 \pi = 1$ on the real projective plane $RP^2$. 
The red edges represents orientation reversing edges characterized by $w(v_iv_j) = 1$. 
See Appendix \ref{app:Cohomology with local coefficient} for the detail. 
[b] The boundary condition twisted by $\theta_x$ on the Klein bottle is not quantized for the twisted $U(1)$ connection. 
}
\label{fig:flux_twisted_u1}
\end{figure*}

\subsection{Fermionic $U(1) \times CT$ symmetry: pin$^c$ structure}
\label{Fermionic U(1)xCT symmetry: pinc structure}

Now we move on to fermionic systems with 
the antiunitary charge-conjugation symmetry $CT$ or the reflection symmetry $R$ that preserves the $U(1)$ charge $e^{i Q}$. 
Let $X$ be an unoriented manifold admitting a pin$^c$ structure. 
The Euclidean action of a free fermion is given by 
\begin{align}
S = \int_X d^n x \sqrt{g} \Big[ 
\bar \psi \gamma^{\mu} D_{\mu} \psi + m \bar \psi \psi \Big]. 
\label{eq:action_fermion_pinc}
\end{align}
Here, 
$\sqrt{g} = (\det g_{\mu \nu})^{1/2}$ with $g_{\mu \nu}$ the metric. 
$\gamma^{\mu}$ is defined by $\gamma^{\mu} = e^{\mu}_a \gamma^a$ with $\gamma^a$ the anticommuting Gamma matrices $\{ \gamma^a, \gamma^b\} = \delta_{ab}$ in the flat spacetime. 
$e^{\mu}_a$ is a vielbein (in $n$-spacetime dimensions) satisfying 
\begin{align}
e^{\mu}_a e^{\nu}_b g_{\mu \nu} = \delta_{ab}, \quad 
e^a_{\mu} e^a_{\nu} = g_{\mu \nu}, \quad 
e^{a}_{\mu} e^{\mu}_b = \delta_{ab}, \quad 
e^a_{\mu} e_a^{\nu} = \delta_{\mu}^{\nu}. 
\end{align}
Latin letters $a,b,\dots $ denote the local frame indices, while Greek ones $\mu,\nu, \dots $ denote the coordinate indices. 
The coordinate indices are raised/lowered by the metric $g_{\mu \nu}$ and its inverse matrix $g^{\mu \nu}$, 
and it is free to raise/lower the local frame indices because of the Euclidean signature. 
$\psi$ and $\bar \psi$ are independent Grassmann fields and their covariant derivative is defined by 
\begin{align}
D_{\mu} \psi
= \p_{\mu} \psi + ( \frac{i}{2} \omega^{ab}_{\mu} \Sigma_{ab} - i A_{\mu}) \psi, 
\end{align}
where $A_{\mu}$ is locally a $U(1)$ field, $\omega^{ab}_{\mu}$ is the spin
connection determined 
by the vielbein $e_a^{\mu}$ and the metric $g_{\mu}$ as~\footnote{
We have employed the Levi-Civita connection (being (i) compatible with the
metric and (ii) torsion free) 
since we are interested only in topological properties of manifolds with background fields. }	 
\begin{align}
  &
\omega^{ab}_{\mu} 
= e^a_{\nu} (\p_{\mu} e_b^{\nu} + \Gamma^{\nu}_{\mu \lambda} e_b^{\lambda}),  
  \nonumber \\
  &
\Gamma^{\nu}_{\mu \lambda} = \frac{1}{2} g^{\nu \sigma}(\p_{\mu} g_{\lambda \sigma} + \p_{\lambda} g_{\mu \sigma} - \p_{\sigma} g_{\mu \lambda}), 
\end{align}
and $\Sigma_{ab} = \frac{1}{4 i} [\gamma^a, \gamma^b]$ are the generators of  $SO(n)$ rotations. 

As mentioned in the beginning of Sec.~\ref{Dirac quantization conditions}, 
in addition to the $CT$ or $R$ symmetry, 
we assume the $SO(n)$ Euclidean spacetime rotation symmetry. 
Under a $\overline{U} \in SO(n)$ rotation in the local frame, 
the Gamma matrices are transformed as $U \gamma^a U^{\dag} = \gamma^b [\overline{U}]_{ba}$, 
where $U$ is a lift $\overline{U} \mapsto \pm U \in \spin(n)$. 
$CT$ or $R$ symmetry induces an orientation-reversing transformation on the Gamma matrices as follows. 
Let ${\cal U}_{CT}$ (${\cal U}_R$) be the matrix associated with the $CT$ ($R$) symmetry that appears in 
\begin{align}
  &
  CT \psi^{\dag}(\bm{x}) (CT)^{-1} =
  {\cal U}_{CT} \psi(\bm{x}),
  \nonumber \\
  &
  (R \psi^{\dag}(x_1,x_2, \dots) R^{-1} = \psi^{\dag}(-x_1,x_2, \dots) {\cal U}_R). 
\end{align}
Let $\hat H = \sum_{\bm{x}}\psi^{\dag}(\bm{x}) ( \sum_{a=1}^{n-1} \Gamma^a (-i \p_{a}) + m \Gamma^n) \psi(\bm{x})$ be the Hamiltonian on the flat space, 
where $\{ \Gamma^a, \Gamma^b \} = 2 \delta_{ab}$. 
The $CT$ ($R$) symmetry implies that ${\cal U}^{\ }_{CT} \Gamma^a {\cal U}_{CT}^{\dag} = - \Gamma^a (a=1, \dots n)$ 
(${\cal U}^{\ }_R \Gamma^1 {\cal U}_R^{\dag} = - \Gamma^1, {\cal U}_R \Gamma^a {\cal U}_R^{\dag} = \Gamma^a (a=2, \dots, n)$). 
The corresponding Euclidean action on the flat spacetime is given by 
\begin{align}
  S[\psi,\bar \psi] = \int d \tau d^{n-1} \bm{x}\,
  \bar \psi \left(\gamma^0 \p_{\tau} + \sum_{a=1}^{n-1} \gamma^a \p_a + m\right) \psi
\end{align}
  with 
$\gamma^0 = \Gamma^n$ and $\gamma^a = -i \Gamma^n \Gamma^a (a=1, \dots, n-1)$. 
Then, it holds that
\begin{align}
  &
  {\cal U}_{CT} \gamma^0 {\cal U}_{CT}^{\dag} = - \gamma^0,
  \quad
    {\cal U}_{CT} \gamma^a {\cal U}_{CT}^{\dag} = \gamma^a
    \quad (a=1,\dots, n-1),
  \nonumber \\
  &
  ({\cal U}_R \gamma^1 {\cal U}_R^{\dag} = - \gamma^1,
    \quad
    {\cal U}_{R} \gamma^a {\cal U}_{R}^{\dag} = \gamma^a
    \quad (a=0,2,\dots, n-1)). 
\end{align}
In sum, the Gamma matrices $\gamma^a$ are changed as a vector $U \gamma^a U^{\dag} = \gamma^b [\overline{U}]_{ba}$ 
for an $O(n)$ rotation $\overline U \in O(n)$, where $U$ is a lift $\overline U \mapsto \pm U$ again. 

The patch transformations on the intersections $U_i \cap U_j$ are summarized as follows
\begin{align}
& 
                \wt \psi(y) = e^{i \alpha_{ij}(y)} U_{ij}(y) \psi(x(y)),
                \nonumber \\
  &
\wt{\bar\psi}(y) = \bar \psi(x(y)) U_{ij}(y)^{\dag} e^{-i \alpha_{ij}(y)}, 
\label{eq:patch_tr_pinc}
\\
& 
\wt A_{\mu}(y) = \frac{\p x^{\beta}}{\p y^{\mu}} A_{\beta}(x(y)) + \p_{\mu} \alpha_{ij}(y), \\
& 
\wt \omega_{\mu}(y)
= \frac{\p x^{\beta}}{\p y^{\mu}} U_{ij}(y) \omega_{\beta}(x(y)) U_{ij}^{\dag}(y) + i \p_{\mu} U_{ij}(y) U_{ij}^{\dag}(y), \\
& 
\wt e^a_{\mu}(y)
= \frac{\p x^{\beta}}{\p y^{\mu}} [\overline{U}_{ij}(y)]_{ab} e^b_{\beta}(x(y)), \\
& 
U_{ij}(y) \gamma^a U_{ij}^{\dag}(y) = \gamma^b [\overline{U}_{ij}(y)]_{ba}. 
\label{eq:patch_tr_gamma}
\end{align}
Here,
on orientation preserving patches,
$U_{ij}(y) = \exp \big(\frac{i}{2} \theta_{ab}(y) \Sigma_{ab} \big) \in
\spin(n)$ is an $SO(n)$ rotation on the spinor space,
while 
on orientation reversing patches, $U_{ij}(y)$ is a form 
$U_{ij}(y) \sim \wt U_{ij}(y) {\cal U}_{CT}$ ($\sim \wt U_{ij}(y) {\cal U}_{R}$)
up to a $U(1)$ phase for $CT$ ($R$) symmetry,
where $\wt{U}_{ij}(y)  = \exp \big(\frac{i}{2} \theta_{ab}(y) \Sigma_{ab} \big) \in \spin(n)$. 
An overall $U(1)$ phase is fixed so that $U_{ij}(y)$ takes values in $\pin_+(n)$
group,
so that $e^{i \alpha_{ij}(y)} U_{ij}(y)$ takes its value in $\pin^c(n)$ group
defined by 
\begin{align}
  \pin^c(n) = \pin_+(n) \times U(1)/ (1,1) \sim (-1,-1),
\end{align}
where $-1 \in \spin(n) \subset \pin_+(n)$. 
For example, in the case of the $R$ symmetry with $R^2 = (-1)^F$, $U_{ij}(y) = \wt{U}_{ij}(y) i {\cal U}_{R}$.
The consistency condition on intersections $U_i \cap U_j \cap U_k$ leads to the cocycle condition 
\begin{align}
e^{i \alpha_{ij}} U_{ij} e^{i \alpha_{jk}} U_{jk}
= e^{i \alpha_{ik}} U_{ik} \quad 
{\rm on}\ \ U_i \cap U_j \cap U_k. 
\label{eq:cocycle_cond_pinc}
\end{align}
Let $\pi: U_{ij} \mapsto \overline{U}_{ij} \in O(n)$ be the surjection of the double cover $\pi: \pin_+(n) \to O(n)$. 
The data $\{ \overline{U}_{ij} \}$ is determined by the tangent bundle $TX$. 
The obstruction to give a pin$_+$ structure is measured by the second Stiefel-Whitney class $w_2(TX)$. 
The failure of the cocycle condition for $U_{ij}$ defines a two-cocycle
$(z_{ijk}) \in Z^2(X; \Z_2)$
representing $w_2(TX)$ by $e^{i \pi z_{ijk}} :=U_{ij} U_{jk} U_{ki} \in \{\pm 1\}$. 
Let $\theta_{ij}: U_i \cap U_j \to \R, \theta_{ji} = - \theta_{ij}$ be a lift of $\alpha_{ij}$ subject to $e^{2 \pi i \theta_{ij}} = e^{i \alpha_{ij}}$. 
The cocycle condition (\ref{eq:cocycle_cond_pinc}) implies that the Dirac
quantization condition of pin$^c$ structures
is twisted by the second Stiefel-Whitney class: 
\begin{align}
c_{ijk}:= \theta_{ij} + \theta_{jk} + \theta_{ki} \in \Z + \frac{z_{ijk}}{2} \in C^2(X;\R), 
\label{eq:dirac_cond_pinc}
\end{align}
where $z_{ijk} \in \{0,1\}$. 
We denote the condition (\ref{eq:dirac_cond_pinc}) by 
\begin{align}
c \equiv \frac{w_2}{2} \mod \Z 
\label{eq:dirac_cond_pinc2}
\end{align}
in short. 
A pin$^c$ structure can be changed by tensoring a complex line bundle $L$, which
means that the classification of the topological sectors
of pin$^c$ structures is given by $H^2(X;\Z)$. 
A pin$^c$ connection is a triple $A = (c,h,\omega) \in C^2(X;\R) \times C^1(X,R) \times \Omega^2(X)$ subject to the conditions 
(\ref{eq:dirac_cond_pinc2}), (\ref{eq:closed_c_u1}), (\ref{eq:connection_cond_2}), and (\ref{eq:2_form_u1}). 
The gauge equivalence is given by the same manner as the $U(1)$ connection.

\subsection{Fermionic $U(1) \rtimes T$ symmetry: pin$^{\tilde c}_{\pm}$ structure}
\label{Fermionic U(1)rxT symmetry: pintildecpm structure}

Let us consider fermionic systems with TRS $T$ or charge-conjugation reflection symmetry $CR$ which flips $U(1)$ charge $e^{i Q}$. 
We deal with both the pin$^{\tilde c}_+$ ($T^2 = (-1)^F$) and the pin$^{\tilde c}_-$ ($T^2 = 1$) structures on an equal footing. 
Here, $\pin^{\tilde c}_{\pm}(n)$ groups are defined as follows. 
Let $\phi : \pin_{\pm}(n) \to \{\pm 1\}$ is the homomorphism which specifies orientation preserving or reversing elements, i.e.\ $\phi(x) = 1$ for even $r$ and $\phi(x) = -1$ for odd $r$. 
(See \eqref{def pin pm groups}
for the definition of $\pin_{\pm}(n)$.)
By the use of $\phi$, we have a new group $\pin_{\pm}(n) \ltimes U(1)$, 
where the multiplication of group structure is defined by
$(x,e^{i \alpha}) \cdot (y, e^{i \beta}) = (xy,e^{i \alpha + i \phi(x) \beta})$. 
The $\pin^{\tilde c}_{\pm}(n)$ group is defined as the quotient 
\begin{align}
\pin^{\tilde c}_{\pm}(n)
:= \pin_{\pm}(n) \ltimes U(1)/ (1,1) \sim (-1,-1), 
\end{align}
where $-1 \in \spin(n) \subset \pin_{\pm}(n)$.

$T$ or $CR$ symmetry induces an orientation reversing transformation on the Gamma matrices as follows. 
Let ${\cal U}_T$ be the matrix associated with TRS $T$ that is defined by $T
\psi^{\dag}(\bm{x}) T^{-1} = \psi^{\dag}(\bm{x}) {\cal U}_T$, 
and satisfy ${\cal U}_T {\cal U}_T^* = 1 (-1)$ for $T^2 = 1$ ($T^2 = (-1)^F$). 
The Time-reversal symmetry implies that ${\cal U}_T (\gamma^{a})^{tr} {\cal
  U}_T^{\dag} = \gamma^{a}$,
irrespective of the sign of ${\cal U}_T {\cal U}_T^* = \pm 1$.~\footnote{
This can be understood from the TRS on the Hamiltonian in the flat space manifold 
$\hat H = \sum_{\bm{x}} \psi^{\dag}(\bm{x}) ( \sum_{\mu=1}^{n-1} \Gamma^{\mu}
(-i \p_{\mu}) + m \Gamma^n) \psi(\bm{x})$
with $\{ \Gamma^{\mu}, \Gamma^{\nu} \} = 2 \delta_{\mu \nu}$ 
for $\mu, \nu = 1, \dots, n$. 
Time-reversal symmetry implies that ${\cal U}_T \Gamma^{\mu} {\cal U}_T^{\dag} = - \Gamma^{\mu} (\mu=1, \dots, n-1)$
and ${\cal U}_T \Gamma^n {\cal U}_T^{\dag} = \Gamma^n$. 
The action on the flat Euclidean space is given by  
$S = \int d^n x \bar \psi [\gamma^0 \p_{\tau} + \sum_{\mu=1}^{n-1} \gamma^{\mu} \p_{\mu} + m] \psi$ 
with $\gamma^0 = \Gamma^n$ and $\gamma^{\mu} = -i \Gamma^n \Gamma^{\mu} (\mu = 1, \dots, n-1)$.} 
A similar relation holds for $CR$ reflection symmetry. 
In the rest of this section, we only describes the case of $T$ symmetry, for simplicity. 

Let $X$ be an unoriented manifold admitting a pin$^{\tilde c}_{\pm}$ structure. 
The Euclidean action of free fermions is (\ref{eq:action_fermion_pinc}) 
and the patch transformations preserving the orientation are the same form as (\ref{eq:patch_tr_pinc} - \ref{eq:patch_tr_gamma}) with $U_{ij}(y)=\wt U_{ij}(y) \in \spin(n)$, 
while the patch transformations reversing the orientation are defined to be
\begin{align}
T: \,  
  &
  \wt \psi(y) = e^{i \alpha_{ij}(y)} \wt U_{ij}(y) (i {\cal U}^{tr}_T) \bar \psi^{tr}(x(y)),
  \nonumber \\
&  
\wt{\bar\psi}(y) = \psi^{tr}(x(y)) (i {\cal U}_T^*) \wt U_{ij}(y)^{\dag} e^{-i \alpha_{ij}(y)}, 
\label{eq:patch_tr_pin_tilde_c}
\end{align}
where $\wt U_{ij}(y) = \exp \big(\frac{i}{2} \theta_{ab}(y) \Sigma_{ab} \big) \in \spin(n)$ representing an orientation preserving rotation 
and the factor $i$ is due to the anticommutation relation of Grassmann variables (see eq.\ (\ref{eq:trs_euclid})). 
To derive the pin$^{\tilde c}_{\pm}$ structure, 
we introduce the ``real'' fermion fields by~\cite{Metlitski2015}
\begin{align}
\chi = \begin{pmatrix}
\chi_{\ua} \\
\chi_{\da} \\
\end{pmatrix}, \quad 
\chi_{\ua} = \frac{\psi + {\cal U}_T^{tr} \bar \psi^{tr}}{2}, \quad 
\chi_{\da} = \frac{\psi - {\cal U}_T^{tr} \bar \psi^{tr}}{2i}.
\end{align}
After some algebras, the action (\ref{eq:action_fermion_pinc}) can be written as
\begin{equation}\begin{split}
S &= 
\int_X d^n x \sqrt{g}\,   
\chi^{tr} {\cal U}_T^{\dag} (-\sigma_y)^s \Big[ \gamma^{\mu} (\p_{\mu} + \frac{i}{2} \omega^{ab}_{\mu} \Sigma_{ab} + i A_{\mu} \sigma_y) - m \sigma_y \Big] \chi, 
\end{split}
\label{eq:action_fermion_pin_tilde_c}
\end{equation}
with $s=0$ for $T^2=1$ and $s=1$ for $T^2=(-1)^F$,
and the orientation reversing patch transformation (\ref{eq:patch_tr_pin_tilde_c}) becomes
\begin{align}
&T:\, 
   \wt \chi_i(y)
  = 
\left\{\begin{array}{ll}
\wt U_{ij}(y) e^{-i \alpha_{ij}(y) \sigma_y} (i \sigma_z)^{w_{ij}} \chi_j(x(y)) & {\rm for\ \ } T^2=1, \\
\wt U_{ij}(y) e^{-i \alpha_{ij}(y) \sigma_y} (\sigma_x)^{w_{ij}} \chi_j(x(y)) & {\rm for\ \ } T^2=(-1)^F, \\
\end{array}\right.
\end{align}
on the intersection $U_i \cap U_j$. 
Now it is evident that the transition function
is an element of $\pin^{\tilde c}_{+}(n)$ ($\pin^{\tilde c}_-(n)$)
group for TRS with $T^2= (-1)^F$ ($T^2 = 1$)
(since $(i \sigma_z)^2 = -1, (\sigma_x)^2=1$).

Next, let us derive the Dirac quantization conditions for the pin$^{\tilde c}_{\pm}$ structures. 
Let $\theta_{ij}: U_i \cap U_j \to \R, \theta_{ji} = - \theta_{ij}$ be a lift of $\alpha_{ij}$ subject to $e^{2 \pi i \theta_{ij}} = e^{i \alpha_{ij}}$ and $\theta_{ji} = - (-1)^{w_{ij}} \theta_{ji}$. 
Since an obstruction to have a pin$_{+}$ (pin$_-$) structure is measured by $w_2$ ($w_2+w_1^2$), 
the cocycle condition leads to the following Dirac quantization condition of pin$^c_{\pm}$ structures: 
\begin{align}
  \tilde c_{ijk} &:= (-1)^{w_{ij}} \theta_{jk} - \theta_{ik} + \theta_{ij}
\in \Z + 
\left\{\begin{array}{ll}
\frac{1}{2} z_{ijk} & ({\rm pin}^{\tilde c}_+) \\
\frac{1}{2} z_{ijk} + \frac{1}{2} w_{ij} w_{jk} & ({\rm pin}^{\tilde c}_-) \\
\end{array}\right.  
\label{eq:dirac_cond_pin_tilde_c}
\end{align}
where $z_{ijk} \in \{0,1\}$ is defined by $(-1)^{z_{ijk}} = g_{ij} g_{jk} g_{ki}$ and $\tilde c \in C^2(X;\wt \R)$. 
In short, the condition is written as 
\begin{align}
&\tilde c \equiv \frac{1}{2}w_2 \mod \Z  \qquad {\rm for} \  {\rm pin}^{\tilde c}_+, \label{eq:dirac_cond_pin_tilde_c_2} \\
&	\tilde c \equiv \frac{1}{2}w_2 + \frac{1}{2}w_1^2 \mod \Z  \qquad {\rm for}\  \ {\rm pin}^{\tilde c}_- \label{eq:dirac_cond_pin_tilde_c-_2}. 
\end{align}
A pin$^{\tilde c}_{\pm}$ structure is changed by tensoring a twisted complex line bundle $L$, 
hence the classification of the topological sectors of pin$^{\tilde c}_{\pm}$ structure is given by $H^2(X;\wt \Z)$. 
A pin$^{\tilde c}$ connection is a triple $A = (\tilde c,h,\omega) \in C^2(X;\wt \R) \times C^1(X,\wt \R) \times \wt \Omega^2(X)$ subject to the conditions 
(\ref{eq:dirac_cond_pin_tilde_c_2}), (\ref{eq:closed_c_twisted_u1}), (\ref{eq:closed_omega_twisted}), and (\ref{eq:2_form_twisted_u1}). 
The gauge equivalence is given by the same manner as the twisted $U(1)$ connection.

\section{
\label{sec:pin_rp2}
More on pin$^c$ and pin$^{\tilde c}_{\pm}$ structures on $\mathbb{R}P^2$}

To further support the Dirac quantization conditions listed in Table~\ref{tab:dirac_cond}, in this subsection, we present the explicit forms of the Dirac operators on $\mathbb{R}P^2$ for pin$^c$ and pin$^{\tilde c}_{\pm}$ structures. 
To this end, we start with the Dirac operator on $S^2$ with a monopole flux and take the projection onto the symmetric components under the antipodal projection to get the Dirac operator on $\mathbb{R}P^2$. 

\subsection{Dirac operator on $S^2$ with the Schwinger gauge}
The Dirac operator on $S^2$ with a monopole flux in the Schwinger gauge is given by
\begin{widetext}
\begin{equation}\begin{split}
\slashed{D}_g
&= \sigma_x \Big( \p_{\theta} + \frac{1}{2} \cot \theta \Big)+ \frac{\sigma_y}{\sin \theta} \Big( \p_{\phi} + i g \cos \theta \Big) \\
&= \begin{pmatrix}
0 & \p_{\theta} - \frac{i}{\sin \theta} (\p_{\phi} + i (g + \frac{1}{2}) \cos \theta) \\ 
\p_{\theta} + \frac{i}{\sin \theta} (\p_{\phi} + i (g - \frac{1}{2})\cos \theta ) & 0 \\
\end{pmatrix}
\end{split}
\label{eq:dirac_op_s2}
\end{equation}
Here, $(\theta,\phi)$ is the spherical coordinate and $g = m_g/2$, $m_g \in \Z$, is the monopole charge. 
In the presence of the magnetic monopole, the vector potential $A_{\phi} = - i g \cos \theta$ is singular at the north and south poles, which is a characteristic of the Schwinger gauge. 
From the index theorem, there are $m_g$ zero modes of $\slashed{D}_{g}$. 
For simplicity, we assume $m_g \geq 0$. 
The Euclidean action with theta term is 
\begin{align}
S_{S^2} = 
\int d \theta d \phi \sin \theta \bar \psi(\theta,\phi) (\slashed{D}_g + M e^{i \Theta \sigma_z}) \psi(\theta,\phi), \quad (M \geq 0)
\label{eq:action_dirac_s2}
\end{align}
and the partition function is given by 
\begin{align}
Z[S^2,m_g]=\int D \bar \psi D \psi e^{-S} = {\rm Det} (\slashed{D}_g + M e^{i \Theta \sigma_z}). 
\end{align}
The eigenfunctions of the Dirac operator $\slashed{D}_g$ are as follows.~\cite{Hasebe2015}
There are $m_g$ zero modes because of the index theorem: 
\begin{align}
&-i \slashed{D}_g \Phi^g_{\lambda_0=0,m}(\theta,\phi)=0, \quad m=-g+\frac{1}{2}, -g+\frac{3}{2}, \dots, g-\frac{1}{2}, \\
&\Phi^g_{\lambda_0=0,m}(\theta,\phi) = \begin{pmatrix}
Y^{g-\frac{1}{2}}_{g-\frac{1}{2},m} (\theta,\phi) \\
0 \\
\end{pmatrix}. 
\end{align}
Nonzero modes of $\slashed{D}_g$ are 
\begin{align}
&-i \slashed{D}_g \Phi^g_{\pm \lambda_n,m}(\theta,\phi)= \pm \sqrt{n(n+2g)} \Phi^g_{\pm \lambda_n,m}(\theta,\phi), \nonumber \\
&n=1,2,\dots, \quad m=-j, -j+1, \dots, j, \quad j=g-\frac{1}{2}+n, \\
&\Phi^g_{\pm \lambda_n,m}(\theta,\phi) = \frac{1}{\sqrt{2}} \begin{pmatrix}
Y^{g-\frac{1}{2}}_{j=g-\frac{1}{2}+n,m} (\theta,\phi) \\
\mp i Y^{g+\frac{1}{2}}_{j=g+\frac{1}{2}+(n-1),m} (\theta,\phi) \\
\end{pmatrix}.
\end{align}
Here, $Y^g_{l,m}(\theta,\phi)$ is the monopole harmonics
\begin{equation}\begin{split}
Y^g_{l,m}(\theta,\phi)
&=(-1)^{l+m}\sqrt{\frac{(2l+1) (l-m)! (l+m)!}{4 \pi (l-g)! (l+g)!}} e^{i m \phi} \\
&\cdot \sum_k (-1)^k \begin{pmatrix}
l-g \\
k \\
\end{pmatrix}
\begin{pmatrix}
l+g\\
g-m+k \\
\end{pmatrix}
(\sin \frac{\theta}{2})^{2l-2k-g+m} 
(\cos \frac{\theta}{2})^{2k+g-m}.
\end{split}\end{equation}
Expanding the complex fermion fields $\psi(\theta,\phi)$ and $\bar \psi(\theta,\phi)$ by the eigenfunctions as 
\begin{align}
&\psi(\theta,\phi)
=\sum_m \Phi^g_{\lambda_0=0,m}(\theta,\phi) \chi_{0,m} 
+\sum_{\pm \lambda_n,m} \Phi^g_{\pm \lambda_n,m}(\theta,\phi) \chi_{\pm \lambda_n,m}, \\
&\bar \psi(\theta,\phi)
=\sum_m \Phi^g_{\lambda_0=0,m}(\theta,\phi)^* \bar \chi_{0,m} 
+\sum_{\pm \lambda_n,m} \Phi^g_{\pm \lambda_n,m}(\theta,\phi)^* \bar \chi_{\pm \lambda_n,m}. 
\end{align}
The action on $S^2$ is re-written as 
\begin{align}
S 
&= \sum_m M e^{i \Theta} \bar \chi_{0,m} \chi_{0,m} 
+ \sum_{\lambda_n>0,m} (\bar \chi_{\lambda_n,m},\bar \chi_{-\lambda_n,m}) 
\begin{pmatrix}
i \sqrt{n(n+2g)} + M \cos \Theta & i M \sin \Theta \\
i M \sin \Theta & -i\sqrt{n(n+2g)} + M \cos \Theta \\
\end{pmatrix}
\begin{pmatrix}
\chi_{\lambda_n,m} \\
\chi_{-\lambda_n,m} \\
\end{pmatrix}.
\end{align}
\end{widetext}
The partition function on $S^2$ is 
\begin{align}
Z[S^2,m_g]
&= (M e^{i \Theta})^{m_g} \prod_{\lambda_n>0,m} \big[ n(n+2g) + M^2 \big] 
\sim e^{i m_g \Theta}.
\end{align}

Some properties of the eigenfunctions which will be used later are listed: 
\begin{itemize}
\item 
The angular momentum depends on the even/odd parity of the magnetic monopole charge: 
\begin{align}
m \in \left\{\begin{array}{ll}
\Z + \frac{1}{2} & (m_g \in 2\Z), \\
\Z & (m_g \in 2\Z+1), \\
\end{array}\right.
\end{align}
which means the the periodicity along the $\phi$-direction depends on the even/odd parity of the magnetic monopole charge as 
\begin{equation}\begin{split}
&\Phi^g_{\lambda_0=0,m}(\theta,\phi+2 \pi) = 
\left\{\begin{array}{ll}
-\Phi^g_{\lambda_0=0,m}(\theta,\phi) & (m_g \in 2 \Z), \\
\Phi^g_{\lambda_0=0,m}(\theta,\phi) & (m_g \in 2 \Z+1), \\
\end{array}\right. \\
&\Phi^g_{\pm \lambda_n,m}(\theta,\phi+2 \pi) = 
\left\{\begin{array}{ll}
-\Phi^g_{\lambda_n,m}(\theta,\phi) & (m_g \in 2 \Z), \\
\Phi^g_{\lambda_n,m}(\theta,\phi) & (m_g \in 2 \Z+1). \\
\end{array}\right.
\end{split}
\label{eq:s2_dirac_periodicity}
\end{equation}
\item The eigenfunctions have the ``$T$'' symmetry (Euclidean TRS) associated with the antipodal map: 
\begin{equation}\begin{split}
&\Phi^g_{\lambda_0=0,m}(\pi-\theta,\phi+\pi)
=e^{i m \pi} \Phi^g_{\lambda_0=0,-m}(\theta,\phi)^*, \\
&\Phi^g_{\pm \lambda_n,m}(\pi-\theta,\phi+\pi)
=(-1)^n e^{i m \pi} \Phi^g_{\pm \lambda_n,-m}(\theta,\phi)^*. 
\end{split}\end{equation}
\item In the case of zero magnetic monopole $m_g=0$, the eigenfunctions have the $R$ symmetry associated the the antipodal map: 
\begin{align}
\Psi^{g=0}_{\pm \lambda_n,m}(\pi-\theta,\phi+\pi) 
&= \mp (-1)^n i \sigma_y \Psi^{g=0}_{\pm \lambda_n,m}(\theta,\phi). 
\label{eq:CT}
\end{align}
\end{itemize}

\subsection{Pin$^c$ structure on $\mathbb{R}P^2$ and the eta invariant}
\label{sec:pin_str_pr2_eta}
As discussed in Sec.~\ref{ex:rp2_pinc}, 
the 2d cobordism of pin$^c$ structure is $\Omega^{\pin^c}_2 = \Z_4$ and its generating manifold is $\mathbb{R}P^2$, and 
the Dirac quantization condition $c \equiv \frac{1}{2} w_2$ implies that the pin$^c$-connection on $\mathbb{R}P^2$ is quantized to $\pm i$ fluxes and a magnetic monopole is forbidden. 
Let us consider the antipodal map
\begin{align}
  R:\,
  \psi(\theta,\phi) &\mapsto e^{i \alpha} \sigma_y \psi(\pi-\theta,\phi+\pi),
  \nonumber \\
\bar \psi(\theta,\phi) &\mapsto \bar \psi(\pi-\theta,\phi+\pi) \sigma_y e^{-i \alpha}, 
\end{align}
where $e^{i \alpha}$ is a constant.~\footnote{
Here, we assumed that $\pin^c$ group is defined by $\pin^c=\pin_+ \times_{\pm 1} U(1)$, i.e., the reflection operator satisfies $R^2=1$.}
The action (\ref{eq:dirac_op_s2}) is invariant under $R$ provided $g=0$ and $\Theta = 0, \pi$. 
We can put the Lagrangian on $\mathbb{R}P^2$ when the $R$ projection 
\begin{align}
\psi(\theta,\phi) &= e^{i \alpha} \sigma_y \psi(\pi-\theta,\phi+\pi), 
    \nonumber \\
\bar \psi(\theta,\phi) &= \bar \psi(\pi-\theta,\phi+\pi) \sigma_y e^{-i \alpha} 
\end{align}
does works. 
From (\ref{eq:CT}), this is the case only when $e^{i \alpha} = \pm i$, 
which corresponds to two inequivalent pin$^c$ structures on $\mathbb{R}P^2$ with $\pm i$ holonomy as expected. 

Let us compute the partition function on $\mathbb{R}P^2$, the eta invariant of pin$^c$ structure. 
The mode expansion with $R$ projection is given by 
\begin{widetext}
\begin{align}
e^{i \alpha}=i: \quad \left\{\begin{array}{ll}
\psi(\theta,\phi)
=\sum_{n>0,n \in {\rm even}, m} \Phi^{g=0}_{\lambda_n,m}(\theta,\phi) \chi_{\lambda_n,m} 
+\sum_{n>0,n \in {\rm odd}, m} \Phi^{g=0}_{-\lambda_n,m}(\theta,\phi) \chi_{-\lambda_n,m}  , \\
\bar \psi(\theta,\phi)
=\sum_{n>0,n \in {\rm even}, m} \Phi^{g=0}_{\lambda_n,m}(\theta,\phi)^{\dag} \bar \chi_{\lambda_n,m} 
+\sum_{n>0,n \in {\rm odd}, m} \Phi^{g=0}_{-\lambda_n,m}(\theta,\phi)^{\dag} \bar \chi_{-\lambda_n,m}, 
\end{array}\right.
\end{align}
\end{widetext}
and a similar form for $e^{i \alpha}=-i$. 
Then, the partition function on $\mathbb{R}P^2$ becomes 
\begin{align}
  &
Z(\mathbb{R}P^2,e^{i \alpha}=i, \pm M)
\prod_{n>0,n \in {\rm even}, m} (in\pm M) 
\prod_{n>0,n \in {\rm odd}, m} (-in\pm M) 
\end{align}
The eta invariant of pin$^c$ manifolds is given by the ratio of the partition functions between that with $M$ and $-M$. 
From a crude calculation, the eta invariant $\eta(X,A)$ can be estimated as 
\begin{equation}\begin{split}
e^{2 \pi i \eta(\mathbb{R}P^2,e^{i \alpha}=i)}
&:= \frac{Z(\mathbb{R}P^2,e^{i \alpha}=i,-M)}{Z(\mathbb{R}P^2,e^{i \alpha}=i,M)}
\nonumber \\
&\xrightarrow{M\to \infty}
\exp \Big[ 
\pi i \sum_{n>0,n \in {\rm even}} 2n 
- \pi i \sum_{n>0,n \in {\rm odd}} 2n \Big] \\
&\sim 
\exp\Big[ [i \pi 2^{1-s} \zeta(s)]_{s \to -1} - 
[i \pi (2-2^{1-s}) \zeta(s)]_{s \to -1} \Big] \\
&= e^{-\pi i/2}. 
\end{split}\end{equation}
This is the desired $\Z_4$ $U(1)$ phase of the pin$^c$ SPT phase. 
For the pin$^c$ structure with $e^{i \alpha} = -i$, the eta invariant is given by $e^{\pi i/2}$. 
For more details of the eta invariant, see \cite{Metlitski2015} for example.

\subsection{$\pin^{\tilde c}_{-}$ structure on $\mathbb{R}P^2$ and $\theta$ term}
\label{sec:pinc_-_str_pr2_theta}
Now let us consider the pin$^{\tilde c}_-$ structure on $\mathbb{R}P^2$. 
In $(1+1)d$, the cobordism is given by $\Omega^{\pin^{\tilde c}_-}_2 = \Z \times \Z_2$.~\cite{Freed2016}
$\Z_2$ is generated by $\mathbb{R}P^2$ with no magnetic monopole, and $\Z$ can be detected by the theta term. 
The general partition function which is cobordism invariant is written as 
\begin{align}
  &
Z_{\pin^{\tilde c}_-}(X;A) 
  = \exp \left[ n \pi i \int_X w_2 + i \theta \int_X \frac{F}{2 \pi} \right],
  \nonumber \\
  & 
n \in \{0,1\}, \quad 
\theta \in \R/2 \pi \Z. 
\end{align}
The $2 \pi$ periodicity of $\theta$ follows from that the Dirac quantization condition (\ref{eq:dirac_cond_pin_tilde_c-_2}) is trivial for every $(1+1)d$ manifold, says, $w_2+w_1^2=0$. 
The $\Z_2$ nontrivial phase with $n=1$ is an interacting enabled phase which is not realized in free fermions, which can be detected by the partial time-reversal transformation discussed in Sec.~\ref{sec:(1+1)AI}. 

The Hamiltonian on $S^1$ is the following two orbital model 
\begin{align}
&\hat H = \int d x \hat \psi^{\dag}(x) [-i \sigma_y \p_x +M \cos \Theta \sigma_z + M \sigma_x \sin \Theta] \hat \psi(x).
\label{eq:hamiltonian_ai_rp2} 
\end{align}
This Hamiltonian has the TRS with $T^2 = 1$ defined by 
\begin{align}
&\hat T \hat \psi^{\dag}(x) \hat T^{-1} = \hat \psi^{\dag}(x), \quad 
\hat T i \hat T^{-1} = -i. 
\end{align}
Correspondingly, the Euclidean action (\ref{eq:action_dirac_s2}) is invariant under the antipodal map 
\begin{align}
  T:\,
  \psi(\theta,\phi) &\mapsto i e^{i \alpha} \bar \psi(\pi-\theta,\phi+\pi),
                      \nonumber \\
\bar \psi(\theta,\phi) &\mapsto i e^{-i \alpha} \psi(\pi-\theta,\phi+\pi) 
\end{align}
with a phase $e^{i \alpha}$. 
The possibility of the antipodal projection by $T$, 
\begin{align}
  \psi(\theta,\phi) &= i e^{i \alpha} \bar \psi(\pi-\theta,\phi+\pi),
  \nonumber \\
\bar \psi(\theta,\phi) &= i e^{-i \alpha} \psi(\pi-\theta,\phi+\pi), 
\end{align}
depends on the monopole charge. 
Taking the $T$ transformation twice, we have 
\begin{align}
\psi(\theta,\phi)
=i e^{i \alpha} \bar \psi(\pi-\theta,\phi+\pi)
=-\psi(\theta,\phi+2\pi), 
\end{align}
which implies that only the eigenfunctions satisfying the periodicity $\psi(\theta,\phi+2\pi) = - \psi(\theta,\phi)$ admits the $T$ projection. 
The phase $e^{i \alpha}$ is free to be chosen, which is consistent with that the flux along the nontrivial $\Z_2$ cycle is not quantized for pin$^{\tilde c}_-$ connections. 
Recall that in the Schwinger gauge the eigenfunctions of the Dirac operators satisfies this anti-periodicity only when the monopole charge is even $m_g \in 2 \Z$, then monopole charge on $\mathbb{R}P^2$ is quantized into integers 
\begin{align}
\int_{\mathbb{R}P^2} \frac{F|_{\pin^{\tilde c}_-}}{2 \pi} = \frac{m_g}{2} \in \Z. 
\end{align}
This is consistent with the Dirac quantization condition of pin$^{\tilde c}_-$ structure $\tilde c \equiv 0$ on $\mathbb{R}P^2$. 
The partition function on $\mathbb{R}P^2$ is just the square root of that on $S^2$, 
\begin{align}
Z(\mathbb{R}P^2, \frac{m_g}{2} \in \Z) 
&= (M e^{i \Theta})^{m_g/2} \sqrt{ \prod_{\lambda_n>0,m} \big[ n(n+2g) + M^2 \big] }
\sim e^{i \frac{m_g}{2} \Theta}.
\end{align}
Note that the periodicity of $\Theta$ is still $2 \pi$.

\subsection{$\pin^{\tilde c}_{+}$ structure on $\mathbb{R}P^2$ and $\theta$ term}
\label{sec:pinc_+_str_pr2_theta}
Now let us consider the pin$^{\tilde c}_+$ structure on $\mathbb{R}P^2$. 
The Hamiltonian on $S^1$ is the Kramers pair of the Hamiltonian (\ref{eq:hamiltonian_ai_rp2}) 
\begin{align}
  &
    \hat H=\hat H_{\ua}+\hat H_{\da},
  \nonumber \\
 & 
   \hat H_{\ua/\da} = \int d x \hat \psi^{\dag}_{\ua/\da}(x)
   [-i \sigma_y \p_x
   +M \cos \Theta \sigma_z + M \sigma_x \sin \Theta] \hat \psi_{\ua/\da}(x). 
\end{align}
One can associate this Hamiltonian with the TRS with the Kramers degeneracy $\hat T^2=(-1)^F$ defined by 
\begin{align}
&\hat T \hat \psi^{\dag}_{\ua}(x) \hat T^{-1} = \hat \psi^{\dag}_{\da}(x), \quad 
\hat T \hat \psi^{\dag}_{\da}(x) \hat T^{-1} = -\hat \psi^{\dag}_{\ua}(x), \quad 
\hat T i \hat T^{-1} = -i. 
\end{align}
Correspondingly, the Euclidean action on $S^2$ 
\begin{align}
  &
S = S_{\ua} + S_{\da},  
  \nonumber \\
  &
S_{\ua/\da} = \int d\theta d \phi \sin \theta \bar \psi_{\ua/\da}(\theta,\phi) (\slashed{D}_g + M e^{i \sigma_z \Theta} ) \psi_{\ua/\da}(\theta,\phi),
\end{align}
has the $T$ antipodal symmetry of a pin$^{\tilde c}_+$ type defined by~\footnote{We have fixed a $U(1)$ phase associated with $T$ to be $e^{i \alpha} = 1$.}
\begin{align}
T: \quad 
\begin{pmatrix}
\psi_{\ua}(\theta,\phi) \\
\psi_{\da}(\theta,\phi) \\
\bar \psi_{\ua}(\theta,\phi) \\
\bar \psi_{\da}(\theta,\phi) \\
\end{pmatrix}
\mapsto 
\begin{pmatrix}
i \bar \psi_{\da}(\pi-\theta,\phi+\pi) \\
-i \bar \psi_{\ua}(\pi-\theta,\phi+\pi) \\
i \psi_{\da}(\pi-\theta,\phi+\pi) \\
-i \psi_{\ua}(\pi-\theta,\phi+\pi) \\
\end{pmatrix}
\end{align}
Taking the $T$ transformation to get 
\begin{align}
&\psi_{\ua}(\theta,\phi)
\mapsto 
i \bar \psi_{\da}(\pi-\theta,\phi+\pi)
\mapsto 
\psi_{\ua}(\theta,\phi+2\pi), \\
&\psi_{\da}(\theta,\phi)
\mapsto 
-i \bar \psi_{\ua}(\pi-\theta,\phi+\pi)
\mapsto 
\psi_{\da}(\theta,\phi+2\pi). 
\end{align}
This means that only the eigenfunctions satisfies the periodicity $\psi_{\ua/\da}(\theta,\phi+2\pi) = \psi_{\ua/\da}(\theta,\phi)$ admits the $T$ projection. 
Because of the periodicity (\ref{eq:s2_dirac_periodicity}), the $T$ antipodal projection is well-defined only when the magnetic monopole has odd charge $m_g \in 2\Z+1$. 
Then, the magnetic monopole inside $\mathbb{R}P^2$ is quantized as (\ref{eq:magnetic_monopole_rp2_pinc+}), as expected. 
The partition function on $\mathbb{R}P^2$ is just the square root of that on $S^2$, that is, the partition function for the single Hamiltonian $H_{\ua}$ on $S^2$, 
\begin{align}
Z_{\pin^{\tilde c}_+}\big( \mathbb{R}P^2, \frac{m_g}{2} \in \Z+\frac{1}{2}\big) 
= (M e^{i \Theta})^{m_g} \prod_{\lambda_n>0,m} \big[ n(n+2g) + M^2 \big]
\sim e^{i m_g \Theta}.
\end{align}
Note that the periodicity of $\Theta$ still $2 \pi$. 
The $\theta$ parameter defined in (\ref{eq:theta_term_rp2_pinc+}) is given by $\theta = 2 \Theta$, hence we reproduced the $4 \pi$ periodicity of the $\theta$ term.

\section{On the reflection swap: the case of $(1+1)d$ class A with reflection symmetry}
\label{app:1d_a+reflection_r_swap}
In Sec.~\ref{sec:+cr_cr_swap}, we saw the important role of the partial fermion
parity flip
(Eq.~(\ref{eq:cr_swap_add_fermion_parity})) for simulating the Klein bottle partition function using the partial reflection swap. 
It is expected that the same prescription holds true in any fermionic systems. 
Here we explicitly check this for $(1+1)d$ class $A$ insulators in the presence of reflection symmetry $R$. 

Let us consider the ground state of the SSH model on a closed chain $[1,L]$,
\begin{align}
\ket{GS}
= \frac{1}{2^{L/2}}
(g_1^\dagger + f_2^\dagger) (g_2^\dagger + f^\dagger_3) \cdots (g^\dagger_L+f^\dagger_1) \ket{0}, 
\end{align}
where $f^{\dag}_j$ and $g^{\dag}_j$ are complex fermion creation operators. 
This wave function is invariant under the reflection 
\begin{align}
R f^\dagger_j R^{-1} = g^\dagger_{L-j+1}, \quad 
R g^\dagger_j R^{-1} = f^\dagger_{L-j+1} 
\end{align}
and the generating model of the $\Z_4$ phase classified by the pin$^c$ cobordism group $\Omega^{{\rm pin}^c}_2 = \Z_4$.~\cite{ShapourianShiozakiRyu2016detection,shiozaki2016many}
We introduce three adjacent intervals 
\begin{align}
I_1 = [1, \dots, M], \qquad 
I_2 = [M+1, \dots, M+N], \qquad 
I_3 = [M+N+1, M+N+M],
\end{align}
and define the reflection swap operator $R_{I_1 \cup I_3}$ by 
\begin{align}
R_{I_1 \cup I_3} f^{\dag}_{j} R_{I_1 \cup I_3}^{-1} = g^\dagger_{M+N+M-j}, \qquad 
R_{I_1 \cup I_3} g^{\dag}_{j} R_{I_1 \cup I_3}^{-1} = f^\dagger_{M+N+M-j}, \qquad 
(j \in I_1, I_3).
\end{align}
It is straightforward to see that
\begin{equation}\begin{split}
&\braket{GS | e^{i \theta_1 \sum_{j \in I_1} \hat n_j} e^{i \theta_2 \sum_{j \in I_2} \hat n_j} e^{i \theta_3 \sum_{j \in I_3} \hat n_j}  R_{I_1 \cup I_3} | GS} \\
&= \frac{(-1)^{M-1}}{8} e^{i M (\theta_1+\theta_3)} e^{i N \theta_2} (
\cos (\theta_1 + \theta_3-\theta_2) - \cos \theta_2).
\end{split}
\label{eq:ref_swap_1d_a_ref}
\end{equation}
Here, $\theta_2$ is the $U(1)$ holonomy along the $\Z_2$ cycle of the Klein bottle which is supposed to be quantized to $\Z_2$ values $\theta_2 \in \{0, \pi\}$, and
$\theta_1$ and $\theta_3$ are additional $U(1)$ phase twists. 
It is clear from  (\ref{eq:ref_swap_1d_a_ref}) that  the ground state expectation value vanishes when $\theta_1=\theta_3=0$, i.e., there is no phase twist on $I_1$ and $I_3$.
Moreover, the consistent background pin$^c$ structure is realized when the
amplitude of (\ref{eq:ref_swap_1d_a_ref}) takes the maximum value, which
determines the proper $U(1)$ phases.
In doing so, we obtain $(\theta_1+\theta_3, \theta_2) =(\pi,0)$ or $(\pi,\pi)$. 
Especially, the choice $(\theta_1,\theta_3) = (0,\pi)$ corresponds to the prescription (\ref{eq:cr_swap_add_fermion_parity}).


\bibliography{ref}

\end{document}